\documentclass[12pt]{article}
\usepackage{graphicx}
\usepackage{natbib}
\usepackage[colorlinks=true, linkcolor=black, citecolor=blue, urlcolor=black, pdfpagemode={UseNone}]{hyperref} 
\usepackage{tikz}
\usepackage{multirow}
\usepackage{dcolumn,booktabs}
\usepackage{algorithm,algorithmicx,algpseudocode,algcompatible}
\usepackage{mathtools} 
\usepackage{authblk} 
\usepackage{url}
\usepackage{xspace}
\usepackage{enumitem}
\usepackage[font=small]{caption}
\usepackage{amsmath,amssymb,bm,amsthm}
\newtheorem{theorem}{Theorem}
\newtheorem{proposition}{Proposition}
\newtheorem{lemma}{Lemma}

\algrenewcommand{\algorithmiccomment}[1]{\hfill$\blacktriangleright$ #1}

\DeclareMathOperator{\Cov}{Cov}

\DeclareMathOperator{\E}{E}

\DeclareMathOperator*{\argmin}{argmin}  
\DeclareMathOperator*{\argmax}{argmax}

\DeclareMathOperator{\diag}{diag}
\DeclareMathOperator{\rank}{rank}

\DeclareMathOperator{\IFu}{IF}
\DeclareMathOperator{\vect}{vec}
\newcommand{\ind}{I}
\newcommand{\MD}{\mbox{MD}}
\newcommand{\p}{\%\xspace}

\newcommand{\MCD}{\text{\tiny\upshape MCD}}
\newcommand{\rk}{k} 
 
\newcommand{\ran}{s}
\newcommand{\rt}{t}

\newcommand{\bZimp}{\bZ^{\mbox{\scriptsize{\upshape imp}}}}
\newcommand{\bhRimp}{\bhR^{\mbox{\scriptsize{\upshape imp}}}}
\newcommand{\sub}{k} 
\newcommand{\zort}{\boldsymbol{z}^{\perp}} 
\newcommand{\xort}{\boldsymbol{x}^{\perp}}
 
\newcommand{\cell}{\mbox{\scriptsize{\upshape cell}}} 
\newcommand{\case}{\mbox{\scriptsize{\upshape case}}} 

\newcommand{\CD}{{\boldsymbol{\cdot}}}
\newcommand{\eps}{\varepsilon}

\newcommand{\bzbar}{\overline{\boldsymbol z}}

\newcommand{\hsigma}{\widehat{\sigma}}

\newcommand{\hme}{\widehat{m}}
\newcommand{\bhme}{\boldsymbol{\widehat{m}}}

\newcommand{\hs}{\widehat{s}}
\newcommand{\bhs}{\boldsymbol{\widehat{s}}}

\newcommand{\hz}{\widehat{z}}
\newcommand{\ki}{\ell}
\newcommand{\bzero}{\boldsymbol 0}
\newcommand{\bone}{\boldsymbol 1}

\newcommand{\ba}{\boldsymbol a}

\newcommand{\bb}{\boldsymbol b}

\newcommand{\bd}{\boldsymbol d}
\newcommand{\bc}{\boldsymbol c}
\newcommand{\bE}{\boldsymbol E}
\newcommand{\bhE}{\boldsymbol{\widehat{E}}}
\newcommand{\be}{\boldsymbol e}

\newcommand{\bL}{\boldsymbol L}

\newcommand{\bp}{\boldsymbol p}

\newcommand{\hr}{\widehat{r}}

\newcommand{\hti}{\widehat t}
\newcommand{\bu}{\boldsymbol u}
\newcommand{\bhu}{\boldsymbol{\widehat{u}}}

\newcommand{\bv}{\boldsymbol v}

\newcommand{\bw}{\boldsymbol w}

\newcommand{\bx}{\boldsymbol x}

\newcommand{\bhx}{\boldsymbol{\widehat{x}}}

\newcommand{\bz}{\boldsymbol z}
\newcommand{\bhz}{\boldsymbol{\widehat{z}}}

\newcommand{\bA}{\boldsymbol A}
\newcommand{\bB}{\boldsymbol B}
\newcommand{\bC}{\boldsymbol C}
\newcommand{\bD}{\boldsymbol D}

\newcommand{\bI}{\boldsymbol I}
\newcommand{\bK}{\boldsymbol K}

\newcommand{\bM}{\boldsymbol M}

\newcommand{\bO}{\boldsymbol O}
\newcommand{\bP}{\boldsymbol P}

\newcommand{\bQ}{\boldsymbol Q}
\newcommand{\bR}{\boldsymbol R}

\newcommand{\bhR}{\boldsymbol{\widehat{R}}}

\newcommand{\bS}{\boldsymbol S}

\newcommand{\bT}{\boldsymbol T}

\newcommand{\bU}{\boldsymbol U}

\newcommand{\bhU}{\boldsymbol{\widehat{U}}}

\newcommand{\bV}{\boldsymbol V}
\newcommand{\btV}{\boldsymbol{\widetilde{V}}}
\newcommand{\bh}{\boldsymbol h}
\newcommand{\bhV}{\boldsymbol{\widehat{V}}}
\newcommand{\bW}{\boldsymbol W}
\newcommand{\btW}{\boldsymbol{\widetilde{W}}}
\newcommand{\bX}{\boldsymbol X}

\newcommand{\bZ}{\boldsymbol Z}
\newcommand{\bhZ}{\boldsymbol{\widehat{Z}}}

\newcommand{\bg}{\boldsymbol g}

\newcommand{\hdelta}{\widehat \delta}

\newcommand{\btheta}{\bm \theta}
\newcommand{\bTheta}{\boldsymbol \Theta}
\newcommand{\bPsi}{\boldsymbol \Psi}

\newcommand{\bPi}{\boldsymbol \Pi}

\newcommand{\bttheta}{\boldsymbol{\widetilde{\theta}}}

\newcommand{\bmu}{\boldsymbol \mu}

\newcommand{\bhmu}{\boldsymbol{\widehat{\mu}}}

\newcommand{\bsigma}{\boldsymbol \sigma}
\newcommand{\bSigma}{\boldsymbol \Sigma}
\newcommand{\bhSigma}{\boldsymbol{\widehat{\Sigma}}}
\newcommand{\btSigma}{\boldsymbol{\widetilde{\Sigma}}}

\newcommand{\bLambda}{\boldsymbol \Lambda}

\newcommand{\bhD}{\boldsymbol{\widehat{D}}}

\newcolumntype{M}[1]{>{\centering\arraybackslash}m{#1}}

\begin{document}

\def\spacingset#1{\renewcommand{\baselinestretch}
{#1}\small\normalsize} \spacingset{1}


\title{\bf Cellwise and Casewise Robust Covariance\\
           in High Dimensions}
\author[1]{Fabio Centofanti}
\author[1]{Mia Hubert}
\author[1]{Peter J. Rousseeuw}

\affil[1]{Section of Statistics and Data Science, Department 
          of Mathematics, KU Leuven, Belgium}

\setcounter{Maxaffil}{0}
\renewcommand\Affilfont{\itshape\small}
\date{May 31, 2026}       
\maketitle

\bigskip
\begin{abstract}
The sample covariance matrix is a cornerstone of 
multivariate statistics, but it is highly sensitive 
to outliers. These can be casewise outliers, such as 
cases belonging to a different population, or 
cellwise outliers, which are deviating cells 
(entries) of the data matrix. Recently some robust 
covariance estimators have been developed that can 
handle both types of outliers, but their computation 
is only feasible up to at most 20 dimensions.
To remedy this we propose the cellRCov method, a 
robust covariance estimator that simultaneously 
handles casewise outliers, 
cellwise outliers, and missing data. 
It relies on a decomposition of the covariance into 
the covariance of fitted points lying in a low-rank 
subspace and the covariance of their residuals, 
leveraging recent work on robust PCA. It also 
employs a ridge-type regularization
to stabilize the estimated covariance matrix.
We establish some theoretical properties of cellRCov, 
including its casewise and cellwise influence functions 
as well as consistency and asymptotic normality.  
A simulation study demonstrates the superior performance
of cellRCov in contaminated and missing data 
scenarios. Furthermore, its practical utility is 
illustrated in a real-world application to anomaly 
detection. We also construct and illustrate the 
cellRCCA method for robust and regularized canonical 
correlation analysis.
\end{abstract}


\newpage
\spacingset{1.5}

\section{Introduction} \label{sec:intro}
The covariance matrix is a key tool for understanding 
relationships between variables, making it a major 
component of many multivariate statistical methods. 
It plays a crucial role in linear and quadratic 
discriminant analysis, multivariate regression, 
cluster analysis, and countless other settings.
The population covariance matrix is typically 
estimated by the sample covariance matrix, which 
is easy to compute and unbiased and consistent at 
distributions with finite second moments. But 
in real-world applications data are often 
contaminated by outliers that can result from 
measurement errors, sensor malfunctions, or rare 
and unexpected events. In such situations the 
sample covariance becomes less reliable.

Multivariate data are typically stored in an 
$n\times p$ matrix, whose rows describe cases and 
columns are variables. Outliers are parts of the
data matrix that deviate substantially from the 
overall pattern. Depending on the context, outliers 
may be undesirable errors that can distort statistical 
analyses, or they can be valuable signals indicating 
interesting phenomena. Regardless of their nature, 
detecting outliers is crucial.

Research on outliers has traditionally focused on
outlying cases, also known as casewise outliers, 
that do not follow the same pattern as the majority 
of the cases. Many robust covariance estimators 
have been developed to address casewise
outliers in order to detect and/or downweight them, 
see e.g.\ \cite{portnoyhe2000} and Chapter 6 
of \cite{maronna2019robust}. For 
instance, M-estimators generalize maximum likelihood 
estimation by assigning lower weights to extreme 
cases. Another method is the Minimum Covariance 
Determinant (MCD) estimator 
\citep{rousseeuw1984least}, which seeks the subset 
with a fixed number of cases that is the most 
concentrated. Combinations of both ideas have
resulted in S-estimators and MM-estimators.
All casewise robust methods require that 
at least 50\p of the cases are clean.

In recent years, cellwise outliers have gained 
increasing attention \citep{alqallaf2009}. 
These are individual entries in the data matrix.
Such deviating cells are particularly common in 
high-dimensional data, when the large number of 
variables makes it more likely that some have
anomalous values. Even a relatively small 
proportion of outlying cells can contaminate many
cases. When random cells are contaminated with
probability $\eps$, the expected fraction of 
contaminated cases is $1 - (1-\eps)^p$. This
grows fast with $p$: even if only 1\p of the 
cells is contaminated in $p=100$ dimensions, 
we can expect 63\p of the cases to be 
contaminated. In such situations casewise 
robust methods become ineffective. 

The existence of outlying cells in the data bears
some resemblance to the missing data situation.
But the new problem is much more challenging, 
because
with missing cells you know exactly where they 
are, whereas outlying cells can be anywhere in
the data matrix. Moreover, they do not have to
stand out in their column. That is, they do not
have to be outlying in the marginal distribution
of their variable, they only need to deviate 
from their conditional distribution given all 
the other cells in the same case. This makes 
outlying cells difficult to detect. Some 
proposals were made by \cite{van2011stahel} and 
\cite{agostinelli2015robust}. The DDC method 
\citep{DDC2018} detects deviating cells by taking 
all the pairwise correlations between the 
variables into account.

In recent years several cellwise robust 
covariance estimators have been developed.
The first was the two-step generalized 
S-estimator (2SGS) of 
\cite{agostinelli2015robust}, based on an
S-estimator for incomplete data. Recently 
\cite{raymaekers2024cellMCD} constructed
a cellwise robust version of the MCD estimator
(cellMCD), that is computationally feasible 
up to about 20 dimensions.
Some cellwise robust methods have been developed 
for other settings, including 
discriminant analysis \citep{aerts2017cellwise},
time series \citep{centofanti2025multivariate},
principal component analysis 
\citep{hubert2019macropca, centofanti2026robust}, 
analysis of tensor data \citep{hirari2025robust}, 
and clustering \citep{zaccaria2024cellwise}. For 
a comprehensive overview of cellwise robust 
methods see \cite{raymaekersChallenges}.

Even apart from the issue of outliers, covariance
estimation has to contend with another challenge, 
the prevalence of datasets with high dimensions. 
Genomics, finance, image processing,
and many other fields generate datasets where 
the number of variables $p$ is comparable to or 
even exceeds the number of observations $n$. 
This exacerbates 
statistical errors, leading to unreliable results 
\citep{pourahmadi2013high}. In high dimensions,
a covariance matrix contains $p(p+1)/2$ entries
that have to be estimated. Even in the absence of 
any outliers, the sample covariance matrix performs 
poorly when $n/p$ is small. It exhibits eigenvalue 
distortions, where large eigenvalues tend to be 
overestimated and small ones underestimated, leading 
to a biased representation of the variance structure. 
These distortions worsen as $n/p$ decreases, and
when $n/p < 1$ the sample covariance matrix becomes 
singular. This severely limits its applicability in 
statistical procedures that require the inverse of
the covariance matrix.

To overcome these limitations of the sample 
covariance in high-dimensional settings, numerous 
regularized covariance estimators have been proposed 
in the literature. Among these,
ridge-type covariance estimators have gained 
widespread adoption due to their simplicity, ease of 
implementation, and ability to consistently improve 
covariance estimates in high-dimensional scenarios. 
They are weighted averages of the sample covariance 
matrix and some target matrix \citep{ledoit2004well,
schafer2005shrinkage,warton2008penalized}.
Other approaches include eigenvalue shrinkage, which 
adjusts eigenvalues to correct estimation errors while 
preserving eigenvectors \citep{ledoit2012nonlinear}, 
and structured covariance estimators, which impose 
sparsity or banded patterns 
\citep{bickel2008regularized,rothman2009generalized}. 
For comprehensive reviews of 
covariance estimation in high-dimensional settings
see e.g. \cite{engel2017overview}
and \cite{lam2020high}.

A few methods for robust high-dimensional covariance 
estimation have been developed. 
\cite{boudt2020minimum} introduced the Minimum 
Regularized Covariance Determinant (MRCD) estimator, 
which extends the MCD approach to high-dimensional 
settings by a ridge-type regularization.
It is casewise robust but not cellwise
robust. Cellwise robust covariance estimation 
methods for high dimensions were developed by 
\cite{croux2016robust}, who use rank correlations, 
and by \cite{tarr2016robust}, who employ other 
robust pairwise covariance estimates. But these 
methods tend to be less casewise robust.

Currently there is no covariance estimator 
that works well for both casewise and cellwise 
outliers and is computationally feasible in high 
dimensions. To address this limitation we
propose the \textit{cellwise Regularized Covariance} 
(cellRCov) method. It relies on a covariance 
decomposition into two components that are in some 
sense orthogonal. The first builds upon a recent 
method for robust dimension reduction, while the 
second is a weighted covariance matrix on a kind of 
residuals. Next it carries out a ridge-type 
regularization to enhance stability.

The main novelty of cellRCov is that it is 
the only covariance method capable of 
simultaneously handling cellwise outliers and 
casewise outliers in high-dimensional 
settings. It can also deal with missing values. 
This makes the methodology suitable for analyzing 
the kind of real data that is becoming increasingly 
common. Another novelty is in the theoretical
properties shown, including consistency, influence
functions, and asymptotic normality. This is the
first time a cellwise influence function of
covariance is obtained.

Section~\ref{sec:method} presents the cellRCov 
estimator, and Section~\ref{sec:theo} provides 
its theoretical properties. 
Section~\ref{sec:implementation} elaborates on 
its implementation. The empirical performance 
of cellRCov is assessed by Monte 
Carlo in Section~\ref{sec:simulation}.  
Section~\ref{sec:realdata} presents a real data
example where cellRCov is used for anomaly 
detection. It also constructs and illustrates
the cellRCCA method for robust and regularized 
canonical correlation analysis. 
Section~\ref{sec:conc} concludes.

\section{Methodology}
\label{sec:method}

\subsection{A covariance decomposition}
Consider a $p$-dimensional random vector $X$, 
with mean $\bmu=\left(\mu_1,\dots,\mu_p\right)^T$ 
and covariance matrix $\bSigma$. The 
spectral decomposition of $\bSigma$ is given by
\begin{equation*}
    \bSigma=\bE\bLambda\bE^T=
    \sum_{j=1}^p\lambda_j\be_j\be_{j}^T
\end{equation*}
where the columns $\be_1, \dots, \be_p$ of the 
matrix $\bE$ are the 
eigenvectors of $\bSigma$ corresponding to the 
eigenvalues $\lambda_1 \geqslant  \cdots 
\geqslant \lambda_p \geqslant 0$, and 
$\bLambda=\diag(\lambda_1,\dots,\lambda_p)$. We 
can express $X$ as the sum of two components: 
one lying in the $\rk$-dimensional subspace of 
$\mathbb{R}^p$ spanned by the first $\rk$ 
eigenvectors of $\bSigma$, and the other in its 
orthogonal complement. That is,
\begin{equation}
    X=X^\sub + X^{\perp}
\end{equation}
where $X^\sub=\bmu+\bE_k\bu$ with 
$\bu=\bE_k^T(X-\bmu)$ in which the columns of 
$\bE_k$ are the first $\rk$ eigenvectors 
$\be_1, \dots, \be_k$\,, and 
$X^{\perp} := X-X^\sub$\,. Then $\bSigma$ can 
be decomposed as
\begin{equation}
\label{eq:covdec}
  \bSigma=\bSigma_{X^\sub}+\bSigma_{X^{\perp}}
\end{equation}
where $\bSigma_{X^\sub}:=\Cov(X^\sub)$ and 
$\bSigma_{X^{\perp}}:=\Cov(X^{\perp})$, because 
$X^\sub$ and $X^{\perp}$ are uncorrelated, i.e.\ 
$\Cov(X^\sub,X^{\perp})=\bzero$.
This decomposition is analogous to the decomposition 
of the response covariance in multivariate regression 
as the sum of the covariance of predicted values and 
the covariance of residuals 
\citep{johnson2002applied}.  

From an estimation perspective, the 
decomposition~\eqref{eq:covdec} can be
used as a guiding principle for covariance 
estimation. The idea is to estimate the 
covariance as the sum of two contributions: 
the estimated covariance of fitted points
lying in a rank-$k$ subspace, that
summarizes the dominant low-dimensional 
structure retained by the fit, and the
estimated covariance of their 
residuals that accounts for the
remaining covariance structure.
Suppose we have $n$ realizations 
$\bx_1, \dots, \bx_n$ of $X$. In classical 
principal component analysis (CPCA), a
data point $\bx_i$ is fitted by the point
$\widehat{\bx^\sub_i}=\bhmu+\bhE_k\bhu_i$.
Here $\bhmu$ is the sample mean, $\bhE_k$ 
is the $p\times k$ matrix whose columns are 
the first $k$ eigenvectors of the sample 
covariance matrix $\bS$, and 
$\bhu_i=\bhE_k^T(\bx_i-\bhmu)$. 
The corresponding residual is
$\widehat{\boldsymbol{x}_i^{\perp}}=
\bx_i-\widehat{\bx^\sub_i}$\,.
Then we can compute $\bhSigma_{\bx^\sub}$ 
and $\bhSigma_{\xort}$ as the sample 
covariance matrices of the 
$\widehat{\bx^\sub_i}$  and 
$\widehat{\boldsymbol{x}_i^{\perp}}$\,.
Then the estimated $\bhSigma :=
\widehat{\bx^\sub_i} +
\widehat{\boldsymbol{x}_i^{\perp}}$ 
from~\eqref{eq:covdec} equals $\bS$. Of 
course, in that setting we gain 
nothing, but~\eqref{eq:covdec} becomes
useful when constructing a robust 
estimator of covariance.

If we were to use CPCA, both the fitted points 
and their residuals would be
affected by outliers. Indeed, the fitted point
$\widehat{\bx^\sub_i}$ depends on the scores
$\bhu_i=\bhE_k^T(\bx_i-\bhmu)$, while the residual
$\widehat{\boldsymbol{x}_i^{\perp}}=
\bx_i-\widehat{\bx^\sub_i}$ also depends
directly on the original observation $\bx_i$\,. 
Therefore, both casewise and cellwise outliers 
in $\bx_i$ would propagate to the fitted 
component and to the residual component. 

\subsection{The cellPCA method}
\label{sec:cellPCA}

We briefly describe the cellPCA method of 
\cite{centofanti2026robust}, that will be used 
as a building block for our
robust covariance matrix. CellPCA obtains a 
robust low-rank approximation of a dataset
in the presence of cellwise outliers, casewise
outliers, and missing values. We emphasize 
that cellPCA is not a new contribution 
of this paper.

Let $\bZ=\lbrace z_{ij}\rbrace$ be an 
$n \times p$ data matrix for which a
rank-$k$ approximation is sought. 
The goal of PCA is to represent $\bZ$ in 
a lower dimensional space, that is
\begin{equation}\label{eq:model}
   \bZ=\bone_n \bmu^T + \bU\bV^T + \bR,
\end{equation}
where $\bone_n$ is a column vector with all 
$n$ components equal to $1$,  the scores matrix
$\bU  =\lbrace u_{i\ell}
\rbrace=\left[\bu^1,\dots,\bu^\rk\right]=
\left[\bu_1,\dots,\bu_n\right]^T$ is 
$n \times \rk$, the loadings matrix
$\bV=\lbrace v_{j\ell}\rbrace=
\left[\bv^1,\dots,\bv^\rk\right]=
\left[\bv_1,\dots,\bv_p\right]^T$ is 
$p\times k$, and the matrix 
$\bR=\left[\zort_1,\dots,\zort_n\right]^T$ is 
the residual term. 
In CPCA, $\bmu$, $\bU$, and $\bV$ are estimated as
$\bhmu_{L^2}$\,, $\widehat{\bU}_{L^2}$\,, and 
$\widehat{\bV}_{L^2}$\, that together minimize
\begin{equation} \label{eq:appF}
 \sum_{i=1}^{n}\sum_{j=1}^{p}
 \left(z_{i j}- \mu_j-\sum_{\ell=1}^\rk 
 u_{i \ell}v_{j \ell} \right)^2
 =\sum_{i=1}^{n}\sum_{j=1}^{p} r_{i j}^2\;,
\end{equation}
where the residuals are denoted 
$r_{i j} := z_{i j}-\mu_j-
\sum_{\ell=1}^k u_{i \ell}v_{j \ell}$\;.
This quadratic loss function 
makes it a least squares fit, that is very 
sensitive to casewise and cellwise outliers.
Another challenge is the presence of missing 
data, as CPCA requires complete data.
  
Instead, cellPCA obtains estimators $\bhV$, $\bhU$, and 
$\bhmu$ of the $\bV$, $\bU$, and $\bmu$ 
in~\eqref{eq:model} by minimizing the loss 
function
\begin{equation} \label{eq:objP}
  L_{\rho_1,\rho_2}(\bZ,\bV,\bU,\bmu) := 
  \frac{\hsigma_2^2}{m}\sum_{i=1}^{n}m_i \rho_2\! 
  \left(\frac{1}{\hsigma_2}\sqrt{\frac{1}{m_i}
  \sum_{j=1}^{p} m_{ij}\, \hsigma_{1,j}^2\,
  \rho_1\!\left(\frac{r_{i j}}
  {\hsigma_{1,j}}\right)}\, \right),
\end{equation}
where $m_{ij}$ is 0 if $x_{ij}$ is missing 
and 1 otherwise, 
$m_i=\sum_{j=1}^{p} m_{ij}$\,,
and $m=\sum_{i=1}^{n} m_i$\,.
The scales $\hsigma_{1,j}$ standardize the 
{\it cellwise residuals} $r_{ij}$\,, and
the scale $\hsigma_2$ standardizes the 
{\it casewise total deviation} defined as
\begin{equation}\label{eq:rt_i}
  \rt_i := \sqrt{\frac{1}{m_i} \sum_{j=1}^{p}
  m_{ij}\, \hsigma_{1,j}^2\, \rho_1\!\left(
  \frac{r_{ij}}{\hsigma_{1,j}} \right)}\;.
\end{equation} 

For $\rho_1(z) = \rho_2(z) = z^2$ the 
objective~\eqref{eq:objP} becomes 
the objective~\eqref{eq:appF} of CPCA,
but here we use the bounded functions 
$\rho_1$ and $\rho_2$ of 
type~\eqref{eq:rhotanh}. This makes
$\bhmu$, $\bhU$, and $\bhV$ robust against 
both cellwise and casewise outliers. 
Indeed, a cellwise outlier in the cell $(i,j)$ 
yields a cellwise residual $r_{ij}$ with a 
large absolute value, but the boundedness of
$\rho_1$ reduces its effect on the estimates.
Similarly, a casewise outlier results in a large 
casewise total deviation $\rt_i$ but its effect
is reduced by $\rho_2$\,.
Note that in the computation of $\rt_i$ the effect 
of cellwise outliers is tempered by the presence 
of $\rho_1$\,. This avoids that a single cellwise 
outlier would always give its case a 
large $\rt_i$\,. The objective~\eqref{eq:objP} 
is minimized by an iteratively reweighted
least squares algorithm described
in \cite{centofanti2026robust}.

We assume the missingness mechanism is 
non-informative with respect to the 
unobserved entries. This includes missing 
completely at random (MCAR) and, more generally, 
missing at random (MAR) settings in which the
missingness pattern can be explained by the 
observed data. 

Note that the loading vectors of classical 
PCA are orthonormal, and the corresponding 
scores are uncorrelated. But here we will only 
use cellPCA to obtain a low-rank
approximation of the dataset. For our purposes
the key quantity is the fitted matrix
$\widehat{\bZ}=\bone_n\widehat{\bmu}^{T} +
\widehat{\bU}\widehat{\bV}^{T}$, rather than 
specific choices of the matrices 
$\widehat{\bU}$ and $\widehat{\bV}$. 
Consequently, imposing orthogonality constraints 
on $\widehat{\bV}$ or correlation constraints on
$\widehat{\bU}$ is not needed here.

We use the decomposition in~\eqref{eq:model}
as a computational device, rather than as a 
signal-plus-noise generative model. In particular, 
we do not assume that the residual component $\bR$ 
is isotropic or Gaussian. Therefore, the term 
$\bU\bV^T$ should not be interpreted as a uniquely 
identifiable signal component. Instead, for a fixed
rank $k$, cellPCA provides the robust low-rank fit
$\bhZ$. These fitted points are the quantities 
needed for the construction of cellRCov. Therefore
\eqref{eq:model} should be interpreted as an
approximation used to construct a covariance 
matrix, not as an identifiable signal/noise
decomposition.

\subsection{The cellRCov estimator}
\label{sec:cellRCov}
We now introduce the cellRCov estimator. Its construction
consists of the following steps.
\begin{enumerate}[nosep]
\item \textit{Robust marginal standardization.}
    Standardize the original data matrix columnwise,
    using robust univariate M-scales.
\item \textit{Robust low-rank approximation.
    Apply the cellPCA method to the standardized data.}
\item \textit{Covariance estimation in the 
    fitted subspace.} Estimate the first term of the 
    decomposition~\eqref{eq:covdec} from
    the robust cellPCA scores and loadings. 
\item \textit{Residual covariance estimation.
    Estimate the second term of the 
    decomposition~\eqref{eq:covdec} 
    from weighted residuals, with cellwise 
    weights downweighting individual contaminated cells, 
    and casewise weights downweighting anomalous cases.}
\item \textit{Regularization.} Regularize the
    residual covariance by ridge-type shrinkage
    to address possible ill-conditioning in 
    high-dimensional settings.
\item \textit{Final covariance estimate.}
    Add the regularized residual covariance matrix
    to the fitted subspace covariance matrix. 
    Then undo the initial marginal standardization.
\end{enumerate}
We now describe these steps in detail.

We store the data $\bx_1,\dots,\bx_n$ in the 
$n \times p$ data matrix $\bX$. It is first
standardized to $\bZ=\lbrace z_{ij}\rbrace=
(\bz_1,\dots,\bz_n)^T=\bX\bhD^{-1}$, where 
$\bhD=\diag(\hsigma^X_1,\dots,\hsigma^X_p)$. 
Here $\hsigma^X_j$ is the M-scale of the 
univariate set 
$(x_{1j}-m_j,\dots,x_{nj}-m_j)$, where 
$m_j$ is the median of the $j$-th variable.
An M-scale of a univariate sample 
$\left(t_1,\dots,t_n\right)$ is the 
solution $\hsigma$ of an equation
\begin{equation}\label{eq:Mscale}
  \frac{1}{n} \sum_{i=1}^n \rho\left(
  \frac{t_i}{a\sigma}\right)=\delta\,.
\end{equation}
For maximal robustness we take 
$\delta = \max(\rho)/2$, and the 
consistency factor $a$ is chosen such 
that $E[\rho(t/a)] = \delta$ for 
$t \sim N(0,1)$, the standard Gaussian 
distribution.
The function $\rho$ is the hyperbolic 
tangent (\textit{tanh}) function $\rho_{b,c}$ 
introduced by~\cite{tanh1981}. It is
defined piecewise by
\begin{equation}\label{eq:rhotanh}
\rho_{b,c}(t) = 
\begin{cases}
 t^2/2 &\mbox{ if } 0 \leqslant |t| \leqslant b,\\
  d - (q_1/q_2) \ln(\cosh(q_2(c - |t|)))    
    &\mbox{ if } b \leqslant |t| \leqslant c,\\
  d &\mbox{ if } c \leqslant |t|,\\
\end{cases}
\end{equation}
where $d = (b^2/2) + (q1/q2)\ln(\cosh(q_2(c - b)))$, 
with the default values $b=1.5$ and $c=4$ with 
$q_1=1.540793$ and $q_2=0.8622731$. With this 
choice, we have $\delta = 1.8811$ and $a = 0.3431$\,.

Then, cellRCov decomposes the 
$\bz_i$ as $\bz^\sub_i + \zort_i$ as 
in~\eqref{eq:covdec}. To obtain this 
decomposition robustly, we apply cellPCA
to $\bZ$, as described in 
Section~\ref{sec:cellPCA}. This yields 
estimates of the quantities appearing in 
the low-rank representation~\eqref{eq:model},
namely the robust center $\bhmu$, the 
score matrix $\bhU$, and the loading
matrix $\bhV$. 

CellRCov next estimates the covariance 
of the fitted points 
$\bz^\sub_i = \bhz_i =\bhmu + \bhV\bhu_i$ as
\begin{equation}\label{eq:SigmaMCD}
  \btSigma_{\bz^\sub}=\bhV
\bhSigma_{\text{\tiny MCD}}(\bhU)\bhV^T\;,
\end{equation} 
where 
$\bhSigma_{\MCD}(\cdot)$ is computed by the 
MCD estimator \citep{rousseeuw1984least} 
with parameter $0.5 \leqslant\alpha<1$.
For this estimation we use 
the fast and robust deterministic algorithm
DetMCD of \cite{hubert2012deterministic}.
Unlike CPCA, the columns of $\bhU$ 
may be correlated. Therefore, the covariance
matrix of the scores is not guaranteed to be
diagonal, so it cannot be estimated 
coordinatewise. This needs to be done 
robustly, because although cellPCA reduces 
the effect of points with large residuals, 
a point $\bz_i$ with small residuals can 
still have a $\bhz_i$ that is far from the 
bulk of the fitted points in the principal
subspace. Such points have outlying scores,
that would strongly affect a non-robust 
covariance estimate of $\bhU$. For this 
reason we use the MCD estimator, to avoid 
that outlying fitted points $\bz_i^\sub$ 
have a large effect on the estimated 
covariance. Appendix~A.1 of 
\cite{hubert2019macropca} illustrates 
why this step is needed.

The use of the MCD, rather than a cellwise 
robust covariance estimator such as
cellMCD \citep{raymaekers2024cellMCD}, is 
motivated by the fact that the notion of 
cellwise outlyingness is tied to the 
original coordinate system of the variables. 
The scores $\bhu_i$ are coordinates in the 
low-dimensional subspace, and a single 
contaminated cell in the original data
does not correspond to a single contaminated
cell in the scores space. As a result, 
the usual notion of cellwise outlyingness is 
no longer meaningful in the scores 
space. Moreover, the influence of cellwise
outliers in the original variable space was 
already reduced when computing the scores 
matrix $\bhU$ and the fitted points 
$\bz_i^\sub$\,. Therefore, once the data 
are represented in the scores space, the 
remaining concern is to protect the scatter 
estimate against casewise outlying scores.

A large $|\hr_{ij}|$ in the residual matrix 
$\bhR=\lbrace \hr_{ij}\rbrace =
\lbrace z_{ij}-\hz_{ij}\rbrace=\bZ-\bhZ$ 
leads us to suspect its $z_{ij}$\,. 
Note that $\bhZ$ is estimated by downweighting 
casewise and cellwise outliers. This avoids 
that the fitted point $\hz_{ij}$ adapts 
to a contaminated cell $z_{ij}$\,. 
If $|\hr_{ij}|$ 
is small, then $z_{ij}$ is well predicted by 
the robust low-rank structure estimated from
the whole data matrix, and is not considered 
outlying relative to the fitted subspace.

To obtain a version of $\bZ$ that is free of
cellwise outliers, we define the imputed data 
matrix as
\begin{equation} \label{eq:imputxi}
   \bZimp := \bhZ + \btW\odot\bhR,
\end{equation}
where the Hadamard product $\odot$ multiplies 
matrices entry by entry. In this expression
$\btW=\bW^{\cell} \odot \bM$, where the 
$n \times p$ matrix 
$\bW^{\cell}=\lbrace w_{ij}^{\cell}\rbrace$ 
contains the {\it cellwise weights}
\begin{equation}\label{eq:cellweight}
   w_{ij}^{\cell}=\psi_1\!\left(
   \frac{\hr_{ij}}{\hsigma_{1,j}}\right)
   \Big/ \frac{\hr_{ij}}{\hsigma_{1,j}}\; , 
   \quad i=1,\dots,n,\quad j=1\dots,p,
\end{equation}
where $\psi_1=\rho_1'$ with the 
convention $w_{ij}^{\cell}(0) = 1$. 
The $n \times p$ matrix $\bM$ contains the 
missingness indicators $m_{ij}$\,.
We then use the imputed residuals 
$\bhRimp=\bZimp-\bhZ=\btW\odot\bhR$ 
to estimate the second term of the
decomposition~\eqref{eq:covdec} as
\begin{equation}\label{eq:Sigmay}
   \btSigma_{\zort} = 
   \frac{1}{b}\sum_{i=1}^n
   w_i^{\case}\btW_i(\bz_i - \bhz_i)
   (\bz_i - \bhz_i)^T\btW_i\;,
\end{equation}
\sloppy where $\btW_i$  is a diagonal matrix 
whose diagonal is the $i$-th row of $\btW$, 
and $b$ is given by
$b=\sum_{i=1}^n\sum_{j=1}^p
\sum_{\ell=1}^p m_{ij}m_{i\ell}w_i^{\case}
w_{ij}^{\cell} w_{i\ell}^{\cell}/p^2$.
Here the {\it casewise weights} $w_i^{\case}$ 
are defined as
\begin{equation}\label{eq:caseweight}
   w_{i}^{\case}=\psi_2\!\left(
   \frac{\hti_i}{\hsigma_2}\right)
   \Big/ \frac{\hti_i}{\hsigma_2}\; , 
   \quad i=1,\dots,n\,,
\end{equation}
with $\psi_2=\rho_2'$\,, and $\hti_i$ is 
obtained from \eqref{eq:rt_i} with 
$\hr_{ij}$ in place of $r_{ij}$.
The constant $b$ is a normalizing 
factor that accounts for the effective 
number of observed and downweighted cell 
pairs contributing to the residual 
covariance estimate.
We stress that the residual component is 
not interpreted as isotropic noise. It 
represents the variation not captured by the 
rank-$k$ fit $\bhZ$, and its covariance
structure need not be diagonal. Since 
cellRCov aims to estimate the full covariance 
matrix, this covariance structure is not 
discarded.

In high-dimensional settings, where the number 
of variables $p$ is comparable to or exceeds the 
number of observations $n$, estimating a covariance 
matrix becomes challenging due to the huge number 
of parameters $(p^2 + p)/2$ it contains. 
In such situations the matrix $\btSigma_{\zort}$
of~\eqref{eq:Sigmay} can degrade, as it inherits 
some of the limitations of the sample covariance. 
Indeed, $\btSigma_{\zort}$ is a kind of weighted 
sample covariance matrix. 
Therefore we stabilize $\btSigma_{\zort}$
by a ridge-type regularization, yielding 
\begin{equation} \label{eq:ridge}
   \btSigma_{\zort}^{R}=(1-\delta) 
   \btSigma_{\zort}+\delta\bT,
\end{equation}
where $\bT=\diag((\btSigma_{\zort})_{11},\,\dots,
(\btSigma_{\zort})_{pp})$ is the diagonal target 
matrix 
with the diagonal entries of $\btSigma_{\zort}$\,, 
and $\left.\delta\in(0,1\right]$ is the tuning 
parameter that controls the amount of shrinkage.
The superscript $R$ in $\btSigma_{\zort}^{R}$
stands for Regularized. As the target 
matrix $\bT$ is positive definite, the 
resulting $\btSigma_{\zort}^{R}$ will be as well 
for $\delta>0$. Several alternative target 
matrices have been proposed in the literature, 
see e.g. \cite{engel2017overview}, but we prefer 
$\bT$ as suggested by \cite{schafer2005shrinkage} 
because it leaves the diagonal entries of 
$\btSigma_{\zort}$ intact, so it does not shrink 
the estimated variances.

The cellRCov estimate $\bhSigma_{\bz}$ of the 
covariance of $\bz$ is defined as
$\bhSigma_{\bz} := \bhSigma_{\bz^\sub} +
\bhSigma_{\zort}^R$\;. Undoing the original
standardization by the diagonal matrix $\bhD$,
this yields
\begin{equation} \label{eq:cellRCov}
  \mbox{cellRCov}(\bX) = \bhSigma =
  \bhD\big(\btSigma_{\bz^\sub}+
  \btSigma_{\zort}^R\big)\bhD\,.
\end{equation}
The standardization step in the
beginning ensures that cellRCov 
is scale equivariant, meaning that for 
any diagonal matrix 
$\bD = \operatorname{diag}(d_1, \dots, d_p)$ 
with $d_j > 0$, the estimator satisfies 
$\bhSigma(\bX\bD) = \bD \bhSigma(\bX) \bD$.
Therefore $\bhSigma$ reacts in the usual
way to changes of variable units.
Moreover, since $\bhSigma$ is positive 
definite, it can be used to compute robust 
Mahalanobis-type distances in 
high-dimensional data.

A detailed complexity analysis is provided in 
Section~\ref{sec:supp-complexity-cellRCov} 
of the Supplementary Material. There we show that, 
when the rank $\rk$ is bounded by a fixed maximum
$\rk_{\max}$\,, the computational complexity of 
cellRCov is only $O(np^2+np\log(n))$. This is not
much more than the $O(np^2)$ of the classical
covariance matrix.

\section{Large-sample properties}
\label{sec:theo}

The influence function (IF) is a key robustness 
tool. It reveals how an estimating functional, 
that is, a mapping from a space of probability 
measures to a parameter space, changes due to a 
small fraction of contamination.
This section presents the influence functions and 
asymptotic normality of cellRCov. 
The proofs are provided in Section~\ref{app:proofs} 
of the Supplementary Material.
For notational simplicity, we assume that the data 
$\bX$ are already standardized, so $\bhD=\bI_p$.  

Consider a $p$-variate random variable $X$
with distribution $H_0$. We then contaminate it,
yielding the variable
\begin{equation} \label{eq:cont_cell_z}
X_{\eps}=A \odot X + (\bone_p-A) \odot \bz 
\end{equation}
where the column vector $A$ is a random variable,
$\bone_p$ is a column vector of ones, and
$\bz=\left(z_1, \ldots, z_p\right)^T$  is a fixed 
$p$-variate vector. The variable $A$ is of the
form $A = A^{\case} \odot A^{\cell}$ and its 
distribution is denoted as $G_\eps$\,.
The \mbox{$p$-variate} variable $A^{\case}$ has 
Bernoulli distributed marginals $A_j^{\case}$ 
for $j=1,\ldots,p$ with success parameter 
$1-\eps^{\case}$, and jointly they are fully 
dependent in the sense that 
$P(A_1^{\case} = \ldots = A_p^{\case})=1$. 
The $p$-variate variable $A^{\cell}$ has 
Bernoulli components $A_j^{\cell}$ with 
success probability $1-\eps^{\cell}_j$.

The usual \textit{casewise influence function}
of \cite{hampel1986} 
uses the \textit{fully dependent contamination
model} (FDCM), which has 
$A := A^{\case}$ with independent $X$ and 
$A^{\case}$ in \eqref{eq:cont_cell_z}. In that 
situation the distribution of $X_{\eps}$ 
simplifies to $(1-\eps^{\case})H_0 + 
\eps^{\case}\Delta_{\bz}$, where  
$\Delta_{\bz}$ is the distribution that puts 
all of its mass in the point $\bz$. We then 
denote $G_\eps$ as $G_\eps^D$, which depends 
on $\eps=\eps^{\case}$, and the distribution 
of $X_{\eps}$ is denoted as $H(G_\eps^D,\bz)$. 
For a functional $\bT$ with values in 
$\mathbb R^p$, the casewise influence 
function is then defined as
\begin{equation}
\label{eq:caseIF}
\IFu_{\case}(\bz, \bT, H_0)=\left.
 \frac{\partial}{\partial \eps} 
 \bT(H(G_\eps^D,\bz))\right|_{\eps=0}=\,
 \lim _{\eps \downarrow 0}
 \frac{\bT\left(H(G_\eps^D,\bz)\right) -
 \bT(H_0)}{ \eps}\;.
\end{equation} 

\cite{alqallaf2009} proposed a cellwise 
version of the IF as well. It considers the 
contaminated variable~\eqref{eq:cont_cell_z} 
under the \textit{fully independent 
contamination model} (FICM), which has  
$A := A^{\cell}$ 
such that $X, A_1^{\cell},\ldots, A_p^{\cell}$
are independent, and 
$P(A_j^{\cell} = 1) = 1-\eps^{\cell}$ for all 
$j=1,\ldots,p$. In this situation we denote 
$G_\eps$ as $G_\eps^I$ which depends on 
$\eps=\eps^{\cell}$, and the distribution of 
$X_{\eps}$ is denoted as $H(G_\eps^I,\bz)$.
The {\it cellwise influence function} 
$\IFu_{\cell}(\bz, \bT, H_0)$ is then 
defined as
\begin{equation}
\label{eq:gIF}
\IFu_{\cell}(\bz, \bT, H_0)=\left.
 \frac{\partial}{\partial \eps}
 \bT(H(G_\eps^I,\bz))\right|_{\eps=0}=\,
 \lim _{\eps \downarrow 0}
 \frac{\bT\left(H(G_\eps^I,\bz)\right)-
 \bT(H_0)}{ \eps}.
\end{equation} 

We denote the functional corresponding to the 
cellRCov estimator $\bhSigma$ in 
\eqref{eq:cellRCov} by $\bSigma(H)$.
We derive the IF of the $p^2 \times 1$
column vector $\vect(\bSigma(H))$, 
where $\vect(\CD)$ converts a matrix to a vector 
by stacking its columns on top of each other.
The derivation of the IF of $\vect(\bSigma(H))$
is long, and involves computing the IF of the
various components of its construction, such
as $\bV(H)$ from~\eqref{eq:objP}, 
$\bSigma^{\bu}_{\MCD}(H)=\bSigma_{\MCD}(\bhU)$ 
from~\eqref{eq:SigmaMCD}, and 
$\bSigma_{\xort}$ from~\eqref{eq:Sigmay}.
Therefore these computations are relegated 
to Section A of the Supplementary Material.

\begin{theorem}\label{IFSigma}
The casewise and cellwise influence functions 
of $\vect(\bSigma)$ are
\begin{align} \label{eq:IFFD_cellRCov}
  \IFu_{\case}\left(\bz,\vect(\bSigma),H_0\right)
  &= \bR_{1}\IFu_{\case}(\bz,\vect(\bV),H_0)
    +\bR_{2}\IFu_{\case}(\bz,\vect(
    \bSigma^{\bu}_{\MCD}),H_0) \nonumber\\
  &\;\;\;\;\;+ 
    \IFu_{\case}(\bz,\vect(\bSigma_{\xort}),H_0)
\end{align}
\begin{align} \label{eq:IFFI_cellRCov}
  \IFu_{\cell}\left(\bz,\vect(\bSigma),H_0\right)
  &= \bR_{1}\IFu_{\cell}(\bz,\vect(\bV),H_0)
    +\bR_{2}\IFu_{\cell}(\bz,\vect(
    \bSigma^{\bu}_{\MCD}),H_0) \nonumber \\
  &\;\;\;\;\;+
    \IFu_{\cell}(\bz,\vect(\bSigma_{\xort}),H_0)\;.
\end{align}
All of these terms are computed in Propositions
\ref{prop1} to \ref{prop3} in 
Section~\ref{app:proofs} of the 
Supplementary Material. In particular
\begin{align*}
  \bR_{1}&=\big(\bV(H_0)\bSigma^{\bu}_{\MCD}(H_0)
  \otimes\bI_p\big)+ \big(\bI_p\otimes\bV(H_0)
  \bSigma^{\bu}_{\MCD}(H_0)\big)\bK_{p,k}\\
  \bR_{2}&=\bV(H_0)\otimes\bV(H_0)
\end{align*}   
in which $\otimes$ is the Kronecker product and
$\bK_{p,k}$ is a $pk \times pk$ permutation matrix. 
\end{theorem}

Let us look at a special case to get a feel for 
these results. For our model distribution we 
choose the bivariate normal
$H_0 = N\Big(\bzero, \begin{bsmallmatrix} 1 & 0.9\\ 
& \\ 0.9 & 1 \end{bsmallmatrix}\Big)$ and we set 
$\rk=1$. The left panel of Figure~\ref{fig_IFs} 
shows the bounded casewise IF of the entry 
$\hsigma_{11}$ of the estimated covariance matrix 
$\bSigma$. Its shape looks complicated at first.
A point $\bz$ with small $|z_2|$ and large
$|z_1|$ gets imputed by moving its first cell,
so its IF is the same as for some smaller $z_1$.
This explains why the IF has a constant shape
in $z_2$ no matter how far we move $z_1$\,. 
That shape is quadratic for centrally located
$z_2$\,, gets lower outside, and stays constant 
further away. The situation is similar for small 
$|z_1|$ and large $|z_2|$.
Moreover, a far-out point $\bz$ close to the 
subspace given by $z_1 = z_2$ will
have both its cells imputed, explaining the 
constant shape of the IF along the first
diagonal.
The cellwise IF in the right panel is bounded
as well, and in its cross-sections parallel to
the axes we again see the shape is quadratic 
in the middle, lower outside, and constant 
further out.

\begin{figure}[ht]
\centering
\vspace{3mm}
\includegraphics[width=0.45\textwidth]
    {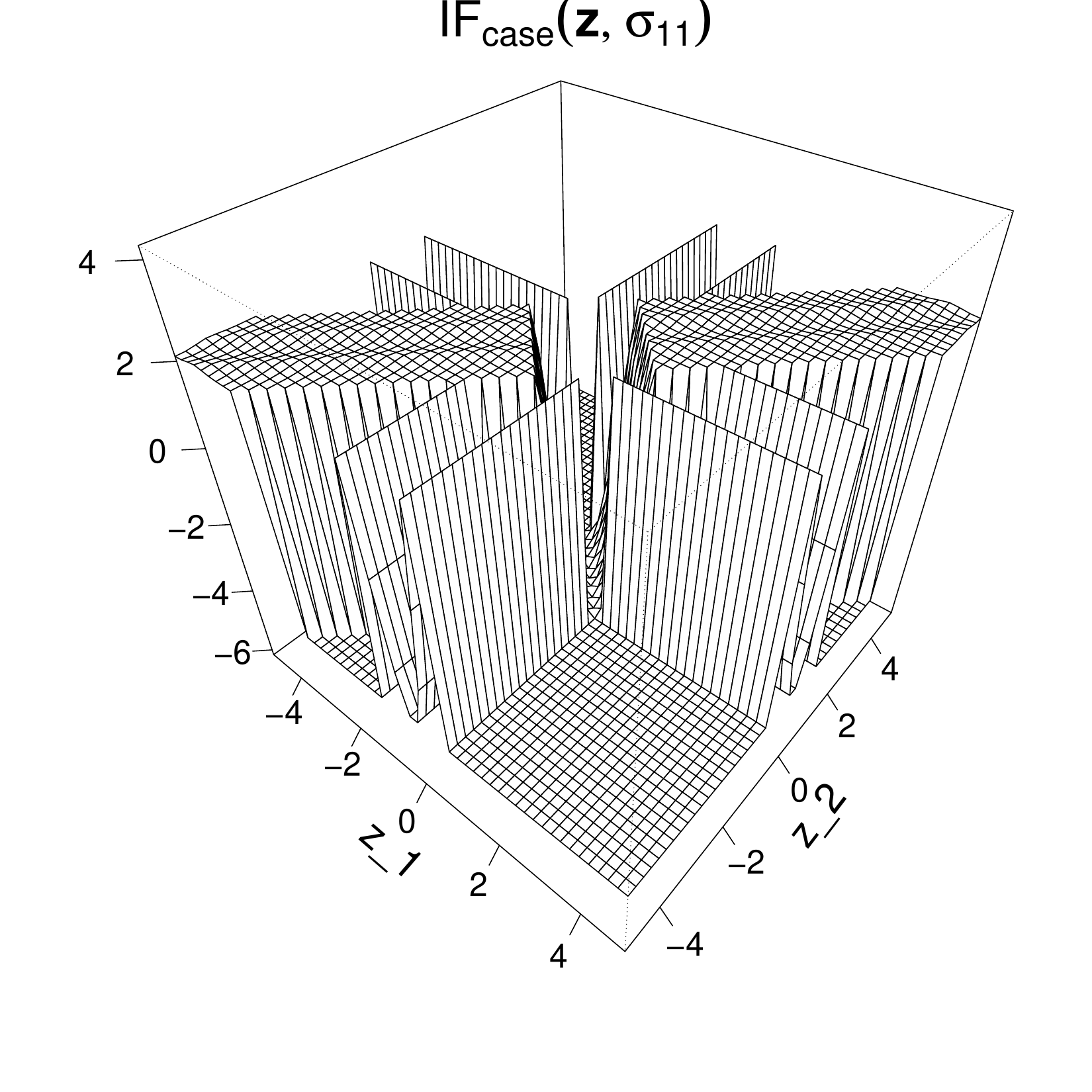}
\includegraphics[width=0.45\textwidth]
    {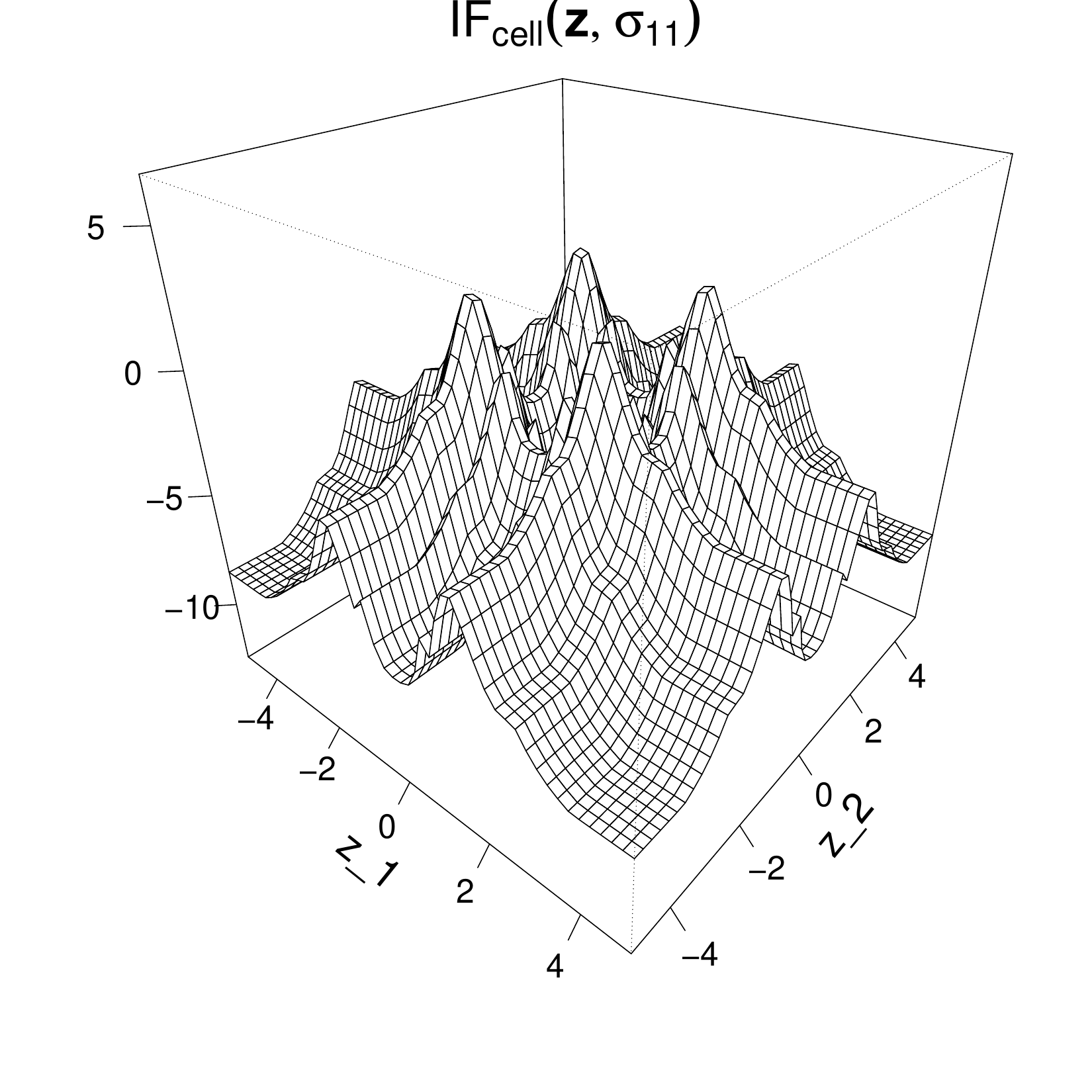}\\
\vspace{-4mm}
\caption{The casewise (left) and cellwise (right)
IF of $\sigma_{11}$ at the bivariate normal $H_0$.}
\label{fig_IFs}
\end{figure}

The following theorem states the consistency 
of the cellRCov estimator. Suppose we obtain 
i.i.d.\ observations $\bx_1,\dots,\bx_n$ 
from $H_0$. We denote the corresponding 
empirical distribution function as $H_n$ and 
put $\bSigma_{n}:=\bSigma(H_n)$.

\begin{theorem}\label{consistSigma}
Under the assumptions of Propositions 
\ref{prop5} to \ref{prop8}
in Section~\ref{app:consistency}
of the Supplementary
Material, $\vect(\bSigma_n)$ is 
consistent for $\vect(\bSigma(H_0))\,$.
\end{theorem}

It turns out that cellRCov is asymptotically normal,
and its asymptotic covariance is related to the
casewise influence function:

\begin{theorem}\label{asynormSigma}
Under suitable regularity conditions on $\bSigma$, 
we have that
 \begin{equation}\label{eq:asym}
   \sqrt{n}\left[\vect(\bSigma_{n})
   -\vect(\bSigma(H_0))\right]
   \; \rightarrow_{\mathcal{D}} \;
   N_{p^2}(\bzero,\bV(\bSigma,H_0))
 \end{equation}
with $\bV(\bSigma,H_0)=\int\IFu_{\case}
\left(\bx,\bSigma,H_0\right)\left(\IFu_{\case}
\left(\bx,\bSigma,H_0\right)\right)^TdH_0(\bx)$.
\end{theorem}
The proof is in Section~\ref{app:consistency} 
of the Supplementary Material.

\section{Matters of implementation}
\label{sec:implementation}

The cellRCov method contains two parameters,
the rank $\rk$ of the fitted subspace in 
the PCA model~\eqref{eq:model} from the
first step, and the factor $\delta$ in the
ridge-type regularization~\eqref{eq:cellRCov}.
They often need to be selected in a
data-dependent way, as described below.

\subsection{Selecting the rank 
\texorpdfstring{$\rk$}{k}}
\label{sec:rank}

\cite{horn1965rationale} developed Parallel 
Analysis (PA) to automatically select the 
number of components of PCA applied to
standardized data.
Its rationale is that sampling variability 
will produce eigenvalues above 1 even if all 
eigenvalues of the population correlation 
matrix are exactly 1 and no large components 
exist, as with independent variates 
\citep{buja1992remarks}. PA first standardizes
the eigenvalues by dividing them by their 
average, and then compares them with those 
from standard normal data with the same 
$n$ and $p$. Those that are larger than the
eigenvalues from the random data are kept.

We use a robust adaptation of PA
to select the dimension $\rk$ of the cellPCA
low-rank approximation. In classical PCA,
each eigenvalue $\ell_s$ is the decrease in 
reconstruction error achieved by moving from
a model with $s-1$ components to a model 
with $s$ components, that is, 
$\ell_s=\nu_{s-1}-\nu_s$ where 
$\nu_s$ is the Frobenius norm 
$\nu_s := ||\bZ-\bhZ_s||_F^2$ and 
$\bhZ_s$ is the best rank-$s$ 
approximation of $\bZ$. In particular, 
$\nu_0 = ||\bZ-\bone_n\bzbar^T||_F^2$ 
where $\bzbar$ is the sample mean.
The robust PA considers the value
$\nu_s^{\mathrm{rob}}$ of the rank-$s$ 
cellPCA objective~\eqref{eq:objP}. The value 
$\nu_0^{\mathrm{rob}}$ is obtained by 
computing~\eqref{eq:objP} on the cellwise 
residuals $z_{ij}-m_j^z$\,, where $m_j^z$ is 
the median of the $j$-th variable of $\bZ$. 
We then define the robust analog of the 
$s$-th eigenvalue as the decrease in the
robust objective, 
$\ell_s^{\mathrm{rob}} :=
\nu_{s-1}^{\mathrm{rob}}-
\nu_s^{\mathrm{rob}}$.
So a component is retained when its 
inclusion decreases the cellPCA 
objective more than would be expected under 
a reference distribution without a dominant 
low-dimensional structure. The reference 
distribution of the eigenvalues is obtained 
by applying CPCA to simulated standard 
normal random datasets, and computing 
$\nu_\ran$ as the 
objective~\eqref{eq:objP} of $\bhZ_\ran$\,. 
Then we select $\rk$ as the number of 
eigenvalues that are larger than the $99$th 
percentile of this distribution. The 
robust PA algorithm is 
provided in Section~\ref{sec:supp-rank}
of the Supplementary Material.

The selected rank should be 
interpreted in light of the low-rank plus 
residual decomposition used by cellRCov. 
The robust PA rule retains the components 
whose contribution to the low-rank fit is 
higher than what would be expected under 
an uninformative reference distribution. 
The retained subspace thus captures the 
dominant low-dimensional variation in 
the fitted component $\bhZ$.
Components not retained by this rule are 
not interpreted as pure noise. They simply 
do not show a sufficiently strong 
low-dimensional structure (relative to
the reference distribution) to be included 
in the low-rank fitted component. Since 
these directions may still contain a weaker
covariance structure, cellRCov accounts for
their contribution through the regularized 
residual covariance estimator.

\subsection{Selecting 
\texorpdfstring{$\delta$}{delta}}
\label{sec:delta}

To select the regularization parameter $\delta$ we 
follow \cite{bickel2008regularized} by using their 
$H$-fold cross-validation procedure with $H=5$. 
We randomly split the original sample $H$ times into 
two subsets of sizes $n_1 = n/3$ and $n_2 = n-n_1$. 
For the $h$-th split, $\btSigma_{\zort}^{R,h}(\delta)$ 
denotes the cellRCov estimate of $\bSigma_{\zort}$ 
computed from the first subset, and $\btSigma_{\zort}^h$ 
represents the unregularized estimate of 
$\bSigma_{\zort}$ computed from the second subset. 
Then we choose $\delta$ as 
\begin{equation}\label{eq:opcv}
  \hdelta = \argmin_{\left.\delta\in (0,1\right]}
  \frac{1}{H}\sum_{h=1}^H \Big|\Big|
  \btSigma_{\zort}^{R,h}(\delta)
  -\btSigma_{\zort}^h\Big|\Big|_F\;.
\end{equation}
In simulations we indeed saw that the selected
$\delta$ tended to grow with the dimension.

The goal of this procedure is to select 
the amount of shrinkage that stabilizes the 
residual covariance estimate, across datasets. 
For a fixed $\delta$, a regularized residual
covariance estimate is computed on the training
portion and compared with the unregularized 
residual covariance computed on the validation 
portion. If the unregularized residual 
covariance is highly variable, small values of 
$\delta$ lead to unstable estimates that do 
not agree well. Larger values
of $\delta$ reduce this variability by 
shrinking the estimate toward the diagonal 
target. On the other hand, if $\delta$ is 
too large, the estimator may be overshrunk and 
fail to capture relevant residual structure.
The selected $\delta$ balances these two 
effects.

The cross-validation for selecting $\delta$ 
does not require much computation, because
the low-rank approximation, the residuals, and
the cellwise and the casewise weights are 
computed only once and kept fixed during the 
cross-validation procedure.
For each split and each candidate value of 
$\delta$, we only need to compute the weighted 
residual covariance matrix and carry out the 
ridge-type shrinkage. This only costs
$O(np^2)$ time, as seen in 
Section~\ref{sec:supp-complexity-cellRCov} 
of the Supplementary Material.

\section{Simulation study}
\label{sec:simulation}
We study the performance of the cellRCov 
method~\eqref{eq:cellRCov} by 
Monte Carlo, where the clean data are generated 
from a multivariate Gaussian with $\bmu=\bzero$ and 
covariance matrix\linebreak 
$\bSigma=\lbrace \sigma_{j\ell}\rbrace$. 
We use the A09 covariance matrix with 
$\sigma_{j\ell}=(-0.9)^{|j-\ell\,|}$. We generate 
$n=100$ data points in dimensions $p = 30,60,120$ 
so both $n>p$ and $n<p$ occur. 

Three contamination types are considered. In the 
cellwise outlier scenario we randomly replace 
$20\p$ of the cells $x_{ij}$ by $\gamma$. 
The parameter $\gamma$ varies from 0 to 10, and
when $\gamma$ is 0 we do not contaminate the data.
In the casewise outlier setting, $20\p$ of the 
cases are generated from 
$N(\gamma\sqrt{p}\,\be/\sqrt{\be^T\bSigma^{-1}\be}
\,, \bSigma)$ 
where $\be$ is the 
eigenvector of $\bSigma$ with smallest eigenvalue.
In the third scenario, the data is contaminated 
by $10\p$ of cellwise outliers as well as $10\p$ 
of casewise outliers. We measure performance by 
the Kullback-Leibler discrepancy between the 
estimated $\bhSigma$ and the true $\bSigma$,
given by
\begin{equation*}
  \mbox{KL}(\bhSigma,\bSigma)=
  \mbox{trace}\left(\bhSigma\bSigma^{-1} 
  - I_p\right) -\log \left(
   \det\left(\bhSigma\bSigma^{-1}\right)\right).
\end{equation*}
For each setting of the simulation parameters 
we generate 200 datasets, and report the average 
Kullback-Leibler discrepancy over these 200 
replications. The figures in this section show 
the average KL for the covariance model A09.
The plots for three other covariance models are 
very similar, see the figures in 
Section~\ref{app:addsim} 
of the Supplementary Material.

\begin{figure}[!ht]
\centering
\begin{tabular}{M{0.0005\textwidth}M{0.29\textwidth}M{0.29\textwidth}M{0.32\textwidth}}
   &\large \textbf{Cellwise}  & \large \textbf{Casewise} &\large{\textbf{Casewise \& Cellwise}} \\
[-4mm]
\rotatebox{90}{\textbf{\footnotesize{$p=30$}}} &
 \includegraphics[width=.31\textwidth]
 {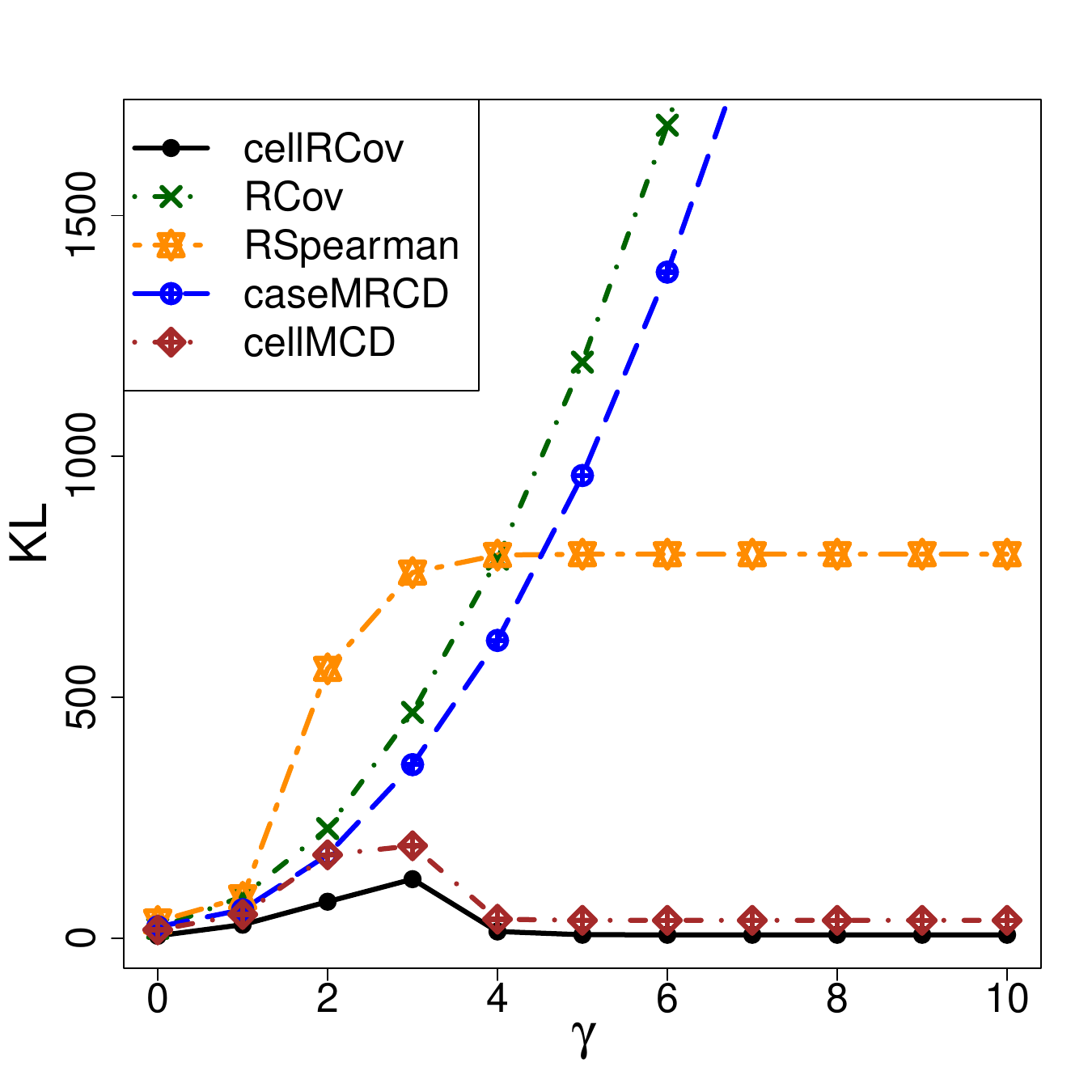} &
 \includegraphics[width=.31\textwidth]
 {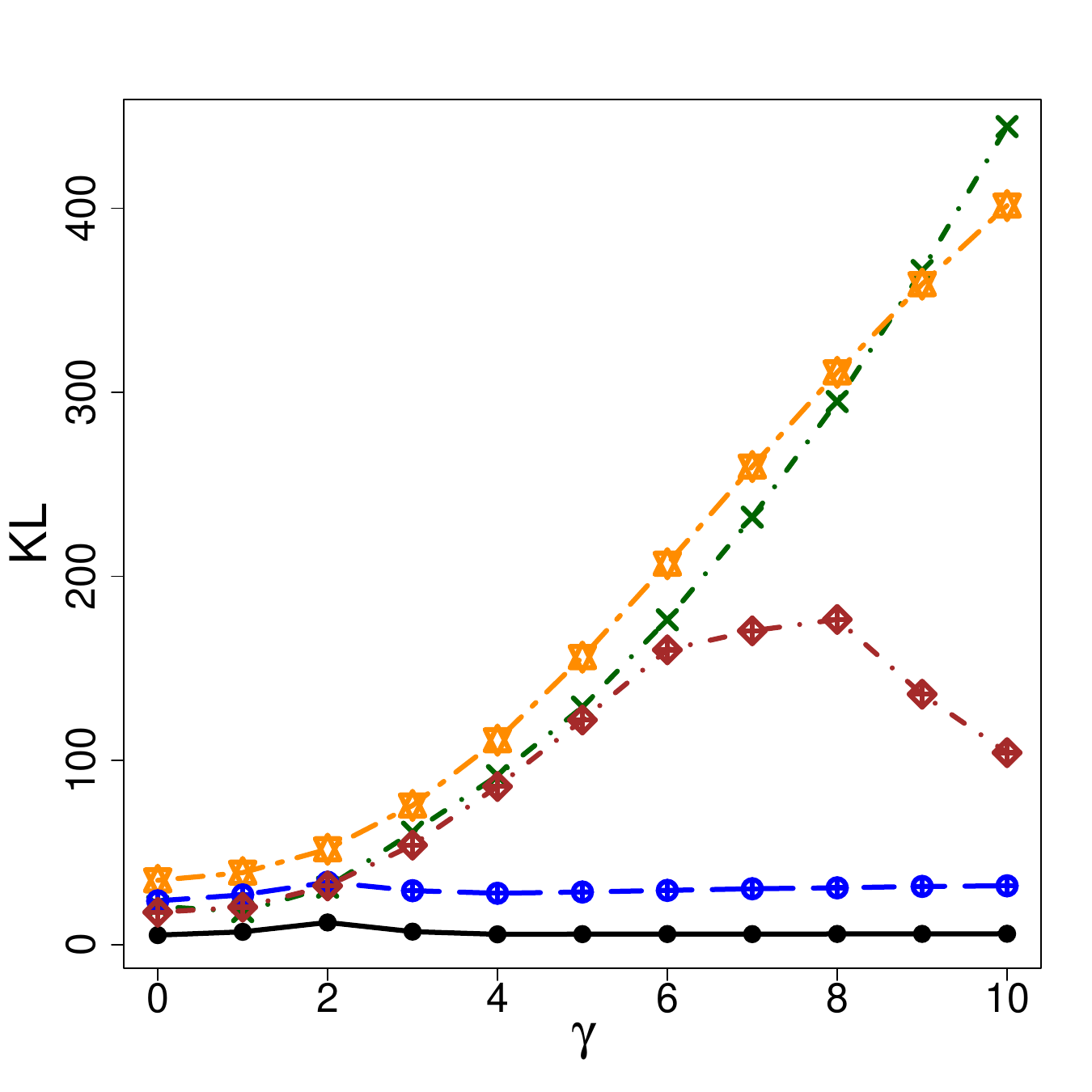} &
 \includegraphics[width=.31\textwidth]
 {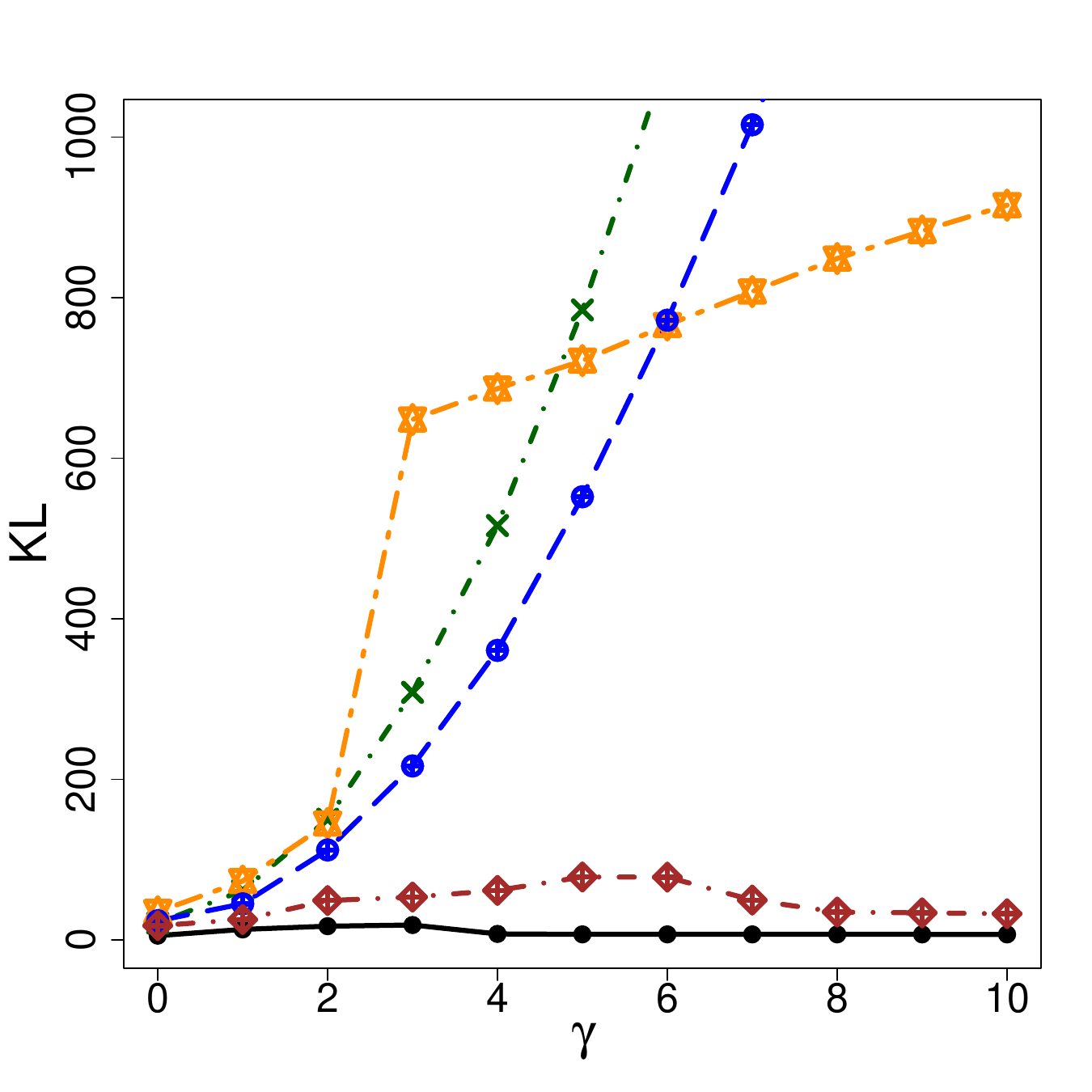} \\ [-4mm]  
 \rotatebox{90}{\textbf{\footnotesize{$p=60$}}} &
 \includegraphics[width=.31\textwidth]
 {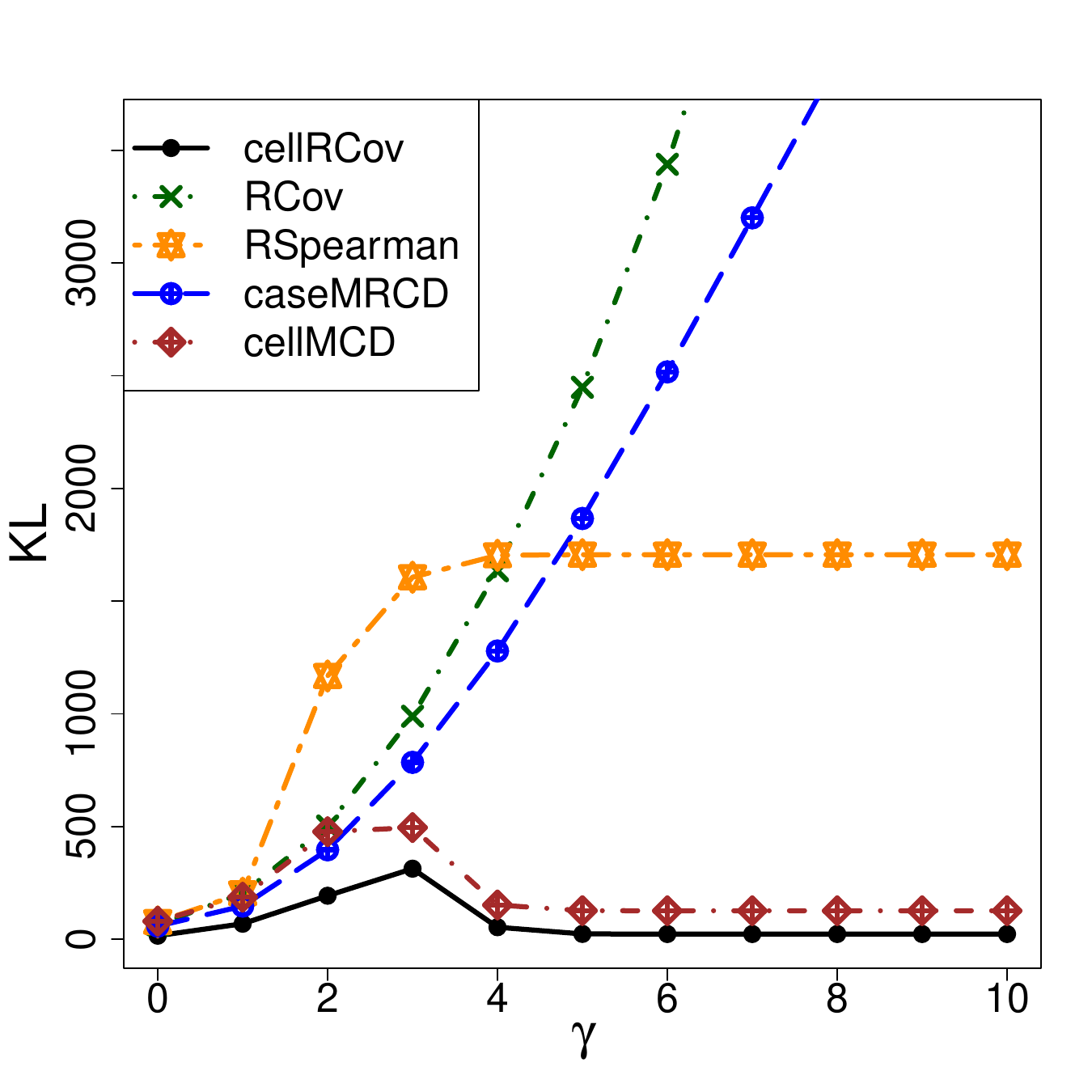} & 
 \includegraphics[width=.31\textwidth]
 {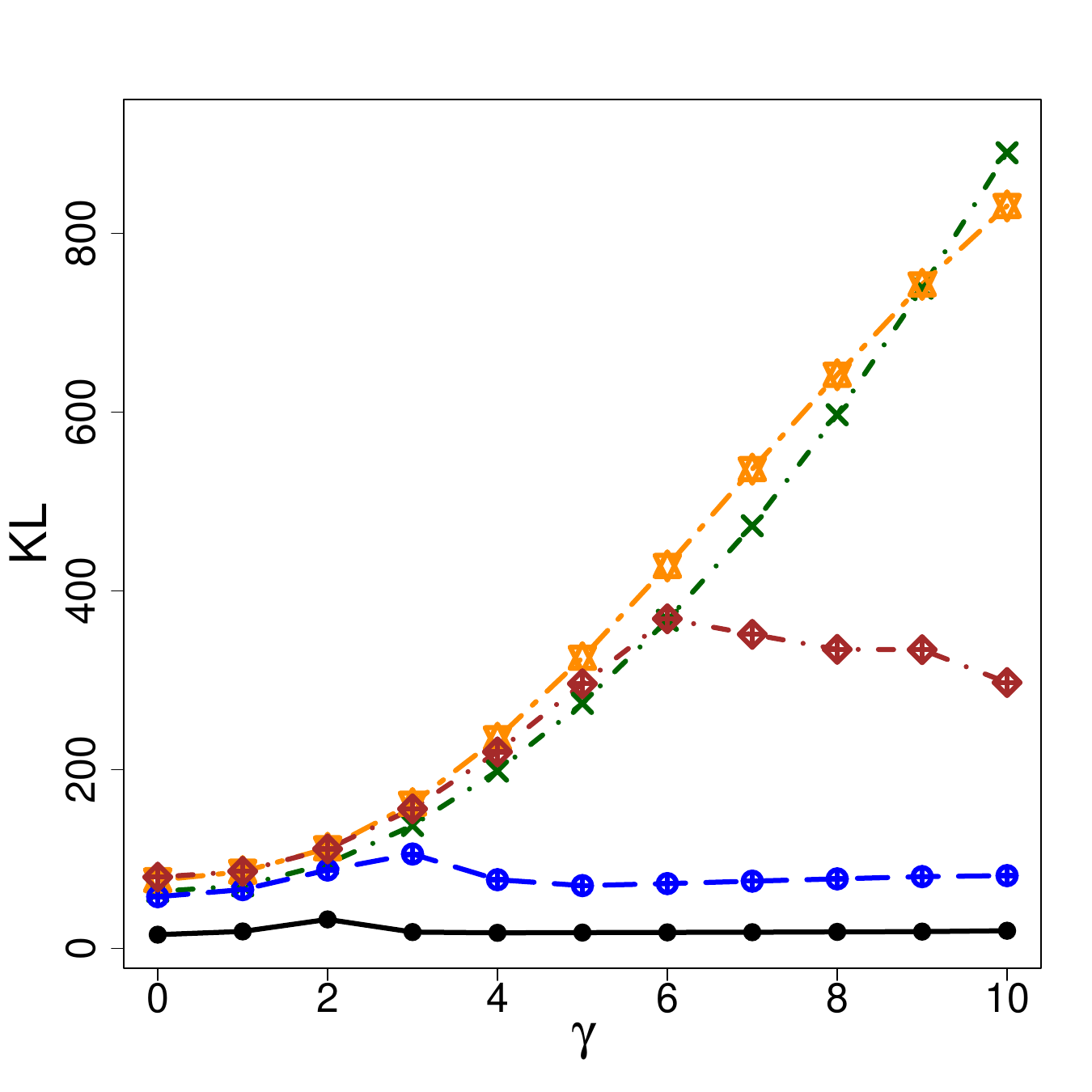} &
 \includegraphics[width=.31\textwidth]
 {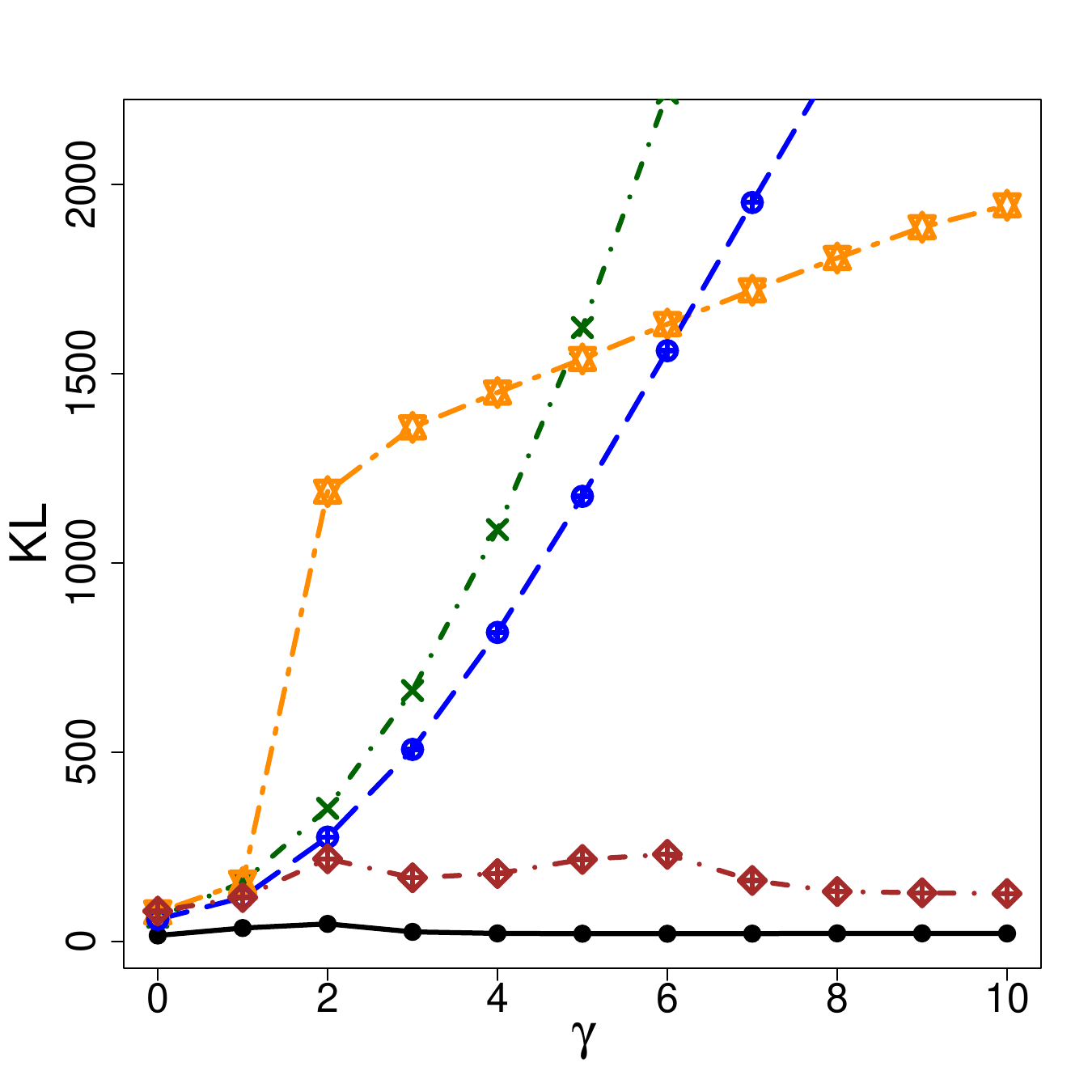} \\ [-4mm]  
 \rotatebox{90}{\textbf{\footnotesize{$p=120$}}} & 
 \includegraphics[width=.31\textwidth]
 {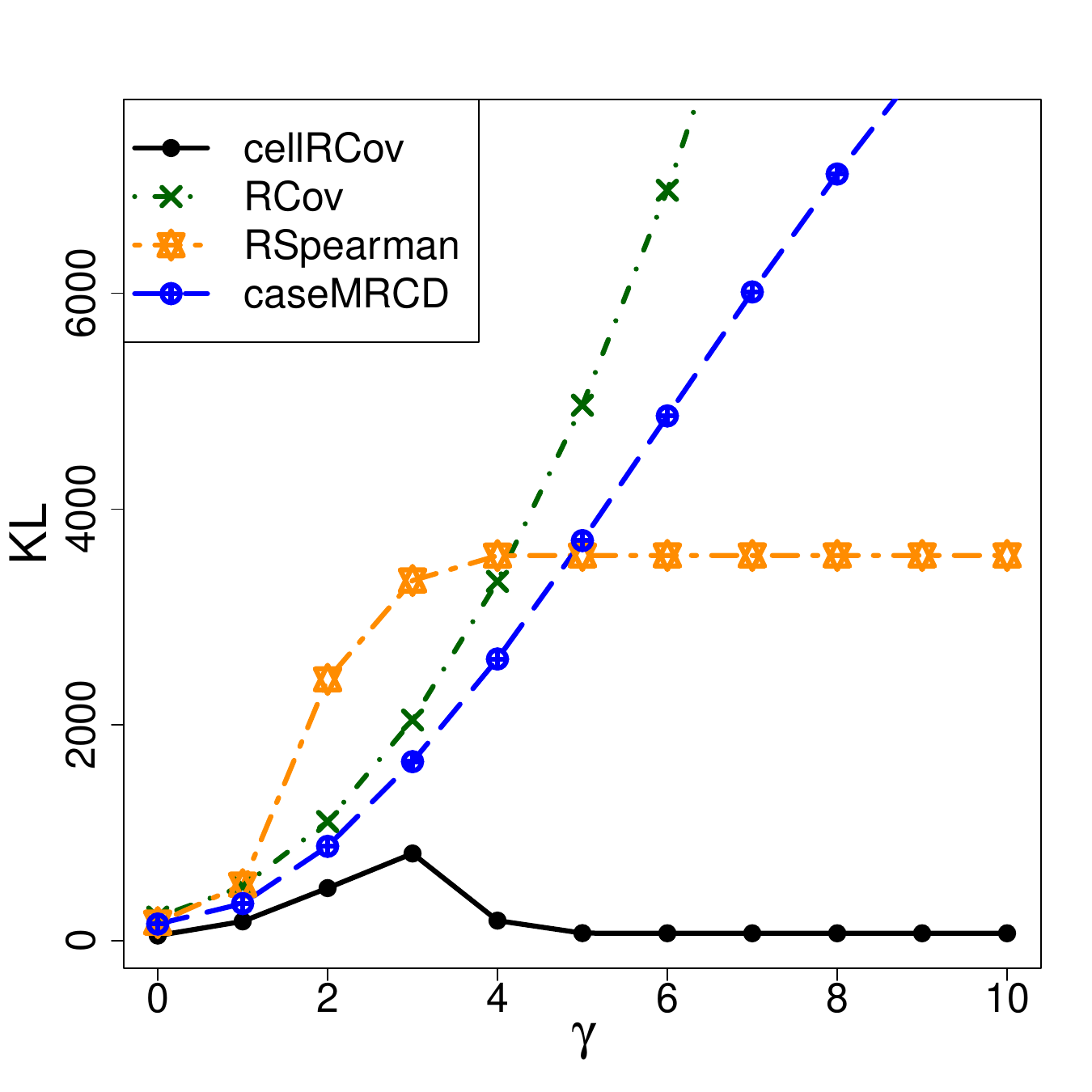} & 
 \includegraphics[width=.31\textwidth]
 {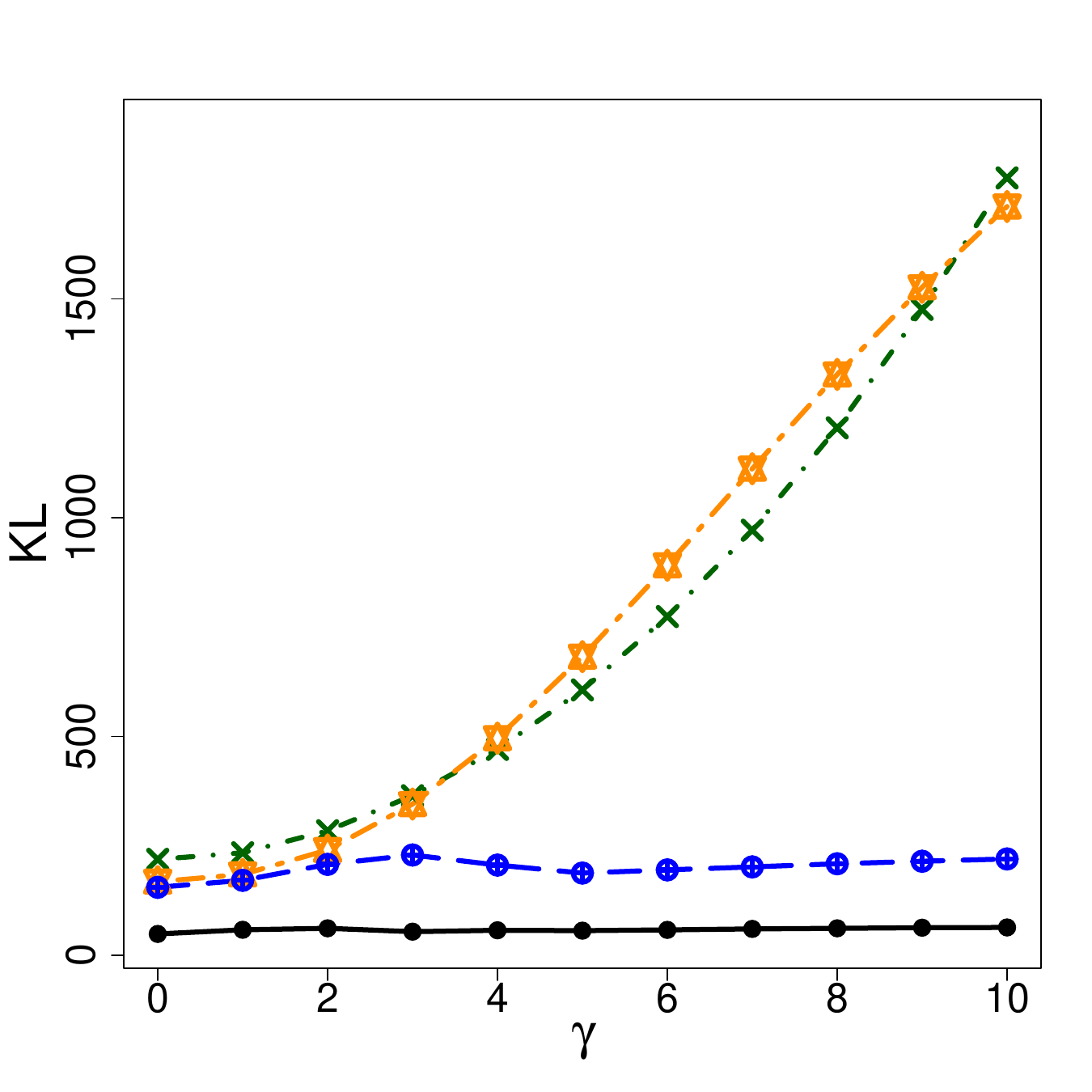} &
 \includegraphics[width=.31\textwidth] 
 {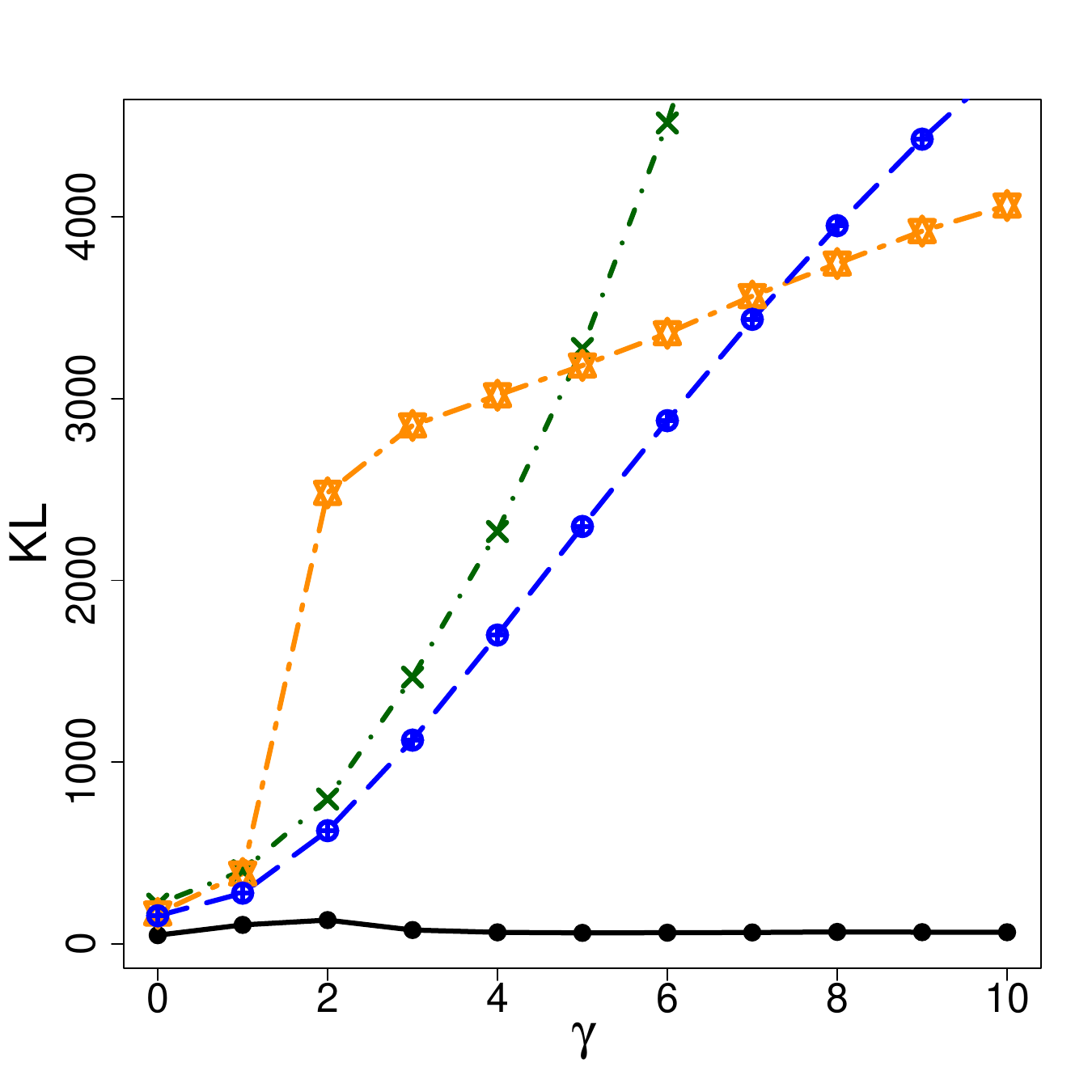}   
\end{tabular}
\caption{Average KL attained by cellRCov, 
RCov, RSpearman, caseMRCD, and cellMCD in the presence 
of cellwise outliers, casewise outliers, or 
both for the A09 covariance model and 
dimensions $p$ in $\lbrace30,60,120\rbrace$.}
\label{fig:results_A09}
\end{figure}

We compare cellRCov with a regularized version 
of the sample covariance denoted by RCov, 
which carries out a ridge-type covariance 
regularization as described in 
Section~\ref{sec:cellRCov}. 
The other competitors are RSpearman, the 
regularized approach based on the Spearman 
correlation of \cite{croux2016robust}, 
the casewise robust minimum regularized 
covariance determinant (caseMRCD) estimator 
of \cite{boudt2020minimum}, and
the cellMCD estimator 
of \cite{raymaekers2024cellMCD}. 
For cellRCov, the rank $\rk$ and the 
regularization parameter $\delta$ are selected 
as described in Sections~\ref{sec:rank} 
and~\ref{sec:delta}. We used the default 
versions of RSpearman, 
caseMRCD, and cellMCD.

The top row of Figure~\ref{fig:results_A09}
shows the results for dimension $p=30$.
We see that the nonrobust RCov exploded in 
the presence of cellwise outliers, casewise 
outliers, and both. The caseMRCD method 
failed when there are cellwise outliers, 
while RSpearman performed poorly for 
casewise outliers.
CellMCD performed well under purely cellwise
contamination, as expected, but its 
performance deteriorated in the presence of 
20\p of casewise outliers, since it is not 
designed to handle that many. It still
did rather well in the mixed setting with
10\p of casewise and 10\p of cellwise
outliers. CellRCov did well for all
three contamination types.
Note that some of the curves have a local
maximum for $\gamma$ around 2 or 3, which is
perfectly natural because that is around the
boundary of the region where a cell or a 
case might be considered outlying. Cells or 
cases with smaller $\gamma$ are not  
downweighted, so they have an effect. The 
middle row of Figure~\ref{fig:results_A09}
shows qualitatively similar results for 
$p=60$, as does the bottom row for $p = 120$ 
where $p > n$. We conclude that cellRCov was 
the only method with satisfactory 
performance in all three outlier settings.

To assess performance in the presence of 
missing data, we repeated the simulation
while setting $20\p$ of randomly selected 
cells to NA. In this situation we cannot 
run caseMRCD, that was not designed for 
incomplete data. For RCov and RSpearman we 
computed the correlations and covariances 
between any two variables on their pairs
of available cells. The resulting 
Figure~\ref{fig:results_A09_NA} looks quite 
similar to Figure~\ref{fig:results_A09}.
Again cellRCov outperforms the other 
approaches over the three outlier settings 
and in all three dimensions. The other 
methods are not competitive in these 
demanding settings.

\begin{figure}[!ht]
\centering
\begin{tabular}{M{0.0005\textwidth}M{0.29\textwidth}
 M{0.29\textwidth}M{0.32\textwidth}} & 
 \large \textbf{Cellwise} & \large \textbf{Casewise} &
 \large{\textbf{Casewise \& Cellwise}} \\ [-4mm]
 \rotatebox{90}{\textbf{\footnotesize{$p=30$}}} &
 \includegraphics[width=.31\textwidth]
 {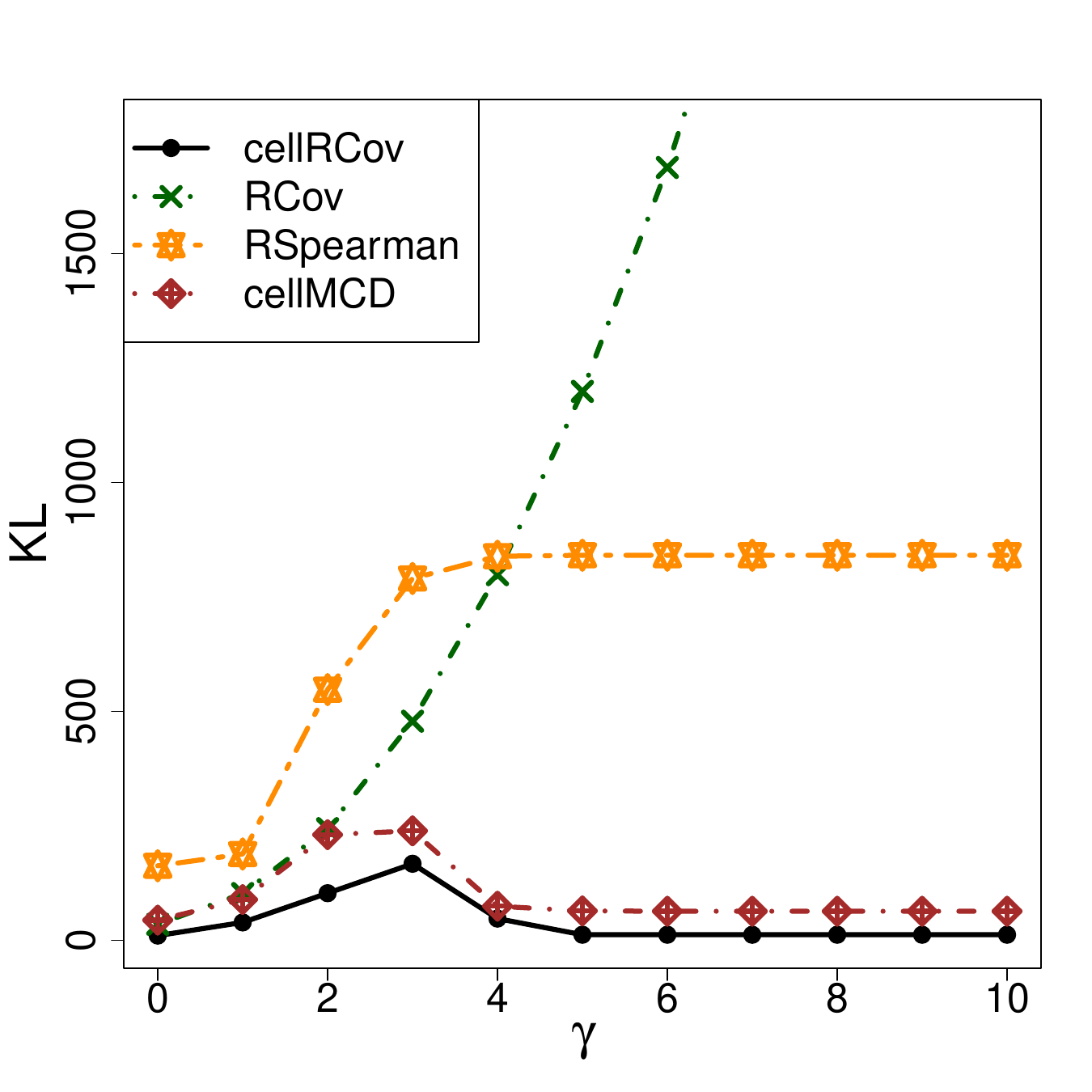} &
 \includegraphics[width=.31\textwidth]
 {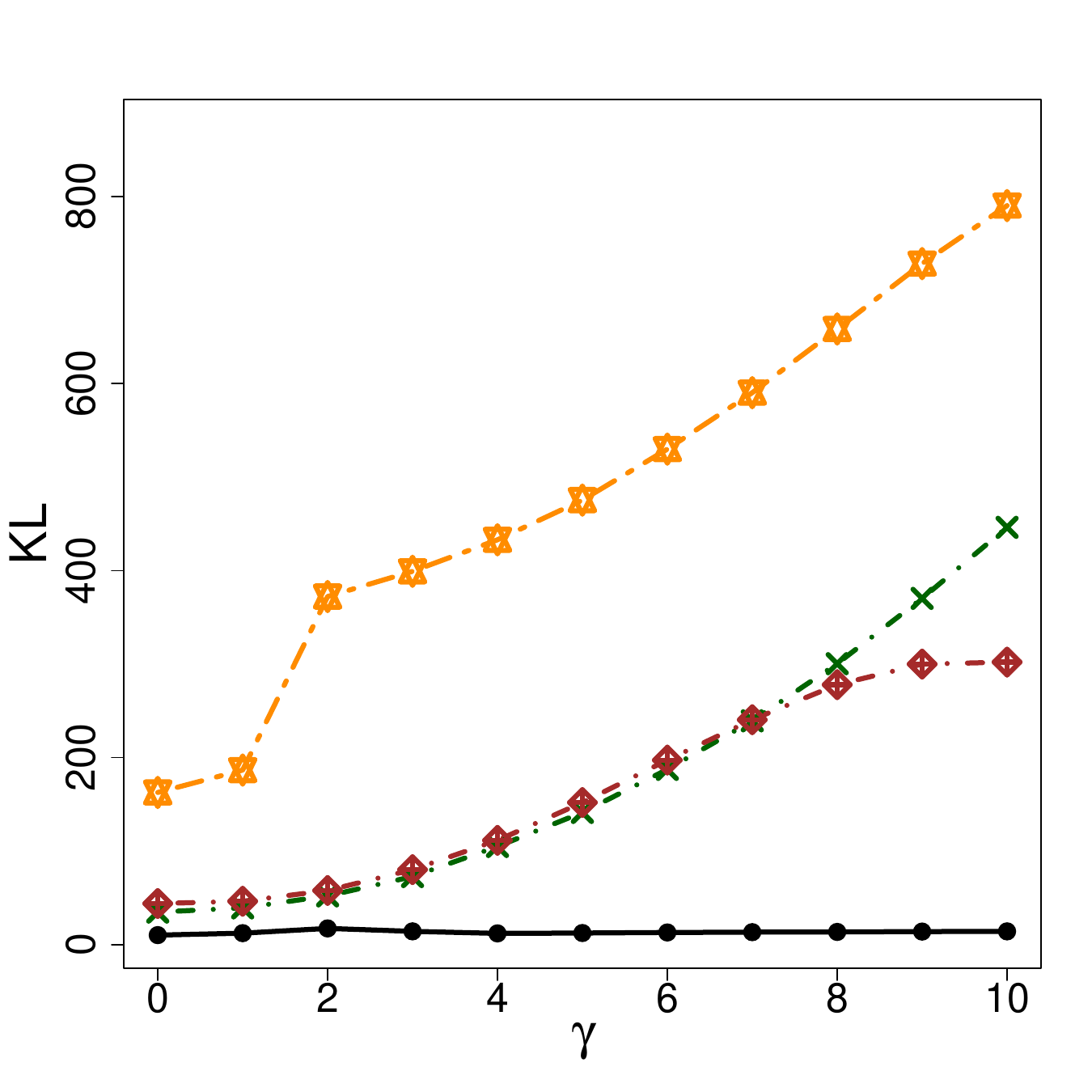} &
 \includegraphics[width=.31\textwidth]
 {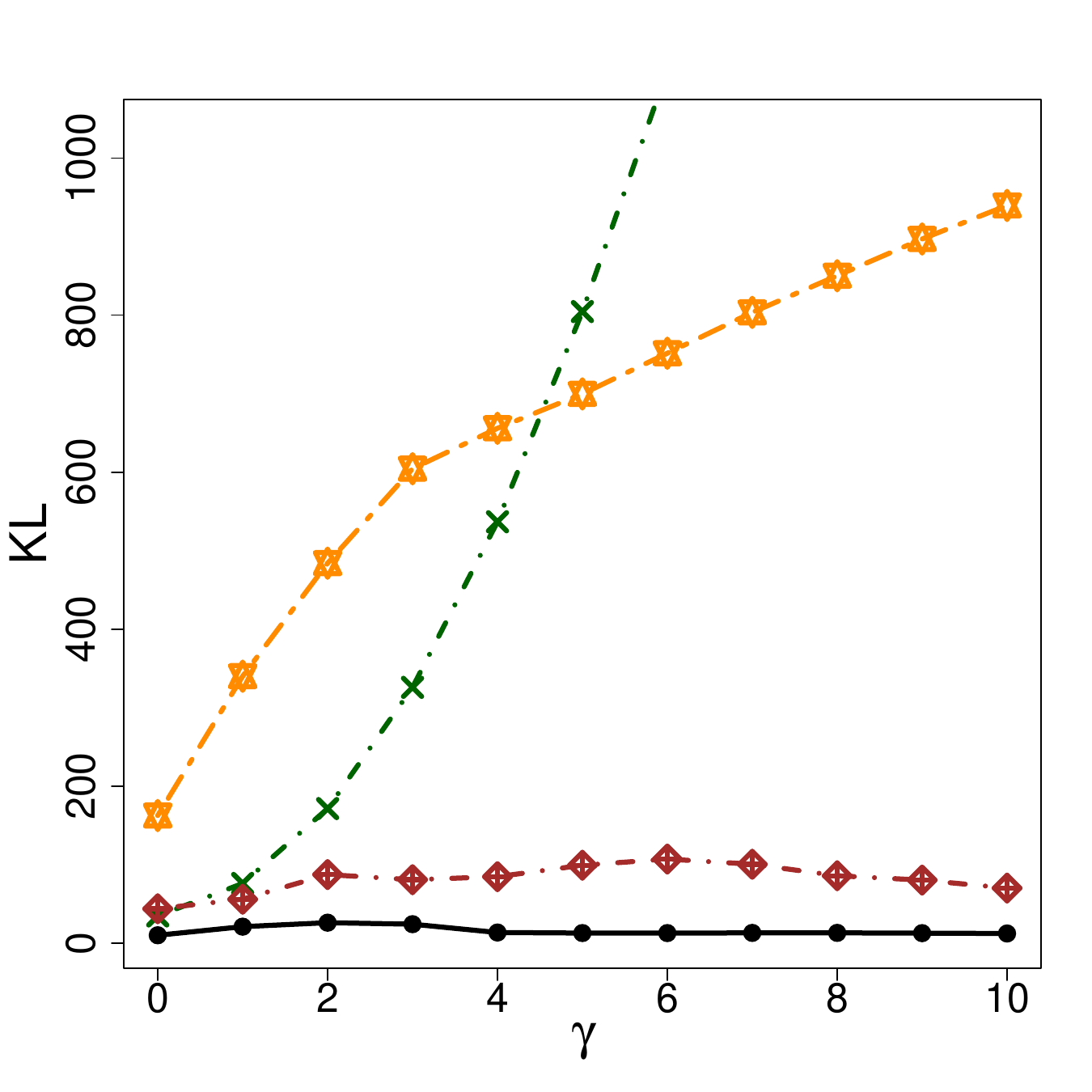} \\ [-4mm]
 \rotatebox{90}{\textbf{\footnotesize{$p=60$}}} &
 \includegraphics[width=.31\textwidth]
 {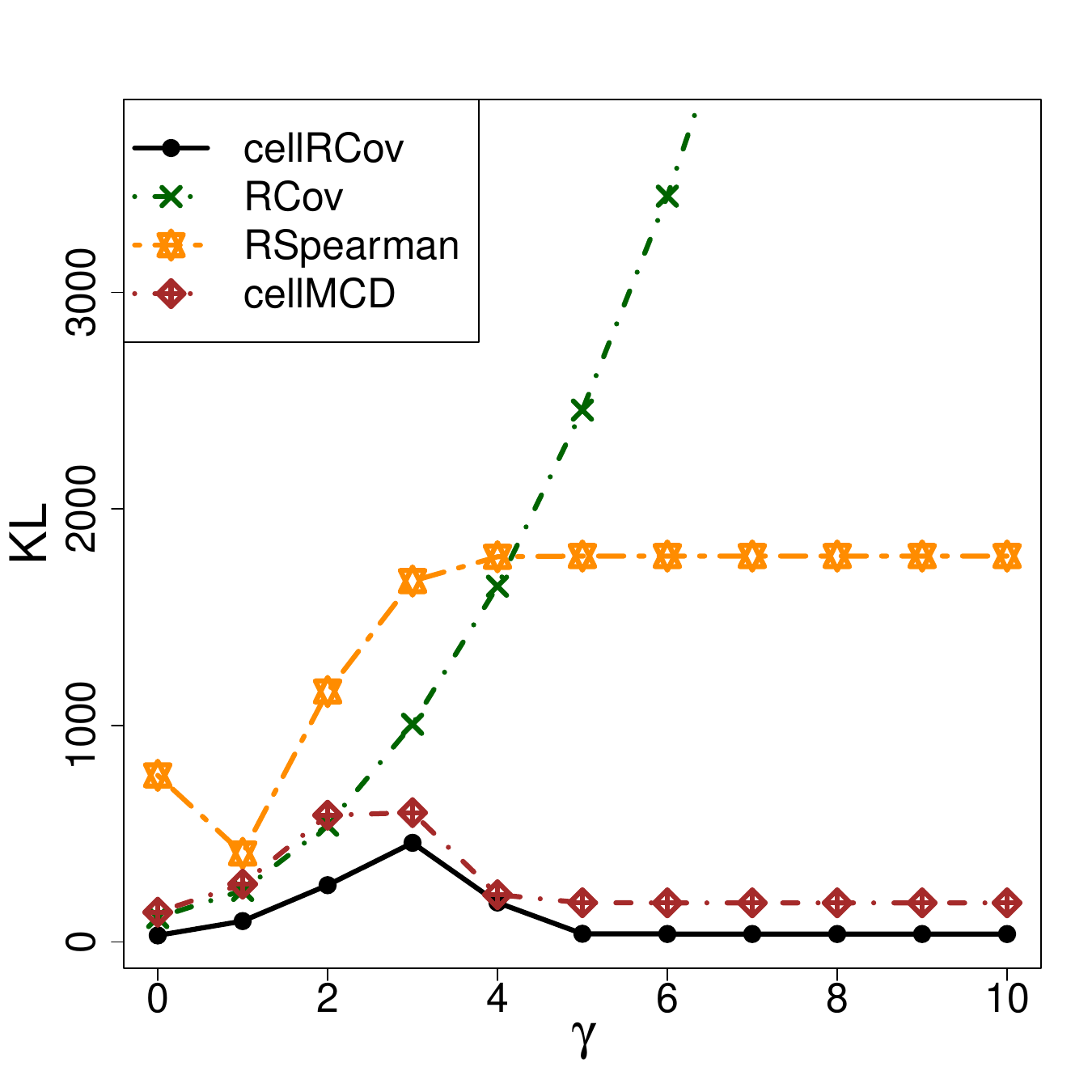} &
 \includegraphics[width=.31\textwidth]
 {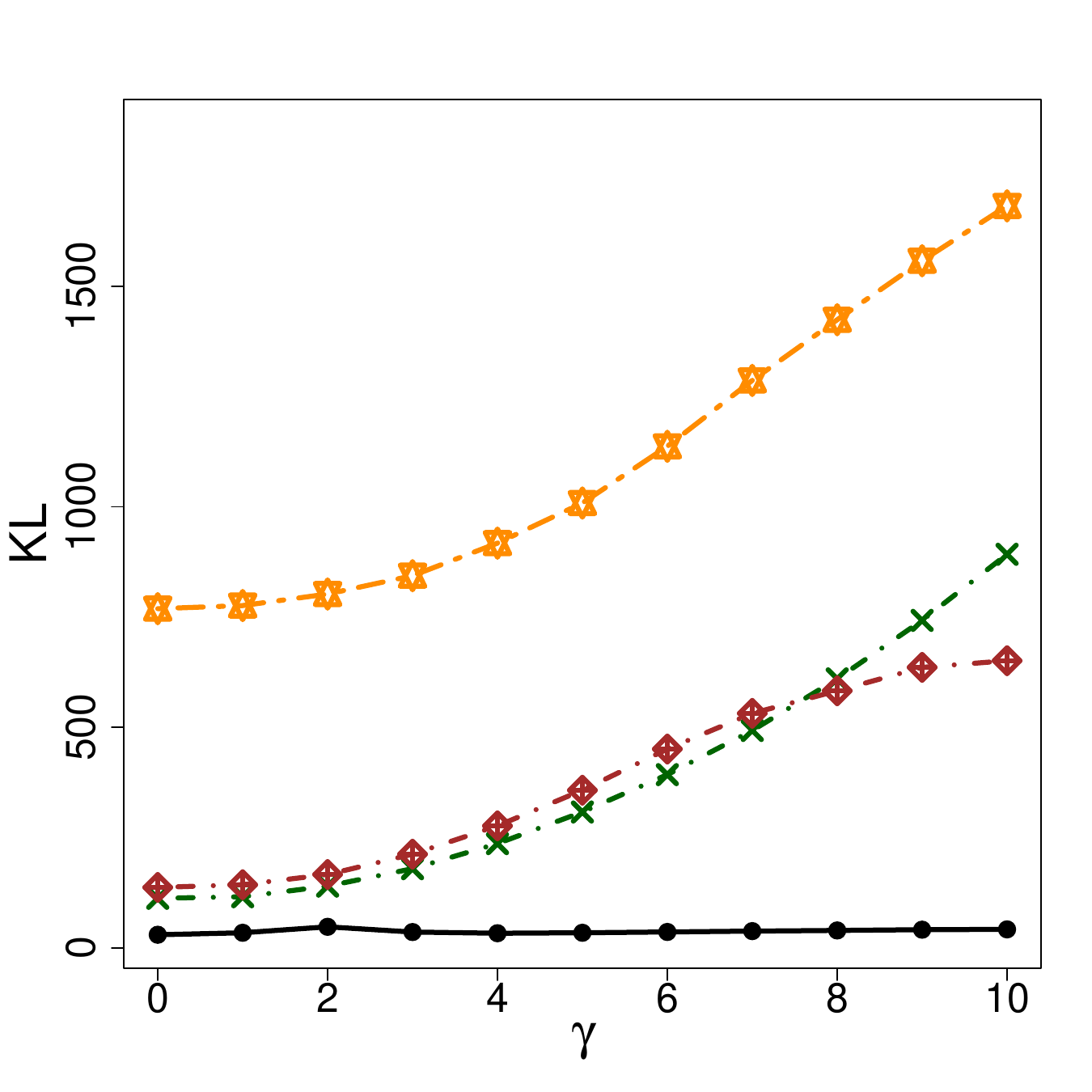} &
 \includegraphics[width=.31\textwidth]
 {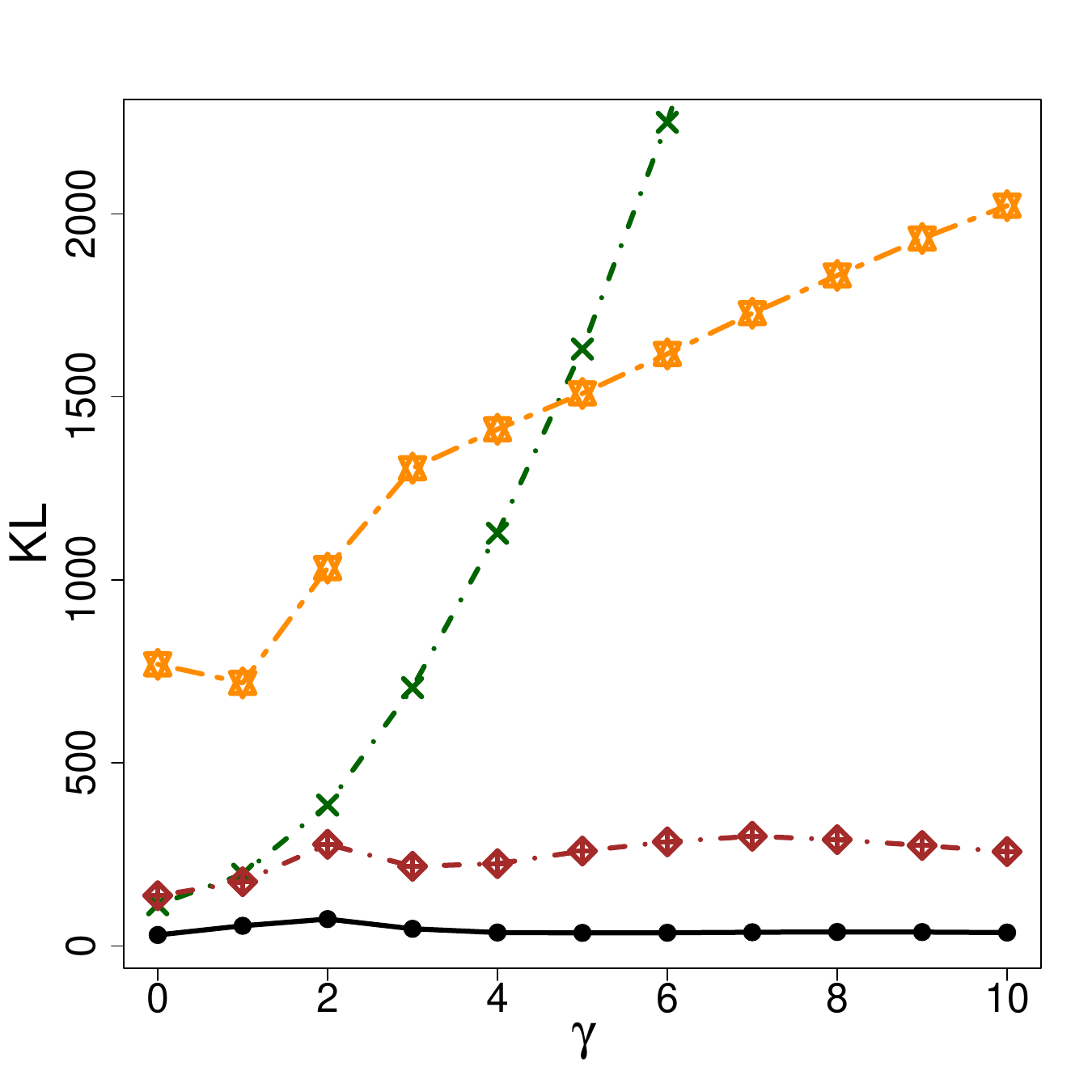} \\ [-4mm] 
 \rotatebox{90}{\textbf{\footnotesize{$p=120$}}} &
 \includegraphics[width=.31\textwidth]
 {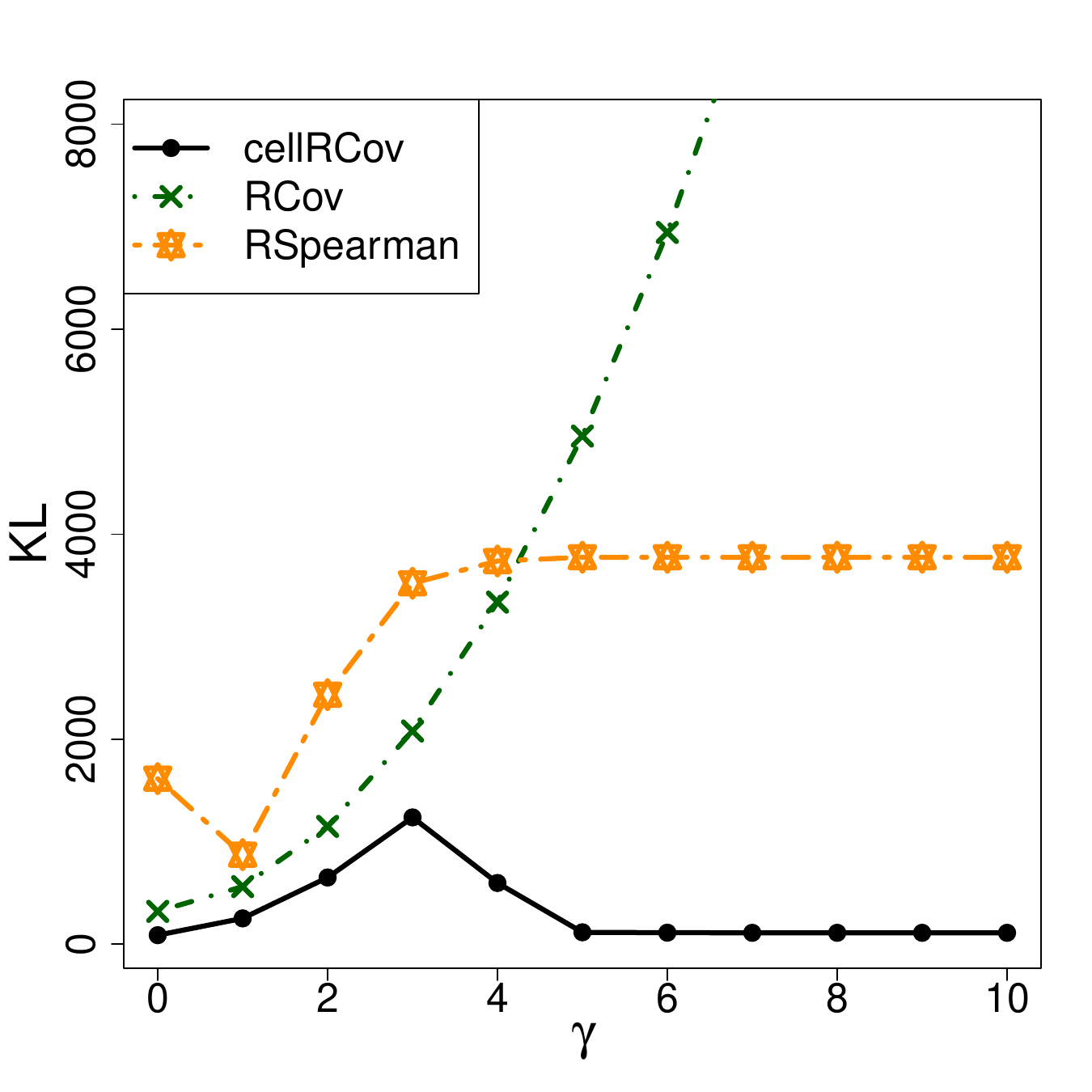} &
 \includegraphics[width=.31\textwidth]
 {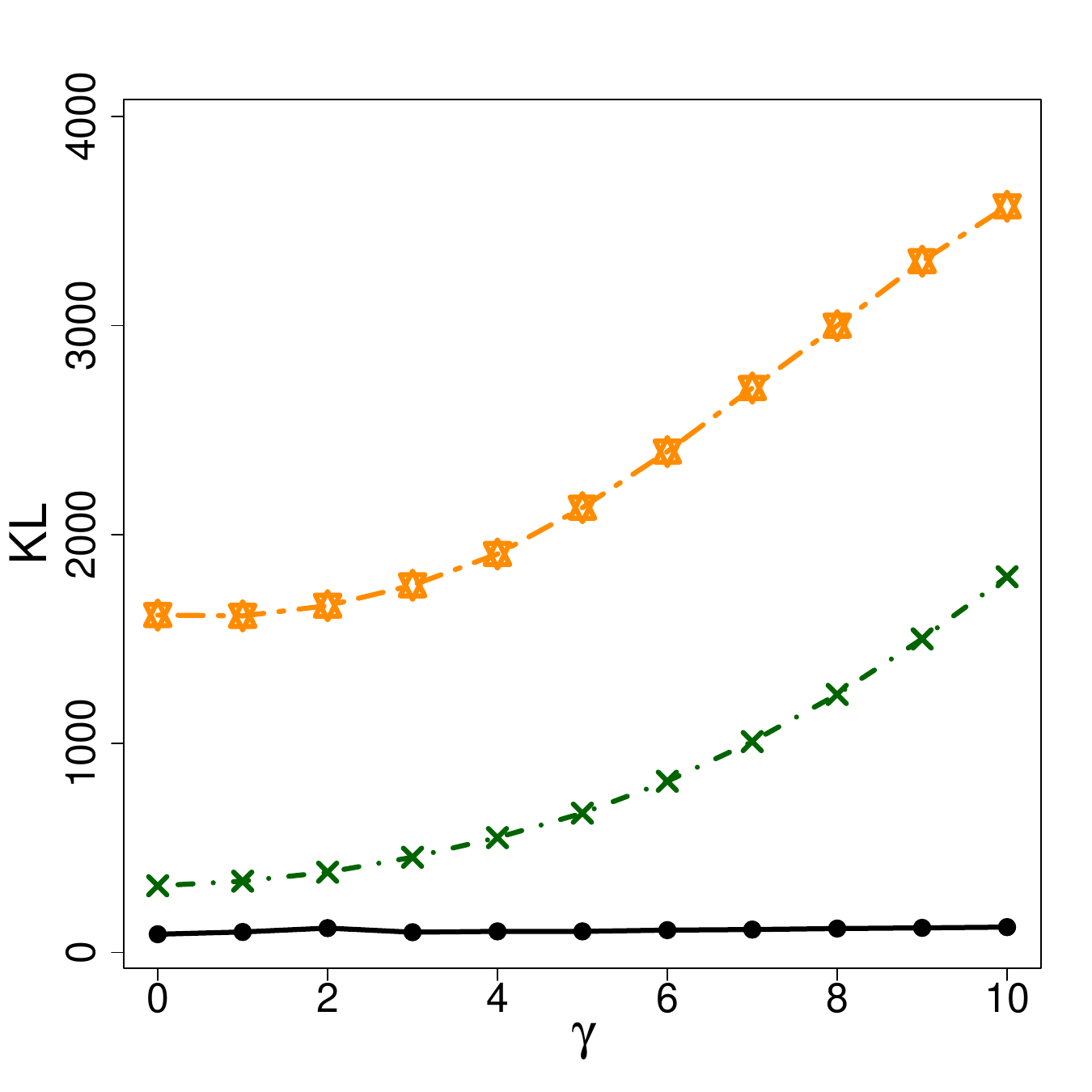} &
 \includegraphics[width=.31\textwidth]
 {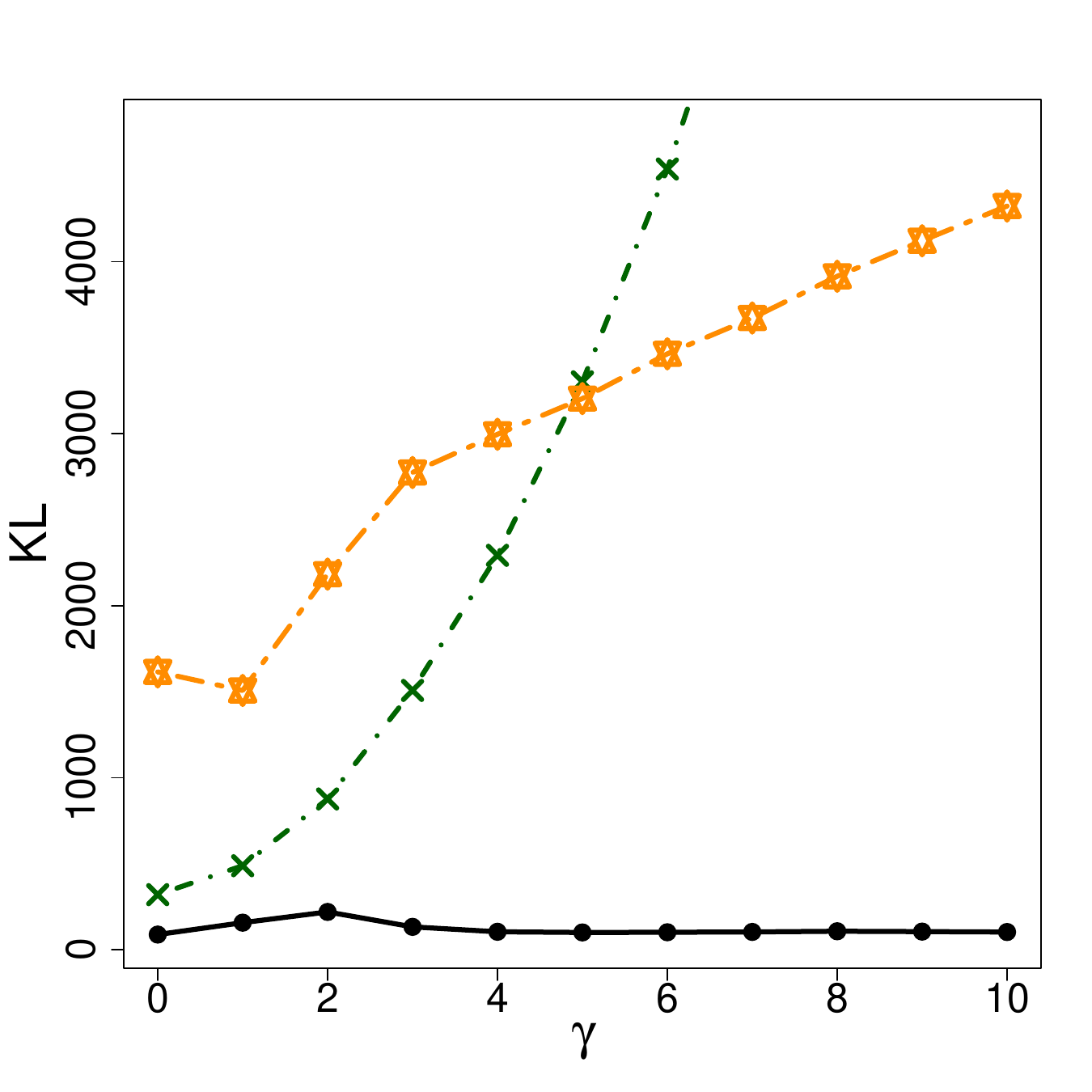}   
\end{tabular}
\caption{Average KL attained by cellRCov, 
RCov, RSpearman, caseMRCD, and cellMCD
in the presence of cellwise outliers, 
casewise outliers, or both for the A09 
covariance model and dimensions $p$ 
in $\lbrace30,60,120\rbrace$, 
with 20\p of missing cells.}
\label{fig:results_A09_NA}
\end{figure}

A sensitivity analysis on the tuning 
parameters $\rk$ and $\delta$ is provided
in Section~\ref{app:sensitivity} of the 
Supplementary Material. It assesses the
effect of varying these parameters across 
the covariance models considered in
the simulation study. 
Section~\ref{app:times} reports runtime 
benchmarks.

\section{Real data examples}
\label{sec:realdata}

\subsection{Anomaly detection in a welding process}
To demonstrate the potential of the proposed covariance estimator, we consider an application about resistance spot welding (RSW) in automotive  manufacturing. RSW is the most common technique employed in joining metal sheets for mass production. In this process two overlapping galvanized steel sheets are joined without the use of any filler material, to guarantee the structural integrity and solidity of the welded items \citep{martin2014assessment}.
The so-called dynamic resistance curve (DRC) shows how the electrical resistance between the pieces of metal being joined changes during the welding operation, providing a technological signature of the metallurgical development of a spot weld. Anomalous conditions, such as excessive welding current and insufficient electrode pressure, cause expulsion in the welding joint \citep{mikno2018analysis}. Expulsion is the ejection of molten metal from either the interfaces of metal sheets or the interfaces between
metal sheets and electrodes, resulting in a serious defect that destroys the welding strength. It also reduces the electrode life. Given these detrimental effects, detecting expulsion is paramount for ensuring consistent welding quality in production lines. 

The dataset consists of $n = 115$ DRC curves measured in $m\Omega$, collected for five spot welding points made by the same welding machine. The data are shown in Figure \ref{fig:data_rsw}. The total number of variables is $p=750$. Expulsion occurred in parts of 70 curves, so more than half of the cases. 

\begin{figure}[ht]
\vspace{2mm}
\centering
\includegraphics[width=0.9\textwidth]
   {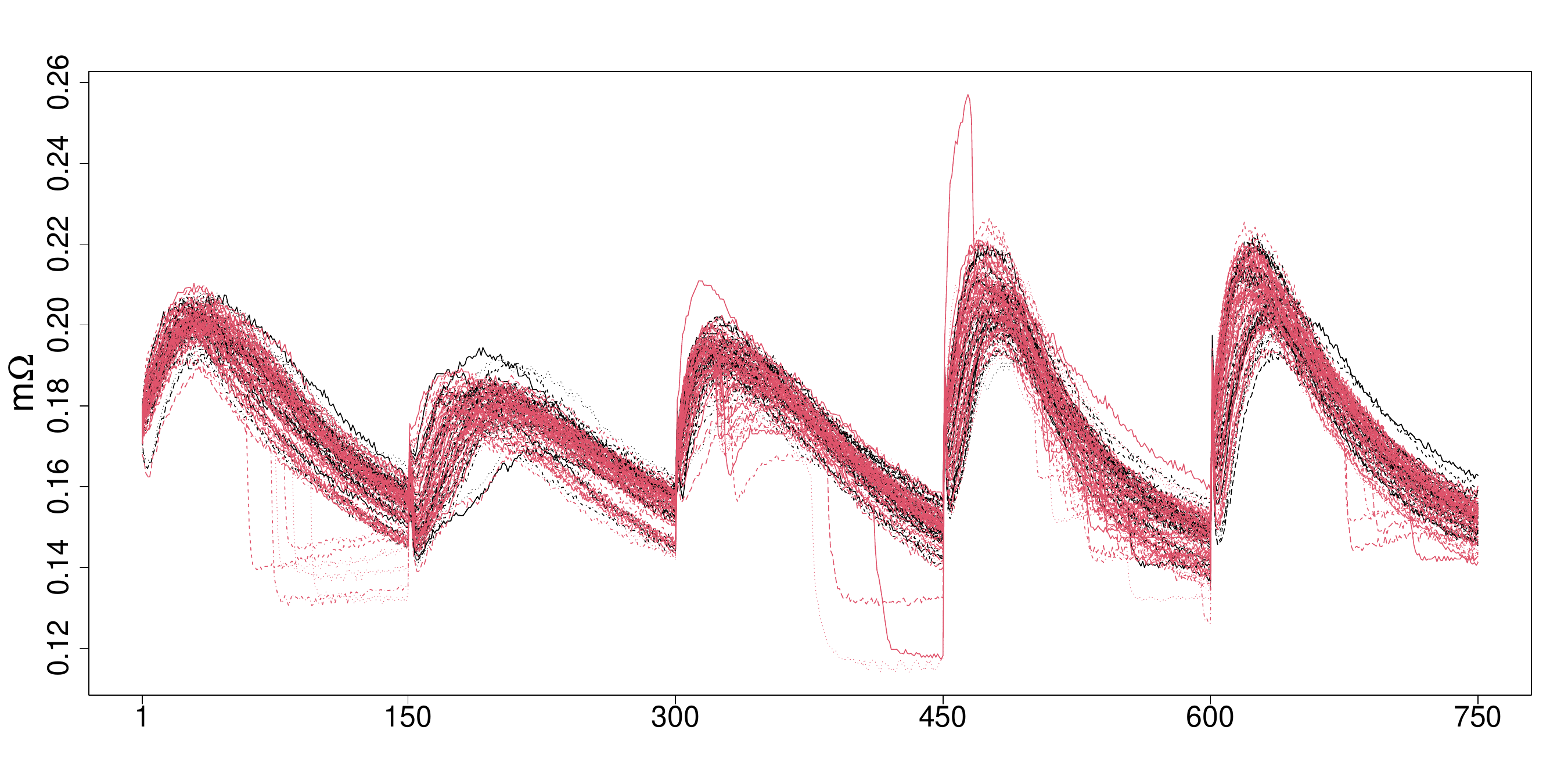}
\caption{The 115 DRC measurements in 
$[m\Omega]$ corresponding to five spot 
welding points. Red curves are cases where 
expulsion occurred.}
\label{fig:data_rsw}
\end{figure}

Since manual labeling of expulsions is 
time-consuming and costly, the goal is to 
develop an unsupervised method for detecting 
joints where expulsion has occurred. For 
this purpose we  consider a detection rule 
where a case $\bx_i$ is classified as 
anomalous if its Mahalanobis distance 
$\MD_i = ((\bx_i-\bhmu)^T\bhSigma^{-1}
(\bx_i-\bhmu))^{1/2}$ is large.
Since $p > n$, we need to use regularized 
estimators. We compare the performance of 
the detection rules using the cellRCov, 
RCov, Spearman, and caseMRCD estimates 
of $\bmu$ and $\bSigma$. Here cellRCov 
selected $k=8$ and $\delta=0.56$.

\begin{figure}[ht]
\centering
\includegraphics[width=0.4\textwidth]{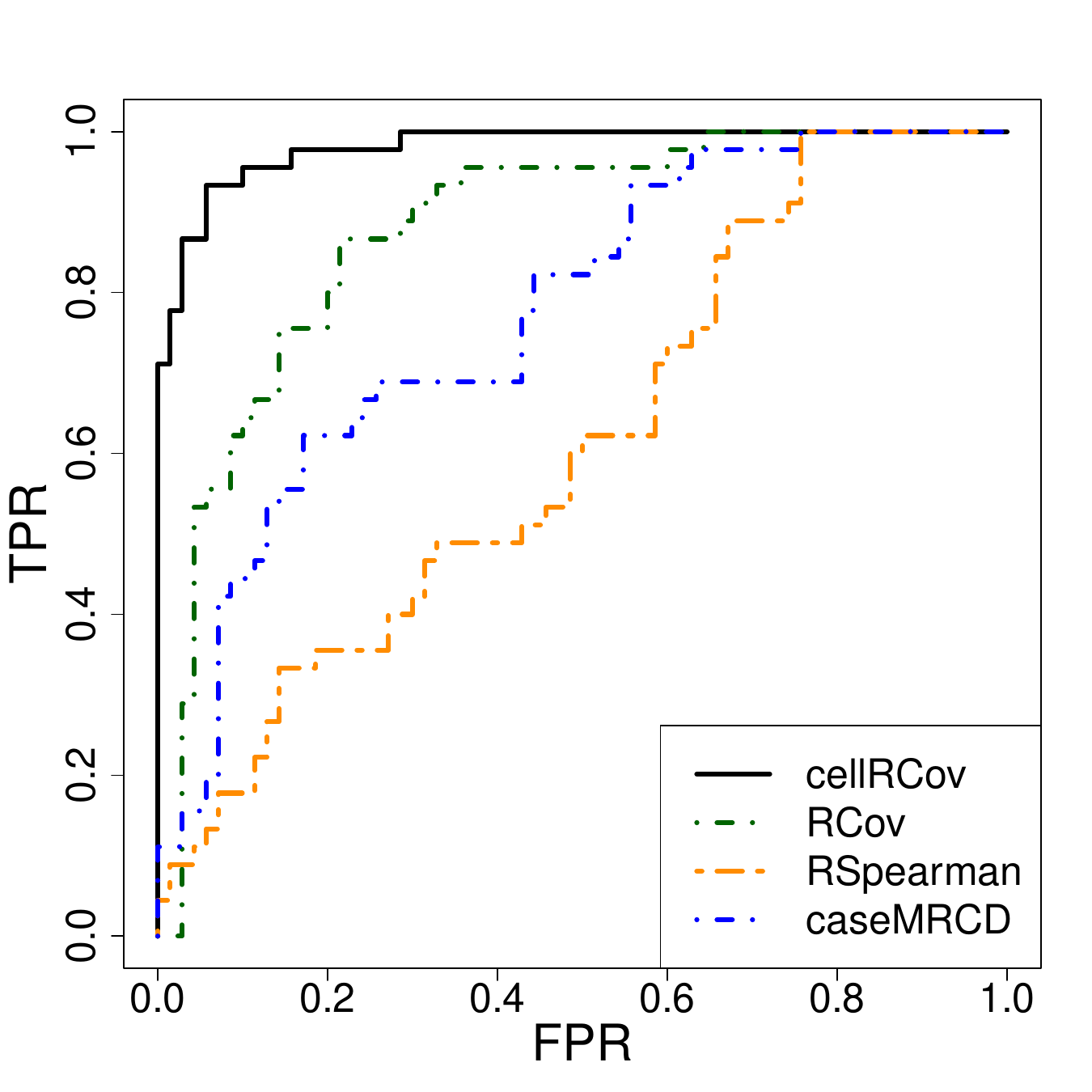}
\caption{ROC curves of the detection rules using 
cellRCov, RCov, RSpearman, and \mbox{caseMRCD}.}
\label{fig:results_rsw}
\end{figure}

The performance of these methods is measured by 
their receiver operating characteristic (ROC) 
curves, which describe the trade-off between the 
true positive rate (TPR) and the false positive 
rate (FPR) at various thresholds. The TPR, also 
known as sensitivity, is the proportion of 
actual positives identified correctly, while 
the true negative rate (TNR), also called 
specificity, is the proportion of actual 
negatives classified correctly.
In Figure~\ref{fig:results_rsw} we
see that cellRCov substantially outperforms 
the competing methods as its ROC curve lies
consistently above the others. This is also
reflected in the area under the ROC curve (AUC),
a standard summary measure of the effectiveness
of a detection rule. Higher values of AUC indicate
better discrimination between normal and anomalous 
cases. The AUC of the cellRCov rule is 0.981, 
which is close to the maximal value of 1 and far 
exceeds those of RCov (0.878), caseMRCD (0.774), 
and Spearman (0.615).

Computing cellRCov on this dataset required 
37 seconds. The computation was performed on 
a laptop equipped with an AMD Ryzen 9 PRO 
8945HS processor with 8 cores and 16 threads, 
and 64 GB of RAM.

To find out which part of the 
cellRCov construction drives the detection of 
anomalous curves, we compute two distances.
Denote the fitted point of $\bz_i$ as 
$\bhz_i$\,. We then compute
$\MD^{\bz^\sub}_{i}
    =
    ((\bhz_i-\bhmu)^T
    \btSigma_{\bz^\sub}^{-1}
    (\bhz_i-\bhmu))^{1/2}$
and $\MD^{\zort}_{i}
    =
    ((\bz_i-\bhz_i)^T
    \left(\btSigma_{\zort}^{R}\right)^{-1}
    (\bz_i-\bhz_i))^{1/2}$.
Therefore $\MD^{\bz^\sub}_{i}$ measures
how far $\bhz_i$ is from the center in the 
fitted subspace, whereas $\MD^{\zort}_{i}$ is 
the distance of its residual component.

\begin{figure}[ht]
\centering
\includegraphics[width=0.45\textwidth]
   {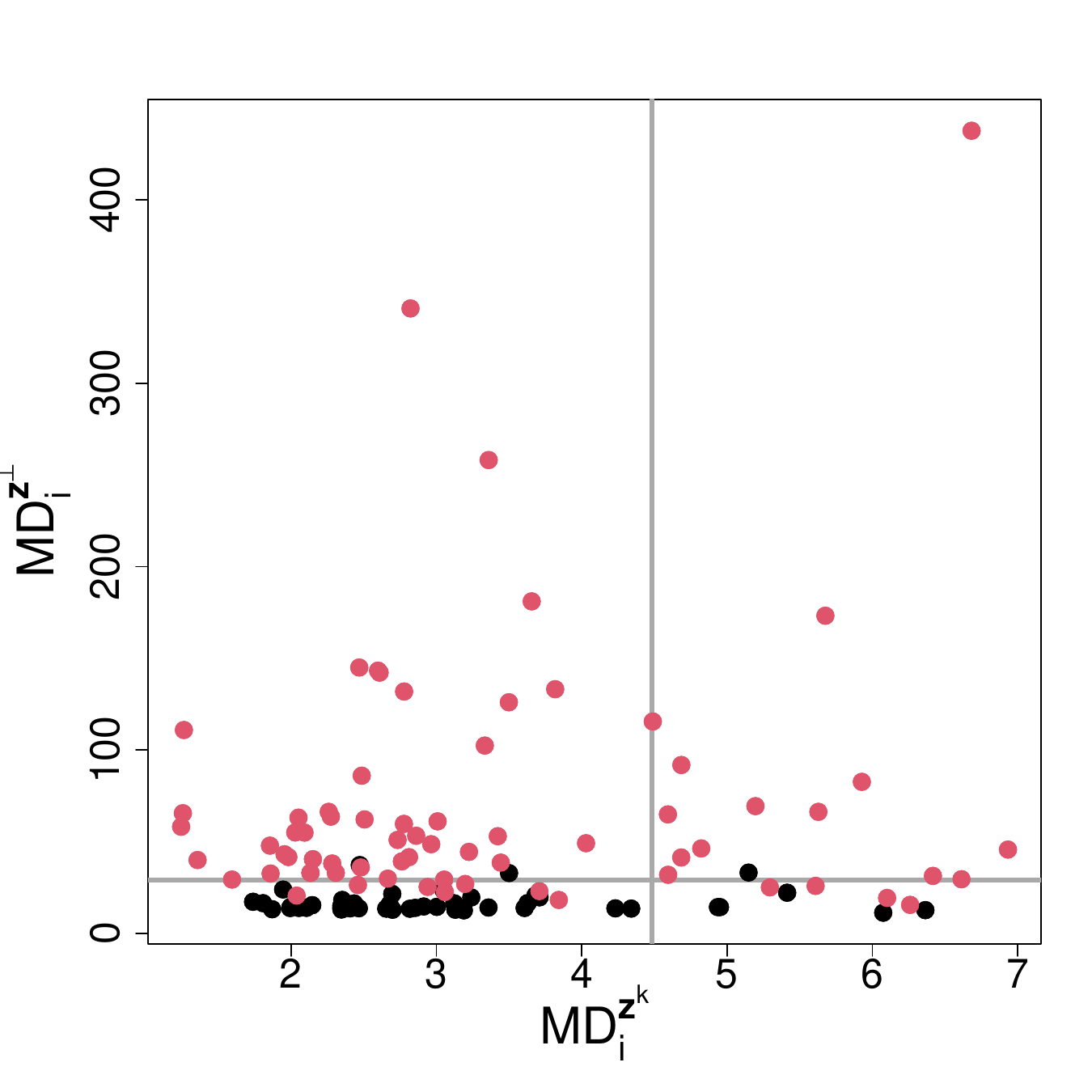}\\
\vspace{-2mm}
\caption{Distances $\MD^{\bz^\sub}_{i}$ 
and $\MD^{\zort}_{i}$ for the RSW data. Cases
with expulsion are shown in red.}
\label{fig:rsw_components}
\end{figure}

Figure~\ref{fig:rsw_components} plots
$\MD^{\zort}_{i}$ versus $\MD^{\bz^\sub}_{i}$. 
Red points correspond to curves where expulsion 
occurred, while black points correspond to 
regular curves. The vertical line is at the 
cutoff $\sqrt{\chi^2_{k,0.99}}$ and the 
horizontal line at $\sqrt{\chi^2_{p-k,0.99}}$\,.
The red and black points are 
mainly separated through $\MD^{\zort}_{i}$, 
whereas $\MD^{\bz^\sub}_{i}$ shows less 
separation. This indicates that, for these data, 
the anomalous behaviour is primarily captured 
by deviations from the fitted subspace rather 
than by unusual positions within that subspace.

\subsection{Robust canonical correlation analysis}
\label{sec:CCA}

Canonical Correlation Analysis (CCA) is 
designed to explore and quantify the relation
between two sets of variables, by identifying 
pairs of linear combinations that are most 
correlated. Instead of describing the full 
covariance structure, CCA extracts a 
lower-dimensional summary that captures the 
strongest associations between the two sets 
of variables. This is particularly useful when 
the goal is to reduce complexity while 
preserving meaningful dependencies, such as in 
feature extraction and visualization.

Take two random vectors 
$X^{(1)}\in\mathbb{R}^p$ 
and $X^{(2)}\in\mathbb{R}^q $ with means 
$\bmu_1$ and $\bmu_2$\,, covariance matrices 
$\bSigma_1$ and $\bSigma_2$\,, and 
cross-covariance $\bSigma_{12} = 
\bSigma_{21}^T = \Cov(X^{(1)},X^{(2)})$.  
CCA seeks a sequence of vectors 
$\ba_\ell \in \mathbb{R}^p$ and 
$\bb_\ell \in \mathbb{R}^q$ for
$\ell=1,\dots,\min(p,q)$ that maximize the 
correlation between the linear projections on 
$X^{(1)}$ and $X^{(2)}$. Formally, CCA
seeks
\begin{equation*}
  \left(\ba_{\ell}, \bb_{\ell}\right) =
  \argmax_{\ba, \bb} \mbox{corr} \left(
  \ba^T X^{(1)},\bb^T X^{(2)}\right)
\end{equation*}
subject to the constraints
\begin{equation*}
  \Cov\left(\ba_{\ell}^T X^{(1)}, 
  \ba_j^T X^{(1)}\right) = 0
  = \Cov\left(\bb_{\ell}^T X^{(2)}, 
  \bb_j^T X^{(2)}\right), 
  \quad j=1, \ldots, \ell-1.
\end{equation*}
The resulting scalar variables
$U_{\ell} = \ba_{\ell}^T X^{(1)}$ and 
$V_{\ell} = \bb_{\ell}^T X^{(2)}$\,, known as 
canonical variables, provide a lower-dimensional 
representation of the original data while 
preserving the strongest correlations. 
The canonical directions $\ba_{\ell}$ and 
$\bb_{\ell}$ indicate how each original variable 
contributes to the correlated structure between 
the two data sets.

It can be shown that $\ba_{\ell}$ and 
$\bb_{\ell}$ are the 
eigenvectors associated with the largest 
eigenvalues of the matrices 
$\bSigma_1^{-1}\bSigma_{1 2}
\bSigma_2^{-1}\bSigma_{2 1}$ and 
$\bSigma_2^{-1}\bSigma_{2 1}
\bSigma_1^{-1}\bSigma_{1 2}$\,,
see e.g.\ \cite{johnson2002applied}.
In classical 
CCA, $\ba_{\ell}$ and $\bb_{\ell}$ 
are estimated from the sample covariance 
estimates $\bS_1$\,, $\bS_2$\,, and 
$\bS_{12}$ of $\bSigma_1$\,, $\bSigma_2$\,, 
and $\bSigma_{12}$\,. However, CCA suffers 
in high-dimensional settings as $\bS_1$ and 
$\bS_2$ can become ill-conditioned or even 
singular. To address this issue regularized
versions of CCA have been proposed
\citep{leurgans1993canonical,WilmsCroux2015},
typically by incorporating regularized 
covariance estimators. Also the presence of 
outliers can severely impact classical CCA by 
making the sample mean and covariance estimates
unreliable. In response to this, several 
casewise robust methods have been introduced 
\citep{branco2005robust,Alfons_MaxAssoc}. 

We propose a CCA method that is cellwise
and casewise robust and able to deal with 
high-dimensional data, by estimating 
$\bSigma_1$\,, $\bSigma_2$\,, and 
$\bSigma_{12}$ as submatrices of cellRCov 
applied to the combined dataset 
$[\bX^{(1)}\, ;\, \bX^{(2)}]$. We refer to 
this procedure as cellRCCA.

To assess the performance of cellRCCA we 
analyze the Corn dataset 
\citep{kalivas1997two}, a well-established 
benchmark in chemometrics. The data is 
available from 
\url{https://www.eigenvector.com/data/Corn/index.html}\,. 
It consists of Near-Infrared 
(NIR) spectra of 80 corn samples. 
Each sample has a spectrum $\bx_i^{(1)}$ ranging 
from 1100 to 2498 nm at 7 nm intervals, 
resulting in $p=200$ variables. The dataset 
also contains a vector $\bx_i^{(2)}$ with $q=4$ 
values for moisture, oil, protein, and starch 
content. This makes it an ideal testbed for 
CCA. Our analysis aims to study the relation 
between the NIR spectral data and the chemical 
composition of the corn samples, assessing how 
well the spectral information can predict the 
underlying physical and chemical properties. 
Studies of such relationships are common in 
the food, chemical, and pharmaceutical
industries.

We want to compare the performance of cellRCCA
to that of rCCA, the regularized CCA method of 
\cite{leurgans1993canonical}, since we cannot 
apply the plain classical CCA because $p>n$
in these data. To evaluate their performance 
we implement a 10-fold cross-validation 
procedure. The dataset is randomly split into 
training and test sets, and for each fold we 
derive the canonical directions from the
training set and then apply these to compute
canonical variables on the test set. From
these variables we then compute the Mean 
Canonical Correlation (MCC) that we define as
\begin{equation*}
\text{MCC}=\frac{1}{L}\sum_{\ell=1}^{L}
  \widehat{\mbox{corr}}_{S}\left(
  \bX^{(1)} \ba_{\ell}\;,
  \bX^{(2)} \bb_{\ell} \right)
\end{equation*}
where 
$\widehat{\operatorname{corr}}_{S}(\cdot)$ 
is the Spearman correlation coefficient, 
and $L$ denotes the number of extracted 
canonical variates. The MCC serves as a
robust measure of performance, with higher 
values indicating better performance, as they 
imply that the method captures and explains 
more of the underlying correlation structure.
Here cellRCCA attains an MCC of 0.951, 
outperforming rCCA which has an MCC of 0.791.
Figure~\ref{fig:CCA_scores} displays the 
canonical variables $U_{\ell}$ and $V_{\ell}$ 
for $\ell=1,2$ for both rCCA and cellRCCA. 
The canonical variables in cellRCCA are closer
to the identity line, resulting in higher 
correlations. 

\begin{figure}[ht]
\centering
 \begin{tabular}{cc}
   \large \textbf{rCCA}  & \large \textbf{cellRCCA}\\
   [-4mm]
\includegraphics[width=.30\textwidth]
  {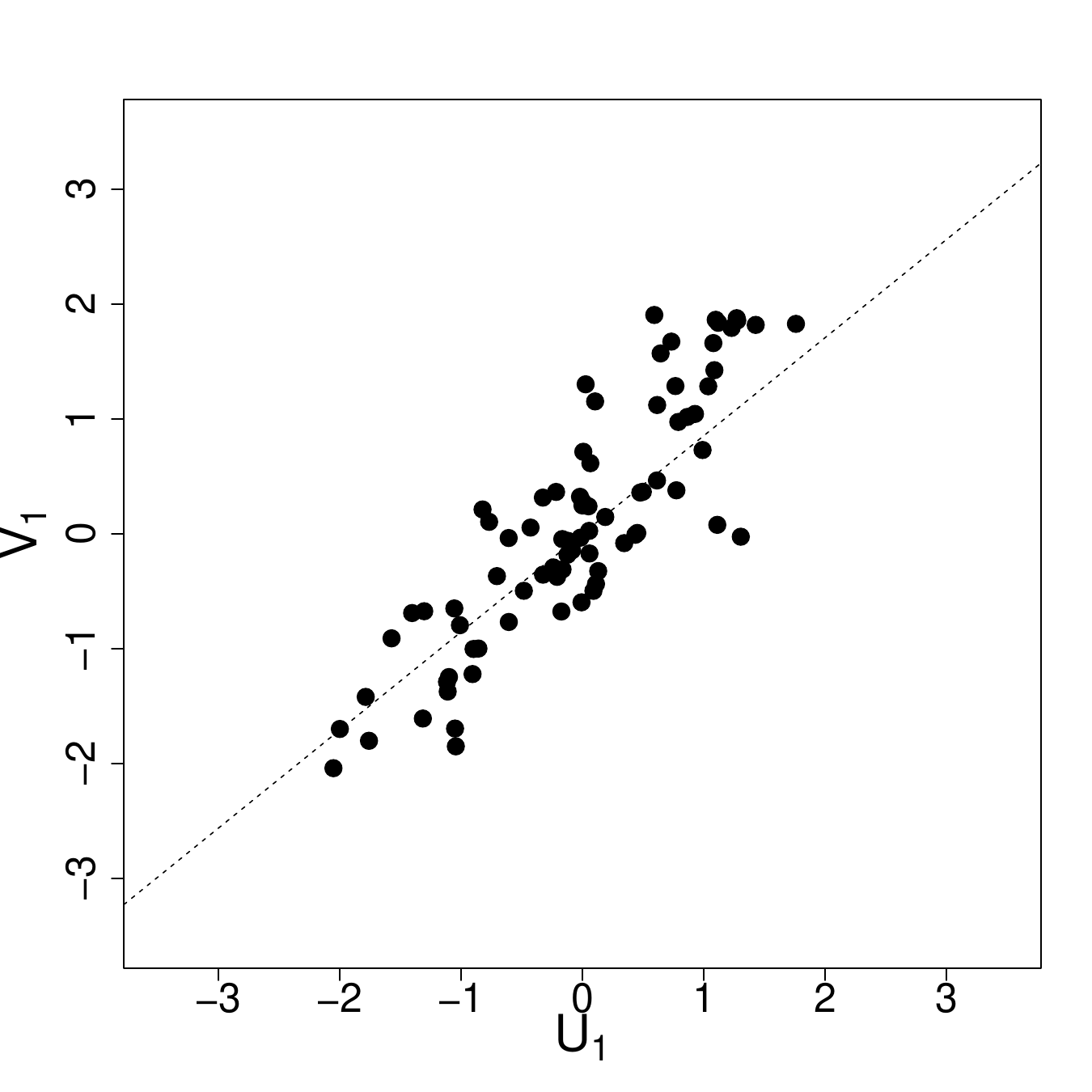} &\includegraphics[width=.30\textwidth]
 {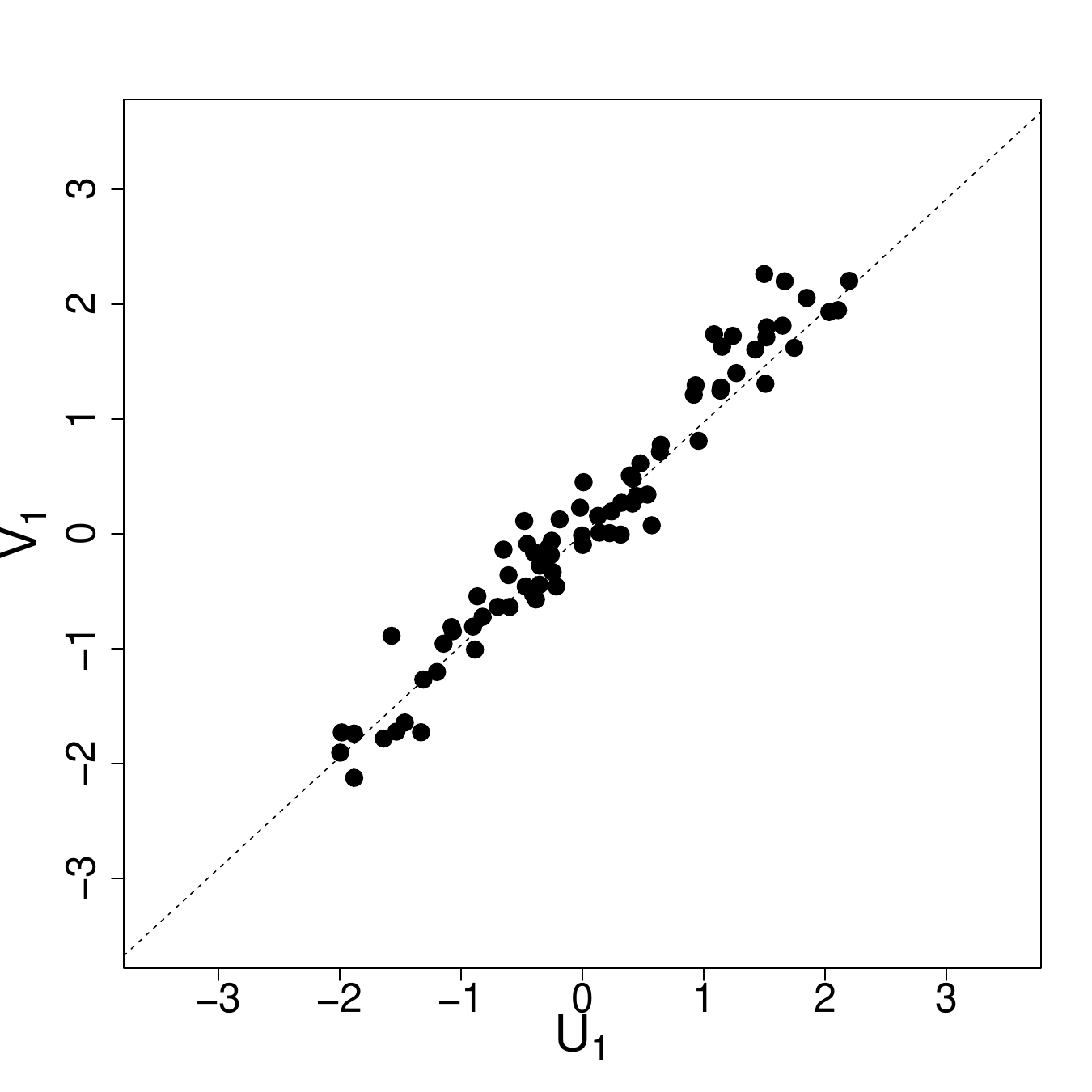} \\  [-4mm]
\includegraphics[width=.30\textwidth]
  {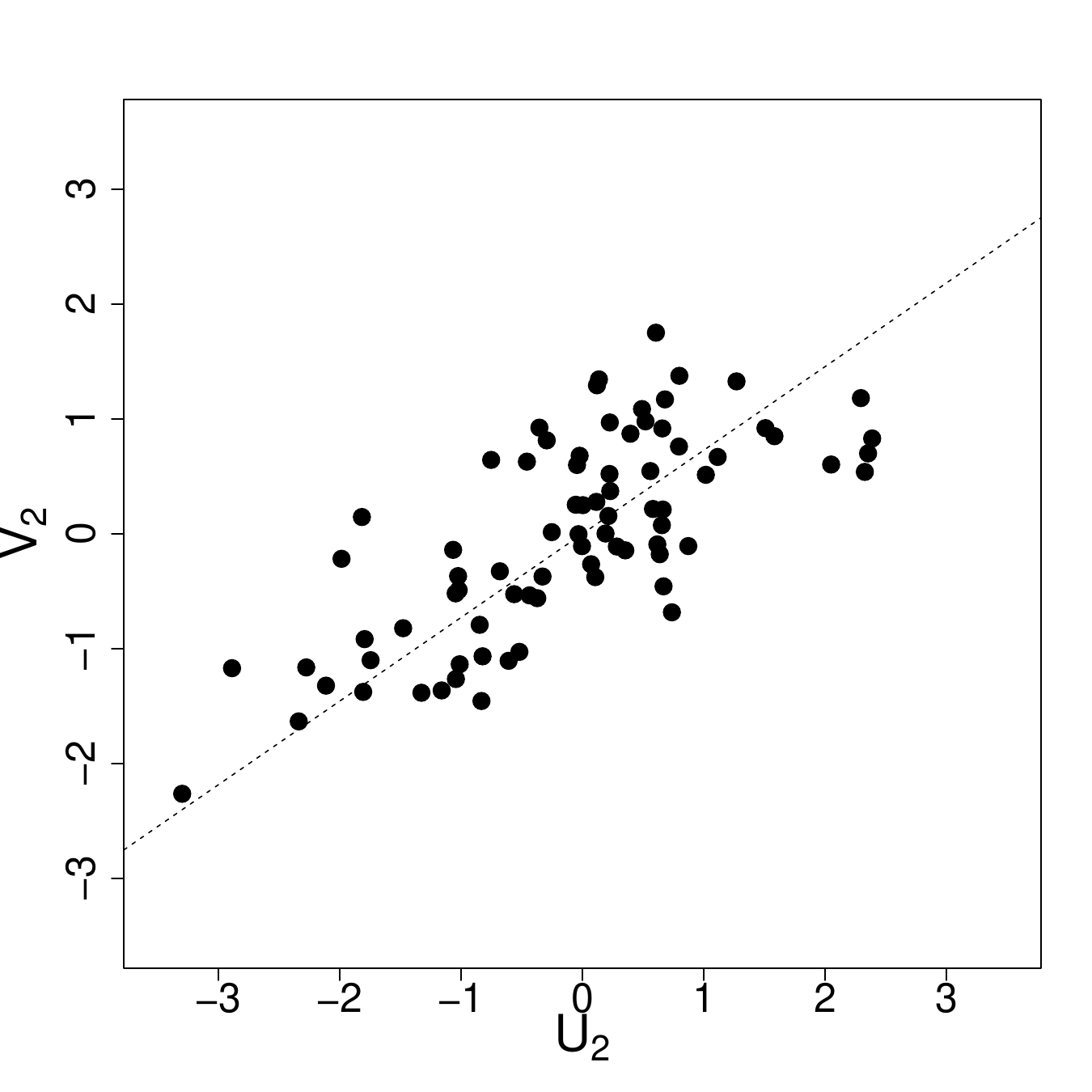} &\includegraphics[width=.30\textwidth]
 {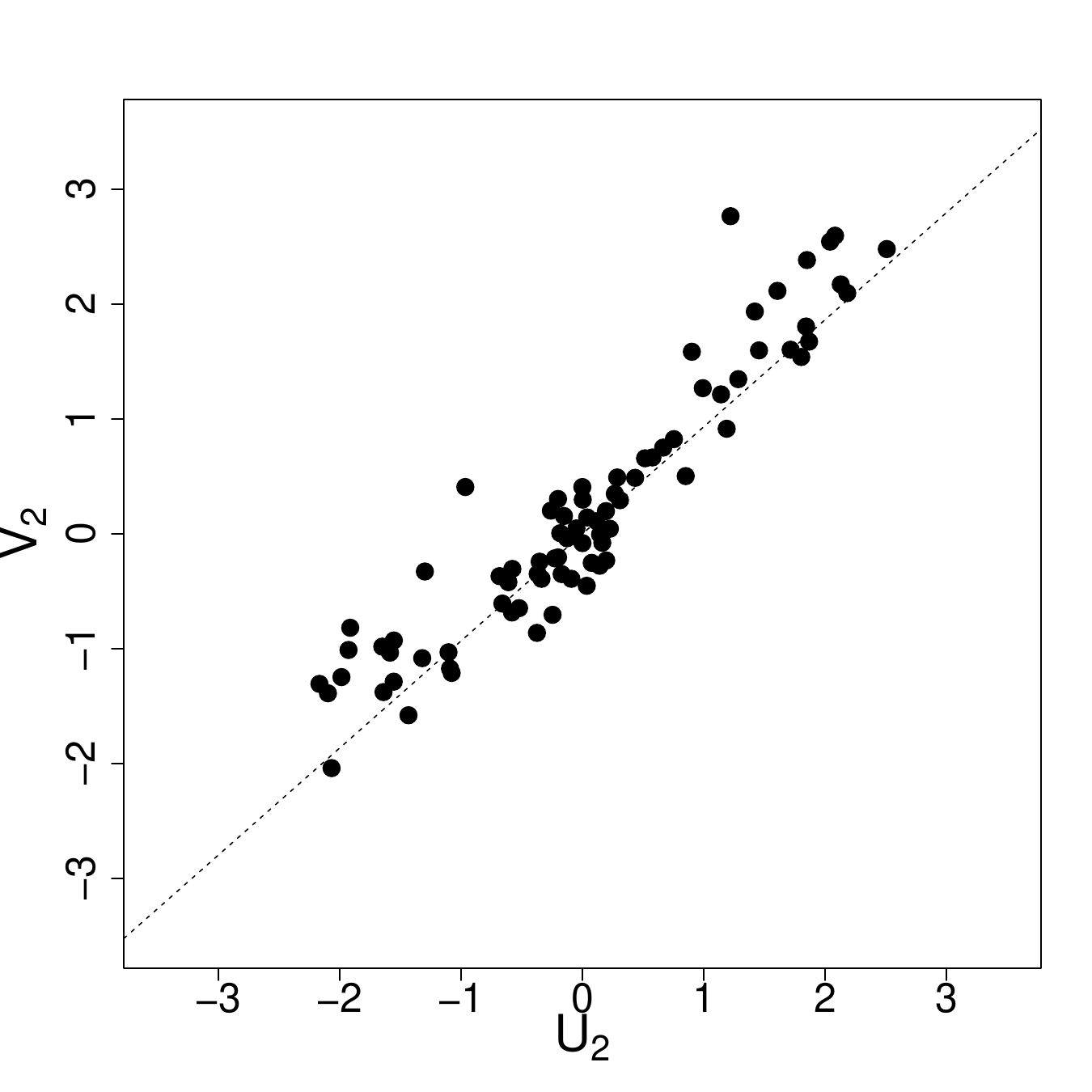} 
\end{tabular}

\vspace{-3mm}
\caption{Canonical pairs $(U_1,V_1)$ and $(U_2,V_2)$
         obtained by rCCA and cellRCCA.}
\label{fig:CCA_scores}
\end{figure}

Computing cellRCov on the combined dataset
$[\bX^{(1)}\, ;\, \bX^{(2)}]$ 
required 3.5 seconds. The computation was 
performed on a laptop equipped with an AMD 
Ryzen 9 PRO 8945HS processor with 8 cores 
and 16 threads, and 64 GB of RAM.

To further validate the performance of cellRCCA, 
we contaminated the Corn dataset with both cellwise 
and casewise outliers. The $n \times (p+q)$ 
combined data matrix was contaminated by $10\p$ 
of cellwise outliers as well as $10\p$ of casewise 
outliers. Cellwise outliers were generated by 
randomly replacing $10\p$ of the cells $z_{ij}$ 
by $\hme_j+\gamma\hs_j$\,, where  $\hme_j$ and 
$\hs_j$ are the median and M-scale~\eqref{eq:Mscale}
of variable $j$ and where $\gamma$ varies from 1 
to 10. The casewise outliers were generated from 
$N_{p+q}\left(\bhme+\gamma\bhs, 
\diag(\hs_1^2,\dots,\hs_{p+q}^2\right))$, where 
$\bhme=\left(\hme_1,\dots,\hme_{p+q}\right)^T$ and 
$\bhs=\left(\hs_1,\dots,\hs_{p+q}\right)^T$.

Figure~\ref{fig:CCA_sim} shows the MCC averaged  
over 200 replications as a function of $\gamma$, 
for rCCA and cellRCCA. It indicates that cellRCCA 
performs well across the entire range of $\gamma$.
In contrast, rCCA lacks robustness and 
deteriorates significantly for large $\gamma$. 
\begin{figure}[ht]
\centering
\includegraphics[width=.45\textwidth]
  {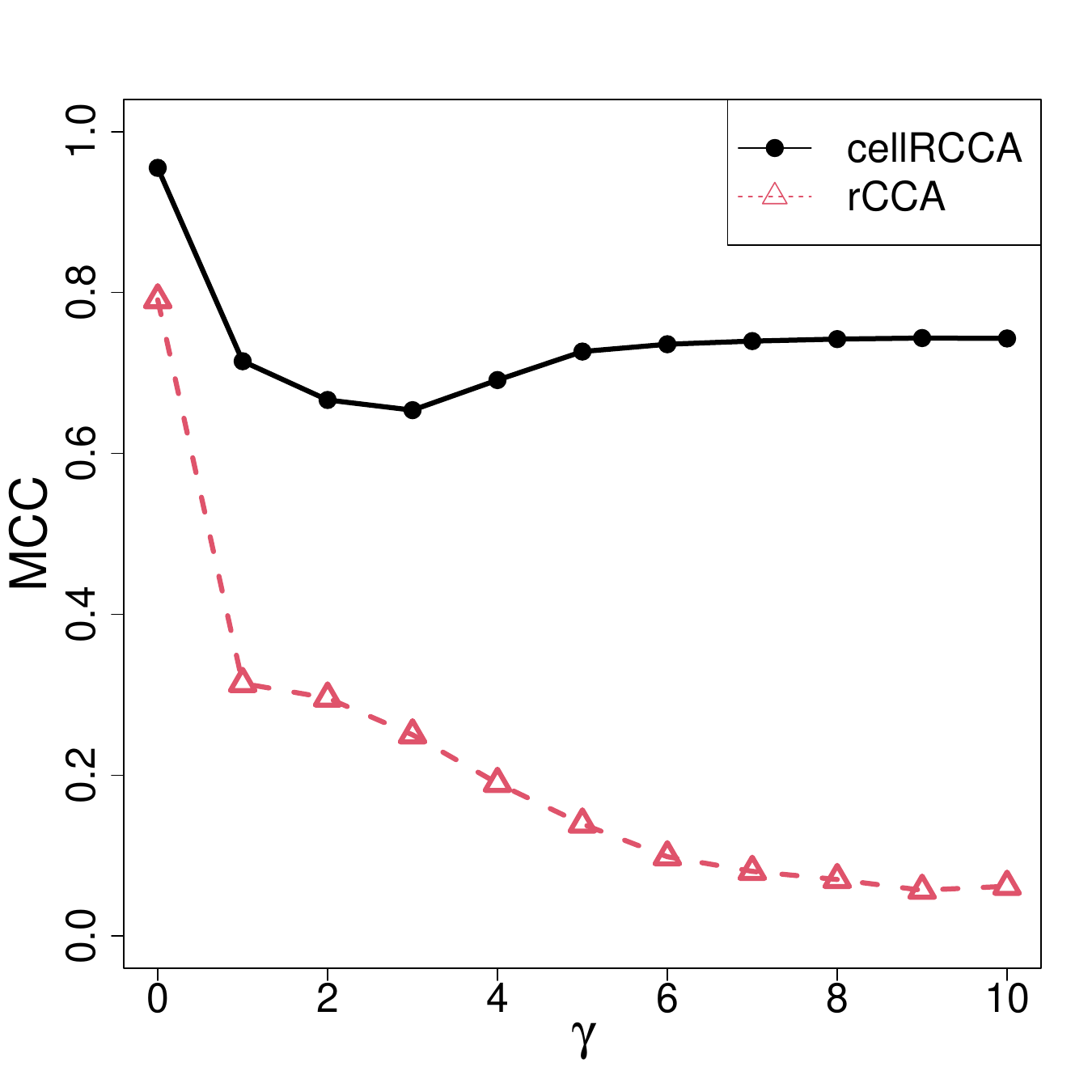} 
\caption{Average MCC attained by rCCA and 
cellRCCA as a function of $\gamma$.}
\label{fig:CCA_sim}
\end{figure}

\section{Conclusions}
\label{sec:conc}
We have introduced the Cellwise Regularized
Covariance (cellRCov) 
method, a novel robust covariance matrix estimator 
capable of simultaneously handling casewise outliers, 
cellwise outliers, and missing data. The method 
builds the covariance matrix from two components.
The first is derived from the subspace 
obtained by a recently developed robust dimension 
reduction method. The second component originates
from the orthogonal subspace, in which a weighted
covariance is computed on robustly imputed residuals.
We have included a ridge-type regularization
to enhance numerical stability, making cellRCov 
particularly well-suited for high-dimensional 
statistical analyses.

We established theoretical properties of the 
cellRCov estimator by deriving its casewise and 
cellwise influence functions, as well as proving
its consistency and asymptotic normality. 

The empirical performance of cellRCov 
was investigated by simulation, confirming
its robustness in contaminated settings with or
without missing cells. Finally, two real data 
examples illustrated practical applications to
anomaly detection and robust canonical correlation 
analysis, both in high dimensions.

In future research, cellRCov could serve as 
a building block for various multivariate 
statistical techniques, such as robust discriminant 
analysis, factor analysis, graphical models, and 
high-dimensional regression.

\noindent{\bf Software availability.} R code for 
cellRCov and a script that reproduces 
the examples are at
\url{https://wis.kuleuven.be/statdatascience/code/cellrcov_r_code.zip}\,.

\spacingset{1.0}


\clearpage
\pagenumbering{arabic}
\appendix
\begin{center}
\phantom{abc}\\ 

\Large{Supplementary Material to: 
  Cellwise and Casewise Robust Covariance
           in High Dimensions}\\
\end{center}
\vspace{17mm}

\setcounter{equation}{0} 
\renewcommand{\theequation}
  {A.\arabic{equation}} 

\spacingset{1.45} 

\section{\large Computational complexity of cellRCov}
\label{sec:supp-complexity-cellRCov}

We determine the computational complexity 
of cellRCov by going through its steps. We
assume that the selected rank 
$\rk$ is bounded above by a given maximum value 
$\rk_{\max}$\,, as in the actual implementation.

The first step is the robust marginal 
standardization of the data matrix. This requires 
computing a univariate robust scale for each 
variable, and therefore has complexity $O(np)$. 
The next step is the robust low-rank approximation
computed by cellPCA. \cite{centofanti2026robust} 
showed that the computational complexity of 
cellPCA when $\rk\leqslant \rk_{\max}$ is
$O(np(\min(n, p)+\log(n)+\log(p)))$.

After the cellPCA fit, cellRCov estimates 
the covariance of the fitted component. This 
requires applying DetMCD to the $n$ scores in 
$\rk$ dimensions, which has complexity 
$O(n\log(n)\rk^2)$ 
\citep{hubert2012deterministic}.
The resulting $\rk \times \rk$ scatter matrix 
is then mapped back to the
original $p$-dimensional space through
$\bhV\bhSigma_{\MCD}(\bhU)\bhV^T$. This matrix 
multiplication costs $O(p\rk^2+p^2\rk)$.

The residual component requires computing the 
fitted points and residuals, which costs 
$O(np\rk)$, and computing the cellwise and 
casewise weights, which costs $O(np)$. The 
weighted residual covariance matrix 
in~\eqref{eq:Sigmay} is obtained as a sum of 
$n$ weighted outer products in dimension $p$, 
and therefore has complexity $ O(np^2)$.
The ridge-type regularization of the residual 
covariance, the addition of the two covariance 
components, and the final back-transformation 
to the original scale all require 
$O(p^2)$ operations.

Combining these terms, the time complexity 
of cellRCov is
\begin{align*}
&O(np)+ O(np(\min(n,p)+\log(n)+\log(p)))
+ O(n\log(n)\rk^2)  \\
& \qquad
+ O(p\rk^2+p^2\rk)
+ O(np\rk)+O(np)
+ O(np^2)
+ O(p^2).
\end{align*}
Since $\rk$ is bounded by $\rk_{\max}$, this 
simplifies to
$O(np(\min(n,p)+\log(n)+\log(p)) + np^2)$.
The first term corresponds to the robust 
low-rank fit, whereas the second term
comes from the computation of the full residual 
covariance matrix. Thus, when $p$ is large, the
computation of the full $p\times p$ covariance 
matrix becomes one of the dominant costs.
Since $\min(n,p)\leq p$ and $\log(p)$ is $o(p)$, 
the terms $np\min(n,p)$ and $np\log(p)$ are 
dominated by $np^2$. Hence the complexity can 
be written as $O(np^2+np\log(n))$, which is not 
much higher than the $O(np^2)$ of the classical 
covariance matrix.

\section{\large Derivation of influence functions}
\label{app:proofs}

We consider the contamination 
model~\eqref{eq:cont_cell_z} given by 
\begin{equation*}
X_{\eps}=A \odot X + (\bone_p-A) \odot \bz 
\end{equation*}
where $\bz=\left(z_1, \ldots, z_p\right)^T$ 
and $A=(A_1,\dots,A_p)^T \sim G_\eps$.

Under both the dependent and independent 
contamination models with 
$P(A_j^{\cell}=1)=1-\eps^{\cell}$ for all 
$j=1,\ldots,p$, 
the distribution of $A$ satisfies
$P\left(A_j=1\right)=1-\eps$, $j=1, \dots, p$, 
and (ii) for any sequence 
$\left(j_1, j_2, \ldots, j_p\right)$ of zeroes 
and ones with $p-\ki$ ones and $\ki$ zeroes, 
$P(A_1=j_1, \ldots, A_p=j_p)$ has the same 
value, denoted as $\delta_\ki(\eps)$. Obviously,
$\eps=\eps^{\case}$ under the fully dependent 
contamination model (FDCM) and 
$\eps=\eps^{\cell}$ under the fully independent 
contamination model (FICM). Under FDCM we have 
that $P(A_1=\cdots=A_p)=1$, and then 
$\delta_0(\eps)=(1-\eps),\; 
\delta_1(\eps)=\cdots=\delta_{p-1}(\eps)=0$, 
and $\delta_p(\eps)=\eps$.
In that situation the 
distribution of $X_{\eps}$ simplifies to 
$(1-\eps)H_0 + \eps\Delta_{\bz}$ where 
$\Delta_{\bz}$ is the distribution which puts 
all of its mass in the point $\bz$. The fully 
independent contamination model (FICM) instead 
assumes that $A_1, \dots, A_p$ are independent, 
hence
$$\delta_h(\eps)=\binom{p}{h}
   (1-\eps)^{p-h} \eps^h, 
  \quad h=0,1, \ldots, p\;. $$
$G_\eps$ is denoted as $G_\eps^D$ in the dependent 
model, and as $G_\eps^I$ in the independent model.

So far influence functions of robust covariance
estimators have only been computed under the 
FDCM and for casewise robust methods, see e.g.
\cite{maronna2019robust}.
Here we will derive both the casewise and 
cellwise IFs of cellRCov.
When there are no missing values we can write 
the functional version $(\bV(H),\bmu(H))$ of 
the minimizer of \eqref{eq:objP} as
\begin{align} \label{eq:ux}
\hspace{-6mm}
  \left(\bV(H),\bmu(H)\right)&=\argmin_{\bV,\bmu}
  \E_H\left[\rho_2 \left(\frac{1}{\sigma_2(H)}
  \sqrt{\frac{1}{p}\sum_{j=1}^{p} \sigma_{1,j}^2(H)
  \rho_1\left(\frac{x_{j}-\mu_j-\bu^T\bv_j}
  {\sigma_{1,j}(H)}\right)}\right)\right]\nonumber\\
  \text{such that}\quad \bu&= \argmin_{\bu}
  \rho_2 
  \left(\frac{1}{\sigma_2(H)}\sqrt{\frac{1}{p}
  \sum_{j=1}^{p} \sigma_{1,j}^2(H)\rho_1 \left(
  \frac{x_{j}-\mu_j-\bu^T\bv_j}{\sigma_{1,j}(H)}
  \right)} \right)
\end{align}
where $\bu=\left(u_{1},\dots,u_{\rk}\right)^T$ 
and $\bx=\left(x_1,\dots,x_p\right)^T\sim H$.
Here $\sigma_{1,j}(H)$ and $\sigma_2(H)$ are 
the initial scale estimates of 
$r_{j} :=x_{j}-\mu_j- \bu^T\bv_j$ and 
$\rt :=\sqrt{\sum_{j=1}^{p} \sigma_{1,j}^2(H)
\rho_1(r_{j}/\sigma_{1,j}(H))/p}$\,. 
This is subject to the first-order conditions 
in \citep{centofanti2026robust} given by
\begin{align} 
  \label{eq:C1}
  \E_{H}\left[ \bW\bV\bu\bu^T\right]&=  
  \E_{H}\left[ \bW(\bx-\bmu) \bu^T\right],\\
\label{eq:C3}
  \E_{H}\left[\bW\bV\bu\right] &=
   \E_{H}\left[\bW 
   \left(\bx-\bmu\right)\right],\\
    \label{eq:C2} 
  \left(\bV^T\bW\bV\right)\bu &=
     \bV^T\bW(\bx-\bmu)\,.
\end{align}
Here $\bW=\diag(\bw)$ for 
    $\bw=\left(w_1,\dots,w_p\right)^T$. 
The components of $\bw$ are 
$w_j=w^{\cell}_jw^{\case}$ 
with  cellwise weights 
$w_{j}^{\cell}=\psi_1\left(\frac{r_{j}}
{\sigma_{1,j}}\right)
\Big/\frac{r_{j}}{\sigma_{1,j}}$ 
and casewise weights
$w^{\case}=\psi_2\left(\frac{\rt}{\sigma_2}
\right)\Big/\frac{\rt}{\sigma_2}$\,, where 
$\psi_1=\rho_1'$ and $\psi_2=\rho_2'$.  
We also denote 
$\btW=\diag(w^{\cell}_1,\dots,w^{\cell}_p)$.
\begin{proposition}\label{prop1}
The casewise and cellwise influence functions
of $\vect(\bV)$ and $\bmu$ are
\begin{equation}\label{eq:IFPFDCM1}
\IFu_{\case}\left(\bz,\vect(\bV),H_0\right)= -\bD_{1}\Big[\bS\IFu_{\case}\left(\bz,\bsigma,H_0\right) + 
   \bg\left(\Delta_{\bz},\vect\left(\bV_0\right),\bmu_0,\bsigma_0\right)\Big],
\end{equation}
\begin{equation}\label{eq:IFPFDCM2}
\IFu_{\case}\left(\bz,\bmu,H_0\right)= -\bD_{2}\Big[\bS\IFu_{\case}\left(\bz,\bsigma,H_0\right)+ 
   \bg\left(\Delta_{\bz},\vect\left(\bV_0\right),\bmu_0,\bsigma_0\right)\Big],
\end{equation}
and
\begin{equation}\label{eq:IFPFICM1}
\IFu_{\cell}\left(\bz,\vect(\bV),H_0\right)= -\bD_{1}\Big[\bS\IFu_{\cell}\left(\bz,\bsigma,H_0\right)+ 
   p\sum_{j=1}^p
  \bg\left(H\left(j,\bz\right),\vect\left(\bV_0\right),\bmu_0,\bsigma_0\right)\Big],
\end{equation}
\begin{equation}\label{eq:IFPFICM2}
\IFu_{\cell}\left(\bz,\bmu,H_0\right)= -\bD_{2}\Big[\bS\IFu_{\cell}\left(\bz,\bsigma,H_0\right)+ 
   p\sum_{j=1}^p
  \bg\left(H\left(j,\bz\right),\vect\left(\bV_0\right),\bmu_0,\bsigma_0\right)\Big],
\end{equation}
with $\bsigma(H)=\left(\sigma_{1,1}(H), \dots, 
\sigma_{1,p}(H), \sigma_{2}(H)\right)^T$, $\bmu_0:=\bmu(H_0)$,
$\bV_0:=\bV(H_0)$, $\bsigma_0:=\bsigma(H_0)$,
\begin{align}\label{eq:gmuL1}
  \bg_1(H,\bmu,\bV,\bsigma)&=
  \vect\left(\E_{H}\big[\bW\right(\bV\bu-\bx+\bmu\left)\bu^T\big]\right),\\\label{eq:gmuL2}
  \bg_2(H,\bmu,\bV,\bsigma)&=
 \E_{H}\big[\bW\left(\bV\bu-\bx+\bmu\right)\big],
\end{align}
and $\bg(H,\bmu,\bV,\bsigma)=\left(\bg_1(H,\bmu,\bV,
\bsigma)^T,\bg_2(H,\bmu,\bV,\bsigma)^T\right)^T$. 
The matrices $\bD_{1}$\,, $\bD_{2}$\,, and 
$\bS$ are described in the proof, and 
$\IFu_{\case}(\bz,\bsigma)$ and $\IFu_{\cell}(\bz,\bsigma)$ 
are the casewise and cellwise influence functions of 
$\bsigma$. 
\end{proposition}
Note that \eqref{eq:gmuL1} and \eqref{eq:gmuL2} 
express two of the first order conditions, but 
the other first-order condition 
$\left(\bV^T\bW\bV\right)\bu =
\bV^T\bW(\bx-\bmu)$ must hold as well, and acts 
as a constraint. Moreover, $\bg_1$ and $\bg_2$ 
depend on $\bsigma$ through $\bW$ and $\bu$\,. 
Also note that $H(j,\bz)$ in \eqref{eq:IFPFICM1} 
and \eqref{eq:IFPFICM2} is the distribution of 
$X \sim H_0$ but with its \mbox{$j$-th} 
component fixed at the constant $z_j$\,. It is 
thus a degenerate distribution concentrated on 
the hyperplane $x_j = z_j$\,.\\

The proof of Proposition \ref{prop1} closely 
follows that of Proposition 2 in 
\cite{centofanti2026robust}. However, we include 
it here as we do not adopt the simplifying 
assumption made in \cite{centofanti2026robust} 
that $\bmu$ is known and equal to $\bzero$.
The proof is based on the implicit function 
theorem, see e.g. Rio Branco
de Oliveira (2012):

\begin{lemma}[Implicit Function Theorem]
Let $\bu(x, \btheta)=(f_1, \ldots, f_p)$ be a function from $\mathbb{R} \times \mathbb{R}^p$ to $\mathbb{R}^p$ that is continuous in $(x_0,\widetilde{\btheta})\in \mathbb{R} \times \mathbb{R}^p$ with $\bu(x_0, \widetilde{\btheta})=\bzero$. Suppose the derivative of $\bu$ exists in a neighborhood $N$ of $(x_0,\widetilde{\btheta})$ and is continuous at $(x_0,\widetilde{\btheta})$,  and that the derivative matrix $\partial \bu/\partial \btheta$ is nonsingular at $(x_0, \widetilde{\btheta})$. Then there are neighborhoods $N_1$ of $x_0$ and $N_p$ of $\widetilde{\btheta}$ with $N_1\times N_p \subset N$, such that for every $x$ in $N_1$ there is a unique $\btheta=\bT(x)$ in $N_p$ for which $\bu(x, \bT(x))=\bzero$. In addition, $\bT$ is differentiable in $x_0$ with derivative matrix given by
\begin{equation*}
  \frac{\partial \bT(x)}{\partial x}\Big|_{x=x_0} 
  = -\Big(\frac{\partial \bu(x_0,\btheta)}
    {\partial \btheta}\Big|_{\btheta=
    \widetilde{\btheta}}\Big)^{-1}
    \frac{\partial \bu(x,\widetilde{\btheta)}}
    {\partial x}\Big|_{x=x_0}.
\end{equation*}
\end{lemma}

\vspace{1mm}
\begin{proof}[\textbf{Proof of Proposition \ref{prop1}}]
The IFs of $\vect(\bV)$ and $\bmu$ at the 
distribution $H_0$ are given by 
\begin{equation*}
   \IFu(\bz,\vect(\bV),G_\eps)=\left.\frac{\partial}{\partial \eps} \vect(\bV(H(G_\eps,\bz)))\right|_{\eps=0}
\end{equation*}
and
\begin{equation*}
    \IFu(\bz,\bmu,G_\eps)=\left.\frac{\partial}{\partial \eps} \bmu(H(G_\eps,\bz))\right|_{\eps=0}.
\end{equation*}
For $\eps=0$ we obtain $\bV(H(G_0,\bz)) = \bV(H_0) = \bV_0$ 
and $\bmu(H(G_0,\bz)) = \bmu(H_0) = \bmu_0$.
Note that in general $\bV(H(G_\eps,\bz))$ and 
$\bmu(H(G_\eps,\bz))$ as well as $\bV_0$ and $\bmu_0$ are
not unique. However, we will show that, given $\bV_0$ and 
$\bmu_0$ and a fixed $\eps$, there exists unique 
$\bV(H(G_\eps,\bz))$ and $\bmu(H(G_\eps,\bz))$ in 
neighborhoods of $\bV_0$ and $\bmu_0$.

Note that $\bV(H(G_\eps,\bz))$ and 
$\bmu(H(G_\eps,\bz))$ have to satisfy the
first-order conditions \eqref{eq:C1} and 
\eqref{eq:C3} saying
\begin{align}
  \label{eq:reparamcond1}
  \E_{H}\left[ \bW\right(\bV\bu-\bx+\bmu\left)\bu^T\right]
   &=  \bzero, \\
\label{eq:reparamcond2}
  \E_{H}\left[\bW\left(\bV\bu-\bx+\bmu\right)\right] 
   &=  \bzero,
\end{align}
hence
\begin{equation}\label{eq:newg1}
  \bg_1(H(G_\eps,\bz),\bT_1(\eps),\bT_2(\eps),
  \bsigma(H(G_\eps,\bz))) =
  \vect\left(\E_{H(G_\eps,\bz)}\big[\bW\right(
  \bV\bu-\bx+\bmu\left)\bu^T\big]\right)=\bzero
\end{equation}
and
\begin{equation}\label{eq:newg2}
  \bg_2(H(G_\eps,\bz),\bT_1(\eps),\bT_2(\eps),
  \bsigma(H(G_\eps,\bz))) =
 \E_{H(G_\eps,\bz)}\big[\bW\left(
 \bV\bu-\bx+\bmu\right)\big]=\bzero
\end{equation}
where the $pk$-variate and $p$-variate column vectors $\bg_1$ and $\bg_2$ are written
as a function of the $pk$-variate and $p$-variate column vectors 
$\bT_1(\eps):=\vect(\bV(H(G_\eps,\bz)))$ and $\bT_2(\eps):=\bmu(H(G_\eps,\bz))$.
Then, $\bg_1$ and $\bg_2$ are combined in the $p(q+1)$-variate column vector $\bg$ as
\begin{equation*}
    \label{eq:newg}
  \bg(H(G_\eps,\bz),\bT(H(G_\eps,\bz)),\bsigma(H(G_\eps,\bz)))=\left[\begin{array}{cc}
       \bg_1(H(G_\eps,\bz),\bT_1(\eps),\bT_2(\eps),\bsigma(H(G_\eps,\bz)))  \\
       \bg_2(H(G_\eps,\bz),\bT_1(\eps),\bT_2(\eps),\bsigma(H(G_\eps,\bz))) 
  \end{array}\right],
  \end{equation*}
  where $\bT(H(G_\eps,\bz)):=\left(\bT_1(\eps),\bT_2(\eps)\right)^T$.

In order to compute $\vect\left(\frac{\partial}{\partial \eps} \bV(H(G_\eps,\bz))\left.\right|_{\eps=0}\right)$ and $\left.\frac{\partial}{\partial \eps} \bmu(H(G_\eps,\bz))\right|_{\eps=0}$ we would like to apply the implicit function theorem in the point $\eps = 0$, but the contaminated distribution $H(G_\eps,\bz)$ is only defined for $\eps > 0$. To circumvent this issue we extend the definition of $\bg$ to negative $\eps$ by defining a function
$\bh$ from $\mathbb{R} \times \mathbb{R}^{pk}$ to $\mathbb{R}^{pk}$ as 
\begin{equation*}
\bh(\eps, \btheta)= \begin{cases}
\bg(H(G_\eps,\bz),\btheta,\bsigma(H(G_\eps,\bz)))
& \mbox{for } \eps \geqslant 0 \\ 
2 \bg(H_0,\btheta,\bsigma(H_0))-\bg(H(G_{|\eps|},\bz),\btheta,\bsigma(H(G_{|\eps|},\bz))) & \mbox{for } \eps<0.\end{cases}
\end{equation*}
We now put $\eps_0=0$ and $\widetilde{\btheta}=\bT(0)=\left(\bT_1(0),\bT_2(0)\right)^T=\left(\vect(\bV_0),\bmu_0\right)^T$. 
Then $\bh(\eps_0,\btheta)=\bh(0,\bT(0))=\bzero$, and assuming that $\bg$ is sufficiently smooth for the differentiability requirements of the implicit function theorem, we can conclude that $\bT(H(G_\eps,\bz))$ and, thus, $\bV(H(G_\eps,\bz))$ and $\bmu(H(G_\eps,\bz))$ are uniquely defined for small $\eps$ in a neighborhood of $\bV_0$ and $\bmu_0$ and that
\begin{align} \label{eq:der}
  \frac{\partial \bT(H(G_\eps,\bz))}{\partial \eps}\Big|_{\eps=0} 
  &= -\Big(\frac{\partial} {\partial \bT} \bh(0,\bT)
   \Big|_{\bT=\bT(0)}\Big)^{-1}
    \;\frac{\partial} {\partial \eps} \bh(\eps,\bT(0))\Big|_{\eps=0}\nonumber\\
  &= -\Big(\frac{\partial} {\partial \bT} \bg(H_0,\bT,\bsigma_0)\Big|_{\bT=\bT(0)}\Big)^{-1}\\
  &\;\;\;\;\; \frac{\partial}{\partial \eps} \bg(H(G_\eps,\bz),\bT\left(0\right),
\bsigma(H(G_\eps,\bz)))\Big|_{\eps=0}\;.
\end{align}
Equivalently
\begin{gather*}
    \left(\begin{array}{cc}
         \frac{\partial} {\partial \eps} \bT_1(\eps)\Big|_{\eps=0}\\ 
          \frac{\partial} {\partial \eps} \bT_2(\eps)\Big|_{\eps=0}
    \end{array}\right)=\\
    \left[\begin{array}{cc}
     \frac{\partial} {\partial \bT_1}\bg_1(H_0,\bT_1,\bT_2(0),\bsigma_0)\Big|_{\bT_1=\bT_1(0)}  & \frac{\partial} {\partial \bT_2}\bg_1(H_0,\bT_1(0),\bT_2,\bsigma_0)\Big|_{\bT_2=\bT_2(0)}  \\
      \frac{\partial} {\partial \bT_1}\bg_2(H_0,\bT_1,\bT_2(0),\bsigma_0)\Big|_{\bT_1=\bT_1(0)}   & \frac{\partial} {\partial \bT_2}\bg_2(H_0,\bT_1(0),\bT_2,\bsigma_0)\Big|_{\bT_2=\bT_2(0)}
    \end{array}\right]^{-1}\\\left(\begin{array}{cc}
        \frac{\partial} {\partial \eps} \bg_1(H(G_\eps,\bz),\bT_1\left(0\right),\bT_2\left(0\right),
\bsigma(H(G_\eps,\bz)))\Big|_{\eps=0}\\
           \frac{\partial} {\partial \eps} \bg_2(H(G_\eps,\bz),\bT_1\left(0\right),\bT_2\left(0\right),
\bsigma(H(G_\eps,\bz)))\Big|_{\eps=0}          
    \end{array}\right).
\end{gather*}
Note that the left hand side is $\left(\vect\left(\frac{\partial}{\partial \eps} \bV(H(G_\eps,\bz))\left.\right|_{\eps=0}\right),\left.\frac{\partial}{\partial \eps} \bmu(H(G_\eps,\bz))\right|_{\eps=0}\right)^T$.
We now have to work out the right-hand side. For 
the first factor we denote the matrix
\begin{equation*}
 \begin{bmatrix}
      \bD_{1} \\
      \bD_{2}
 \end{bmatrix}:= \left(\frac{\partial} {\partial \bT} \bg(H_0,\bT,\sigma_0)\Big|_{\bT=\bT\left(0\right)}\right)^{-1}, 
\end{equation*}
which does not depend on $\bz$ and can be computed numerically. For the second factor, 
from \eqref{eq:newg1} and \eqref{eq:newg2}
we know that $\bg_1(H(G_\eps,\bz),\bT_1\left(0\right),\bT_2\left(0\right),
\bsigma(H(G_\eps,\bz)))$ and $\bg_2(H(G_\eps,\bz),\bT_1\left(0\right),\bT_2\left(0\right),
\bsigma(H(G_\eps,\bz)))$
are expectations over the mixture distribution
$H(G_\eps,\bz)$, so we can write them as linear 
combinations with coefficients $\delta_h(\eps)$. 
For the FDCM model we know that $G_\eps^D$ has 
$\delta_0(\eps)=(1-\eps),\; 
\delta_1(\eps)=\cdots=\delta_{p-1}(\eps)=0$ and 
$\delta_p(\eps)=\eps$, so $\bg_1$ can be written as
\begin{align*}
\bg_1&(H(G_\eps^D,\bz),\vect(\bV_0),\bmu_0,
\bsigma(H(G_\eps^D,\bz)))\\
&= \delta_0(\eps) \bg_1(H_0,\vect(\bV_0),\bmu_0,
\bsigma(H(G_\eps^D,\bz)))\\
&\;\;\; + 
\delta_p(\eps) \bg_1(H(\{1, \dots, p\},\bz),\vect(\bV_0),\bmu_0,
\bsigma(H(G_\eps^D,\bz)))\\
&= (1-\eps)\bg_1(H_0,\vect(\bV_0),\bmu_0,\bsigma(H(G_\eps^D,\bz))) + 
\eps\bg_1(\Delta_{\bz},\vect(\bV_0),\bmu_0,\bsigma(H(G_\eps^D,\bz))),
\end{align*}
which yields the derivative
\begin{align}
\frac{\partial}{\partial \eps}& \bg_1(H(G_\eps^D,\bz),\bT_1\left(0\right),\bT_2\left(0\right),
\bsigma(H(G_\eps^D,\bz)))\Big|_{\eps=0}  \nonumber\\
  &=-\bg_1(H_0,\vect(\bV_0),\bmu_0,\bsigma_0) 
  + \frac{\partial }{\partial \eps}\bg_1(H_0,
  \vect(\bV_0),\bmu_0,\bsigma(H(G_\eps^D,\bz)))
  \Big|_{\eps=0} \nonumber\\ 
  &\;\;\;\;\; + \bg_1(\Delta_{\bz},\vect(\bV_0),\bmu_0,\bsigma_0) \nonumber\\
  &= \bzero + \frac{\partial \bg_1}{\partial \bsigma} 
  (H_0, \vect(\bV_0),\bmu_0, \bsigma)\Big|_{\bsigma=\bsigma_0}
   \frac{\partial \bsigma(H(G_\eps^D,\bz))}{\partial \eps} 
   \Big|_{\eps=0} + 
   \bg_1(\Delta_{\bz},\vect(\bV_0),\bmu_0,
   \bsigma_0) \nonumber\\
   \label{eq:IFcaseV1}
  &= \bS_1\IFu_{\case}(\bz,\bsigma)+
  \bg_1(\Delta_{\bz},\vect(\bV_0),\bmu_0,\bsigma_0)
\end{align}
where $\bS_1:=\left.\frac{\partial}{\partial \bsigma} 
\bg_1(H_0, \vect(\bV_0),\bmu_0, \bsigma)\right|_{\bsigma=\bsigma_0}$ 
and  $\IFu_{\case}(\bz,\bsigma)$ is the influence
function of $\bsigma$ under FDCM. 
Analogously 
\begin{align}
\frac{\partial}{\partial \eps} \bg_2(H(G_\eps^D,\bz),\bT_1\left(0\right)&,\bT_2\left(0\right),
\bsigma(H(G_\eps^D,\bz)))\Big|_{\eps=0}  \nonumber\\
   \label{eq:IFcaseV2}
  &=\bS_2\IFu_{\case}(\bz,\bsigma)+
  \bg_2(\Delta_{\bz},\vect(\bV_0),\bmu_0,\bsigma_0),
\end{align}
where $\bS_2:=\left.\frac{\partial}{\partial \bsigma} 
\bg_2(H_0, \vect(\bV_0),\bmu_0, \bsigma)\right|_{\bsigma=\bsigma_0}$. 
Under the FICM model, the second factor is different. 
We have $\delta_0(\eps)=(1-\eps)^p$, $\delta_0(0)=1$, 
$\delta_1(\eps) = p (1-\eps)^{p-1}\eps$ so $\delta_1(0)=0$ 
and $\delta_1^{\prime}(0)=p$, and $\delta_i(0)=
\delta_i^{\prime}(0)=0$ for $i \geqslant 2$. Therefore 
$\bg_1$ can be written as the sum
\begin{align*}
\bg_1&(H(G_\eps^I,\bz),\vect(\bV_0),\bmu_0,
\bsigma(H(G_\eps^I,\bz)))\\
&= \delta_0(\eps) \bg_1(H_0,\vect(\bV_0),\bmu_0,
\bsigma(H(G_\eps^I,\bz)))
+ \delta_1(\eps) \sum_{j=1}^p 
\bg_1(H(j,\bz),\vect(\bV_0),\bmu_0,
\bsigma(H(G_\eps^I,\bz)))\\
&= (1-\eps)^p\bg_1(H_0,\vect(\bV_0),
   \bmu_0,\bsigma(H(G_\eps^I,\bz)))\\
&\;\;\;\;\; +p (1-\eps)^{p-1}\eps\sum_{j=1}^p 
    \bg_1(H(j,\bz),\vect(\bV_0),\bmu_0,
     \bsigma(H(G_\eps^I,\bz)))
\end{align*}
where $H(j,\bz)$ is the distribution of 
$X \sim H_0$ but with its $j$-th component fixed 
at the constant $z_j$\,. It is thus a degenerate 
distribution concentrated on the hyperplane 
$x_j = z_j$\,. 
The derivative now becomes
\begin{align}
\frac{\partial}{\partial \eps} \bg_1(H(G_\eps^I,\bz),\bT_1\left(0\right)&,\bT_2\left(0\right),
\bsigma(H(G_\eps^I,\bz)))\Big|_{\eps=0}  \nonumber\\
   \label{eq:IFcellV1}
  &= \bS_1\IFu_{\cell}(\bz,\bsigma)+p\sum_{j=1}^p 
  \bg_1(H\left(j,\bz\right),\vect(\bV_0),\bmu_0,\bsigma_0),
\end{align}
where $\bS$ is the same as before but 
$\IFu_{\cell}(\bz,\bsigma)$ is now the cellwise 
influence function of $\bsigma$. Similarly,
\begin{align}
\frac{\partial}{\partial \eps} \bg_2(H(G_\eps^I,\bz),\bT_1\left(0\right)&,\bT_2\left(0\right),
\bsigma(H(G_\eps^I,\bz)))\Big|_{\eps=0}  \nonumber\\
   \label{eq:IFcellV2}
  &= \bS_2\IFu_{\cell}(\bz,\bsigma)+p\sum_{j=1}^p 
  \bg_2(H\left(j,\bz\right),\vect(\bV_0),\bmu_0,\bsigma_0).
\end{align}
Let us define $\bS=\begin{bmatrix}
      \bS_{1}^T &
      \bS_{2}^T
 \end{bmatrix}^T$.
Combining \eqref{eq:IFcaseV1}, \eqref{eq:IFcaseV2}, 
\eqref{eq:IFcellV1} and \eqref{eq:IFcellV2} in 
\eqref{eq:der} yields \eqref{eq:IFPFDCM1}, 
\eqref{eq:IFPFDCM2}, \eqref{eq:IFPFICM1}, and 
\eqref{eq:IFPFICM2}.
\end{proof}

To derive the IF of cellRCov, another important 
piece is the IF of the functional  
$\bSigma^{\bu}_{\MCD}(H)$  corresponding to the
MCD estimator of covariance $\bSigma_{\MCD}$ 
applied to $\bhU$ under both FDCM and FICM.  
That is, $\bSigma^{\bu}_{\MCD}(H)$ corresponds 
to the MCD functional $\bSigma_{\MCD}(\cdot)$ 
of scatter with parameter $0.5<\alpha<1$, applied 
to the distribution $H^{\bu}(H, \bT(H))$ of 
$\bu$ when $\bx$ is distributed as $H$, for 
$\bT(H):=\left(\vect(\bV(H)), \bmu(H),
\bsigma(H)\right)^T$. We also define 
$\bmu^{\bu}_{\MCD}(H)$ as the MCD functional 
$\bmu_{\MCD}(\cdot)$ of location with parameter 
$\alpha$ applied to $H^{\bu}(H, \bT(H))$.

\begin{proposition}\label{prop2}
The casewise and cellwise influence functions 
of $\vect(\bSigma^{\bu}_{\MCD})$ are
\begin{align*}
\IFu_{\case}(\bz,\vect(\bSigma^{\bu}_{\MCD}),H_0)= 
  & -\bD^{\bu}_{1}\Big[\bS^{\bu}  
  \IFu_{\case}\left(\bz,\bsigma,H_0\right) 
  + \bB^{\bu}_{1} \IFu_{\case}
    \left(\bz,\vect(\bV),H_0\right)\\
 &\;\; + \bB^{\bu}_{2} \IFu_{\case}\left(
   \bz,\bmu,H_0\right)+\bg^{\bu}(\Delta_{\bz}, 
   \bT_0,\bSigma^{\bu}_{\MCD,0},
   \bmu^{\bu}_{\MCD,0},q_{\alpha,0})\Big]
\end{align*}
and
\begin{align*}
  \IFu_{\cell}(\bz,\vect(\bSigma^{\bu}_{\MCD}),H_0)= 
  & -\bD^{\bu}_{1}\Big[\bS^{\bu} 
    \IFu_{\cell}\left(\bz,\bsigma,H_0\right) +
    \bB^{\bu}_{1} \IFu_{\cell}
    \left(\bz,\vect(\bV),H_0\right)\\
  &\;\; +
    \bB^{\bu}_{2} \IFu_{\cell}\left(\bz,\bmu,H_0\right)\\
  &\;\; +
    p\sum_{j=1}^p\bg^{\bu}(H(j,\bz), 
    \bT_0,\bSigma^{\bu}_{\MCD,0},
    \bmu^{\bu}_{\MCD,0},q_{\alpha,0}) \Big]\,,
\end{align*}
with $\bsigma(H)=\left(\sigma_{1,1}(H), \dots, 
\sigma_{1,p}(H), \sigma_{2}(H)\right)^T$, 
$\bmu^{\bu}_{\MCD,0}:=\bmu^{\bu}_{\MCD}(H_0)$, 
$\bSigma^{\bu}_{\MCD,0}:=\bSigma^{\bu}_{\MCD}(H_0)$,
$\bT_0:=\left(\vect(\bV_0), \bmu_0,\bsigma_0\right)^T$ and
\begin{equation}
    \label{eq:newgf2}
  \bg^{\bu}(H,\bT,\bSigma^{\bu}_{\MCD},\bmu^{\bu}_{\MCD},q_{\alpha})=\left[\begin{array}{cc}
       \bg^{\bu}_1(H,\bT,\bSigma^{\bu}_{\MCD},\bmu^{\bu}_{\MCD},q_{\alpha}) \\
       \bg^{\bu}_2(H,\bT,\bSigma^{\bu}_{\MCD},\bmu^{\bu}_{\MCD},q_{\alpha}) \\
       g^{\bu}_3(H,\bT,\bSigma^{\bu}_{\MCD},\bmu^{\bu}_{\MCD},q_{\alpha})
  \end{array}\right].
\end{equation}
The functions $\bg^{\bu}_1$, $\bg^{\bu}_2$, and 
$g^{\bu}_3$ are defined as
\begin{align*}\label{eq:newgf1}
    \bg^{\bu}_1(H,\bT,\bSigma^{\bu}_{\MCD},\bmu^{\bu}_{\MCD},q_{\alpha})
  &=\vect\left(\E_{H^{\bu}(H, \bT)}\right.\Big[\ind\left(
    {\bx \in A(\bSigma^{\bu}_{\MCD},\bmu^{\bu}_{\MCD},q_{\alpha})}\right)\\
  &\;\;\;\;\; \left.\left(c_{\alpha}\left(\bx-\bmu^{\bu}_{\MCD}\right)\left(\bx-\bmu^{\bu}_{\MCD}
  \right)^T-\bSigma^{\bu}_{\MCD}\right)\Big]\right)\,,
\end{align*}
\begin{equation*}\label{eq:newgf2.b}
    \bg^{\bu}_2(H,\bT,\bSigma^{\bu}_{\MCD},\bmu^{\bu}_{\MCD},q_{\alpha})\nonumber=\E_{H^{\bu}(H, \bT)}\Big[\ind\left(
    {\bx \in A(\bSigma^{\bu}_{\MCD},\bmu^{\bu}_{\MCD},q_{\alpha})}\right)\left(\bx-\bmu^{\bu}_{\MCD}\right)\Big],
\end{equation*}
and
\begin{equation*}\label{eq:newgf2.c}
    g^{\bu}_3(H,\bT,\bSigma^{\bu}_{\MCD},\bmu^{\bu}_{\MCD},q_{\alpha})=\E_{H^{\bu}(H, \bT)}\Big[\ind\left(
    {\bx \in A(\bSigma^{\bu}_{\MCD},\bmu^{\bu}_{\MCD},q_{\alpha})}\right)-(1-\alpha)\Big]
\end{equation*}
where 
\begin{align*}
   A(\bSigma^{\bu}_{\MCD},\bmu^{\bu}_{\MCD},q_{\alpha})=\lbrace \bx \in \mathbb{R}^k:\left(\bx-\bmu^{\bu}_{\MCD}\right)^T\left(\bSigma^{\bu}_{\MCD}\right)^{-1}\left(\bx-\bmu^{\bu}_{\MCD}\right)\leqslant q_{\alpha}\rbrace.
\end{align*}
The functional $q_{\alpha}(H)$  satisfies
\begin{equation*}
    \int\ind\left(
    {\bx \in A( \bSigma^{\bu}_{\MCD}(H),\bmu^{\bu}_{\MCD}(H),q_{\alpha}(H))}\right)dH^{\bu}(H, \bT)(\bx)=1-\alpha,
\end{equation*}
in which $q_{\alpha,0}:=q_{\alpha}(H_0)$ and
$c_{\alpha}$ is chosen in such a way that 
consistency is obtained at a prespecified model. 
The matrices $\bB^{\bu}_{1}$, 
$\bB^{\bu}_{2}$, and $\bD^{\bu}_{1}$ are 
defined in the proof, and $H(j,\bz)$ is the distribution of 
$X \sim H_0$ but with its $j$-th component fixed 
at the constant $z_j$\,.
\end{proposition}

\begin{proof}[\textbf{Proof of Proposition \ref{prop2}}]
Consider the distribution $H(G_\eps,\bz)$. The IFs of 
$\vect(\bSigma^{\bu}_{\MCD})$ and $\bmu^{\bu}_{\MCD}$ 
are given by 
\begin{equation*}
  \IFu(\bz,\vect(\bSigma^{\bu}_{\MCD}),G_\eps)=
  \left.\frac{\partial}{\partial \eps}
  \vect(\bSigma^{\bu}_{\MCD}(H(G_\eps,\bz)))
  \right|_{\eps=0}
\end{equation*}
and
\begin{equation*}
  \IFu(\bz,\bmu^{\bu}_{MCD},G_\eps)=
  \left.\frac{\partial}{\partial \eps} 
  \bmu^{\bu}_{\MCD}(H(G_\eps,\bz))\right|_{\eps=0}.
\end{equation*}
For $\eps=0$ we obtain $\bSigma^{\bu}_{\MCD}(H(G_0,\bz)) = \bSigma^{\bu}_{\MCD}(H_0) = \bSigma^{\bu}_{\MCD,0}$ and $\bmu^{\bu}_{\MCD}(H(G_0,\bz)) = \bmu^{\bu}_{\MCD}(H_0) = \bmu^{\bu}_{\MCD,0}$.
  
The functionals $\bmu^{\bu}_{\MCD}(H)$ and $\bSigma^{\bu}_{\MCD}(H)$ with parameter $0.5<\alpha<1$ correspond to the MCD functionals 
$\bmu_{\MCD}(\cdot)$ and $\bSigma_{\MCD}(\cdot)$ 
of location and scatter applied to the the distribution $H^{\bu}(H, \bT(H))$ of $\bu$ when $\bx$ is distributed as $H$, for $\bT(H):=\left(\vect(\bV(H)), \bmu(H),\bsigma(H)\right)^T$. That is, $\bmu^{\bu}_{\MCD}(H)=\bmu_{\MCD}(H^{\bu}(H, \bT(H)))$ and $\bSigma^{\bu}_{\MCD}(H)=\bSigma_{\MCD}(H^{\bu}(H, \bT(H)))$.
Therefore $\bmu^{\bu}_{\MCD}(H)$ and $\bSigma^{\bu}_{\MCD}(H)$
are given by
\begin{align}\label{eq:funcfsigma}
    \bSigma^{\bu}_{\MCD}(H)=\frac{c_{\alpha}}{1-\alpha}\int&\ind\left(
    {\bx \in A(\bSigma^{\bu}_{\MCD}(H),\bmu^{\bu}_{\MCD}(H),q_{\alpha}(H))}\right)&\nonumber\\&\left(\bx-\bmu^{\bu}_{\MCD}(H)\right)\left(\bx-\bmu^{\bu}_{\MCD}(H)\right)^TdH^{\bu}(H, \bT(H))(\bx),&
\end{align}
\begin{align}\label{eq:funcfmu}
    \bmu^{\bu}_{\MCD}(H)=\frac{1}{1-\alpha}\int \ind\left(
    {\bx \in A(\bSigma^{\bu}_{\MCD}(H),\bmu^{\bu}_{\MCD}(H),q_{\alpha}(H))}\right)\bx dH^{\bu}(H, \bT(H))(\bx),&
\end{align}
with $q_{\alpha}(H)$ that satisfies
\begin{equation*}
    \int\ind\left(
    {\bx \in A( \bSigma^{\bu}_{\MCD}(H),\bmu^{\bu}_{\MCD}(H),q_{\alpha}(H))}\right)dH^{\bu}(H, \bT)(\bx)=1-\alpha,
\end{equation*}
where 
\begin{align*}
   A(\bSigma^{\bu}_{\MCD},\bmu^{\bu}_{\MCD},q_{\alpha})=\lbrace \bx \in \mathbb{R}^k:\left(\bx-\bmu^{\bu}_{\MCD}\right)^T\left(\bSigma^{\bu}_{\MCD}\right)^{-1}\left(\bx-\bmu^{\bu}_{\MCD}\right)\leqslant q_{\alpha}\rbrace,
\end{align*}
and $c_{\alpha}$ is chosen in such a way that 
consistency is obtained at a prespecified model. In 
order to use these expressions of the MCD functionals 
we have assumed that $H^{\bu}(H, \bT(H))$ has 
no mass on the boundary of $A$ 
\citep{cator2012central}. Then
\begin{align}\label{eq:newgf1}
    \bg^{\bu}_1&(H,\bT,\bSigma^{\bu}_{\MCD},\bmu^{\bu}_{\MCD},q_{\alpha})&\nonumber\\
  &=\vect\left(\E_{H^{\bu}(H, \bT)}\right.\Big[\ind\left(
    {\bx \in A(\bSigma^{\bu}_{\MCD},\bmu^{\bu}_{\MCD},q_{\alpha})}\right)\nonumber\\&\hspace{4cm}\left.\left(c_{\alpha}\left(\bx-\bmu^{\bu}_{\MCD}\right)\left(\bx-\bmu^{\bu}_{\MCD}\right)^T-\bSigma^{\bu}_{\MCD}\right)\Big]\right),&
    \end{align}
    \begin{align}\label{eq:newgf2.d}
    \bg^{\bu}_2&(H,\bT,\bSigma^{\bu}_{\MCD},\bmu^{\bu}_{\MCD},q_{\alpha})&\nonumber\\&=\E_{H^{\bu}(H, \bT)}\Big[\ind\left(
    {\bx \in A(\bSigma^{\bu}_{\MCD},\bmu^{\bu}_{\MCD},q_{\alpha})}\right)\left(\bx-\bmu^{\bu}_{\MCD}\right)\Big],&
\end{align}
and
\begin{align}\label{eq:newgf2.e}
    g^{\bu}_3&(H,\bT,\bSigma^{\bu}_{\MCD},\bmu^{\bu}_{\MCD},q_{\alpha})&\nonumber\\&=\E_{H^{\bu}(H, \bT)}\Big[\ind\left(
    {\bx \in A(\bSigma^{\bu}_{\MCD},\bmu^{\bu}_{\MCD},q_{\alpha})}\right)-(1-\alpha)\Big].&
\end{align}
We combine $\bg^{\bu}_1$, $\bg^{\bu}_2$ and $g^{\bu}_3$ 
in the $(k(k+1)+1)$-variate column vector $\bg^{\bu}$ as
\begin{equation}\label{eq:newgf2.f}
  \bg^{\bu}(H,\bT,\bSigma^{\bu}_{\MCD},\bmu^{\bu}_{\MCD},q_{\alpha})=\left[\begin{array}{cc}
       \bg^{\bu}_1(H,\bT,\bSigma^{\bu}_{\MCD},\bmu^{\bu}_{\MCD},q_{\alpha}) \\
       \bg^{\bu}_2(H,\bT,\bSigma^{\bu}_{\MCD},\bmu^{\bu}_{\MCD},q_{\alpha}) \\
       g^{\bu}_3(H,\bT,\bSigma^{\bu}_{\MCD},\bmu^{\bu}_{\MCD},q_{\alpha})
  \end{array}\right]=\bzero.
\end{equation}
  
The IFs of $\vect(\bSigma^{\bu}_{\MCD})$ and $\bmu^{\bu}_{\MCD}$ under the FDCM  and FICM are obtained by differentiating \eqref{eq:newgf2} in $\eps=0$  when $H$ coincides with $H(G_\eps^D,\bz)$ and $H(G_\eps^I,\bz)$, respectively, and $\bT=\bT(H)$.
Note that from \eqref{eq:cont_cell_z}, $H(G_\eps,\bz)$ propagates in the distribution $H^{\bu}(H(G_\eps,\bz), \bT(H(G_\eps,\bz)))$ as
\begin{equation}
\label{eq:contmod2f}
\delta_0(\eps)H^{\bu}(H_0, \bT(H(G_\eps,\bz)))+\sum_{j=1}^p\delta_j(\eps)\sum_{I\in\ell_j}H^{\bu}(H(I,\bz),\bT(H(G_\eps,\bz))),
\end{equation}
where $\ell_j=\lbrace I=\lbrace i_1, \dots, i_j\rbrace: i_1<\cdots<i_j\rbrace$.
Consider the distribution $H(G_\eps^D,\bz)$. Then we
know from \eqref{eq:newg1} and \eqref{eq:newg2} that 
$$\bg^{\bu}(H(G_\eps,\bz), \bT(H(G_\eps,\bz)),
\bSigma^{\bu}_{\MCD}(H(G_\eps,\bz)),\bmu^{\bu}_{\MCD}(H(G_\eps,\bz)),q_{\alpha}(H(G_\eps,\bz)))$$ 
is an expectation over the mixture distribution
$H^{\bu}(H(G_\eps,\bz), \bT(H(G_\eps,\bz)))$, so we 
can write it as a linear combination with coefficients 
$\delta_h(\eps)$ for $h=0,1,\ldots,p$\,.

For the FDCM model we know that $G_\eps^D$ has 
$\delta_0(\eps)=(1-\eps),\; 
\delta_1(\eps)=\cdots=\delta_{p-1}(\eps)=0$ and 
$\delta_p(\eps)=\eps$, so $\bg^{\bu}$ can be written as
\begin{align*}
\bg^{\bu}&(H^D_{\eps,\bz}, \bT(H^D_{\eps,\bz}),\bSigma^{\bu}_{\MCD}(H^D_{\eps,\bz}),\bmu^{\bu}_{\MCD}(H^D_{\eps,\bz}),q_{\alpha}(H^D_{\eps,\bz}))&\\
&= \delta_0(\eps) \bg^{\bu}(H_0, \bT(H^D_{\eps,\bz}),\bSigma^{\bu}_{\MCD}(H^D_{\eps,\bz}),\bmu^{\bu}_{\MCD}(H^D_{\eps,\bz}),q_{\alpha}(H^D_{\eps,\bz}))\\ 
&\;\;\;\;\; + 
\delta_p(\eps) \bg^{\bu}(H(\{1, \dots, p\},\bz), \bT(H^D_{\eps,\bz}),\bSigma^{\bu}_{\MCD}(H^D_{\eps,\bz}),\bmu^{\bu}_{\MCD}(H^D_{\eps,\bz}),q_{\alpha}(H^D_{\eps,\bz}))\\
&= (1-\eps) \bg^{\bu}(H_0, \bT(H^D_{\eps,\bz}),\bSigma^{\bu}_{\MCD}(H^D_{\eps,\bz}),\bmu^{\bu}_{\MCD}(H^D_{\eps,\bz}),q_{\alpha}(H^D_{\eps,\bz}))\\ 
&\;\;\;\;\; + 
\eps \bg^{\bu}(\Delta^{\bu}_{\bz}, \bT(H^D_{\eps,\bz}),\bSigma^{\bu}_{\MCD}(H^D_{\eps,\bz}),\bmu^{\bu}_{\MCD}(H^D_{\eps,\bz}),q_{\alpha}(H^D_{\eps,\bz}))
\end{align*}
where $H^D_{\eps,\bz}:=H(G_\eps^D,\bz)$, which yields the derivative
\begin{align}\label{eq:derFDCMf}
\frac{\partial}{\partial \eps}& \bg^{\bu}(H^D_{\eps,\bz}, \bT(H^D_{\eps,\bz}),\bSigma^{\bu}_{\MCD}(H^D_{\eps,\bz}),\bmu^{\bu}_{\MCD}(H^D_{\eps,\bz}),q_{\alpha}(H^D_{\eps,\bz}))\Big|_{\eps=0}  \nonumber\\
  &=-\bg^{\bu}(H_0, \bT_0,\bSigma^{\bu}_{\MCD,0},\bmu^{\bu}_{\MCD,0},q_{\alpha,0})\nonumber\\
  &\;\;\;\;\; + \frac{\partial }{\partial \eps}\bg^{\bu}(H_0, \bT(H^D_{\eps,\bz}),\bSigma^{\bu}_{\MCD}(H^D_{\eps,\bz}),\bmu^{\bu}_{\MCD}(H^D_{\eps,\bz}),q_{\alpha}(H^D_{\eps,\bz}))\Big|_{\eps=0} \nonumber\\
  &\;\;\;\;\; + \bg^{\bu}(\Delta^{\bu}_{\bz}, \bT_0,\bSigma^{\bu}_{\MCD,0},\bmu^{\bu}_{\MCD,0},q_{\alpha,0}) \nonumber\\
  &= \bzero + \frac{\partial }{\partial \eps}\bg^{\bu}(H_0, \bT(H^D_{\eps,\bz}),\bSigma^{\bu}_{\MCD,0},\bmu^{\bu}_{\MCD,0},q_{\alpha,0})\Big|_{\eps=0}\nonumber \\
  &\;\;\;\;\; +\frac{\partial }{\partial \eps}\bg^{\bu}(H_0, \bT_0,\bSigma^{\bu}_{\MCD}(H^D_{\eps,\bz}),\bmu^{\bu}_{\MCD}(H^D_{\eps,\bz}),q_{\alpha}(H^D_{\eps,\bz}))\Big|_{\eps=0}\nonumber\\
  &\;\;\;\;\; + \bg^{\bu}(\Delta^{\bu}_{\bz}, \bT_0,\bSigma^{\bu}_{\MCD,0},\bmu^{\bu}_{\MCD,0},q_{\alpha,0}). 
\end{align}
where $\bT_0=\bT(0):=\left(\vect(\bV(H_0)), \bmu(H_0),\bsigma(H_0)\right)^T$ and $q_{\alpha,0}=q_{\alpha}(H_0)$.
Note that 
\begin{align*}
  &\frac{\partial }{\partial \eps}\bg^{\bu}(H_0, \bT_0,\bSigma^{\bu}_{\MCD}(H^D_{\eps,\bz}),\bmu^{\bu}_{\MCD}(H^D_{\eps,\bz}),q_{\alpha}(H^D_{\eps,\bz}))\Big|_{\eps=0}\\
 &=\left(\begin{array}{ccc}
     \bC_{11}  &  \bC_{12}&  \bc_{13}\\
     \bC_{21}& \bC_{22}& \bc_{23}\\
     \bc_{31}^T& \bc_{32}^T& c_{33}
    \end{array}\right)  
    \left(\begin{array}{c}
    \frac{\partial} {\partial \eps}\vect \left(\bSigma^{\bu}_{\MCD}(H^D_{\eps,\bz})\right)\Big|_{\eps=0}\\ 
    \frac{\partial} {\partial \eps} \bmu^{\bu}_{\MCD}(H^D_{\eps,\bz})\Big|_{\eps=0}\\
    \frac{\partial} {\partial \eps}q_{\alpha}(H^D_{\eps,\bz})\Big|_{\eps=0}
    \end{array}\right),
  \end{align*}
where
\begin{gather*}
    \bC_{11}:=\frac{\partial} {\partial \vect\left(\bSigma\right)}\bg_1^{\bu}(H_0, \bT_0,\bSigma,\bmu^{\bu}_{\MCD,0},q_{\alpha,0})\Big|_{\bSigma=\bSigma^{\bu}_{\MCD,0}},\\
    \bC_{12}:=\frac{\partial} {\partial \bmu}\bg_1^{\bu}(H_0, \bT_0,\bSigma^{\bu}_{\MCD,0},\bmu,q_{\alpha,0})\Big|_{\bmu=\bmu^{\bu}_{\MCD,0}},\\
    \bC_{21}:=\frac{\partial} {\partial \vect\left(\bSigma\right)}\bg_2^{\bu}(H_0, \bT_0,\bSigma,\bmu^{\bu}_{\MCD,0},q_{\alpha,0})\Big|_{\bSigma=\bSigma^{\bu}_{\MCD,0}},\\
    \bC_{22}:=\frac{\partial} {\partial \bmu}\bg_2^{\bu}(H_0, \bT_0,\bSigma^{\bu}_{\MCD,0},\bmu,q_{\alpha,0})\Big|_{\bmu=\bmu^{\bu}_{\MCD,0}},\\
    \bc_{13}:=\frac{\partial} {\partial q_{\alpha}}\bg_1^{\bu}(H_0, \bT_0,\bSigma^{\bu}_{\MCD,0},\bmu^{\bu}_{\MCD,0},q_{\alpha})\Big|_{q_{\alpha}=q_{\alpha,0}},\\
     \bc_{23}:=\frac{\partial} {\partial q_{\alpha}}\bg_2^{\bu}(H_0, \bT_0,\bSigma^{\bu}_{\MCD,0},\bmu^{\bu}_{\MCD,0},q_{\alpha})\Big|_{q_{\alpha}=q_{\alpha,0}},\\
    \bc_{31}:=\frac{\partial} {\partial \vect\left(\bSigma\right)}g_3^{\bu}(H_0, \bT_0,\bSigma,\bmu^{\bu}_{\MCD,0},q_{\alpha,0})\Big|_{\bSigma=\bSigma^{\bu}_{\MCD,0}},\\
    \bc_{32}:=\frac{\partial} {\partial \bmu}\bg_3^{\bu}(H_0, \bT_0,\bSigma^{\bu}_{\MCD,0},\bmu,q_{\alpha,0})\Big|_{\bmu=\bmu^{\bu}_{\MCD,0}},\\
    c_{33}:=\frac{\partial} {\partial q_{\alpha}}g_3^{\bu}(H_0, \bT_0,\bSigma^{\bu}_{\MCD,0},\bmu^{\bu}_{\MCD,0},q_{\alpha})\Big|_{q_{\alpha}=q_{\alpha,0}}.
\end{gather*}
Moreover,
\begin{align*}
    &\frac{\partial }{\partial \eps}\bg^{\bu}(H_0, \bT(H^D_{\eps,\bz}),\bSigma^{\bu}_{\MCD,0},\bmu^{\bu}_{\MCD,0},q_{\alpha,0})\Big|_{\eps=0}\\
    &=\frac{\partial }{\partial \eps}\bg^{\bu}(H_0, \left[\vect(\bV(H^D_{\eps,\bz})), \bmu(H^D_{\eps,\bz}),\bsigma(H^D_{\eps,\bz})\right],\bSigma^{\bu}_{\MCD,0},\bmu^{\bu}_{\MCD,0},q_{\alpha,0})\Big|_{\eps=0}\\
     &=\left(\begin{array}{c}
          \frac{\partial }{\partial \bsigma}\bg^{\bu}_1(H_0, \left[\vect(\bV_0), \bmu_0,\bsigma\right],\bSigma^{\bu}_{\MCD,0},\bmu^{\bu}_{\MCD,0},q_{\alpha,0})\Big|_{\bsigma=\bsigma_0}  \\
          \frac{\partial }{\partial \bsigma}\bg^{\bu}_2(H_0, \left[\vect(\bV_0), \bmu_0,\bsigma\right],\bSigma^{\bu}_{\MCD,0},\bmu^{\bu}_{\MCD,0},q_{\alpha,0})\Big|_{\bsigma=\bsigma_0} \\
          \frac{\partial }{\partial \bsigma}g^{\bu}_3(H_0, \left[\vect(\bV_0), \bmu_0,\bsigma\right],\bSigma^{\bu}_{\MCD,0},\bmu^{\bu}_{\MCD,0},q_{\alpha,0})\Big|_{\bsigma=\bsigma_0} 
     \end{array}\right)\frac{\partial \bsigma(H^D_{\eps,\bz})}{\partial \eps}\Big|_{\eps=0}\\
     &\;\;+\left(\begin{array}{c}
          \frac{\partial }{\partial \vect(\bV)}\bg^{\bu}_1(H_0, \left[\vect(\bV), \bmu_0,\bsigma_0\right],\bSigma^{\bu}_{\MCD,0},\bmu^{\bu}_{\MCD,0},q_{\alpha,0})\Big|_{\bV=\bV_0}  \\
          \frac{\partial }{\partial \vect(\bV)}\bg^{\bu}_2(H_0, \left[\vect(\bV), \bmu_0,\bsigma_0\right],\bSigma^{\bu}_{\MCD,0},\bmu^{\bu}_{\MCD,0},q_{\alpha,0})\Big|_{\bV=\bV_0}\\
          \frac{\partial }{\partial \vect(\bV)}g^{\bu}_3(H_0, \left[\vect(\bV), \bmu_0,\bsigma_0\right],\bSigma^{\bu}_{\MCD,0},\bmu^{\bu}_{\MCD,0},q_{\alpha,0})\Big|_{\bV=\bV_0}
     \end{array}\right)\frac{\partial \vect(\bV(H^D_{\eps,\bz}))}{\partial \eps}\Big|_{\eps=0}\\
     &\;\;+\left(\begin{array}{c}
          \frac{\partial }{\partial \bmu}\bg^{\bu}_1(H_0, \left[\vect(\bV_0), \bmu,\bsigma_0\right],\bSigma^{\bu}_{\MCD,0},\bmu^{\bu}_{\MCD,0},q_{\alpha,0})\Big|_{\bmu=\bmu_0}   \\
        \frac{\partial }{\partial \bmu}\bg^{\bu}_2(H_0, \left[\vect(\bV_0), \bmu,\bsigma_0\right],\bSigma^{\bu}_{\MCD,0},\bmu^{\bu}_{\MCD,0},q_{\alpha,0})\Big|_{\bmu=\bmu_0}  \\
        \frac{\partial }{\partial \bmu}\bg^{\bu}_3(H_0, \left[\vect(\bV_0), \bmu,\bsigma_0\right],\bSigma^{\bu}_{\MCD,0},\bmu^{\bu}_{\MCD,0},q_{\alpha,0})\Big|_{\bmu=\bmu_0} 
     \end{array}\right)\frac{\partial \bmu(H^D_{\eps,\bz})}{\partial \eps}\Big|_{\eps=0}\\
     &= \bS^{\bu}\IFu_{\case}\left(\bz,\bsigma,H_0\right)
     + \bB^{\bu}_{1} \IFu_{\case}\left(\bz,\bV,H_0\right)
     + \bB^{\bu}_{2} \IFu_{\case}\left(\bz,\bmu,H_0\right).
\end{align*}
By setting \eqref{eq:derFDCMf} equal to $\bzero$, we obtain
\begin{align*}
\IFu_{\case}(\bz,\vect(\bSigma^{\bu}_{\MCD}),H_0)
   =& -\bD^{\bu}_{1}\Big[\bS^{\bu}  
      \IFu_{\case}\left(\bz,\bsigma,H_0\right)
       + \bB^{\bu}_{1} \IFu_{\case}
       \left(\bz,\vect(\bV),H_0\right)\\
   & +\bB^{\bu}_{2} \IFu_{\case}
     \left(\bz,\bmu,H_0\right)+\bg^{\bu}
     (\Delta^{\bu}_{\bz}, \bT_0,\bSigma^{\bu}_{\MCD,0},
     \bmu^{\bu}_{\MCD,0},q_{\alpha,0})\Big]
\end{align*}
and
\begin{align*}
\IFu_{\case}(\bz,\bmu^{\bu}_{\MCD},H_0)
  =& -\bD^{\bu}_{2}\Big[\bS^{\bu} 
     \IFu_{\case}\left(\bz,\bsigma,H_0\right)
     + \bB^{\bu}_{1} \IFu_{\case}
       \left(\bz,\vect(\bV),H_0\right)\\
  &+ \bB^{\bu}_{2} \IFu_{\case}\left(\bz,\bmu,H_0\right)
     +\bg^{\bu}(\Delta^{\bu}_{\bz}, 
     \bT_0,\bSigma^{\bu}_{\MCD,0},
     \bmu^{\bu}_{\MCD,0},q_{\alpha,0})\Big],
  \end{align*}
where 
\begin{equation*}
    \begin{bmatrix} 
\bD^{\bu}_{1}  \\ 
\bD^{\bu}_{2}  \\
(\bd^{\bu}_{3})^T 
\end{bmatrix}
=\left(\begin{array}{ccc}
 \bC_{11}  &  \bC_{12}&  \bc_{13}\\
 \bC_{21}& \bC_{22}& \bc_{23}\\
 \bc_{31}^T& \bc_{32}^T& c_{33}
 \end{array}\right)^{-1}.
\end{equation*}

Under the FICM model we have $\delta_0(\eps)=(1-\eps)^p$, $\delta_0(0)=1$, 
$\delta_1(\eps) = p (1-\eps)^{p-1}\eps$ so $\delta_1(0)=0$ 
and $\delta_1^{\prime}(0)=p$, and $\delta_i(0)=
\delta_i^{\prime}(0)=0$ for $i \geqslant 2$. Therefore 
$\bg^{\bu}$ can be written as the sum
\begin{align*}
\bg^{\bu}&(H^I_{\eps,\bz}, \bT(H^I_{\eps,\bz}),\bSigma^{\bu}_{\MCD}(H^I_{\eps,\bz}),\bmu^{\bu}_{\MCD}(H^I_{\eps,\bz}),q_{\alpha}(H^I_{\eps,\bz}))&\\
&= \delta_0(\eps) \bg^{\bu}(H_0, \bT(H^I_{\eps,\bz}),\bSigma^{\bu}_{\MCD}(H^I_{\eps,\bz}),\bmu^{\bu}_{\MCD}(H^I_{\eps,\bz}),q_{\alpha}(H^I_{\eps,\bz}))\\ 
&\;\;\;\;\; + 
\delta_1(\eps) \sum_{j=1}^p \bg^{\bu}(H(j,\bz), \bT(H^I_{\eps,\bz}),\bSigma^{\bu}_{\MCD}(H^I_{\eps,\bz}),\bmu^{\bu}_{\MCD}(H^I_{\eps,\bz}),q_{\alpha}(H^I_{\eps,\bz}))\\
&= (1-\eps)^p \bg^{\bu}(H_0, \bT(H^I_{\eps,\bz}),\bSigma^{\bu}_{\MCD}(H(G_\eps^D,\bz)),\bmu^{\bu}_{\MCD}(H^I_{\eps,\bz}),q_{\alpha}(H^I_{\eps,\bz}))\\ 
&\;\;\;\;\; + 
p (1-\eps)^{p-1}\eps \sum_{j=1}^p \bg^{\bu}(H(j,\bz), \bT(H^I_{\eps,\bz}),\bSigma^{\bu}_{\MCD}(H^I_{\eps,\bz}),\bmu^{\bu}_{\MCD}(H^I_{\eps,\bz}),q_{\alpha}(H^I_{\eps,\bz})).
\end{align*}
where $H^I_{\eps,\bz}:=H(G_\eps^I,\bz)$.
The derivative becomes
\begin{align}\label{eq:derFICMf}
\frac{\partial}{\partial \eps}& \bg^{\bu}(H^I_{\eps,\bz}, \bT(H^I_{\eps,\bz}),\bSigma^{\bu}_{\MCD}(H^I_{\eps,\bz}),\bmu^{\bu}_{\MCD}(H^I_{\eps,\bz}),q_{\alpha}(H^I_{\eps,\bz}))\Big|_{\eps=0}&  \nonumber\\
  &= \frac{\partial }{\partial \eps}\bg^{\bu}(H_0, \bT(H^I_{\eps,\bz}),\bSigma^{\bu}_{\MCD,0},\bmu^{\bu}_{\MCD,0},q_{\alpha,0})\Big|_{\eps=0}\nonumber \\
  &\;\;\;\;\;+\frac{\partial }{\partial \eps}\bg^{\bu}(H_0, \bT_0,\bSigma^{\bu}_{\MCD}(H^I_{\eps,\bz}),\bmu^{\bu}_{\MCD}(H^I_{\eps,\bz}),q_{\alpha}(H^I_{\eps,\bz}))\Big|_{\eps=0}\nonumber\\
  &\;\;\;\;\; + p\sum_{j=1}^p\bg^{\bu}(H(j,\bz), \bT_0,\bSigma^{\bu}_{\MCD,0},\bmu^{\bu}_{\MCD,0},q_{\alpha,0}). 
\end{align}
By setting \eqref{eq:derFICMf} equal to $\bzero$ we obtain
\begin{align*}
\IFu_{\cell}(\bz,\vect(\bSigma^{\bu}_{\MCD})&,H_0)
  = -\bD^{\bu}_{1}\Big[\bS^{\bu} 
     \IFu_{\cell}\left(\bz,\bsigma,H_0\right)
     + \bB^{\bu}_{1} \IFu_{\cell}
       \left(\bz,\vect(\bV),H_0\right)\\
  &+ \bB^{\bu}_{2} \IFu_{\cell}\left(\bz,\bmu,H_0\right)
      + p\sum_{j=1}^p\bg^{\bu}(H(j,\bz), \bT_0,
      \bSigma^{\bu}_{\MCD,0},
      \bmu^{\bu}_{\MCD,0},q_{\alpha,0}) \Big]
\end{align*}
and
\begin{align*}
\IFu_{\cell}(\bz,\bmu^{\bu}_{\MCD},H_0)
  =& -\bD^{\bu}_{2}\Big[\bS^{\bu} \IFu_{\cell}
      \left(\bz,\bsigma,H_0\right)
     +\bB^{\bu}_{1} \IFu_{\cell}
     \left(\bz,\vect(\bV),H_0\right)\\
  &+ \bB^{\bu}_{2} \IFu_{\cell}\left(\bz,\bmu,H_0\right)
     + p\sum_{j=1}^p\bg^{\bu}(H(j,\bz), 
     \bT_0,\bSigma^{\bu}_{\MCD,0},
     \bmu^{\bu}_{\MCD,0},q_{\alpha,0})\Big].
\end{align*}

In the proof we have assumed that 
$H^{\bu}(H(G_\eps^D,\bz), \bT(H(G_\eps^D,\bz)))$ and 
$H^{\bu}(H(G_\eps^I,\bz), \bT(H(G_\eps^D,\bz)))$ have
no mass on the boundary of $A$ in order to use the 
expressions of the MCD functionals in 
\eqref{eq:funcfsigma} and \eqref{eq:funcfmu}, and that $\bB^{\bu}_{1}$, $\bB^{\bu}_{2}$, $\bD^{\bu}_{1}$, $\bD^{\bu}_{2}$,  and  $\bS^{\bu}$ exist.
Moreover, we have assumed that $\bSigma^{\bu}_{\MCD,0}$ and $\bmu^{\bu}_{\MCD,0}$ are unique, which holds when $H^{\bu}(H_0, \bT(H_0))$ has a density (Cator and Lopuh\"aa 2012), that the IFs exist, and that $\left(\bT(H(G_\eps,\bz)),\bSigma^{\bu}_{\MCD}(H(G_\eps,\bz)),\bmu^{\bu}_{\MCD}(H(G_\eps,\bz))\right)$ converges to $\left(\bT_0,\bSigma^{\bu}_{\MCD,0},\bmu^{\bu}_{\MCD,0}\right)$ as $\eps \downarrow 0$.
\end{proof}

We now denote the functional corresponding to the  
estimator $\btSigma_{\xort}^{R}$ by $\btSigma_{\xort}^{R}(H)$.
\begin{proposition}\label{prop3}
The casewise and cellwise influence functions 
of $\vect(\btSigma_{\xort}^{R})$ are
\begin{align*}
  \IFu_{\case}\left(\bz,\vect(\btSigma_{\xort}^{R}),H_0\right)
   =&\;\bD^{\xort}\Big[\bB^{\xort}_{1}\IFu_{\case}
      \left(\bz,\bmu,H_0\right)+\bB^{\xort}_{2}\IFu_{\case}
      \left(\bz,\vect(\bV),H_0\right)\\
   &+\bS^{\xort}\IFu_{\case}\left(\bz,\bsigma,H_0\right)
      +\bg^{\xort}(\Delta_{\bz},
       \bT_0,\bSigma_{\xort,0}^{R})\Big]
\end{align*}
and
\begin{align*}
  \IFu_{\cell}\left(\bz,\vect(\btSigma_{\xort}^{R}),H_0\right)
  =&\;\bD^{\xort}\Big[\bB^{\xort}_{1}\IFu_{\cell}
     \left(\bz,\bmu,H_0\right)+\bB^{\xort}_{2}
     \IFu_{\cell}\left(\bz,\vect(\bV),H_0\right)\\
   &+\bS^{\xort}\IFu_{\cell}\left(\bz,\bsigma,H_0\right)
    +p\sum_{j=1}^p \bg^{\xort}(H\left(j,\bz\right),
    \bT_0,\bSigma_{\xort,0}^{R})\Big],
\end{align*}
with  $ \bSigma_{\xort,0}^{R}= \btSigma_{\xort}^{R}(H_0)$, 
$\bT_0=\left(\vect(\bV_0), \bmu_0,\bsigma_0\right)^T$ and
\begin{multline*}
  \bg^{\xort}(H,\bT,\btSigma_{\xort}^{R}):=
  \vect\Big(\E_{H}\Big[b\btSigma_{\xort}^{R}
  - (1-\delta)w^{\case}\btW(\bx - \bmu-\bV\bu)(\bx - \bmu-\bV\bu)^T
  \btW\\-\delta w^{\case}\left(\btW(\bx - \bmu-\bV\bu)(\bx - \bmu-\bV\bu)^T
  \btW\right)\odot\bI_p\Big]\Big),
\end{multline*}
where $\bT=\left(\vect(\bV), \bmu,\bsigma\right)^T$, 
$b=\sum_{j=1}^p\sum_{\ell=1}^p w^{\case}w_{j}^{\cell}w_{\ell}^{\cell}/p^2$\,, 
 $\bu$ depends on $\bx$ and $\bT$ through \eqref{eq:ux}, and $\odot$ is the Hadamard product.
The matrices $\bD^{\xort}$, $\bB_1^{\xort}$, 
$\bB_2^{\xort}$, and $\bS^{\xort}$  are 
defined in the proof. 
\end{proposition}

\vspace{3mm}
\begin{proof}[\textbf{Proof of Proposition \ref{prop3}}]
The IF of $\vect(\btSigma_{\xort}^{R})$ at the 
distribution $H_0$ is given by 
\begin{equation*}
  \IFu(\bz,\vect(\btSigma_{\xort}^{R}),G_\eps)=
  \left.\frac{\partial}{\partial \eps}
  \vect(\btSigma_{\xort}^{R}(H(G_\eps,\bz)))
  \right|_{\eps=0}\,.
\end{equation*}
For $\eps=0$ we obtain $\btSigma_{\xort}^{R}(H(G_0,\bz))
= \btSigma_{\xort}^{R}(H_0) = \bSigma_{\xort,0}^{R}$.
Note that $\btSigma_{\xort}^{R}(H)$ satisfies 
\begin{equation}
\label{eq:functionalerrorfun}
  \bg^{\xort}(H,\bT,\btSigma_{\xort}^{R})=\bzero
\end{equation}
where $\bT=\bT(H):=\left(\vect(\bV(H)), 
\bmu(H),\bsigma(H)\right)^T$.

The IFs of $\vect(\btSigma_{\xort}^{R})$ under the FDCM
and FICM models are obtained by differentiating 
\eqref{eq:functionalerrorfun} in $\eps=0$ when $H$ 
is either $H(G_\eps^D,\bz)$ or $H(G_\eps^I,\bz)$.

For FDCM, we know that $G_\eps^D$ has 
$\delta_0(\eps)=(1-\eps),\; 
\delta_1(\eps)=\cdots=\delta_{p-1}(\eps)=0$ and 
$\delta_p(\eps)=\eps$, so $ \bg^{\xort}$ can be 
written as
\begin{align*}
\bg^{\xort}&(H(G_\eps^D,\bz), \bT(H(G_\eps^D,\bz)), \btSigma_{\xort}^{R}(H(G_\eps^D,\bz)))&\\
&= \delta_0(\eps)\bg^{\xort}(H_0, \bT(H(G_\eps^D,\bz)),  \btSigma_{\xort}^{R}(H(G_\eps^D,\bz)))\\
&\;\;\;\;\; + 
\delta_p(\eps)\bg^{\xort}(H(\{1, \dots, p\},\bz), \bT(H(G_\eps^D,\bz)),  \btSigma_{\xort}^{R}(H(G_\eps^D,\bz)))\\
&= (1-\eps)\bg^{\xort}(H_0, \bT(H(G_\eps^D,\bz)),  \btSigma_{\xort}^{R}(H(G_\eps^D,\bz)))\\
&\;\;\;\;\; + \eps\bg^{\xort}(\Delta_{\bz}, \bT(H(G_\eps^D,\bz)),
   \btSigma_{\xort}^{R}(H(G_\eps^D,\bz))),
\end{align*}
which yields the derivative
\begin{align}\label{eq:derFDCMf2}
\frac{\partial}{\partial \eps}&\bg^{\xort}(H(G_\eps^D,\bz), \bT(H(G_\eps^D,\bz)),  \btSigma_{\xort}^{R}(H(G_\eps^D,\bz)))\Big|_{\eps=0}  \nonumber\\
  &=-\bg^{\xort}(H_0, \bT_0,\bSigma_{\xort,0}^{R})
  + \frac{\partial }{\partial \eps}\bg^{\xort}(H_0, \bT(H(G_\eps^D,\bz)),  \btSigma_{\xort}^{R}(H(G_\eps^D,\bz)))\Big|_{\eps=0} \nonumber\\
  &\;\;\;\;\; +\bg^{\xort}(\Delta_{\bz}, \bT_0,\bSigma_{\xort,0}^{R}) \nonumber\\
  &= \bzero + \frac{\partial }{\partial \eps}\bg^{\xort}(H_0, \bT(H(G_\eps^D,\bz)),\bSigma_{\xort,0}^{R})\Big|_{\eps=0}\nonumber +\frac{\partial }{\partial \eps}\bg^{\xort}(H_0, \bT_0, \btSigma_{\xort}^{R}(H(G_\eps^D,\bz)))\Big|_{\eps=0}\nonumber\\
  &\;\;\;\;\; +\bg^{\xort}(\Delta_{\bz}, \bT_0,\bSigma_{\xort,0}^{R})
\end{align}
 \sloppy where $\bT_0=\left(\vect(\bV(H_0)), \bmu(H_0),\bsigma(H_0)\right)^T=\left(\vect(\bV_0), \bmu_0,\bsigma_0\right)^T$ and  $\bSigma_{\xort,0}^{R}=\vect(\bSigma_{\xort,0}^{R})$.
Note that
\begin{align*}
    \frac{\partial }{\partial \eps}&\bg^{\xort}(H_0, \bT(H(G_\eps^D,\bz)),\bSigma_{\xort,0}^{R})\Big|_{\eps=0}\\
      &=  \frac{\partial }{\partial \bsigma}\bg^{\xort}(H_0, \left[\vect(\bV_0), \bmu_0,\bsigma\right],\bSigma_{\xort,0}^{R})\Big|_{\bsigma=\bsigma_0}\IFu_{\case}\left(\bz,\bsigma,H_0\right)  \\
     &\;\;\;\;\; +  \frac{\partial }{\partial \vect(\bV)}\bg^{\xort}(H_0, \left[\vect(\bV), \bmu_0,\bsigma_0\right],\bSigma_{\xort,0}^{R})\Big|_{\bV=\bV_0}\IFu_{\case}\left(\bz,\bV,H_0\right)  \\
      &\;\;\;\;\; +  \frac{\partial }{\partial \bmu}\bg^{\xort}(H_0, \left[\vect(\bV_0), \bmu,\bsigma_0\right],\bSigma_{\xort,0}^{R})\Big|_{\bmu=\bmu_0}\IFu_{\case}\left(\bz,\bmu,H_0\right)  \\
      &=\bS^{\xort}\IFu_{\case}\left(\bz,\bsigma,H_0\right)+\bB^{\xort}_{1}\IFu_{\case}\left(\bz,\bV,H_0\right)+\bB^{\xort}_{2}\IFu_{\case}\left(\bz,\bmu,H_0\right).
\end{align*}
By setting \eqref{eq:derFDCMf2} equal to $\bzero$,
we obtain
\begin{align*}
   &\bS^{\xort}\IFu_{\case}\left(\bz,\bsigma,H_0\right)
     +\bB^{\xort}_{1}\IFu_{\case}
       \left(\bz,\vect(\bV),H_0\right)
     +\bB^{\xort}_{2}\IFu_{\case}
       \left(\bz,\bmu,H_0\right)\\
   &\;\;\;\;\; + \frac{\partial }{\partial \btSigma_{\xort}^{R}}
    \bg^{\xort}(H_0, \bT_0, \btSigma_{\xort}^{R})
    \Big|_{ \btSigma_{\xort}^{R}=\bSigma_{\xort,0}^{R}}
    \frac{\partial \btSigma_{\xort}^{R}(H(G_\eps^D,\bz))}
    {\partial \eps}\Big|_{\eps=0}
    +\bg^{\xort}(\Delta_{\bz}, \bT_0,\bSigma_{\xort,0}^{R})\\
   &=\bS^{\xort}\IFu_{\case}\left(\bz,\bsigma,H_0\right)
     +\bB^{\xort}_{1}\IFu_{\case}
      \left(\bz,\vect(\bV),H_0\right)
      +\bB^{\xort}_{2}\IFu_{\case}\left(\bz,\bmu,H_0\right)\\
   &\;\;\;\;\; + \bC_{\xort}\IFu_{\case}
     \left(\bz,\vect(\btSigma_{\xort}^{R}),H_0\right)
     +\bg^{\xort}(\Delta_{\bz},\bT_0,\bSigma_{\xort,0}^{R})
     =\bzero,
\end{align*}
where $\bC_{\xort}=\bI_{p^2}\E_{H_0}\left[b\right]$.
Finally,
\begin{align*}
  \IFu_{\case}\left(\bz,\vect(\btSigma_{\xort}^{R}),H_0\right)
  =&\; \bD^{\xort}\Big[\bS^{\xort}\IFu_{\case}
    \left(\bz,\bsigma,H_0\right)
    +\bB^{\xort}_{1}\IFu_{\case}
    \left(\bz,\vect(\bV),H_0\right)\\
   &\;\;\; +\bB^{\xort}_{2}\IFu_{\case}
     \left(\bz,\bmu,H_0\right)
     +\bg^{\xort}(\Delta_{\bz}, \bT_0,
     \bSigma_{\xort,0}^{R})\Big]
\end{align*}
where $\bD^{\xort}=\left(\bC_{\xort}\right)^{-1}$.

Under the FICM model we have $\delta_0(\eps)=(1-\eps)^p$, 
$\delta_0(0)=1$ and $\delta_1(\eps)=p(1-\eps)^{p-1}\eps$\,, 
so $\delta_1(0)=0$ and $\delta_1^{\prime}(0)=p$, and 
$\delta_i(0)=\delta_i^{\prime}(0)=0$ for $i \geqslant 2$. 
Therefore $\bg^{\xort}$ can be written as the sum
\begin{align*}
\bg^{\xort}&(H(G_\eps^I,\bz), \bT(H(G_\eps^I,\bz)), 
   \btSigma_{\xort}^{R}((H(G_\eps^I,\bz)))&\\
&= \delta_0(\eps)\bg^{\xort}(H_0, \bT(H(G_\eps^I,\bz)),  \btSigma_{\xort}^{R}(H(G_\eps^I,\bz)))\\
&\;\;\;\;\; + 
\delta_1(\eps)\sum_{j=1}^p\bg^{\xort}(H(j,\bz), \bT(H(G_\eps^I,\bz)),  \btSigma_{\xort}^{R}(H(G_\eps^I,\bz)))\\
&= (1-\eps)^p \bg^{\xort}(H_0, \bT(H(G_\eps^I,\bz)),  \btSigma_{\xort}^{R}(H(G_\eps^I,\bz)))\\ 
&\;\;\;\;\; + 
p (1-\eps)^{p-1}\eps \sum_{j=1}^p \bg^{\xort}(H(j,\bz), \bT(H(G_\eps^I,\bz)),  \btSigma_{\xort}^{R}(H(G_\eps^I,\bz))).
\end{align*}
Following the same steps as for the FDCM, we obtain
\begin{align*}
  \IFu_{\cell}\left(\bz,\vect(\btSigma_{\xort}^{R}),H_0\right)
  =&\; \bD^{\xort}\Big[\bS^{\xort}\IFu_{\cell}
     \left(\bz,\bsigma,H_0\right)
     +\bB^{\xort}_{1}\IFu_{\cell}
      \left(\bz,\vect(\bV),H_0\right)\\
   &\;\;\; +\bB^{\xort}_{2}\IFu_{\cell}
     \left(\bz,\bmu,H_0\right)
    +\bg^{\xort}(\Delta_{\bz},\bT_0,\bSigma_{\xort,0}^{R})\Big].
\end{align*}
This ends the proof of Proposition~\ref{prop3}. \end{proof}

\vspace{3mm}
\begin{proof}[\textbf{Proof of Theorem \ref{IFSigma}}]
Here we will compute the IF of $\vect(\bSigma)$ at $H_0$ 
as 
\begin{equation*}
  \IFu(\bz,\vect(\bSigma),G_\eps)=
  \vect\Big(\left.\frac{\partial}{\partial \eps} 
  \bSigma(H(G_\eps,\bz))\right|_{\eps=0}\Big)\,
\end{equation*}
where $G_\eps = G_\eps^D$ in the FDCM model, and
$G_\eps = G_\eps^I$ in the FICM model. In either model
$\bSigma(H)=\bSigma_{\bx^\sub}(H)+\btSigma_{\xort}^{R}(H)$ 
with $\bSigma_{\bx^\sub}(H)=\bV(H)\bSigma^{\bu}_{\MCD}(H)
\left(\bV(H)\right)^T$, so we obtain 
\begin{align*}
    &\left.\frac{\partial}{\partial \eps} \bSigma(H(G_\eps,\bz))\right|_{\eps=0}\\
    &=\left.\frac{\partial}{\partial \eps}\bV(H(G_\eps,\bz))\bSigma^{\bu}_{\MCD}(H(G_\eps,\bz))\left(\bV(H(G_\eps,\bz))\right)^T\right|_{\eps=0}+\left.\frac{\partial}{\partial \eps}\btSigma_{\xort}^{R}(H(G_\eps,\bz))\right|_{\eps=0}\\
    &=\left.\frac{\partial}{\partial \eps}\bV(H(G_\eps,\bz))\right|_{\eps=0}\bSigma^{\bu}_{\MCD,0}\bV_0^T+\bV_0\left.\frac{\partial}{\partial \eps}\bSigma^{\bu}_{\MCD}(H(G_\eps,\bz))\right|_{\eps=0}\bV_0^T\\
    &\;\;\;\;\; +\bV_0\bSigma^{\bu}_{\MCD,0}\left.\frac{\partial}{\partial \eps}\left(\bV(H(G_\eps,\bz))\right)^T\right|_{\eps=0}+\left.\frac{\partial}{\partial \eps}\btSigma_{\xort}^{R}(H(G_\eps,\bz))\right|_{\eps=0}\;,
\end{align*}
where $\bV_{0}=\bV(H_0)$, 
$\bSigma^{\bu}_{\MCD,0}=\bSigma^{\bu}_{\MCD}(H_0)$, 
and $\bSigma_{\xort,0}^{R}=\btSigma_{\xort}^{R}(H_0)$.
Applying $\vect(\CD)$ to both sides yields 
\begin{align*}
    &\vect\left(\left.\frac{\partial}{\partial \eps} \bSigma(H(G_\eps,\bz))\right|_{\eps=0}\right)\\
    &=\left(\bV_0\bSigma^{\bu}_{\MCD,0}\otimes\bI_p\right)\vect\left(\left.\frac{\partial}{\partial \eps}\bV(H(G_\eps,\bz))\right|_{\eps=0}\right)\\
    &\;\;\;\;\; +\left(\bV_0\otimes\bV_0\right)\vect\left(\left.\frac{\partial}{\partial \eps}\bSigma^{\bu}_{\MCD}(H(G_\eps,\bz))\right|_{\eps=0}\right)\\
    &\;\;\;\;\; +\left(\bI_p\otimes\bV_0\bSigma^{\bu}_{\MCD,0}\right)\bK_{p,k}\vect\left(\left.\frac{\partial}{\partial \eps}\left(\bV(H(G_\eps,\bz))\right)\right|_{\eps=0}\right)
      +\IFu(\bz,\vect(\btSigma_{\xort}^{R}),H_0)\\
    &=\left(\left(\bV_0\bSigma^{\bu}_{\MCD,0}\otimes\bI_p\right)
      +\left(\bI_p\otimes\bV_0\bSigma^{\bu}_{\MCD,0}\right)
      \bK_{p,k}\right)\IFu(\bz,\vect(\bV),H_0)\\
    &\;\;\;\;\; +\left(\bV_0\otimes\bV_0\right)
       \IFu(\bz,\vect(\bSigma^{\bu}_{\MCD}),H_0)
       +\IFu(\bz,\vect(\btSigma_{\xort}^{R}),H_0)\\
    &=\bR_{1}\IFu(\bz,\vect(\bV),H_0)
      +\bR_{2}\IFu(\bz,\vect(\bSigma^{\bu}_{\MCD}),H_0)
      +\IFu(\bz,\vect(\btSigma_{\xort}^{R}),H_0),
\end{align*}
where the matrix $\bK_{p,k}$ is a $pk \times pk$ permutation 
matrix that rearranges the entries of the column 
vector $\vect\left(\left.\frac{\partial}{\partial \eps}
\left(\bV(H(G_\eps,\bz))\right)\right|_{\eps=0}\right)$ 
to become those of 
$\vect\left(\left.\frac{\partial}{\partial \eps}
\left(\bV(H(G_\eps,\bz))\right)^T\right|_{\eps=0}\right)$.
The matrices $\bR_{1}$ and $\bR_{2}$ are
\begin{align*}
    \bR_{1}&=\left(\bV_0\bSigma^{\bu}_{\MCD,0}\otimes\bI_p\right)+\left(\bI_p\otimes\bV_0\bSigma^{\bu}_{\MCD,0}\right)\bK_{p,k}\\
    \bR_{2}&=\bV_0\otimes\bV_0\
\end{align*}
in both the casewise and the cellwise settings. \end{proof}

\section{\large Consistency and asymptotic normality}
\label{app:consistency}

Here we will prove that $\bSigma_{n}$ is a 
consistent estimator of $\bSigma_{0}\,:=\,\bSigma(H_0)$\,, that is,  
$\bSigma_{n}$ converges to $\bSigma_{0}$ in 
probability.
To obtain this result we first study the 
consistency properties of cellPCA. 
Note that, for every $p$-dimensional vector 
$\bmu_0$, we can represent $\bhx=\bmu+\bV\bu$ 
equivalently as 
$\bhx=\left(\bmu+\bV\bu_0\right) +
\bV\left(\bu-\bu_0\right)$. Moreover, for 
every $k\times k$ nonsingular matrix $\bO$ we 
can represent $\bhx$ as 
$\bhx=\bmu+(\bV\bO)(\bO^{-1}\bu)$. To resolve this 
arbitrariness of the representation of the fitted 
point $\bhx$ we consider an equivalent 
parametrization of model~\eqref{eq:model} by
the  principal subspace model used in 
\cite{centofanti2025multivariate}, given by
\begin{equation} \label{eq:model2}
  \bx =  \bmu+\bP\bx^0 + \be
\end{equation}
where $\bx=\left(x_1,\dots,x_p\right)^T\sim H$, 
$\bx^0$ is a $p$-dimensional vector, $\bP$ is 
a $p \times p$ projection matrix of 
rank $\rk$, that is,
$\bP^T=\bP$, $\bP^2=\bP$ and $\rank(\bP)=\rk$,
which projects $\bx^0$ on itself, i.e.
$\bP\bx^0 = \bx^0$, and with $\bmu$ such
that $\bP\bmu = \bzero$. The image of $\bP$ 
is a $\rk$-dimensional linear subspace 
$\bPi_0$ through the origin. The 
predicted datapoints $\bhx = \bx^0 + \bmu$
lie on the affine subspace
$\bPi = \bPi_0 + \bmu$, which is called the
{\it principal subspace}. 
Note that any $\rk$-dimensional
affine subspace $\bPi$ determines a 
unique $\bP$ and $\bmu$ satisfying the 
constraints $\rank(\bP)=\rk$, $\bP^T=\bP$, 
$\bP^2=\bP$, and $\bP\bmu = \bzero$, and that
any such $\bP$ and $\bmu$ determine a unique 
subspace $\bPi$ of dimension $\rk$. The
cellPCA method estimates a principal subspace 
$\bPi$ in the form of a couple $(\bP,\bmu)$.
The functional version $(\bP(H),\bmu(H))$ of 
cellPCA is given by 
\begin{align} \label{eq:ux2}
\hspace{-6mm}
  \left(\bP(H),\bmu(H)\right)&=\argmin_{\bP,\bmu}
  \E_H\left[\rho_2 \left(\frac{1}{\sigma_2(H)}\sqrt{\frac{1}{p}
  \sum_{j=1}^{p} \sigma_{1,j}^2(H)\rho_1\left(\frac{
  x_{j}-\mu_j-\bp_j^T\bx^0}{\sigma_{1,j}(H)}\right)}
  \right)\right]\nonumber\\
  \text{such that}\quad \bx^0&= \argmin_{\bx^0}
  \rho_2 
  \left(\frac{1}{\sigma_2(H)}\sqrt{\frac{1}{p}\sum_{j=1}^{p} 
  \sigma_{1,j}^2(H)\rho_1 \left(\frac
  {x_{j}-\mu_j-\bp_j^T\bx^0}{\sigma_{1,j}(H)}
  \right)} \right)
\end{align}
where $\bx=\left(x_1,\dots,x_p\right)^T\sim H$,
$\bx^0=(x^0_{1},\dots,x^0_{p})^T$, $(\bP,\bmu)$ 
satisfies the above constraints, and
$\bp_{1},\dots,\bp_{p}$ are the columns of $\bP$.
Let us rewrite the objective function  
\eqref{eq:ux2} as
\begin{equation} \label{eq:G}
G(\btheta,\bsigma, H):=\E_{H} \left[
  g_{\btheta,\bsigma}(\bx) \right]
\end{equation}
where
\begin{equation} \label{eq:g}
g_{\btheta,\bsigma}(\bx) :=\rho_2 \left(
  \frac{1}{\sigma_2}\sqrt{\frac{1}{p}
  \sum_{j=1}^{p} \sigma_{1,j}^2\rho_1\left(
  \frac{x_{j}-\mu_j-\bp_j^T\bx^0}{\sigma_{1,j}}
  \right)}\right),
\end{equation}
in which $\bx^0$ is 
defined as in~\eqref{eq:ux2}. The parameters are 
$\btheta=(\vect(\bP),\bmu)^T \in \bTheta$, where 
\begin{multline*}
\bTheta=\lbrace (\vect(\bP),\bmu)\in \mathbb{R}^{p(p+1)}:
   \rank(\bP)=\rk, \bP^T=\bP, \bP^2=\bP, \bmu\in\bTheta_{\bmu} \text{ and } \bP\bmu = \bzero
\rbrace,
\end{multline*}
with $\bTheta_{\bmu}$ a compact subset of $\mathbb{R}^{p}$.
Note that $\bTheta$ is a compact subset of $\mathbb{R}^{p(p+1)}$, as the set of projection matrices $\bP$ is a compact subset of $\mathbb{R}^{p^2}$ and $\bTheta_{\bmu}$ is compact.
Moreover,  $\bsigma=\left(\sigma_{1,1}, \dots, \sigma_{1,p}, \sigma_{2}\right)^T\in\bTheta_{\bsigma}$ with
$\bTheta_{\bsigma}$ a compact subset of
$\left(\mathbb{R}^{+}\right)^{p+1}$. We denote 
$M_n(\btheta,\bsigma):=G(\btheta,\bsigma, H_n)$, 
$\bsigma_n:=\bsigma(H_n)$, and 
$M(\btheta,\bsigma):=G(\btheta,\bsigma, H_0)$. We 
aim to study the convergence properties of a minimizer 
$\btheta_n=\left(\vect(\bP_n),\bmu_n\right)^T$ of 
$ M_n(\btheta,\bsigma_n)$ where $\bP_n:=\bP(H_n)$ 
and $\bmu_n:=\bmu(H_n)$, to a minimizer 
$\btheta_0=\left(\vect(\bP_0),\bmu_0\right)^T$ of 
the population quantity $M(\btheta,\bsigma_0)$. 
Since $M(\btheta,\bsigma_0)$ need not have a unique
minimizer, we consider the set 
$\bTheta_0=\lbrace \btheta_0\in \bTheta: 
M(\btheta_0,\bsigma_0)=\inf_{\btheta\in\bTheta}
M(\btheta,\bsigma_0)\rbrace $, which is nonempty by 
the compactness of $\bTheta$ if 
$\btheta\rightarrow M(\btheta,\bsigma_0)$ is a 
continuous function.

\begin{proposition}\label{prop4}
Let us assume that $\bsigma_n\rightarrow_p \bsigma_0$
with $\bsigma_n, \bsigma_0\in \bTheta_{\bsigma}$
and that
\begin{equation}\label{assu_cons2}
\sup _{\btheta\in \bTheta: d\left(\btheta, \bTheta_0
\right) \geqslant \varepsilon} M(\btheta,\bsigma_0)
> M\left(\btheta_0,\bsigma_0\right) , \quad 
\text{for } \btheta_0\in \bTheta_0\,.
\end{equation}
Then it holds for any sequence of estimators 
$\btheta_n$ in $\bTheta$ with 
$M_n\left(\btheta_n,\bsigma_n\right) \leqslant 
M_n\left(\btheta_0,\bsigma_n\right)+o_p(1)$ for some 
$\btheta_0\in \bTheta_0 $ that 
$d(\btheta_n,\bTheta_0):=\inf_{\btheta\in\bTheta_0}
||\btheta_n-\btheta||\rightarrow_p 0$\,.
\end{proposition}

\begin{proof}[\textbf{Proof of Proposition \ref{prop4}}]
Note that the function $\bTheta\times \bTheta_{\bsigma}
\rightarrow \mathbb{R}:(\btheta,\bsigma)\rightarrow 
g_{\btheta,\bsigma}(\bx)$ is continuous for every 
$\bx$ because $\rho_1$ and $\rho_2$ in~\eqref{eq:g}
are continuous and the denominators are in compact 
sets away from zero. We know that $\btheta_n$ satisfies 
$M_n\left(\btheta_n,\bsigma_n\right) \leqslant 
M_n\left(\btheta_0,\bsigma_n\right)+o_p(1)$.
We now adapt the proof of Theorem 5.7 in 
Van der Vaart (2000) by noting that
\begin{align*}
  \sup_{\btheta \in \Theta}
  |M_n(\btheta,\bsigma_n)-M(\btheta,\bsigma_0)|
  \leqslant
  \sup_{\btheta \in \Theta}
  |M_n(\btheta,\bsigma_n)-M(\btheta,\bsigma_n)|+
  \sup_{\btheta \in \Theta}
  |M(\btheta,\bsigma_n)-M(\btheta,\bsigma_0)|\\
  \leqslant
  \sup _{\btheta \in \Theta}
  |M_n(\btheta,\bsigma_n)-M(\btheta,\bsigma_n)|+
  |M(\btheta^*,\bsigma_n)-M(\btheta^*,\bsigma_0)|,
\end{align*}
where $\btheta^*:=\argmax_{\btheta \in \Theta}
|M(\btheta,\bsigma_n)-M(\btheta,\bsigma_0)|$
which exists since the supremum is a maximum 
because the function 
$\bTheta\times \bTheta_{\bsigma}\rightarrow 
\mathbb{R}:(\btheta,\bsigma)\rightarrow 
M(\btheta,\bsigma)$ is continuous. 
Indeed, consider a sequence 
$(\btheta_l,\bsigma_l)_{l \in \mathbb{N}}$ 
in $\bTheta\times \bTheta_{\bsigma}$ that 
converges to $(\btheta,\bsigma)$ for 
$l \rightarrow \infty$. Then
\begin{align*}
  \lim_{l \rightarrow \infty} 
  M(\btheta_l,\bsigma_l) 
&= \lim_{l \rightarrow \infty} 
   \int g_{\btheta_l,\bsigma_l}(\bx) 
   d(H_0)(\bx)\\
&= \int \lim _{l \rightarrow \infty} 
    g_{\btheta_l,\bsigma_l}(\bx) d(H_0)(\bx) \\
&= \int g_{\btheta,\bsigma}(\bx) d(H_0)(\bx) \\
&= M(\btheta,\bsigma),
\end{align*}
where we have used the dominated convergence theorem 
in the second equality, which is possible because 
the function $\mathbb{R}^p\rightarrow \mathbb{R}:
\bx\rightarrow g_{\btheta,\bsigma}(\bx)$ is 
bounded due to the form of $\rho_2$ 
in~\eqref{eq:rhotanh}, so it is dominated by an 
integrable function for every 
$\btheta\in\bTheta$ and 
$\bsigma\in \bTheta_{\bsigma}$\,. The third 
equality uses the fact that the function 
$\bTheta\times \bTheta_{\bsigma}\rightarrow 
\mathbb{R}:(\btheta,\bsigma)\rightarrow 
g_{\btheta,\bsigma}(\bx)$ is continuous for every 
$\bx$. By using the continuous mapping theorem 
(Theorem 2.3 in Van der Vaart (2000))
it follows that 
$|M(\btheta^*,\bsigma_n)-M(\btheta^*,\bsigma_0)|
\rightarrow_p 0$ as  $\bsigma_n\rightarrow_p 
\bsigma_0$\,. Moreover,  
$\sup_{\btheta \in \Theta}|M_n(\btheta,\bsigma_n)-
M(\btheta,\bsigma_n)|\rightarrow_p 0$ because the 
set of functions $\lbrace g_{\btheta,\bsigma}: 
\bsigma\in \bTheta_{\bsigma}\rbrace$ is of 
Glivenko-Cantelli type (Example 19.8 in 
Van der Vaart (2000))
because 
$\bTheta_{\bsigma}$ is compact, and the function 
$\bTheta_{\bsigma}\rightarrow \mathbb{R}:
\bsigma\rightarrow g_{\btheta,\bsigma}(\bx)$ is
continuous and dominated by an integrable 
function for every $\bx$ and $\btheta$.
Therefore $M_n(\cdot,\bsigma_n)$ converges 
uniformly to $M(\cdot,\bsigma_0)$, i.e., 
$\sup_{\btheta \in \Theta}|M_n(\btheta,\bsigma_n)-
M(\btheta,\bsigma_0)| \rightarrow_p 0$.
This implies 
$M_n(\btheta_0,\bsigma_n) \rightarrow_p 
M(\btheta_0,\bsigma_0)$. From the assumption we 
have $M_n(\btheta_n,\bsigma_n) \leqslant 
M_n(\btheta_0,\bsigma_n)+o_p(1)$, and thus 
$M_n(\btheta_n,\bsigma_n) \leqslant 
M(\btheta_0,\bsigma_0)+o_p(1)$. Therefore
\begin{align}\label{eq:proofcon}
M(\btheta_n,\bsigma_0)-M(\btheta_0,\bsigma_0)& 
\leqslant M(\btheta_n,\bsigma_0)-
M_n(\btheta_n,\bsigma_n)+o_p(1)\nonumber \\
& \leqslant \sup_{\btheta \in \Theta}
|M_n(\btheta,\bsigma_n)-M(\btheta,\bsigma_0)|+
o_p(1) \rightarrow_p 0\,.
\end{align}
Moreover, by~\eqref{assu_cons2}, for every 
$\varepsilon>0$ and $\bttheta_0\in \bTheta_0$ 
there exists a number $\eta>0$ such that 
$M(\btheta,\bsigma_0) > M(\bttheta_0,
\bsigma_0)+\eta$ for every $\btheta$ with 
$d\left(\btheta, \bTheta_0\right) \geqslant 
\varepsilon$. Therefore the event 
$\lbrace d\left(\btheta_n, \bTheta_0\right) 
\geqslant \varepsilon\rbrace$ is contained in 
the event $\lbrace M(\btheta_n,\bsigma_0) > 
M\left(\btheta_0,\bsigma_0\right)+\eta\rbrace$. 
The probability of the latter event converges 
to 0 in view of \eqref{eq:proofcon}.
\end{proof}

However, Proposition~\ref{prop4} does not yet
ensure that the sequence $\btheta_n$ converges 
in probability to some $\btheta_0 \in \bTheta_0$. 
For this we need the additional assumption that 
there exists a compact subset $\bTheta_s$ of 
$\bTheta$ on which $M(\btheta,\bsigma_0)$ has 
a unique global minimum point $\btheta^*_0$.

\begin{proposition}\label{prop5}
Assume that $\bsigma_n\rightarrow_p \bsigma_0$  
with $\bsigma_n, \bsigma_0\in \bTheta_{\bsigma}$
and that $M(\cdot,\bsigma_0)$ has a unique global 
minimum point $\btheta^*_0=\left(\vect(\bP^*_0),
\bmu^*_0\right)^T$ in the compact set $\bTheta_s$.
Then for any sequence of estimators 
$\btheta_n=(\vect(\bP_n),\bmu_n)^T$ in 
$\bTheta_s$ with $M_n(\btheta_n,
\bsigma_n) \leqslant M_n(\btheta^*_0,
\bsigma_n)+o_p(1)$ we have that 
$(\vect(\bP_n),\bmu_n) \rightarrow_p 
(\vect(\bP^*_0),\bmu^*_0)$. 
\end{proposition}

\begin{proof}[\textbf{Proof of Proposition \ref{prop5}}]
The proof is analogous to the proof of Proposition 
\ref{prop4}, by noting that \eqref{assu_cons2} is 
automatically satisfied as $\bTheta_s$ is compact, 
and the function $\bTheta\rightarrow \mathbb{R}:
\btheta\rightarrow M(\btheta,\bsigma)$  is continuous 
for every $\bsigma\in \bTheta_{\bsigma}$ 
(Problem 5.27, Van der Vaart (2000)).
\end{proof}

For the next step in the consistency of cellRCov 
we need the consistency of $\bSigma_{\MCD}(\bhU)$. 
We use~\eqref{eq:ux} which obtains $\bu$ as
\begin{equation*}
\bu(\bx,\btheta^r,\bsigma)=\argmin_{\bu\in \mathbb{R}^k}g^{\bu}_{\btheta^r,\bsigma}(\bx,\bu),
\end{equation*}
where
\begin{equation*}
g^{\bu}_{\btheta^r,\bsigma}(\bx,\bu) :=
 \rho_2 \left(\frac{1}{\sigma_2}\sqrt{\frac{1}{p}
 \sum_{j=1}^{p} \sigma_{1,j}^2\rho_1\left(\frac{
 x_{j}-\mu_j-\bu^T\bv_j}{\sigma_{1,j}}\right)}
 \right),
\end{equation*}
where $\btheta^r=(\vect(\bV),\bmu)^T \in \bTheta^r$ 
with
\begin{equation*}
\bTheta^r=\lbrace (\vect(\bV),\bmu)\in 
  \mathbb{R}^{p(\rk+1)}: 
  \bmu\in\bTheta_{\bmu} \text{ and } 
  \bV(\bV^T\bV)^{-1}\bV^T\bmu = \bzero
\rbrace.
\end{equation*} 
Here $\bV$ is a parametrization of the projection
matrix $\bP$ in the sense that 
$\bV(\bV^T\bV)^{-1}\bV^T = \bP$. We know that
such a matrix $\bV$ is not unique, but we will 
see later on that different choices of $\bV$
lead to the same result for cellRCov. Denote by 
$H^{\bu}_n=H^{\bu}(H_n,\btheta^r_n,\bsigma_n)$ 
the distribution of 
$\bu(\bx,\btheta^r_n,\bsigma_n)$ when $\bx$ is 
distributed as $H_n$\,, with $\btheta^r_n=\left(\vect(\bV_n),\bmu_n\right)^T$ 
and with 
$H^{\bu}_0=H^{\bu}(H_0,\btheta^r_0,\bsigma_0)$ the 
distribution of $\bu(\bx,\btheta^r_0,\bsigma_0)$.
In Proposition~\ref{prop6} we will assume that
$\btheta^r_n\rightarrow_p \btheta^r$, and in
the proof of Proposition~\ref{prop7} it will be
shown that such a sequence of representations
$\bV_n$ of $\bP_n$ for which
$\vect(\bV_n) \rightarrow_p \vect(\bV)$ always
exists, using the fact that 
$\vect(\bP_n) \rightarrow_p \vect(\bP)$
from Proposition~\ref{prop5}. We denote 
$\bSigma^{\bu}_{\MCD,n}=\bSigma_{\MCD}(H^{\bu}_n)$, 
$\bmu^{\bu}_{\MCD,n}=\bmu_{\MCD}(H^{\bu}_n)$ and 
$\bSigma^{\bu}_{\MCD,0}=\bSigma_{\MCD}(H^{\bu}_0)$, 
$\bmu^{\bu}_{\MCD,0}=\bmu_{\MCD}(H^{\bu}_0)$, where 
$\bmu_{\MCD}(\cdot)$ and $\bSigma_{\MCD}(\cdot)$  
are the MCD functionals of location and scatter 
with parameter $0.5<\alpha<1$.

\begin{proposition}\label{prop6} 
Suppose $H^{\bu}_0$ has a unimodal elliptic
density, $H^{\bu}_0(A)<\alpha$ for 
every hyperplane $A \subset \mathbb{R}^k$, and 
that $H^{\bu}_0$  and $H^{\bu}_n$ have no mass on 
the boundary of the minimizing ellipsoid generated 
by $\bSigma^{\bu}_{\MCD,0}$, $\bmu^{\bu}_{\MCD,0}$ 
and $\bSigma^{\bu}_{\MCD,n}$, $\bmu^{\bu}_{\MCD,n}$.
Moreover, consider the sequences $\bsigma_n$ 
and $\btheta^r_n$ with 
$\bsigma_n\rightarrow_p \bsigma_0$ and 
$\btheta^r_n\rightarrow_p \btheta^r_0$, and assume 
that there exists a set 
$\bTheta_s^{\bu}\subseteq\mathbb{R}^k$ such that $\bTheta_s^{\bu}\rightarrow \mathbb{R}:\bu
\rightarrow g^{\bu}_{\btheta^r_0,\bsigma_0}(\bx,\bu)$ 
is a unique global minimizer for each 
$\bx\in \mathbb{R}^p$,  
$\bu(\bx,\btheta^r_0,\bsigma_0)$ and the sequence $\bu(\bx,\btheta^r_n,\bsigma_n)$ are in 
$\bTheta_s^{\bu}$, and
$\mathbb{R}^p\times\bTheta_s^{\bu}\rightarrow 
\mathbb{R}:(\bx,\bu)\rightarrow 
g^{\bu}_{\btheta^r_0,\bsigma_0}(\bx,\bu)$ is 
continuous. Then
$(\bSigma^{\bu}_{\MCD,n},\bmu^{\bu}_{\MCD,n})
\rightarrow_p 
(\bSigma^{\bu}_{\MCD,0},\bmu^{\bu}_{\MCD,0})$.
\end{proposition}

\begin{proof}[\textbf{Proof of Proposition \ref{prop6}}]
The proof is directly derived from the application of Corollary~4.1 of  Cator and Lopuh\"aa (2012).
Indeed, $H^{\bu}_0$ satisfies (3.1) and (4.1) of Cator and Lopuh\"aa (2012), and guarantees that $ \bSigma^{\bu}_{\MCD,0}$, $ \bmu^{\bu}_{\MCD,0}$ are unique as it has a unimodal elliptical contour density \citep{butler1993asymptotics}.
Moreover, to prove that $H^{\bu}_n$ converges weakly to $H^{\bu}_0$, we have to show that $\bu(\bx,\btheta^r_n,\bsigma_n)$ converges in distribution to $\bu(\bx,\btheta^r_0,\bsigma_0)$.
Let us define
\begin{equation*}
\bu_s(\bx,\btheta^r,\bsigma)=\argmin_{\bu\in \bTheta_s^{\bu}}g^{\bu}_{\btheta^r,\bsigma}(\bx,\bu),
\end{equation*}
then by assumption $\bu_s(\bx,\btheta^r_n,\bsigma_n)=\bu(\bx,\btheta^r_n,\bsigma_n)$ and $\bu_s(\bx,\btheta^r_0,\bsigma_0)=\bu(\bx,\btheta^r_0,\bsigma_0)$ and thus convergences in distribution of $\bu_s(\bx,\btheta^r_n,\bsigma_n)$ to $\bu_s(\bx,\btheta^r_0,\bsigma_0)$ implies  that $\bu(\bx,\btheta^r_n,\bsigma_n)$ converges in distribution to $\bu(\bx,\btheta^r_0,\bsigma_0)$.
To use the continuous mapping theorem (Theorem 2.3, 
Van der Vaart (2000)),
we have to show that the function $\mathbb{R}^p\rightarrow \bTheta_s^{\bu}:\bx\rightarrow \bu_s(\bx,\btheta^r_0,\bsigma_0)$ is continuous. This comes from the application of the maximum theorem (Theorem 9.14, \cite{sundaram1996first})  by using the fact that   by assumption
$\mathbb{R}^p\times\bTheta_s^{\bu}\rightarrow \mathbb{R}:(\bx,\bu)\rightarrow g^{\bu}_{\btheta^r_0,\bsigma_0}(\bx,\bu)$ is continuous, and  $\bu_s(\bx,\btheta^r_0,\bsigma_0)$ is unique for each  $\bx\in \mathbb{R}^p$. Indeed $\mathbb{R}^p\rightarrow \bTheta_s^{\bu}:\bx\rightarrow \bu_s(\bx,\btheta^r_0,\bsigma_0)$ is a single-valued correspondence, i.e., a correspondence where each element in the domain is associated with at most one element in the codomain  \citep{sundaram1996first}. 
From the maximum theorem, we have that the mapping 
$\mathbb{R}^p \to \bTheta_s^{\bu} : \bx \mapsto \bu_s(\bx, \btheta^r_0, \bsigma_0)$
is a semicontinuous correspondence. Since it is also single-valued, it is continuous when viewed as a function (Theorem 9.12, \cite{sundaram1996first}).
Thus, $\bu_s(\bx,\btheta^r_n,\bsigma_n)$ converges in distribution to $\bu_s(\bx,\btheta^r_0,\bsigma_0)$. 
Because $H^{\bu}_0$  and $H^{\bu}_n$  have no mass on the boundary of the minimizing ellipsoids generated by $ \bSigma^{\bu}_{\MCD,0}$, $ \bmu^{\bu}_{\MCD,0}$ and $ \bSigma^{\bu}_{\MCD,n}$, $ \bmu^{\bu}_{\MCD,n}$,  we have that the indicator function corresponding to these ellipsoids are the MCD functional minimizing functions by Theorem 3.2 of Cator and Lopuh\"aa (2012). By Lemma A.7 of Cator and Lopuh\"aa (2012), there exists a ball with centre zero that contains the set of points where the indicator function corresponding to the minimizing ellipsoid is different from zero. Then the convergence result follows.
\end{proof}
Let define $\bSigma^{\bx^0}_n:=\bSigma^{\bx^0}(H_n)=\bV_n\bSigma^{\bu}_{\MCD,n}\bV_n^T$ and $\bSigma^{\bx^0}_0:=\bSigma^{\bx^0}(H_0)=\bV_0\bSigma^{\bu}_{\MCD,0}\bV_0^T$.

\begin{proposition}\label{prop7}
Under the assumptions of Propositions~\ref{prop5} and 
\ref{prop6}, we have that 
$\bSigma^{\bx^0}_n\rightarrow_p\bSigma^{\bx^0}_0$ 
regardless of the choice of the sequence $\bV_n$\,, and 
where $\bSigma^{\bx^0}_0$ does not depend on $\bV_0$.
\end{proposition}
\begin{proof}[\textbf{Proof of Proposition \ref{prop7}}]
Under the assumption of Proposition \ref{prop5}, we have that $\bP_n\rightarrow_p\bP^*_0$ and $\bmu_n\rightarrow_p\bmu_0$. 
However, to apply the results of Proposition \ref{prop6}, we 
additionally need that $\bV_n\rightarrow_p\bV_0$.
We will first show that $\bP_n\rightarrow_p\bP^*_0$ implies that there we can find a sequence $\bV_n$ and $\bV_0$ such that $\bV_n\rightarrow_p\bV_0$. Then, we will show that consistency results are valid regardless of which sequence $\bV_n$ and $\bV_0$ is selected.
Note that $\bP_n=\bV_n(\bV_n^T\bV_n)^{-1}\bV_n^T$ and $\bP^*_0=\bV_0(\bV_0^T\bV_0)^{-1}\bV_0^T$. Let us consider the QR decompostion of $\bV_n=\btV_n\bR_n$ and  $\bV_n=\btV_0\bR_0$, where $\btV_n$ and $\btV_0$ are orthonormal matrices with the same dimension as $\bV_n$, and $\bR_n$ and $\bR_0$ are upper triangular invertible $\rk \times \rk$ matrices.
Then, $\bP_n=\btV_n\btV_n^T$ and $\bP^*_0=\btV_0\btV_0^T$. The condition $\bP_n\rightarrow_p\bP^*_0$ implies that the singular values of the matrix $\bA_n=\btV_0^T\btV_n$ tend to 1 in probability. Then, by applying the singular value decomposition to $\bA_n$, we have that $\bA_n=\bL_n\bS_n\bR_n$,where $\bL_n$ and $\bR_n$ are $\rk \times \rk$ orthonormal matrices and $\bS_n=\diag(s_{1,n},\dots,s_{\rk,n})$, with $s_{\ell,n}\rightarrow_p 1$. Let us define $\bQ_n:=\bR_n\bL_n^T$, then $\bA_n\bQ_n=\bL_n\bS_n\bL_n^T$ and  $\btV_0^T\btV_n\bQ_n\rightarrow_p \bI_{\rk}$. Thus  $||\btV_n\bQ_n-\btV_0||_F^2\rightarrow_p 0$, which implies $\btV_n\bQ_n\rightarrow_p\btV_0$.
For every sequence $\bV_n$ and $\bV_0$ we can find a new sequence $\btV_n\bQ_n$ and $\btV_0$ such that $\btV_n\bQ_n\rightarrow_p\btV_0$, which can be used to apply the results of Proposition \ref{prop6}. Then by direct application of the continuous mapping theorem (Theorem 2.3,
Van der Vaart (2000))
and the results of Proposition \ref{prop6}, we have that $\btV_n\bQ_n\btSigma^{\bu}_{\MCD,n}(\btV_n\bQ_n)^T\rightarrow_p\btV_0\btSigma^{\bu}_{\MCD,0}\btV_0^T$, where $\btSigma^{\bu}_{\MCD,n}$ and $\btSigma^{\bu}_{\MCD,0}$ are computed based on the distribution of $\bu$ corresponding to $\btV_n\bQ_n$ and $\btV_0$.
Note that, $\btV_n\bQ_n=\btV_n\bR_n\bR_n^{-1}\bQ_n=\bV_n\bR_n^{-1}\bQ_n$, and the distribution of $\bu$ corresponding to $\bV_n$ is equal to $\bR_n^{-1}\bQ_n$ the distribution of $\bu$ corresponding to $\btV_n\bQ_n$. Then, $\btV_n\bQ_n\btSigma^{\bu}_{\MCD,n}(\btV_n\bQ_n)^T=\bV_n\bR_n^{-1}\bQ_n\btSigma^{\bu}_{\MCD,n}(\bV_n\bR_n^{-1}\bQ_n)^T=\bV_n\bSigma^{\bu}_{\MCD,n}\bV_n^T$, where the last equality comes from the affine equivariance of the MCD covariance estimator. Similarly, we have that $\btV_0\btSigma^{\bu}_{\MCD,0}\btV_0^T=\bV_0\bSigma^{\bu}_{\MCD,0}\bV_0^T$. Thus, $\bV_n\bSigma^{\bu}_{\MCD,n}\bV_n^T\rightarrow_p\bV_0\bSigma^{\bu}_{\MCD,0}\bV_0^T$.
This result is the same for any sequence $\bV_n$ and $\bV_0$.
To prove that $\bSigma^{rs}_0$ is independent of $\bV_0$ let consider two different $\bV_0$, i.e., $\bV^1_0$ and  $\bV_{0,2}$. By assumption $\bP^*_0$ is unique, then, we have $\bV_{0,1}(\bV_{0,1}^T\bV_{0,1})^{-1}\bV_{0,1}^T=\bV_{0,2}(\bV_{0,2}^T\bV_{0,2})^{-1}\bV_{0,2}^T$ and $\bV_{0,1}=\bV_{0,2}(\bV_{0,2}^T\bV_{0,2})^{-1}\bV_{0,2}^T\bV_{0,1}$. Let us denote with $\bSigma^{\bu,1}_{\MCD,0}$ and $\bSigma^{\bu,2}_{\MCD,0}$, the MCD covariance estimator applied to $\bu$ corresponding to $\bV_{0,1}$ and $\bV_{0,2}$. Then $\bV_{0,1}\bSigma^{\bu,1}_{\MCD,0}\bV_{0,1}^T=\bV_{0,2}(\bV_{0,2}^T\bV_{0,2})^{-1}\bV_{0,2}^T\bV_{0,1}\bSigma^{\bu,1}_{\MCD,0}(\bV_{0,2}(\bV_{0,2}^T\bV_{0,2})^{-1}\bV_{0,2}^T\bV_{0,1})^T=\bV_{0,2}\bSigma^{\bu,2}_{\MCD,0}\bV_{0,2}^T$, where last equality comes from the affine equivariance of the MCD covariance estimator and the fact that the distribution of $\bu$ corresponding to $\bV_{0,1}$ is equal to $(\bV_{0,2}^T\bV_{0,2})^{-1}\bV_{0,2}^T\bV_{0,1}$ the distribution of $\bu$ corresponding to $\bV_{0,2}$.
Thus the proposition follows.
\end{proof}

The final step in the construction of cellRCov
is orthogonal to the estimated principal subspace.
Let us put $\bSigma_{\xort,n}^R:=\bSigma_{\xort}^R(H_n)$. 
In the following we will prove that 
$\bSigma_{\xort,n}^R$ is a consistent estimator of $\bSigma_{\xort,0}^R:=\bSigma_{\xort}^R(H_0)$, that is,  
$\bSigma_{\xort,n}^R$ converges in probability to 
$\bSigma_{\xort,0}^R$.

Consider the function 
\begin{equation*}
G^{\xort}(\bSigma_{\xort}^R,\bT,H):=
  \Big|\Big|\bSigma_{\xort}^R-
  \frac{\E_{H}\left[\bg^{\xort}_{1,\bT}(\bx)\right]}
  {\E_{H}\left[g^{\xort}_{2,\bT}(\bx)\right]}
  \Big|\Big|,
\end{equation*}
where $\bT=\left(\vect(\bP), \bmu,\bsigma\right)^T$ and $\bSigma_{\xort}^R\in \mathbb{R}^{p^2}$,  
\begin{multline*}
\bg^{\xort}_{1,\bT}(\bx) :=\vect\Big( (1-\delta)w^{\case}\btW(\bx - \bmu-\bV\bu)(\bx - \bmu-\bV\bu)^T
  \btW\\+\delta w^{\case}\left(\btW(\bx - \bmu-\bV\bu)(\bx - \bmu-\bV\bu)^T
  \btW\right)\odot\bI_p\Big),
\end{multline*}
and \begin{equation*}
g^{\xort}_{2,\bT}(\bx) := \frac{1}{p^2}
  \sum_{j=1}^p\sum_{\ell=1}^p 
  w^{\case}w_{j}^{\cell}w_{\ell}^{\cell},
\end{equation*}
and $\bx^0$ is 
defined as in~\eqref{eq:ux2}.
Then let us indicate with $M^{\xort}_n(\bSigma_{\xort}^R,\bT):=G^{\xort}(\bSigma_{\xort}^R,\bT, H_n)$, and  $M^{\xort}(\bSigma_{\xort}^R,\bT):=G^{\xort}(\bSigma_{\xort}^R,\bT, H_0)$.  Let us further consider compact subsets $\bTheta_{\bSigma_{\xort}^R}$ of $\mathbb{R}^{p^2}$ and $\bTheta_{\bT}$ of $\mathbb{R}^{(p(p+1))}\times \bTheta_{\bsigma}$.
We aim to study the convergence properties of a minimizing value $\bsigma_{\xort,n}^R:=\vect(\bSigma_{\xort,n}^R)^T$ of $M^{\xort}_n(\bSigma_{\xort}^R,\bT_n)$  to a minimizing value $\bsigma_{\xort,0}^R:=\vect(\bSigma_{\xort,0}^R)^T$ of $M^{\xort}(\cdot,\bT_0)$, where $\bT_n:=\left(\vect(\bP_n), \bmu_n,\bsigma_n\right)^T$ and   $\bT_0:=\left(\vect(\bP_0), \bmu_0,\bsigma_0\right)^T$.

\begin{proposition}\label{prop8}
Assume that $\bT_n\rightarrow_p \bT_0$, with 
$\bT_n, \bT_0\in \bTheta_{\bT}$, the functions 
$\bTheta_{\bT}\rightarrow \mathbb{R}^{p^2}:
\bT\rightarrow \bg^{\xort}_{1,\bT}(\bx)$ and 
$\bTheta_{\bT}\rightarrow \mathbb{R}:
\bT\rightarrow g^{\xort}_{2,\bT}(\bx)$ are 
continuous for every $\bx$, and that the functions
$\mathbb{R}^p\rightarrow \mathbb{R}^{p^2}:
\bx\rightarrow \bg^{\xort}_{1,\bT}(\bx)$ and 
$\mathbb{R}^p\rightarrow \mathbb{R}:
\bx\rightarrow g^{\xort}_{2,\bT}(\bx)$ are 
dominated by an integrable function for every 
$\bT\in\bTheta_{\bT}$.
Then we have for any sequence of estimators 
$\bsigma_{\xort,n}^R$ in $\bTheta_{\bsigma_{\xort}}$ with $M_n\left(\bsigma_{\xort,n}^R,\bT_n\right) \leqslant M_n\left(\bsigma_{\xort,0}^R,\bT_n\right)+o_p(1)$ that 
$||\bsigma_{\xort,n}^R-\bsigma_{\xort,0}^R||
\rightarrow_p 0$.
\end{proposition}

\begin{proof}[\textbf{Proof of Proposition \ref{prop8}}]
Let us define 
\begin{equation*}
\bPsi_1(\bT,H)=\E_{H}\left[\bg^{\xort}_{1,\bT}(\bx)\right]
\end{equation*}
and 
\begin{equation*}
\Psi_2(\bT,H)=\E_{H}\left[g^{\xort}_{2,\bT}(\bx)\right].
\end{equation*}
Moreover, let us further consider
$\bPsi_{1,n}(\bT):=\bPsi_1(\bT,H_n)$, $\bPsi_{1,0}(\bT):=\bPsi_1(\bT,H_0)$, $\Psi_{2,n}(\bT):=\Psi_2(\bT,H_n)$, and $\Psi_{2,0}(\bT):=\Psi_2(\bT,H_0)$.
Note that \begin{equation*}
    |\Psi_{2,n}(\bT_n)-\Psi_{2,0}(\bT_0)|\leqslant
    |\Psi_{2,n}(\bT_n)-\Psi_{2,0}(\bT_n)|+|\Psi_{2,0}(\bT_n)-\Psi_{2,0}(\bT_0)|.
\end{equation*}
The function $\bTheta_{\bT}\rightarrow \mathbb{R}:\bT\rightarrow \Psi_{2,0}(\bT)$  is continuous. In fact, denote $(\bT_l)_{l \in \mathbb{N}}$ a sequence in $\bTheta_{\bT}$ that converges to $\bT^*$ for $l \rightarrow \infty$. We have
\begin{align}
\lim _{l \rightarrow \infty} \Psi_{2,0}(\bT_l) & =\lim _{l \rightarrow \infty} \int g^{\xort}_{2,\bT_l}(\bx) d(H_0)(\bx) \\
& =\int \lim _{l \rightarrow \infty} g^{\xort}_{2,\bT_l}(\bx) d(H_0)(\bx) \\
& =\int  g^{\xort}_{2,\bT^*}(\bx) d(H_0)(\bx) \\
& =\Psi_{2,0}(\bT^*),
\end{align}
where we have used the dominated convergence theorem in the second equality, which is possible because the function $\mathbb{R}^p\rightarrow \mathbb{R}:\bx\rightarrow g^{\xort}_{2,\bT}(\bx)$, is dominated by an integrable function for every $\bT\in\bTheta_{\bT}$, and the fact that the function $\bTheta_{\bT}\rightarrow \mathbb{R}:\bT\rightarrow g^{\xort}_{2,\bT}(\bx)$ is continuous for every $\bx$ in the third equality.
Thus, by using the continuous mapping theorem, we have that $|\Psi_{2,0}(\bT_n)-\Psi_{2,0}(\bT_0)|\rightarrow_p 0$ as  $\bT_n\rightarrow_p \bT_0$.

Moreover,  $ |\Psi_{2,n}(\bT_n)-\Psi_{2,0}(\bT_n)|\rightarrow_p 0$ being the set of functions   $\lbrace g^{\xort}_{2,\bT}: \bT\in \bTheta_{\bT}\rbrace$ \textit{Glivenko-Cantelli} (Example 19.8, 
Van der Vaart (2000))
because $\bTheta_{\bT}$ is compact, the function $\bTheta_{\bT}\rightarrow \mathbb{R}:\bT\rightarrow g^{\xort}_{2,\bT}(\bx)$ is continuous and is dominated by an integrable function for every $\bx$.
Thus, $\Psi_{2,n}(\bT_n) \rightarrow_p \Psi_{2,0}(\bT_0)$.
For the same arguments, we have $\bPsi_{1,n}(\bT_n) \rightarrow_p \bPsi_{1,0}(\bT_0)$.

As for each $\bsigma_{\xort}^R \in \bTheta_{\bsigma_{\xort}} $, $G^{\xort}(\bsigma_{\xort}^R,\bT, H)$ is a continuous function of $\bPsi_1(\bT,H)$ and $\Psi_2(\bT,H)$, by using the continuous mapping theorem, we have that $|M^{\xort}_n(\bsigma_{\xort}^R,\bT_n)-M^{\xort}(\bsigma_{\xort}^R,\bT_0)|\rightarrow_p 0$ as  $n\rightarrow \infty$ and, thus,  $M^{\xort}_n(\bsigma_{\xort}^R,\bT_n) \rightarrow_p M^{\xort}(\bsigma_{\xort}^R,\bT_0)$. By hypothesis, we have $M^{\xort}_n\left(\bsigma_{\xort,n}^R,\bT_n\right) \leqslant M_n\left(\bsigma_{\xort,0}^R,\bT_n\right)+o_p(1)$, and, thus, $M^{\xort}_n\left(\bsigma_{\xort,n}^R,\bT_n\right) \leqslant M_n\left(\bsigma_{\xort,0}^R,\bT_0\right)+o_p(1)$. Hence
\begin{align}\label{eq:proofconerror}
M^{\xort}\left(\bsigma_{\xort,n}^R,\bT_0\right)-&M^{\xort}\left(\bsigma_{\xort,0}^R,\bT_0\right) \leqslant M^{\xort}\left(\bsigma_{\xort,n}^R,\bT_0\right)-M^{\xort}_n\left(\bsigma_{\xort,n}^R,\bT_n\right)+o_p(1)\nonumber \\
& \leqslant  \sup _{\bsigma_{\xort}^R \in \Theta_{\bsigma_{\xort}^R}}|M^{\xort}_n(\bsigma_{\xort}^R,\bT_n)-M^{\xort}(\bsigma_{\xort}^R,\bT_0)|+o_p(1) \rightarrow_p 0 .
\end{align}
As $\bTheta_{\bsigma_{\xort}}\rightarrow \mathbb{R}:\bsigma_{\xort}^R\rightarrow M^{\xort}(\bsigma_{\xort}^R,\bT_0)$  is a continuous function in the compact set $\bTheta_{\bsigma_{\xort}}$ and $\bsigma_{\xort,0}^R$ is unique, we have that for every $\varepsilon>0$, there exists a number $\eta>0$ such that $M^{\xort}(\bsigma_{\xort}^R,\bT_0)>M^{\xort}\left(\bsigma_{\xort,0}^R,\bT_0\right)+\eta$ for every $\bsigma_{\xort}^R$ with $d\left(\bsigma_{\xort}^R, \bsigma_{\xort,0}^R\right) := ||\bsigma_{\xort,n}^R-\bsigma_{\xort,0}^R|| \geqslant \varepsilon$. Thus, the event $\lbrace d\left(\bsigma_{\xort,n}^R, \bsigma_{\xort,0}^R\right) \geqslant \varepsilon\rbrace$ is contained in the event $\lbrace M^{\xort}(\bsigma_{\xort,n}^R,\bT_0)>M^{\xort}\left(\bsigma_{\xort,0}^R,\bT_0\right)+\eta\rbrace$. The probability of the latter event converges to 0 in view of \eqref{eq:proofconerror}.
\end{proof}

\begin{proof}[\textbf{Proof of 
Theorem~\ref{consistSigma}}]
Note that 
$\bSigma_n=\bSigma^{\bx^0}_{\MCD,n}+\bSigma_{\xort,n}^R$ and $\bSigma_n=\bSigma^{\bx^0}_{\MCD,0}+\bSigma_{\xort,0}^R$\,.
Under the assumptions of Propositions~\ref{prop5} 
and~\ref{prop6} we have that $\bSigma^{\bx^0}_{\MCD,n}
\rightarrow_p\bSigma^{\bx^0}_{\MCD,0}$. 
Under the assumptions of Proposition \ref{prop8} we have 
that $\bSigma_{\xort,n}^R\rightarrow_p\bSigma_{\xort,0}^R$\,.
The theorem follows by direct application of the 
continuous mapping theorem.
\end{proof}
We now look at asymptotic normality, where the
casewise IF of $\Sigma$ provides the expression
of the asymptotic covariance matrix.

\begin{proof}[\textbf{Proof of Theorem \ref{asynormSigma}}]
Let us consider a multivariate distribution $G$.
Then the first order von Mises expansion 
\citep{fernholz2012mises} around $H_0$
of the functional $\bSigma$ at $G$ is 
\begin{equation*}
  \vect\left(\bSigma(G)\right)=
  \vect\left(\bSigma_0\right)
  +\int \IFu_{\case}\left(\bx,\bSigma,H_0\right)
  d(G-H_0)(\bx)+\mbox{Rem}(G-H_0)\,,
\end{equation*}
where $\mbox{Rem}(G-H_0)$ is a remainder term. 
Replacing $G$ by the empirical distribution
$H_n$ and using the fact that 
$\int \IFu_{\case}\left(\bx,\bSigma,
H_0\right)dH_0(\bx)=0$, we have
\begin{equation*}
    \vect\left(\bSigma_n \right)=\vect\left(\bSigma_0\right)+\int \IFu_{\case}\left(\bx,\bSigma,H_0\right)dH_n(\bx)+\mbox{Rem}(H_n-H_0).
\end{equation*}
Since $\int \IFu_{\case}\left(\bx,\bSigma,H_0\right)dH_n(\bx)=\frac{1}{n}\sum_{i=1}^{n}\IFu_{\case}\left(\bx_i,\bSigma,H_0\right)$, we have
\begin{equation*}
    \sqrt{n}\left[\vect\left(\bSigma_n \right)-\vect\left(\bSigma_0\right)\right]=\frac{\sqrt{n}}{n}\sum_{i=1}^{n}\IFu_{\case}\left(\bx_i,\bSigma,H_0\right)+\sqrt{n}\,\mbox{Rem}(H_n-H_0).
\end{equation*}
If $\sqrt{n}\,\mbox{Rem}(H_n-H_0)$ becomes negligible 
as $n\rightarrow \infty$, the multivariate 
central limit theorem yields \eqref{eq:asym}.

The assumption that $\sqrt{n}\,\mbox{Rem}(H_n - H_0)$ 
becomes negligible is a reasonable regularity 
condition because it contains an additional 
factor of $1/\sqrt{n}$ compared to the leading 
term in the expansion, and $H_n$ converges to 
$H_0$ by the Glivenko-Cantelli theorem 
(Van der Vaart, 2000).
Rigorous treatment of this condition has been 
addressed in various ways in the literature. 
For example, \cite{mises1947asymptotic} 
assumed that the remainder could be expressed 
as a second derivative term combined with a 
higher-order remainder. Other authors have 
explored specific forms of differentiation, 
such as the Hadamard (or compact) derivative 
\citep{fernholz2012mises} and the Fréchet 
derivative \citep{clarke2018robustness}.
\end{proof}

\newpage

\section{\large Robust Parallel Analysis for 
 Selecting the Rank \texorpdfstring{$\rk$}{k}}
\label{sec:supp-rank}
The robust PA procedure is presented in 
Algorithm~\ref{alg:robustPA}. In our implementation 
we use $B=100$ reference samples and set
$\alpha=0.01$, so that the threshold is the 
empirical $(1-\alpha)$-quantile of the reference 
distribution. The maximum rank $k_{\max}$ is 
chosen large enough to contain all relevant 
components, while remaining computationally
reasonable.
\begin{algorithm}[H]
\caption{Robust Parallel Analysis for selecting $\rk$}
\label{alg:robustPA}
\begin{algorithmic}[1]
\REQUIRE The standardized data $\bZ$, 
  maximum rank $k_{\max}$, number of reference samples $B$.
\STATE Fit cellPCA to $\bZ$ for ranks $s=1,\ldots,k_{\max}$\,.
\STATE Compute the robust objective values $\nu_s^{\mathrm{rob}}$
  from~\eqref{eq:objP}.
\STATE Compute $\nu_0^{\mathrm{rob}}$ from~\eqref{eq:objP} 
  using medians of the columns of $\bZ$.
\STATE Compute the $\ell_s^{\mathrm{rob}}$ for $s=1,\ldots,k_{\max}$\,.
   
\FOR{$b=1,\ldots,B$}
    \STATE Generate an $n\times p$ reference matrix $\bZ_b^0$ with independent
    standard normal entries.
    \STATE Standardize each variable of $\bZ_b^0$\,.
    \STATE Apply CPCA to $\bZ_b^0$ for ranks $s=1,\ldots,k_{\max}$\,.
    \STATE For each $s$, evaluate the objective~\eqref{eq:objP} at the
    CPCA rank-$s$ fit, obtaining $\nu_{s,b}^{\mathrm{rob}}$.
    \STATE Compute $\nu_{0,b}^{\mathrm{rob}}$ from~\eqref{eq:objP} using
    medians.
    \STATE Compute $\ell_{s,b}^{\mathrm{rob}}
        =
        \nu_{s-1,b}^{\mathrm{rob}}
        -
        \nu_{s,b}^{\mathrm{rob}}$.
\ENDFOR
    \STATE Let $c_s$ be the empirical $(1-\alpha)$-quantile of
$\ell_{s,1}^{\mathrm{rob}},\ldots,\ell_{s,B}^{\mathrm{rob}}$.
\STATE Select the rank $\rk$ as the number of consecutive components satisfying $
    \ell_s^{\mathrm{rob}} > c_s$.
If no such $s$ exists, set $\rk$ as $k_{\max}$.
\end{algorithmic}
\end{algorithm}

\section{\large Additional simulation results}
\label{app:addsim}

In addition to the A09 covariance model in 
Section~\ref{sec:simulation}, we also consider
the A06 covariance matrix with  
$\sigma_{j\ell}=(-0.6)^{|j-\ell\,|}$.
The {\it planar} covariance matrix is the 
correlation matrix with the eigenvectors of the 
A09 covariance and eigenvalues such that the first 
component explains 53\p of the total variance, 
the first two together explain 90\p, and all 
subsequent eigenvalues are tiny. Finally, 
the {\it dense} covariance matrix has
$\sigma_{j\ell}=1$ for $j=\ell$ and 
$\sigma_{j\ell}=0.8$ otherwise. 

Figures~\ref{fig:results_A06}, 
\ref{fig:results_planar} and 
\ref{fig:results_dense} show the average KL in 
the presence of either cellwise outliers, 
casewise outliers, or both, again for 
$p=\lbrace30,60,120\rbrace$ but now for the 
covariance models A06, planar, and dense.
Figures~\ref{fig:results_A06_NA}, 
\ref{fig:results_planar_NA}, and
\ref{fig:results_dense_NA} show the 
corresponding results when 20\p of randomly 
selected cells were made NA. All of these 
curves are qualitatively similar to those 
for the A09 covariance model in
Section~\ref{sec:simulation}.\\

\vspace{5mm}

\begin{figure}[!ht]
\centering
 
\begin{tabular}{M{0.0005\textwidth}M{0.29\textwidth}M{0.29\textwidth}M{0.32\textwidth}}
   &\large \textbf{Cellwise}  & \large \textbf{Casewise} &\large{\textbf{Casewise \& Cellwise}} \\
   [-4mm]
   \rotatebox{90}{\textbf{\footnotesize{$p=30$}}}&\includegraphics[width=.31\textwidth]
  {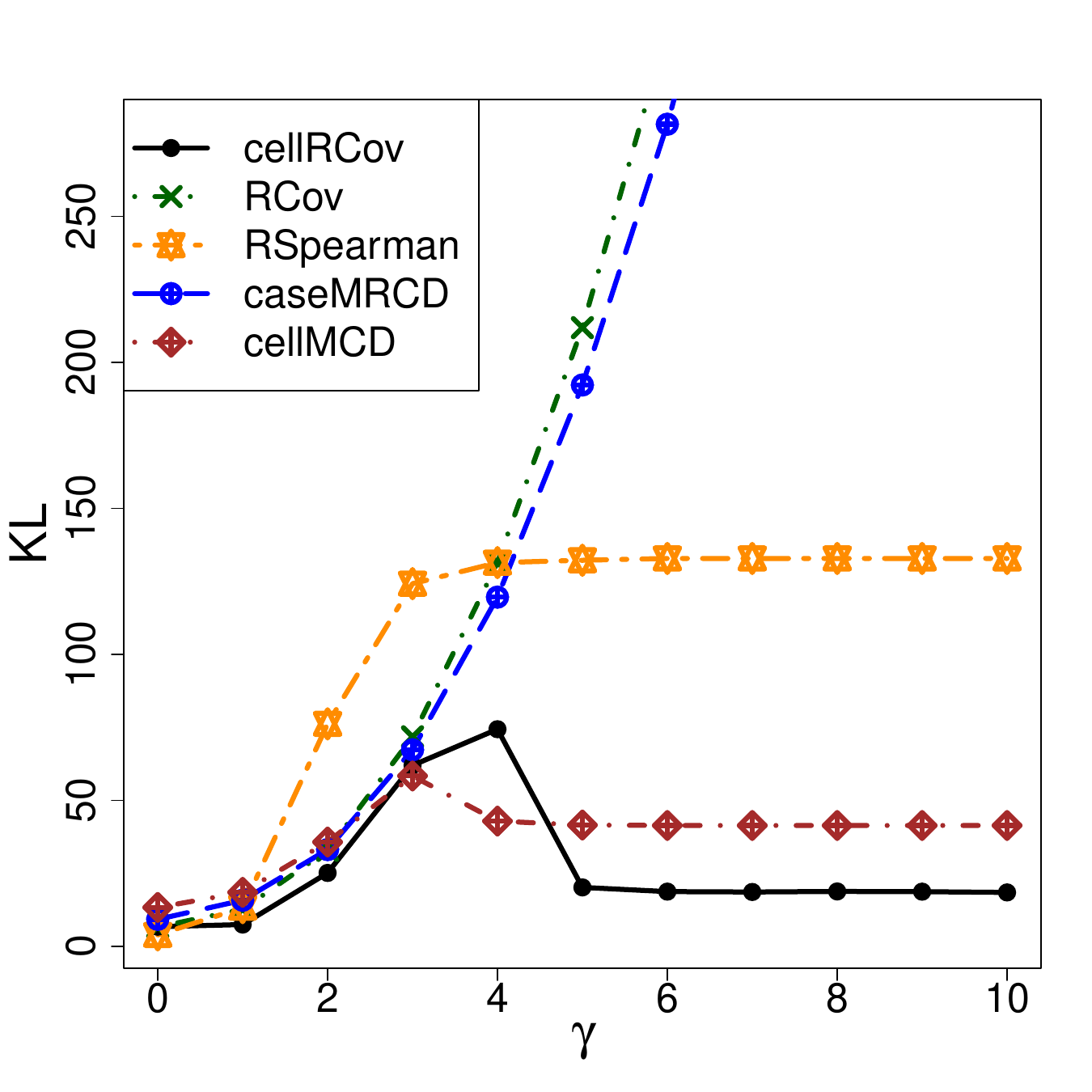}&\includegraphics[width=.31\textwidth]
  {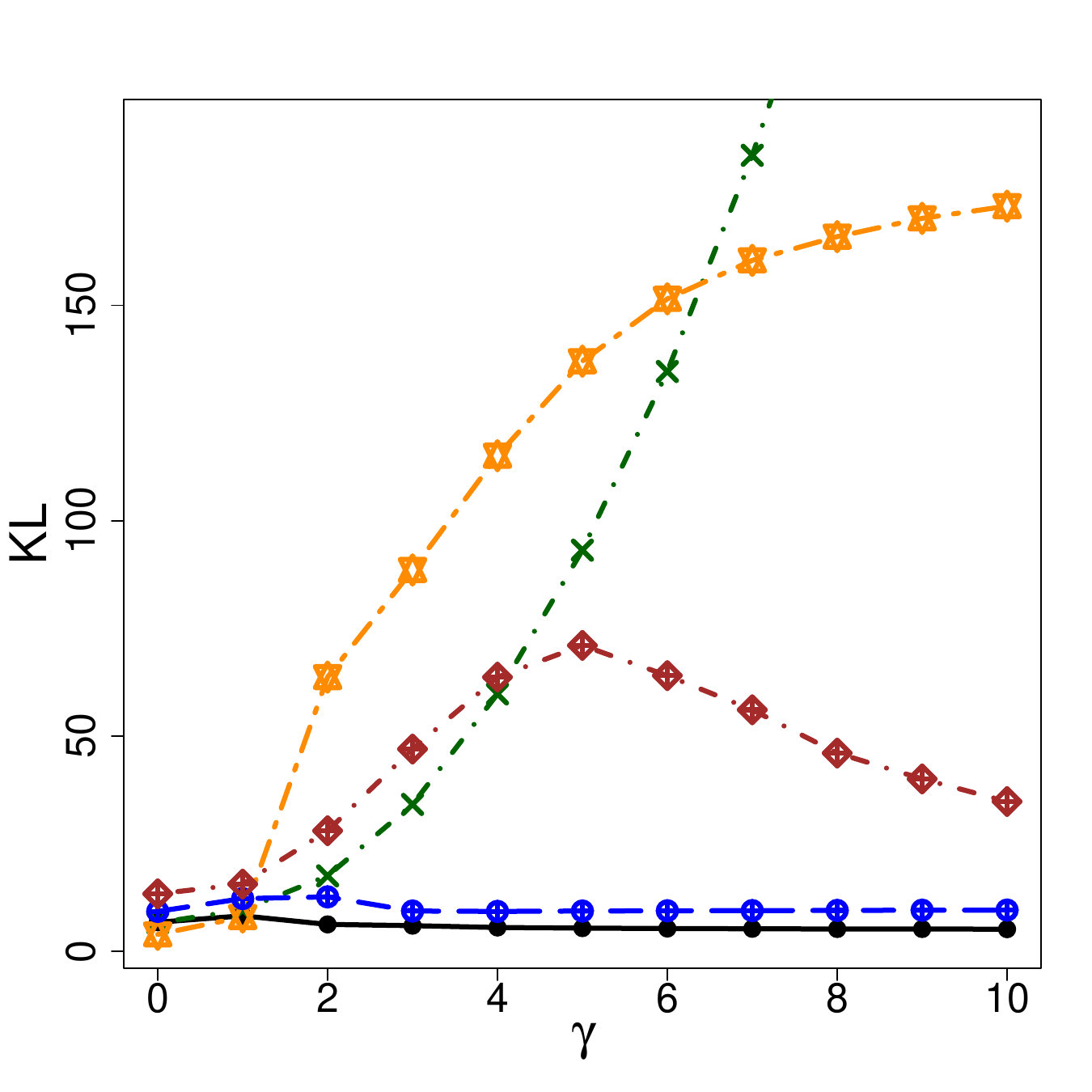}&\includegraphics[width=.31\textwidth]
  {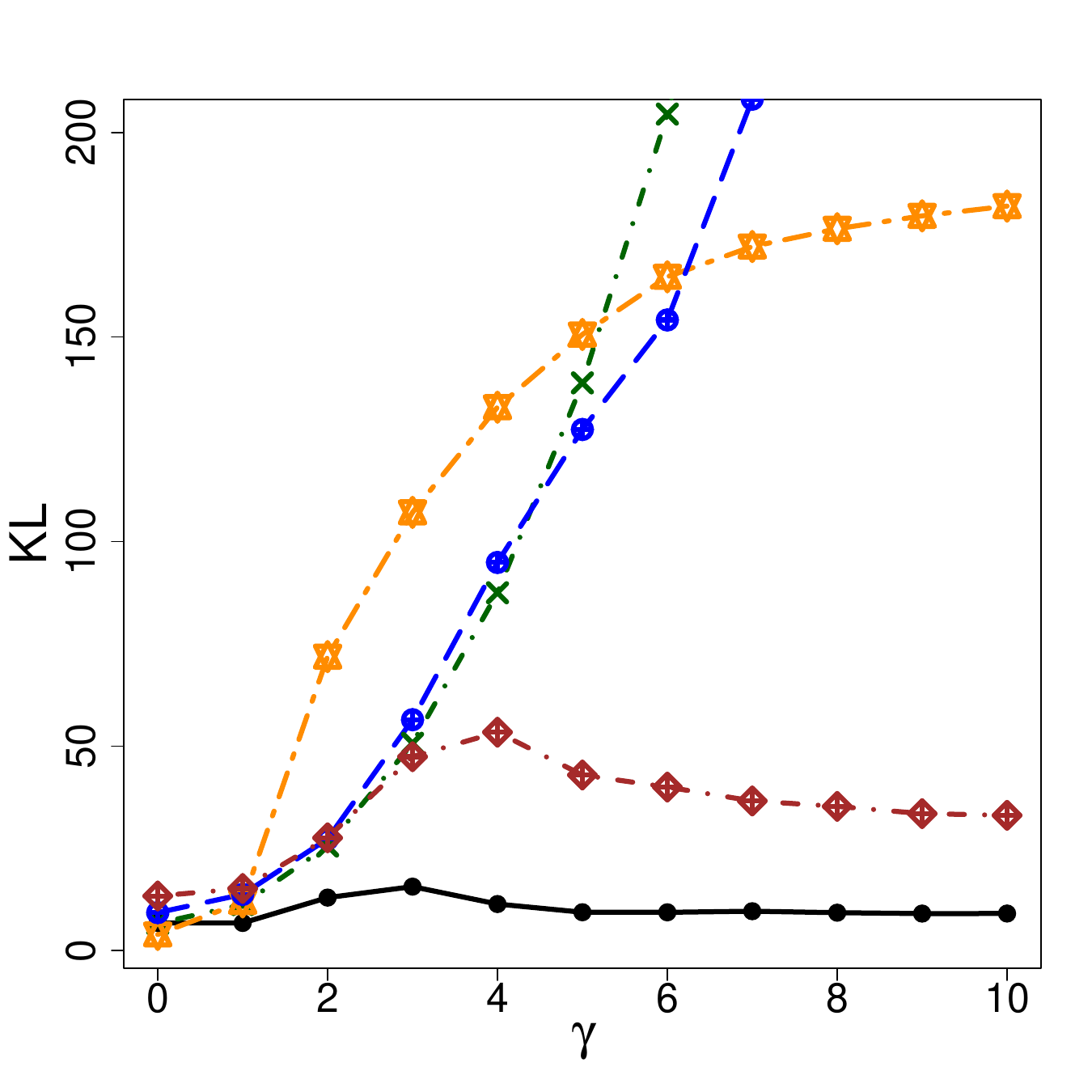}   \\ [-4mm]  \rotatebox{90}{\textbf{\footnotesize{$p=60$}}}&\includegraphics[width=.31\textwidth]
  {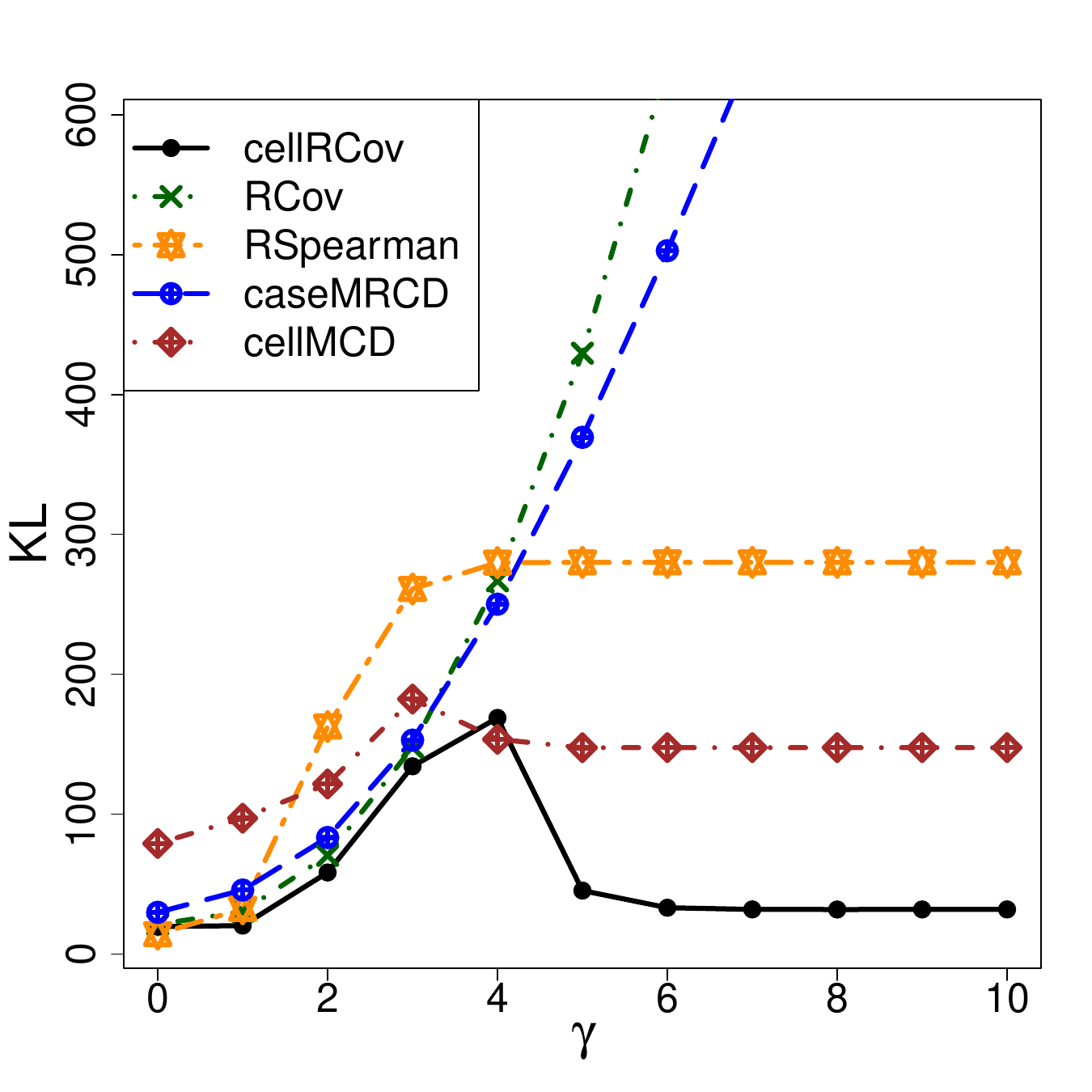}&\includegraphics[width=.31\textwidth]
  {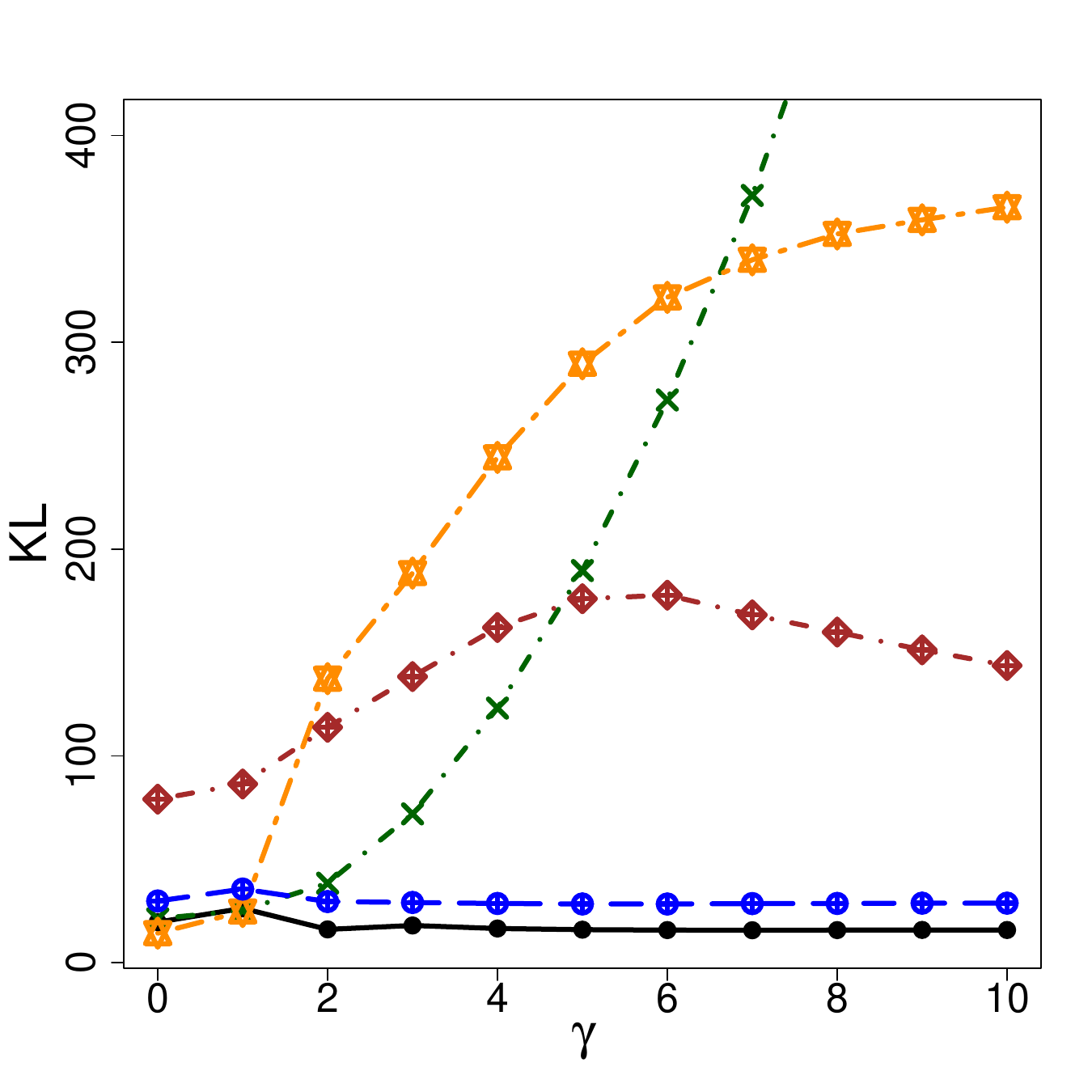}&\includegraphics[width=.31\textwidth]
  {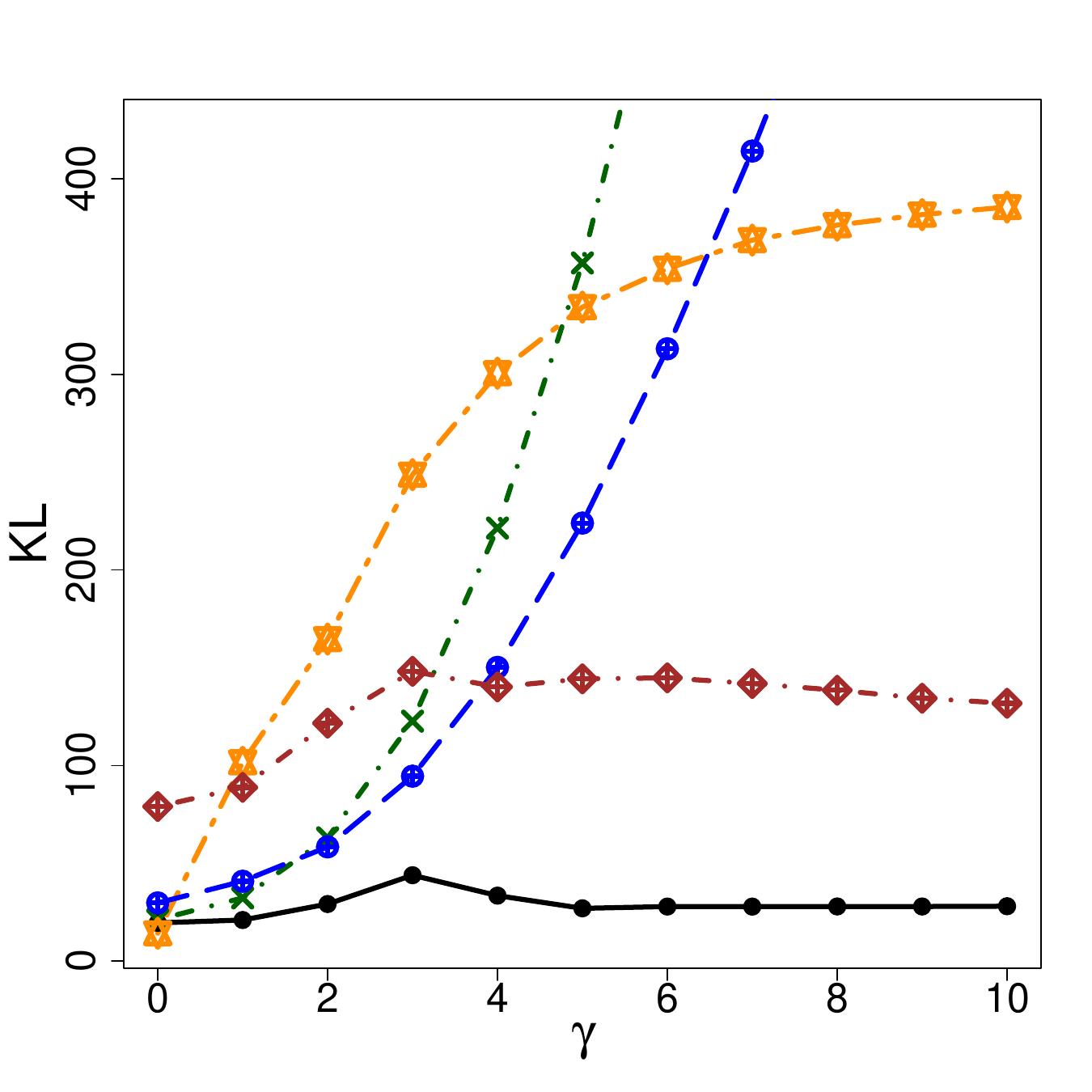}   \\ [-4mm]  \rotatebox{90}{\textbf{\footnotesize{$p=120$}}}&\includegraphics[width=.31\textwidth]
  {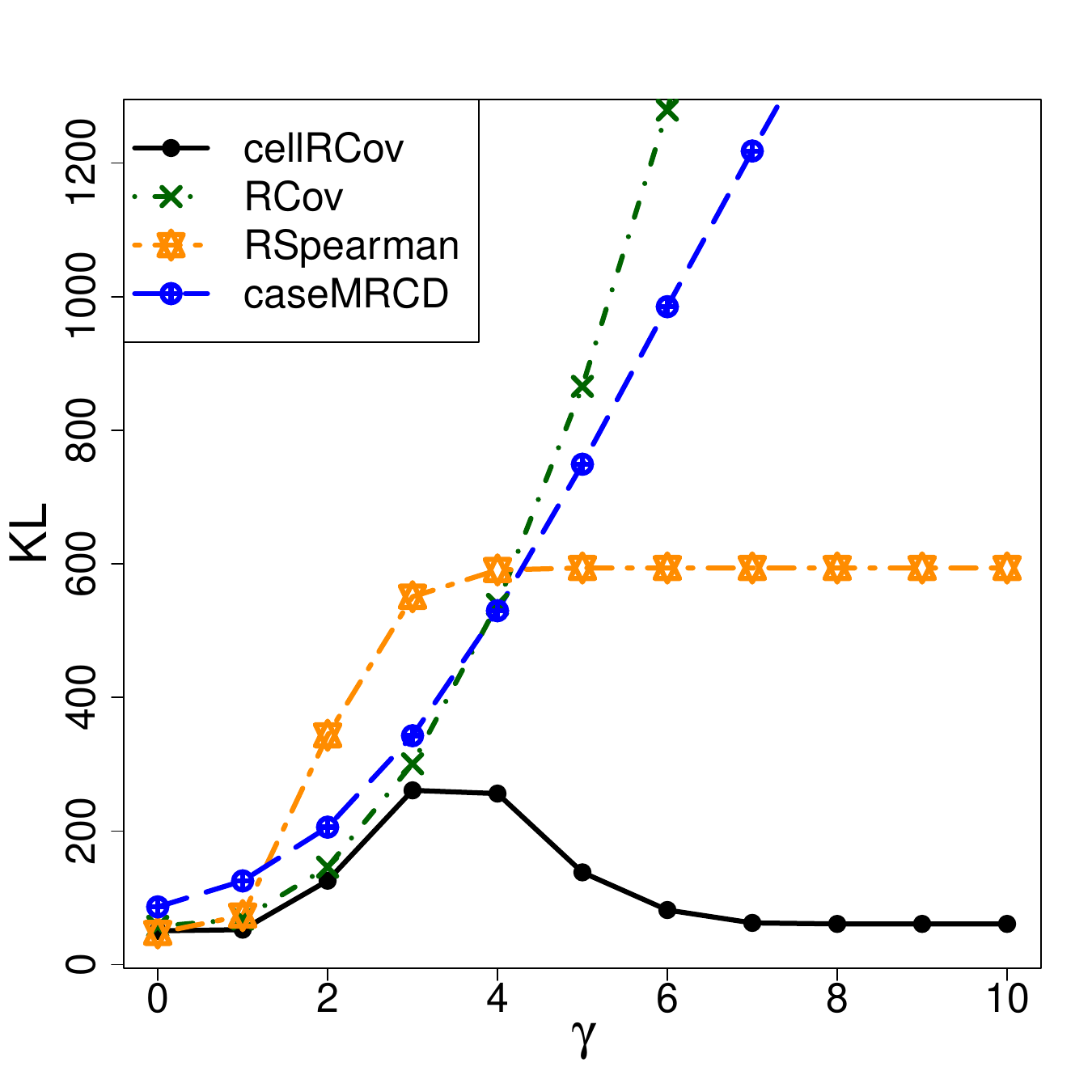}&\includegraphics[width=.31\textwidth]
  {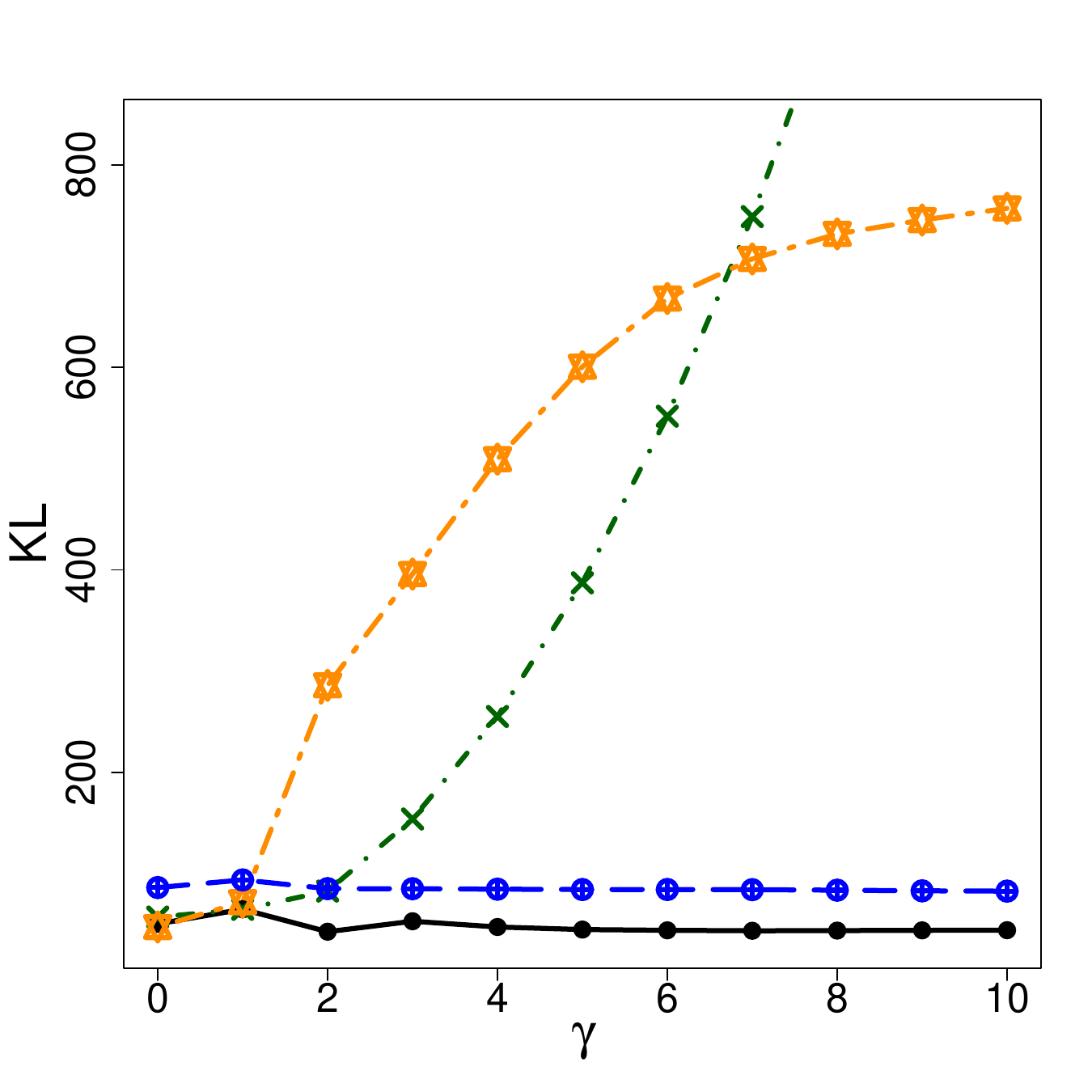}&\includegraphics[width=.31\textwidth]
  {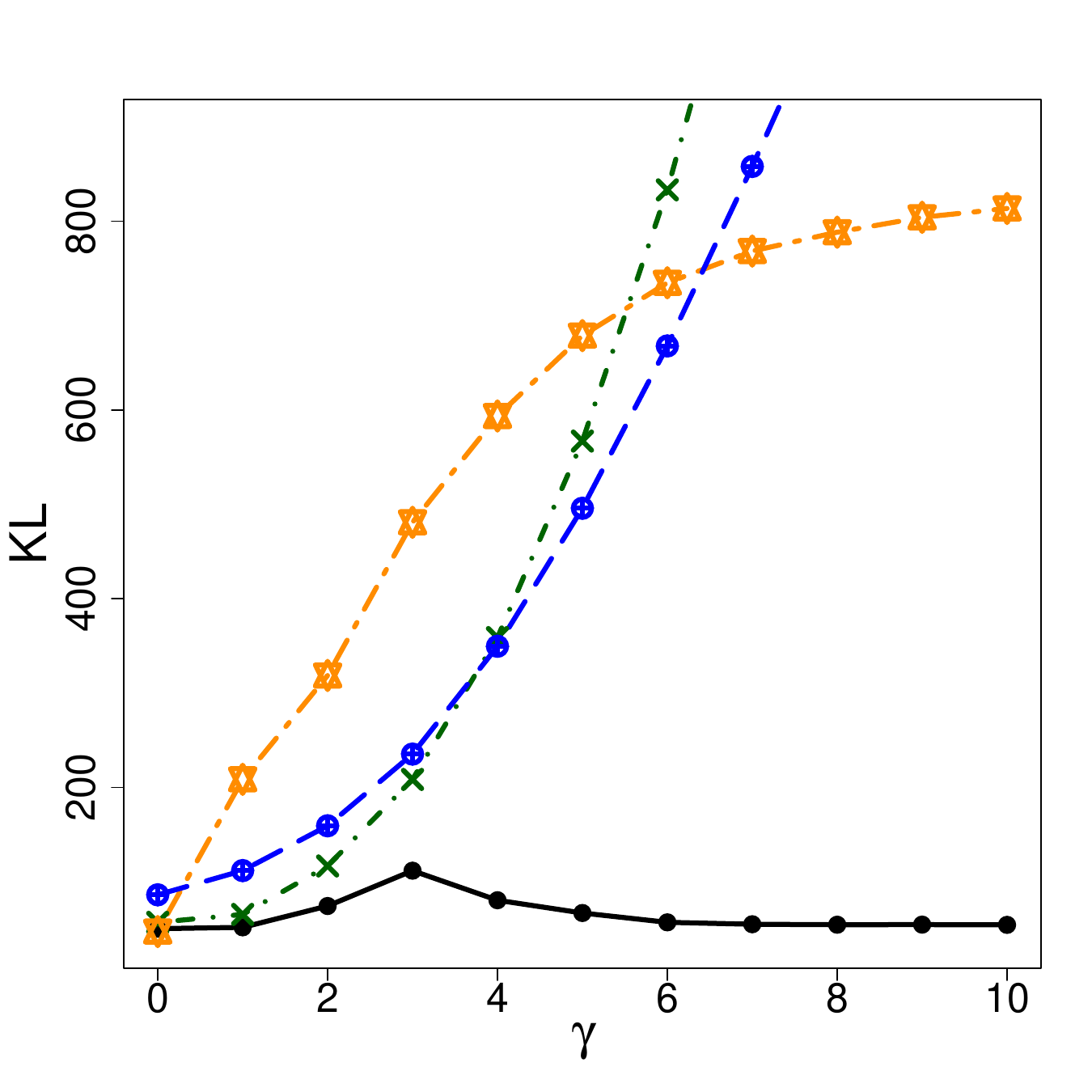}   
\end{tabular}
\caption{Average KL attained by cellRCov, RCov, 
RSpearman, caseMRCD, and cellMCD in the presence of either 
cellwise outliers, casewise outliers, or both, 
for the A06 covariance model in dimensions 
$p$ in $\lbrace30,60,120\rbrace$.}
\label{fig:results_A06}
\end{figure}

\begin{figure}[!ht]
\centering
 
 \begin{tabular}{M{0.0005\textwidth}M{0.29\textwidth}M{0.29\textwidth}M{0.32\textwidth}}
   &\large \textbf{Cellwise}  & \large \textbf{Casewise} &\large{\textbf{Casewise \& Cellwise}} \\
   [-4mm]
   \rotatebox{90}{\textbf{\footnotesize{$p=30$}}}&\includegraphics[width=.31\textwidth]
  {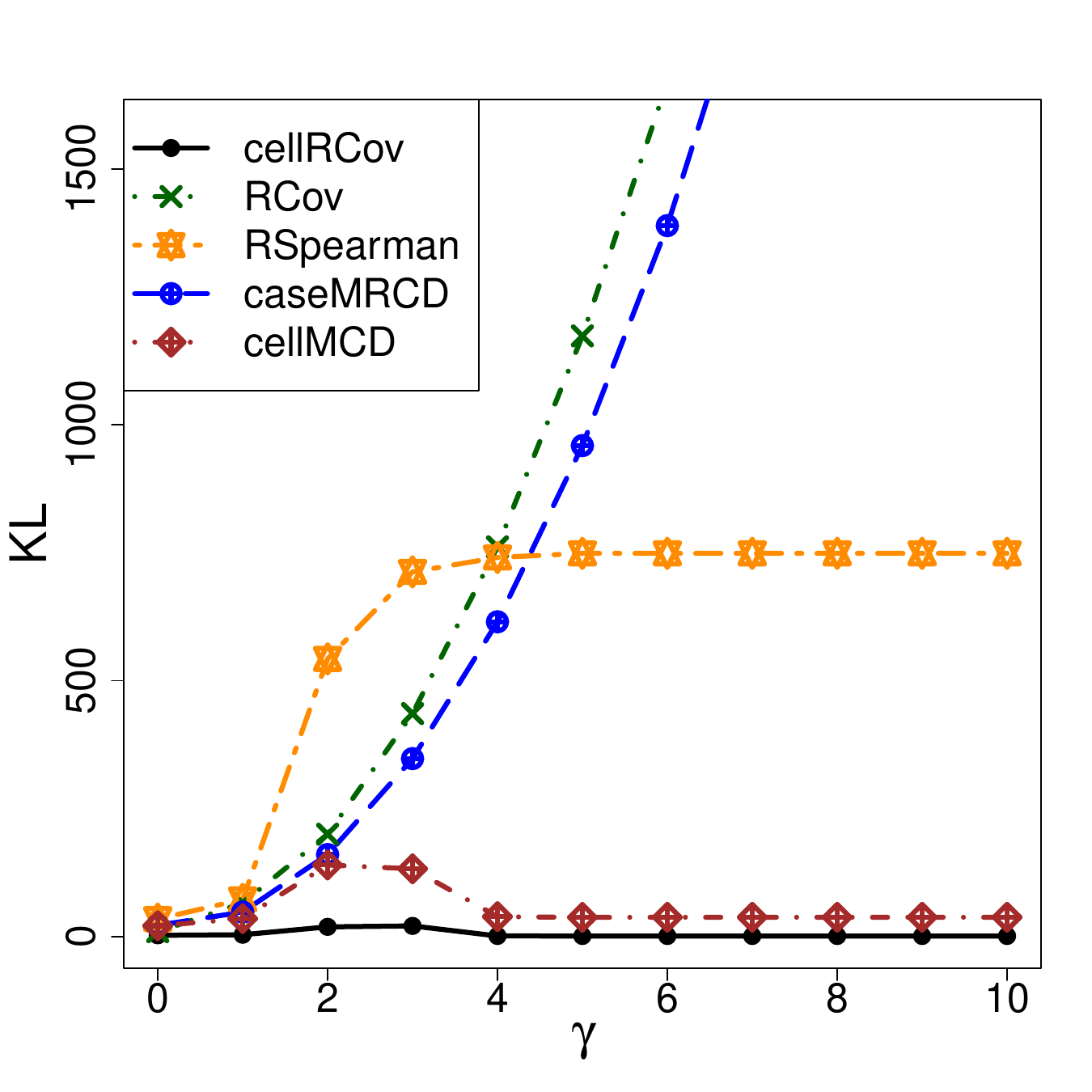}&\includegraphics[width=.31\textwidth]
  {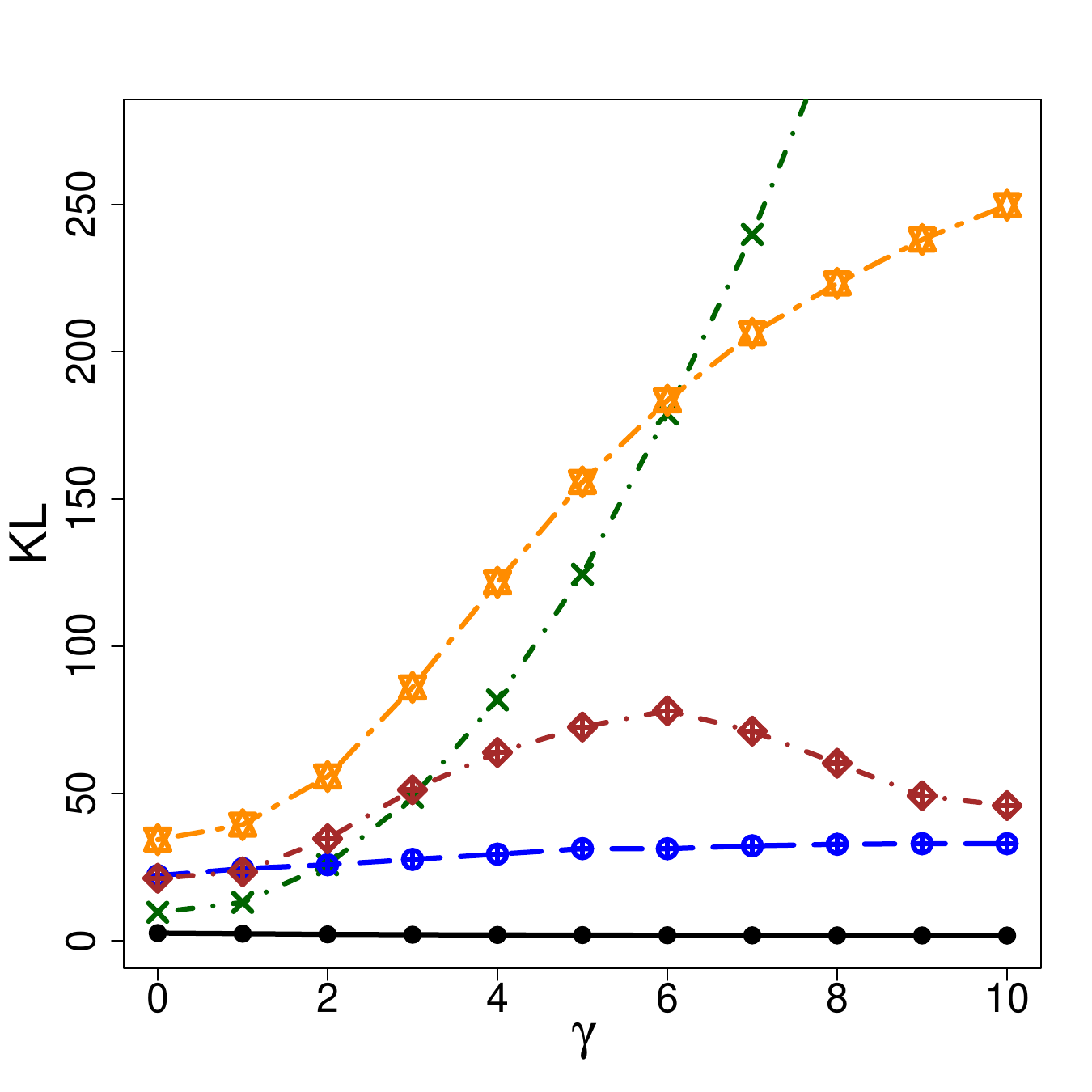}&\includegraphics[width=.31\textwidth]
  {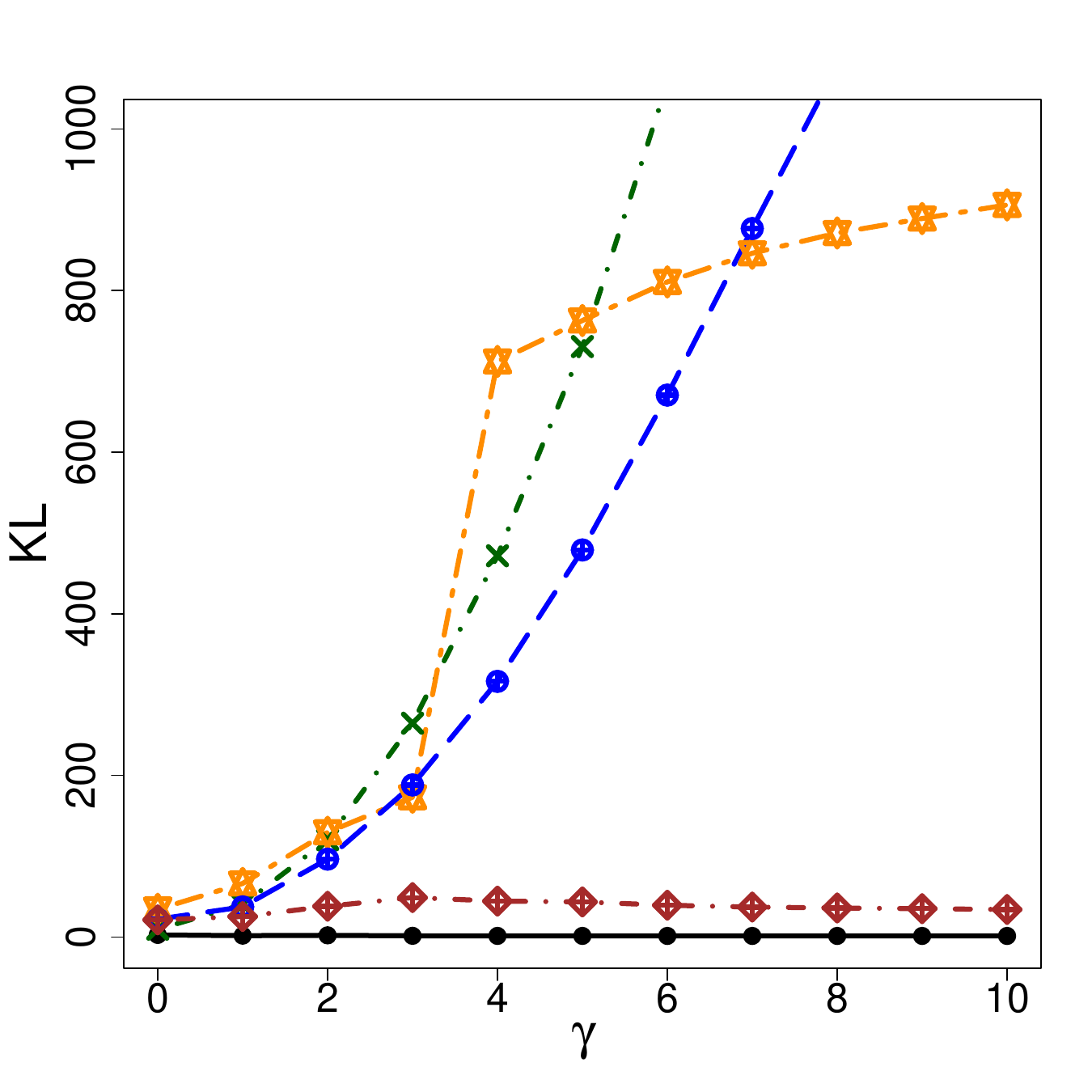}   \\ [-4mm]  \rotatebox{90}{\textbf{\footnotesize{$p=60$}}}&\includegraphics[width=.31\textwidth]
  {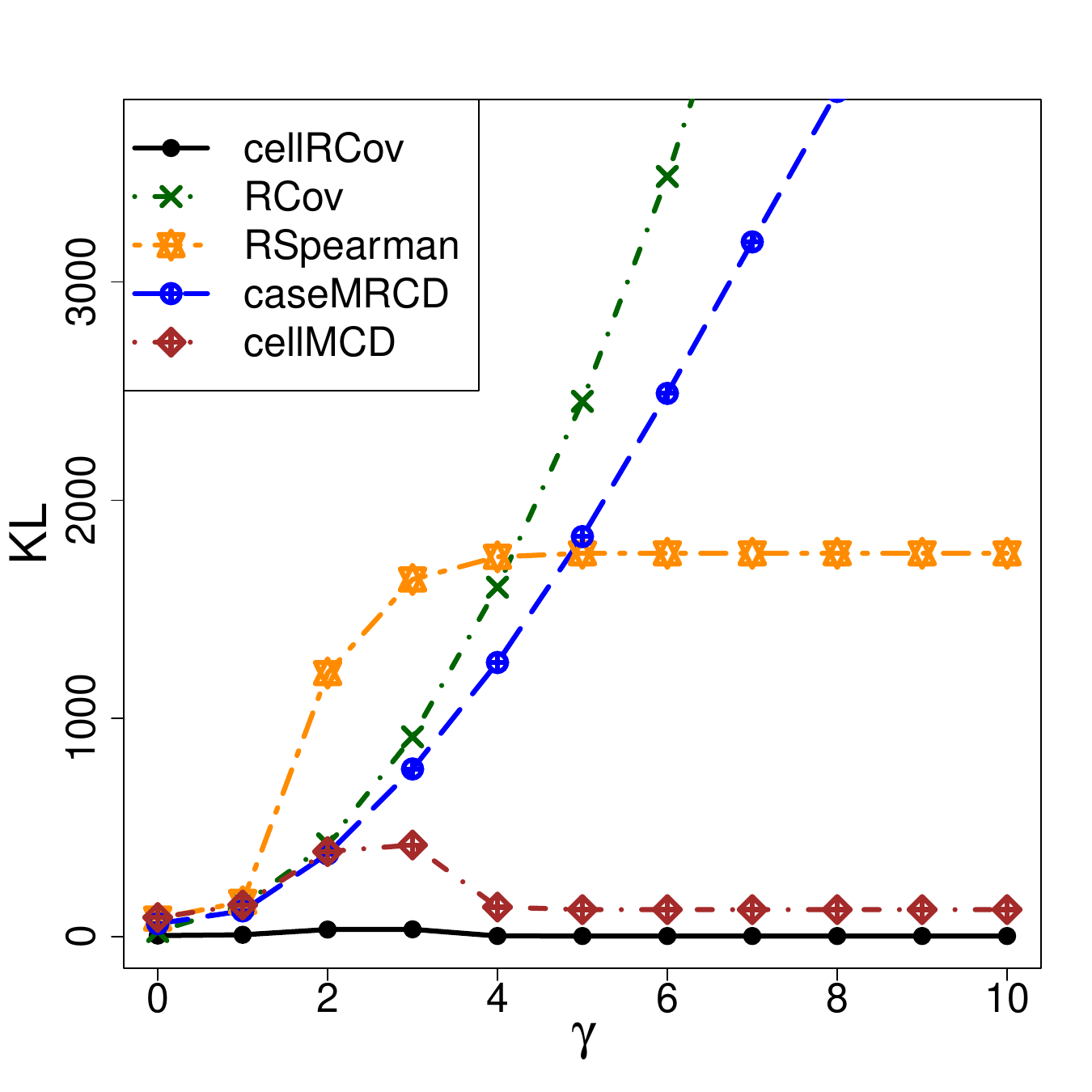}&\includegraphics[width=.31\textwidth]
  {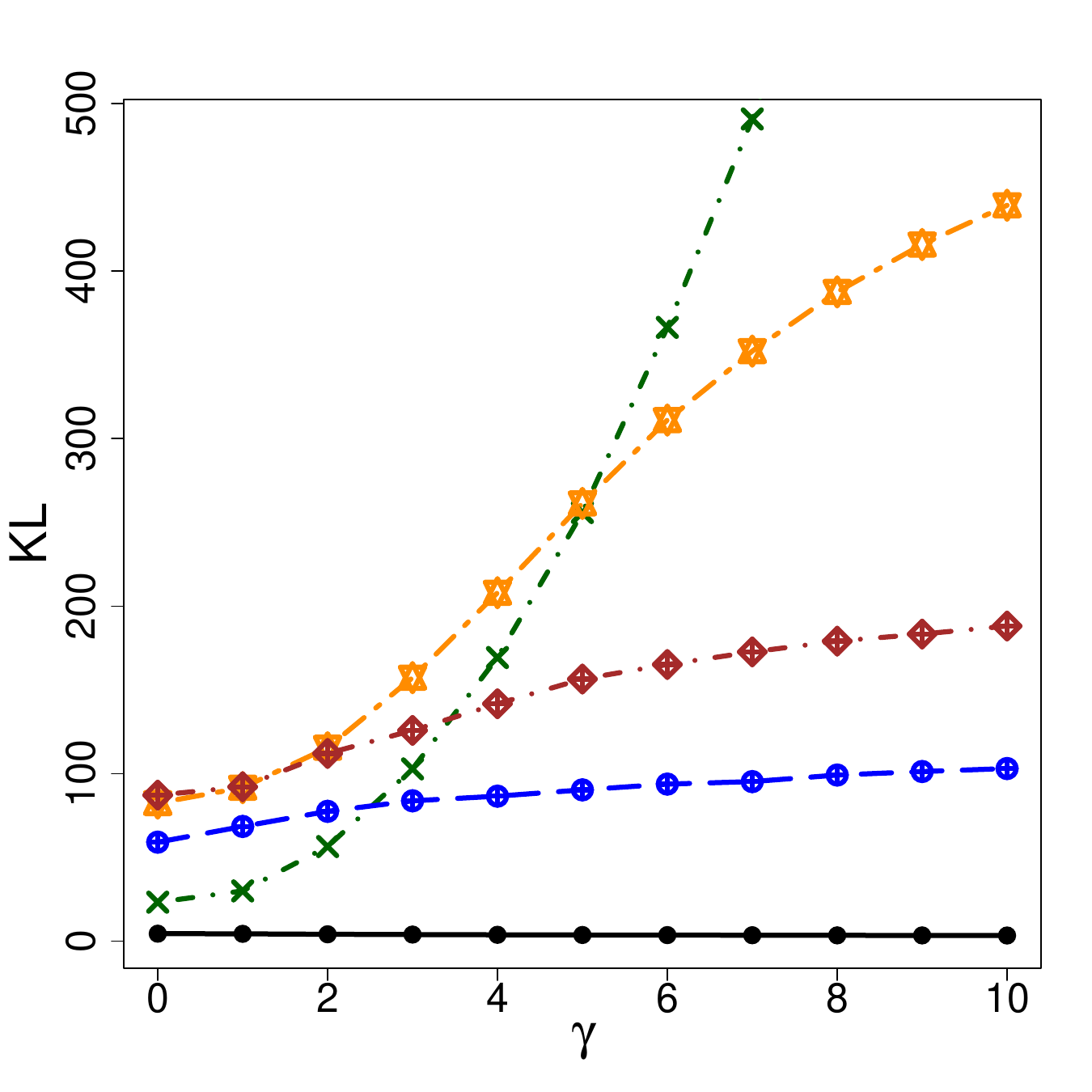}&\includegraphics[width=.31\textwidth]
  {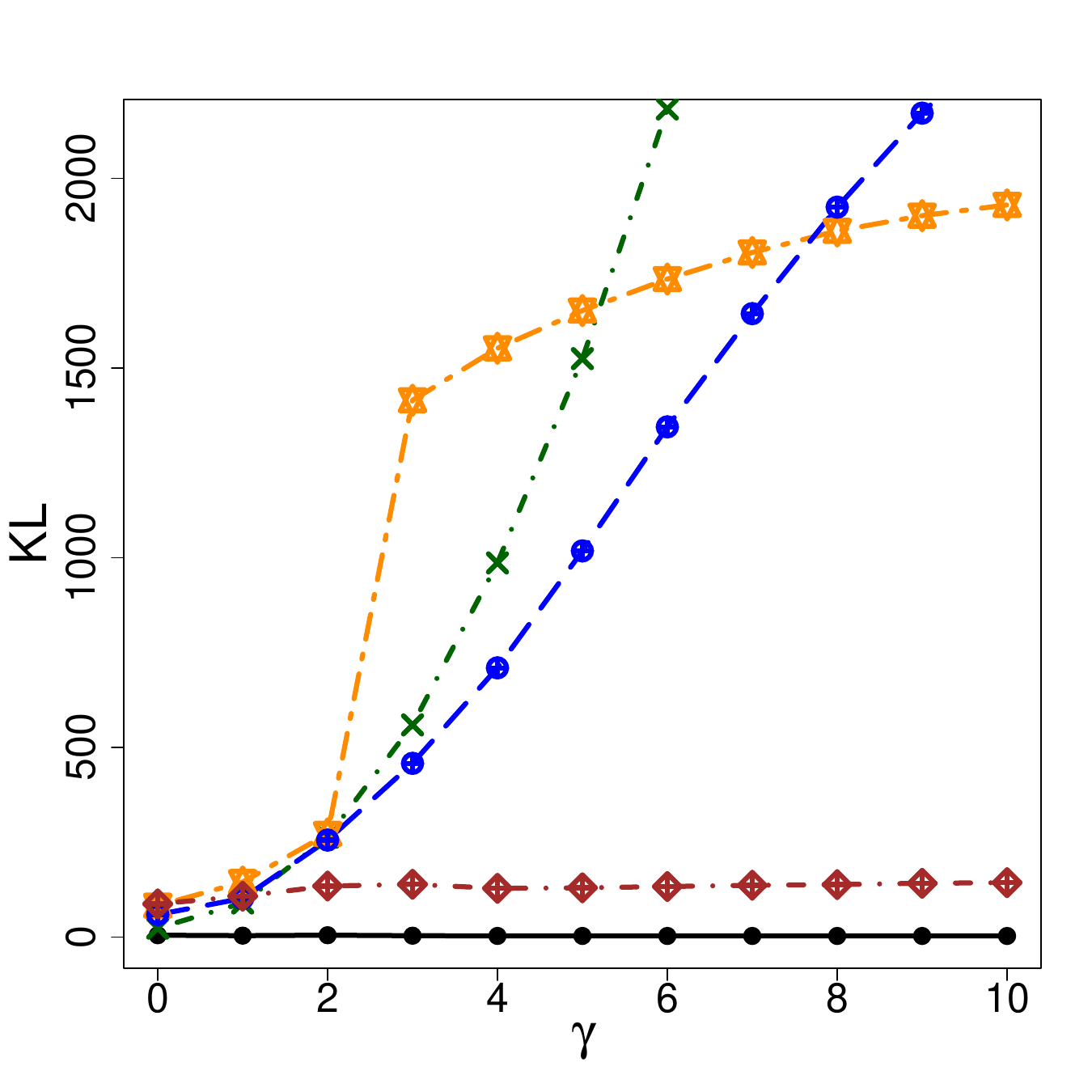}   \\ [-4mm]  \rotatebox{90}{\textbf{\footnotesize{$p=120$}}}&\includegraphics[width=.31\textwidth]
  {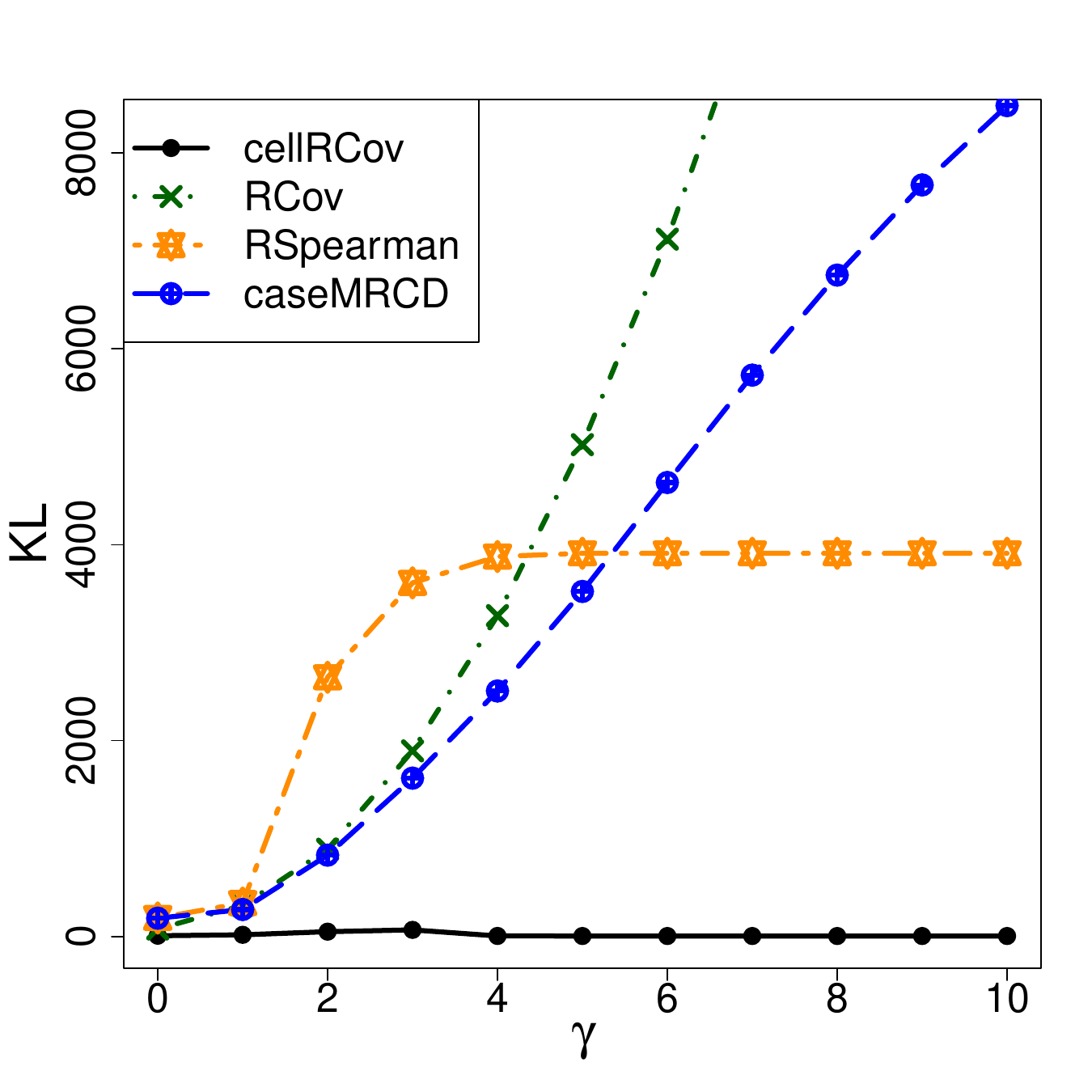}&\includegraphics[width=.31\textwidth]
  {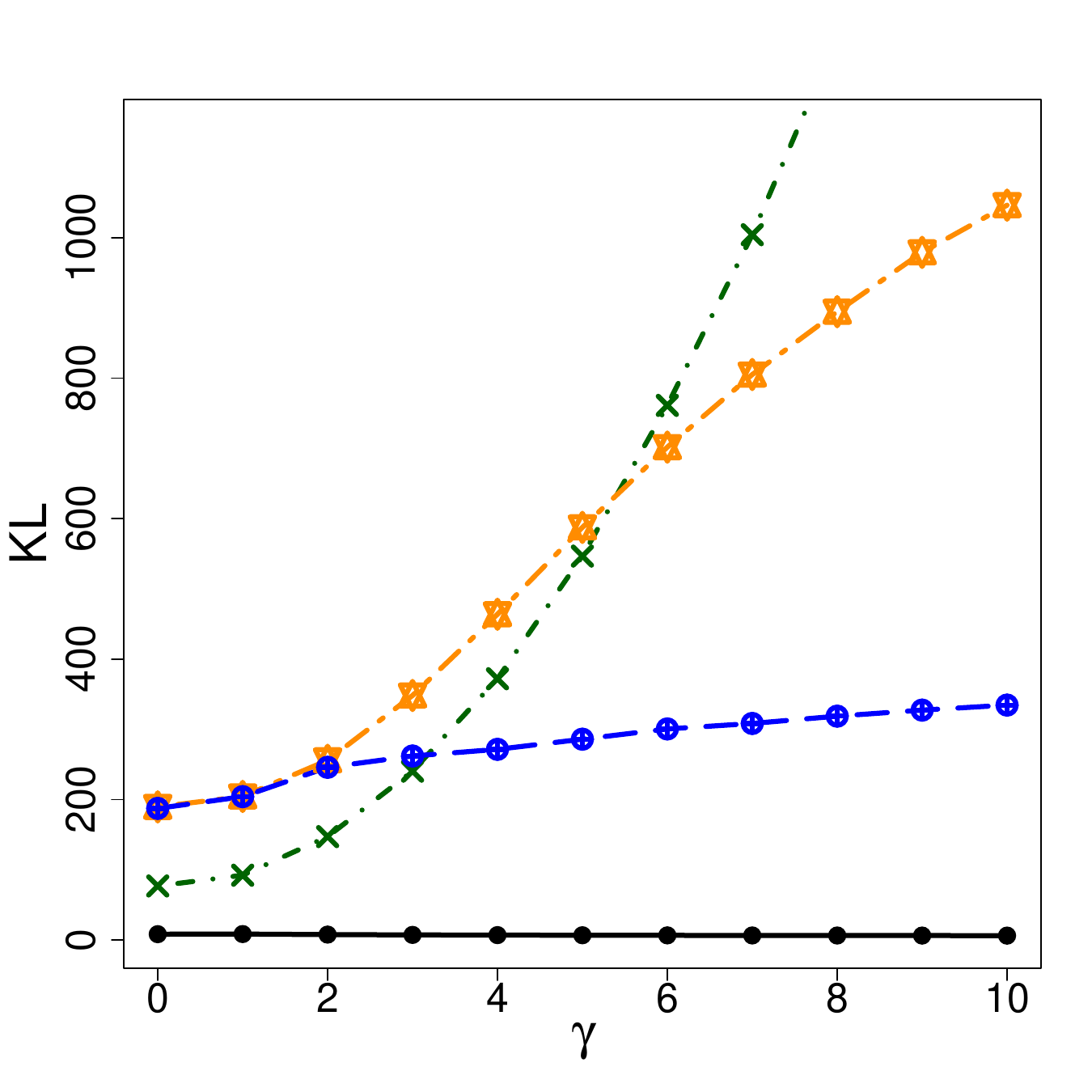}&\includegraphics[width=.31\textwidth]
  {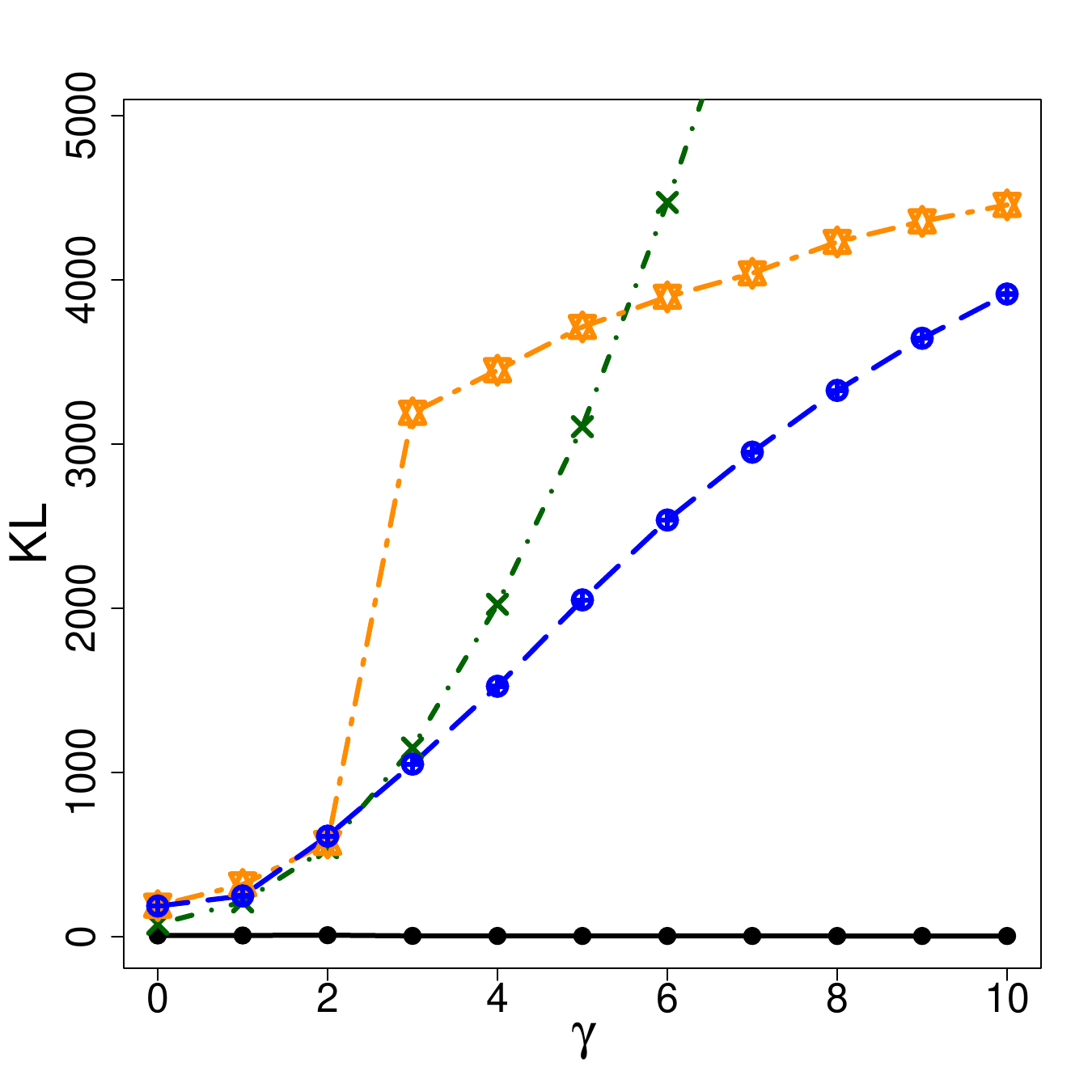}   
\end{tabular}
\caption{Average KL attained by cellRCov, RCov,
RSpearman, caseMRCD, and cellMCD in the presence of either 
cellwise outliers, casewise outliers, or both, 
for the planar covariance model in dimensions 
$p$ in $\lbrace30,60,120\rbrace$.}
\label{fig:results_planar}
\end{figure}

\begin{figure}[!ht]
\centering
 
 \begin{tabular}{M{0.0005\textwidth}M{0.29\textwidth}M{0.29\textwidth}M{0.32\textwidth}}
   &\large \textbf{Cellwise}  & \large \textbf{Casewise} &\large{\textbf{Casewise \& Cellwise}} \\
   [-4mm]
   \rotatebox{90}{\textbf{\footnotesize{$p=30$}}}&\includegraphics[width=.31\textwidth]
  {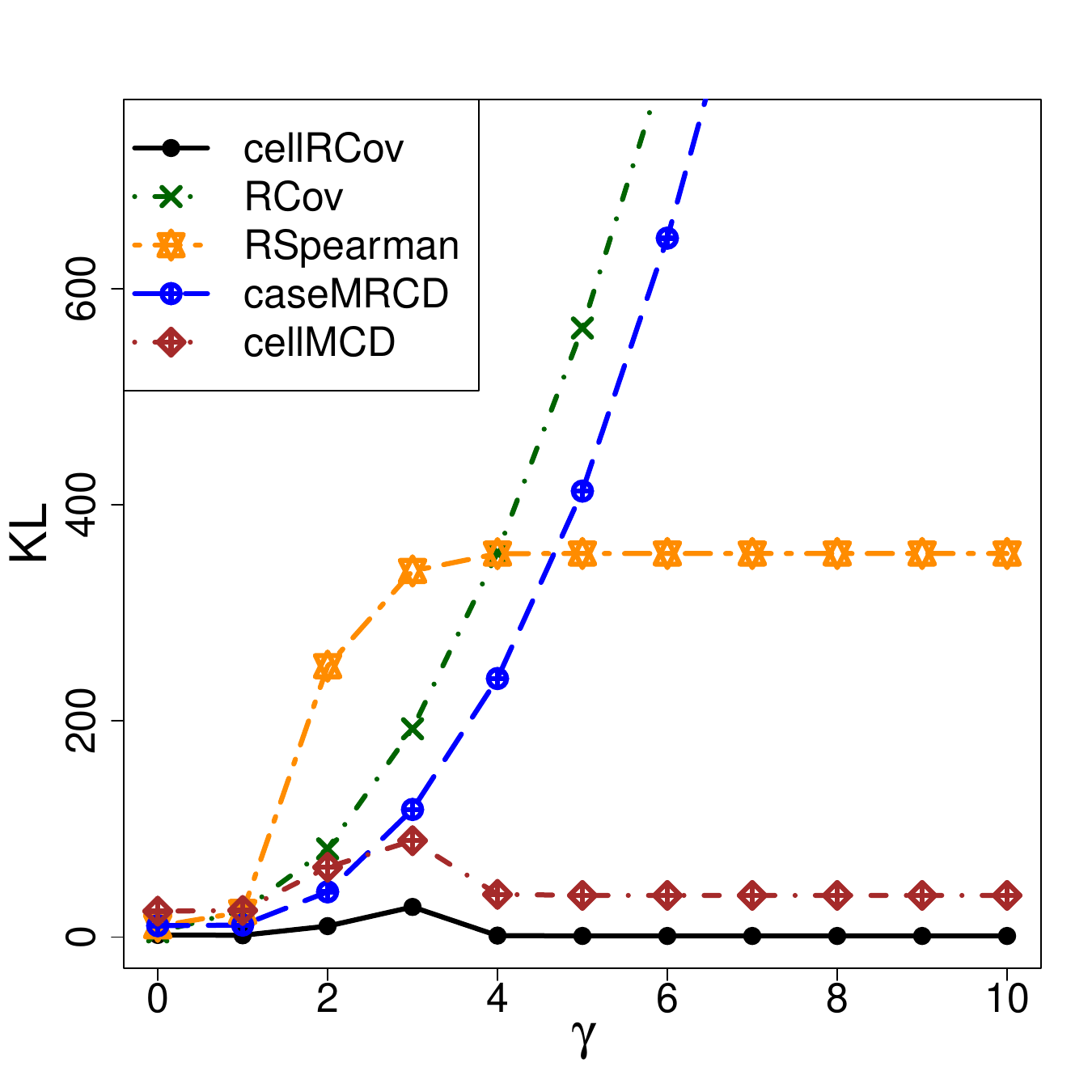}&\includegraphics[width=.31\textwidth]
  {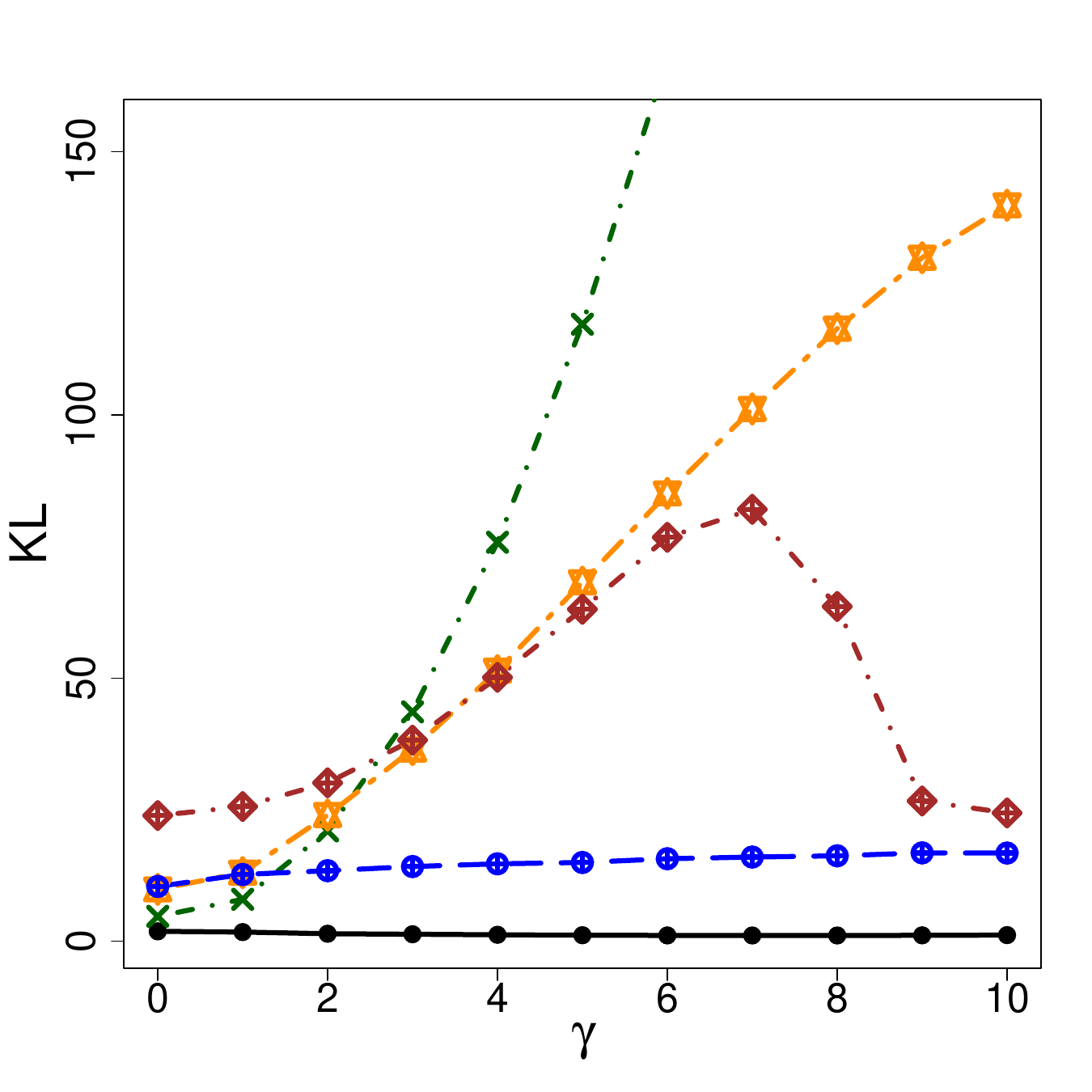}&\includegraphics[width=.31\textwidth]
  {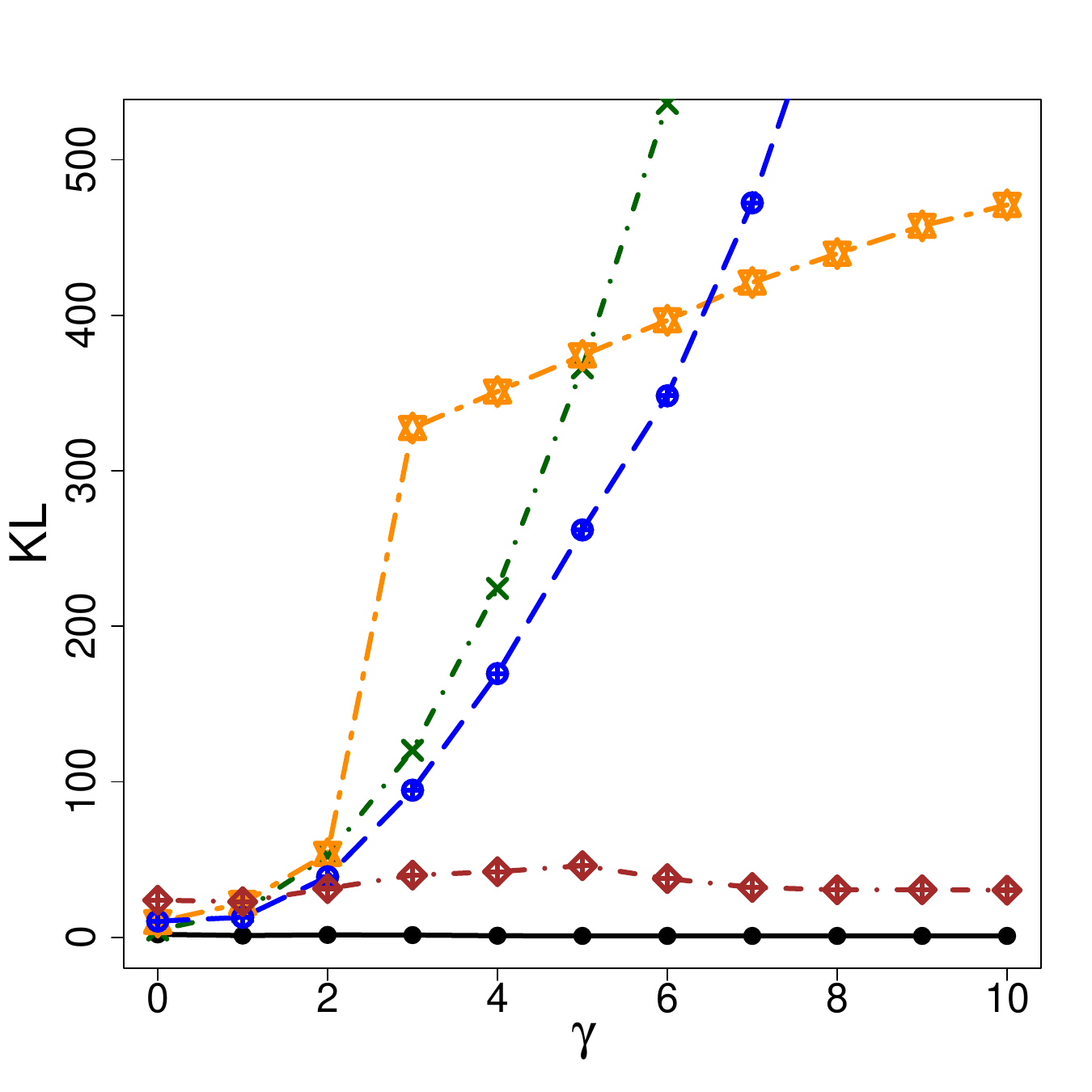}   \\ [-4mm]  \rotatebox{90}{\textbf{\footnotesize{$p=60$}}}&\includegraphics[width=.31\textwidth]
  {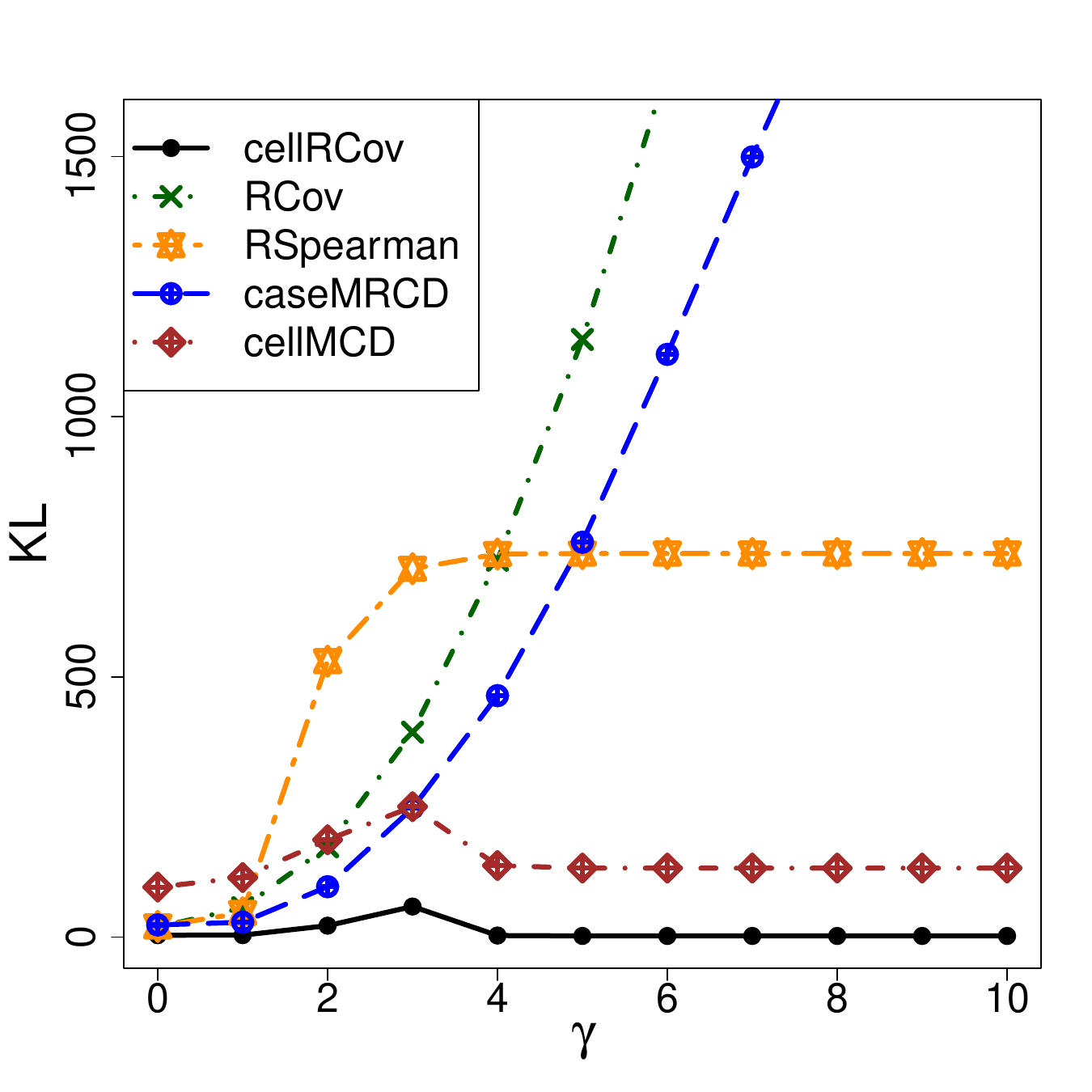}&\includegraphics[width=.31\textwidth]
  {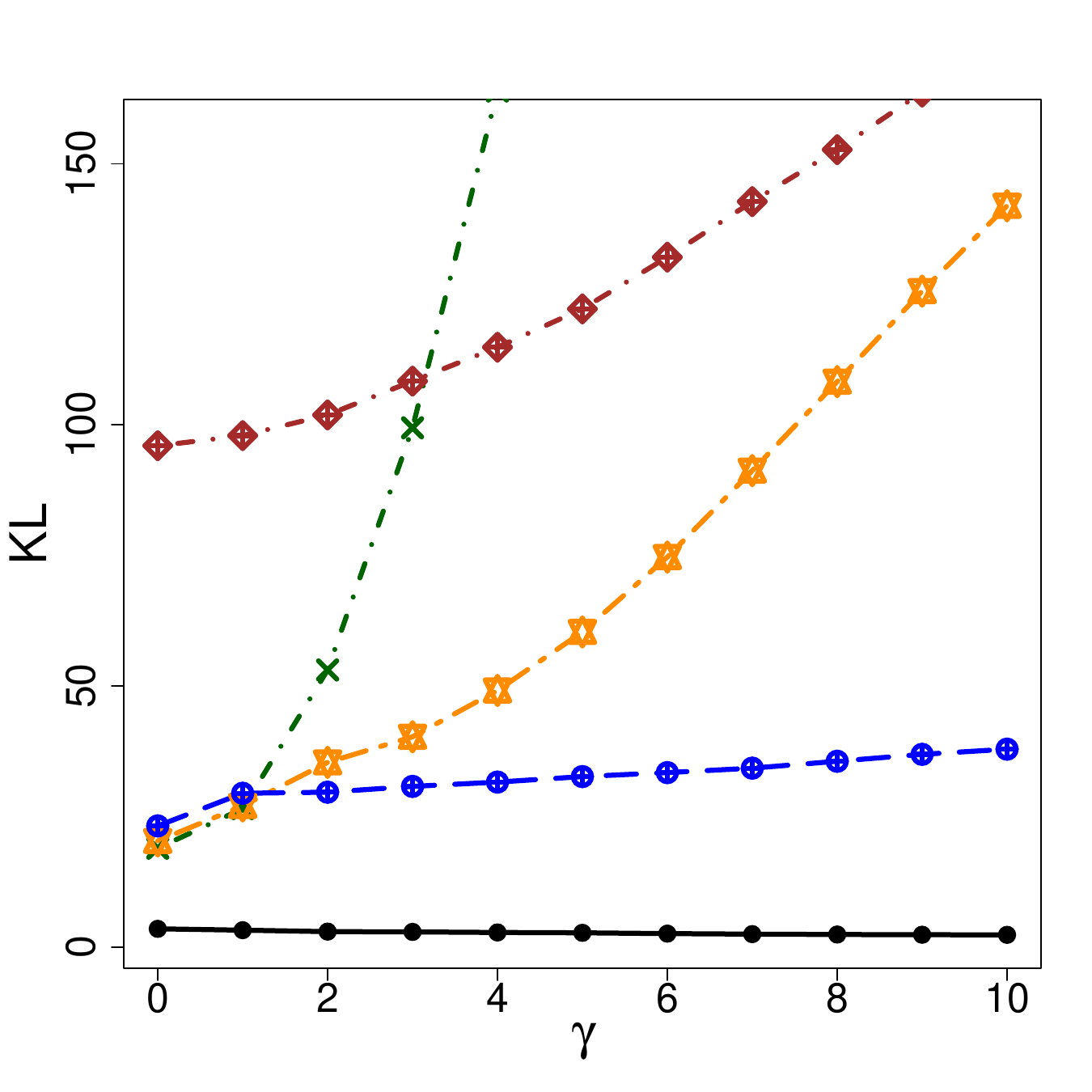}&\includegraphics[width=.31\textwidth]
  {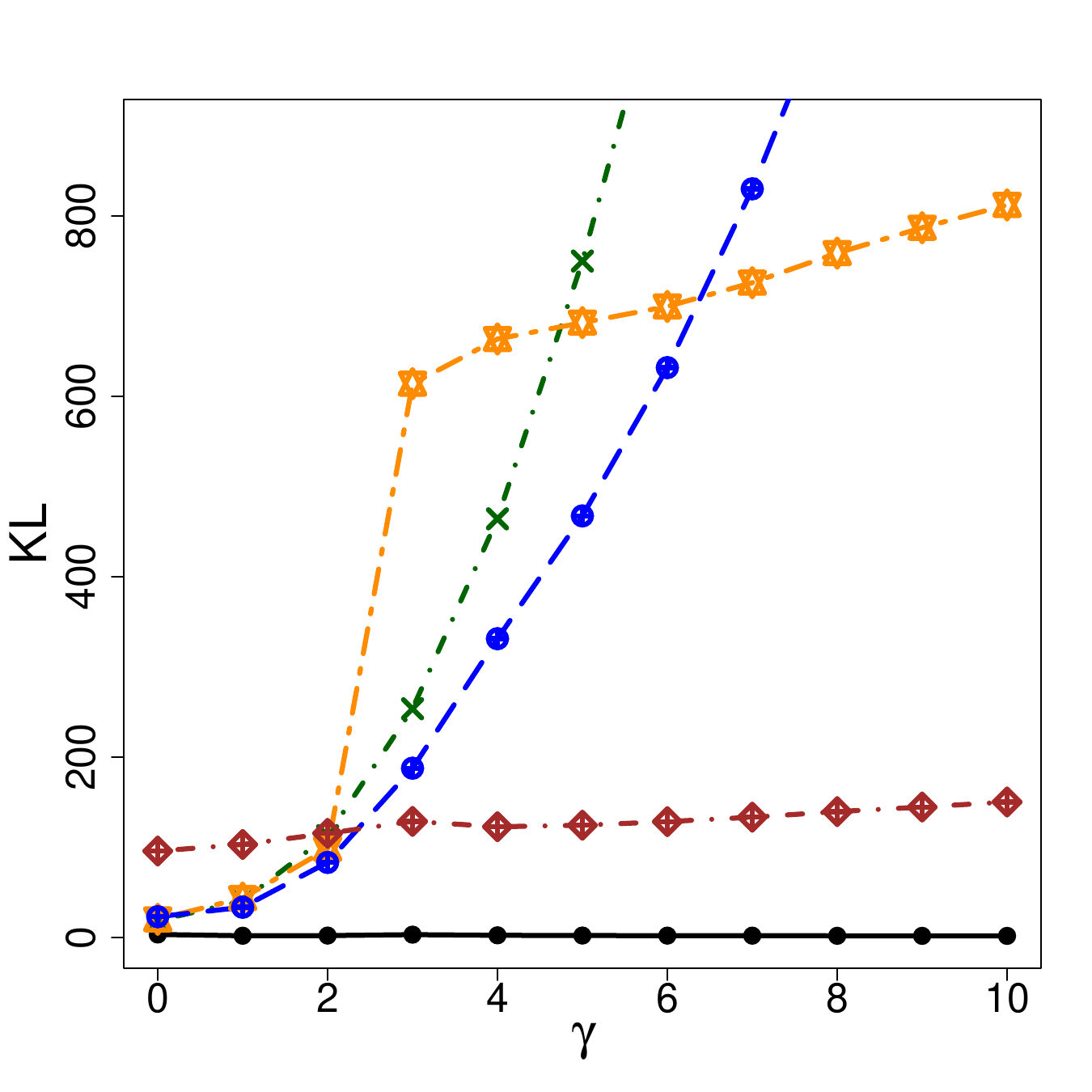}   \\ [-4mm]  \rotatebox{90}{\textbf{\footnotesize{$p=120$}}}&\includegraphics[width=.31\textwidth]
  {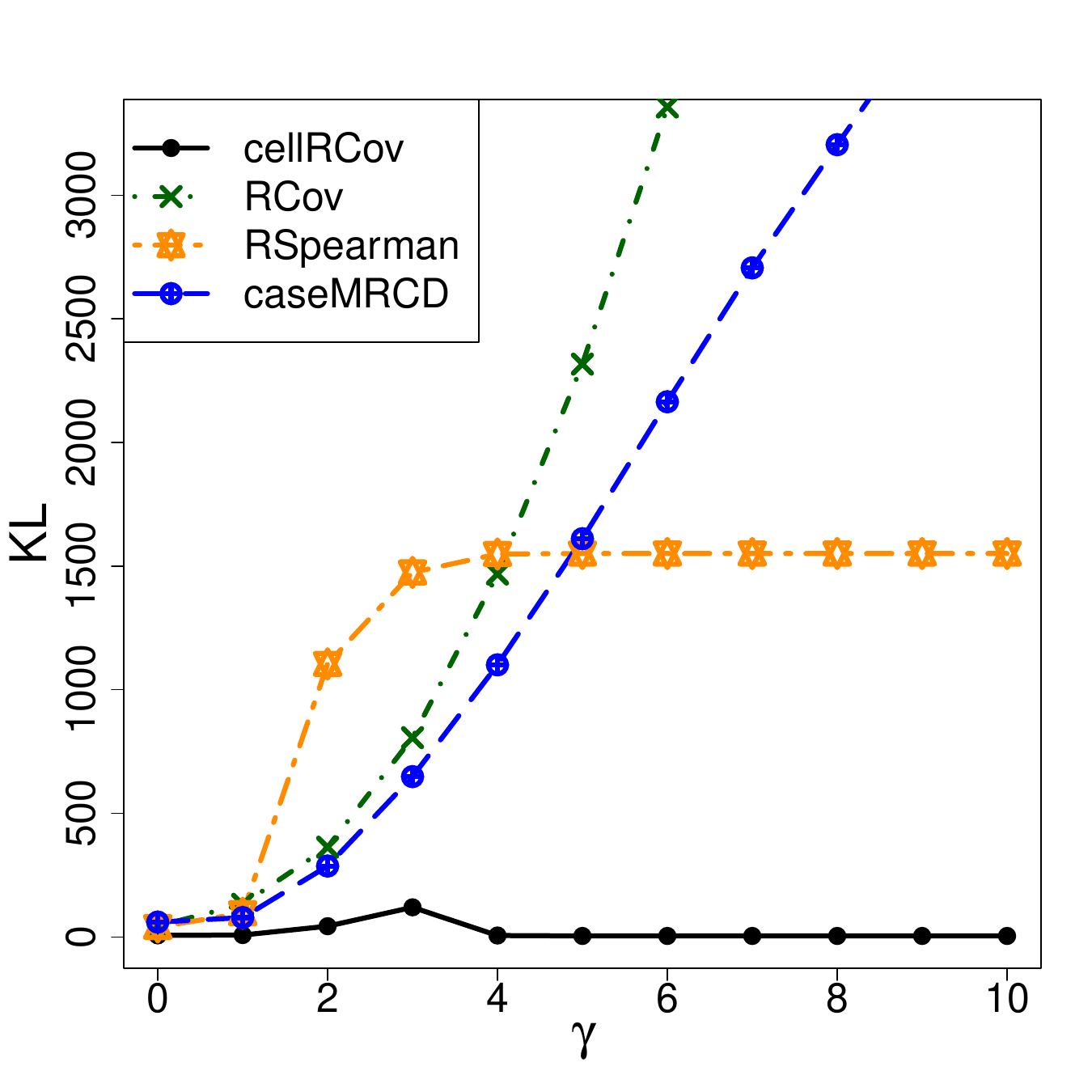}&\includegraphics[width=.31\textwidth]
  {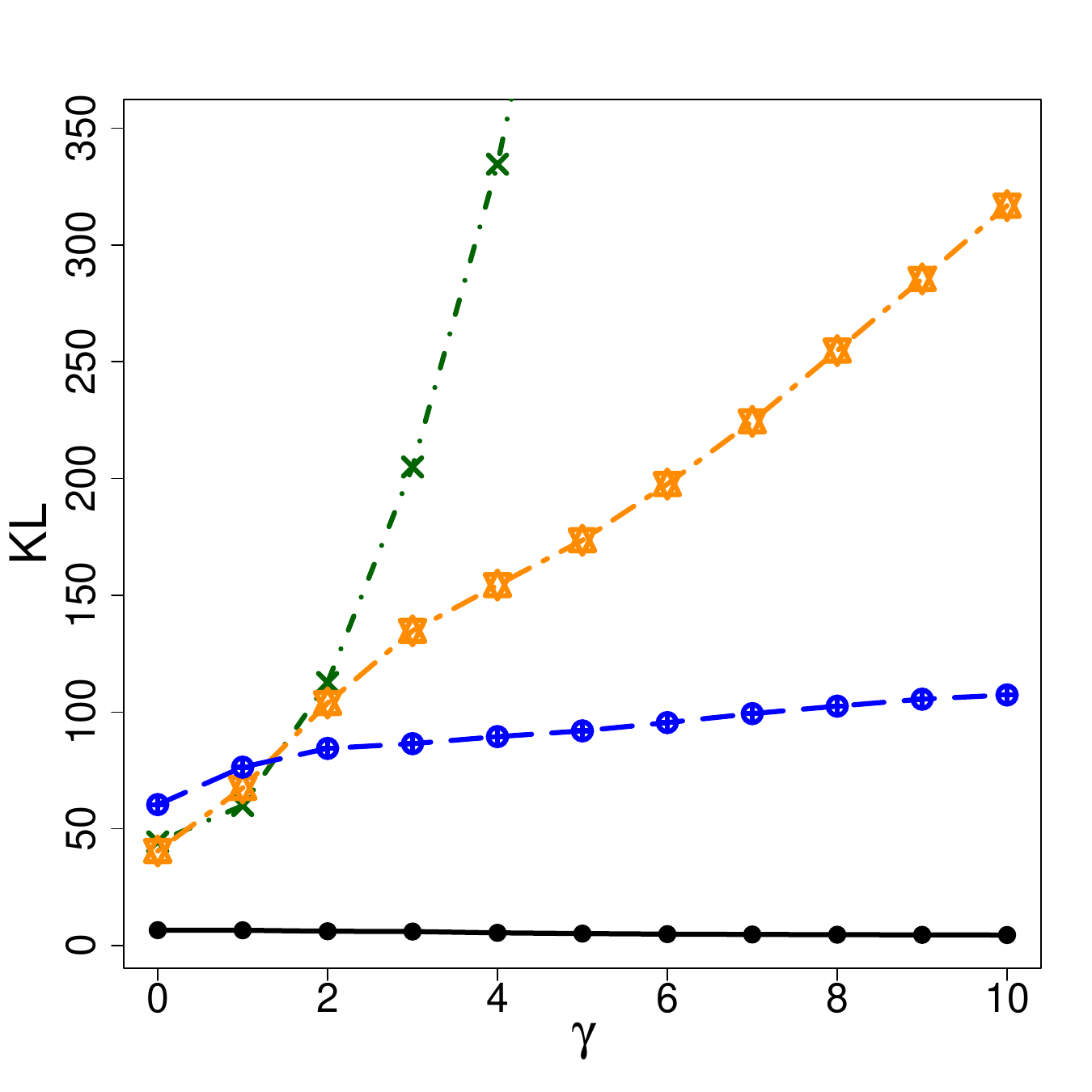}&\includegraphics[width=.31\textwidth]
  {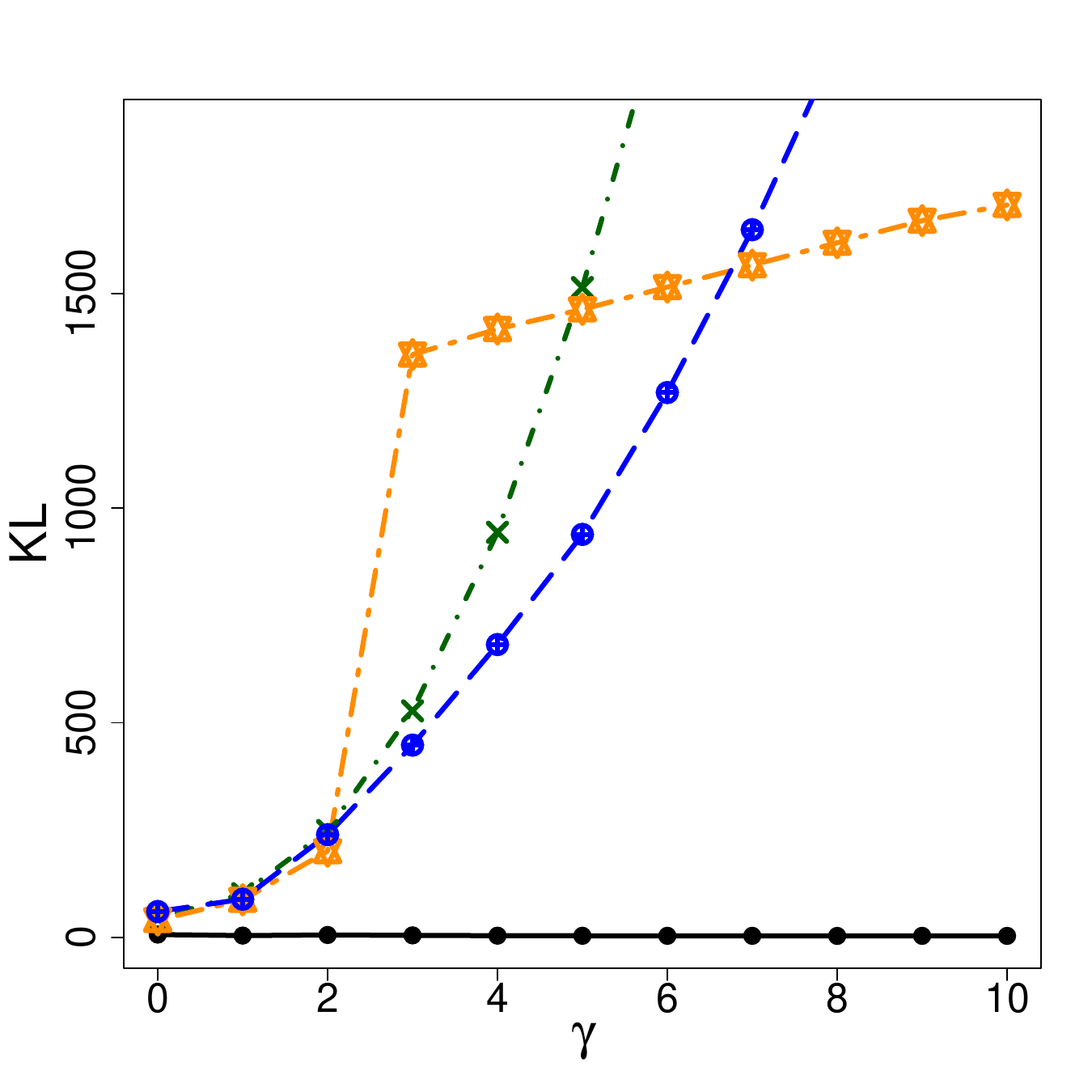}   
\end{tabular}
\caption{Average KL attained by cellRCov, RCov, 
RSpearman, caseMRCD, and cellMCD in the presence of either 
cellwise outliers, casewise outliers, or both, 
for the dense covariance model in dimensions 
$p$ in $\lbrace30,60,120\rbrace$.}
\label{fig:results_dense}
\end{figure}

\begin{figure}[!ht]
\centering
 
 \begin{tabular}{M{0.0005\textwidth}M{0.29\textwidth}M{0.29\textwidth}M{0.32\textwidth}}
   &\large \textbf{Cellwise}  & \large \textbf{Casewise} &\large{\textbf{Casewise \& Cellwise}} \\
   [-4mm]
   \rotatebox{90}{\textbf{\footnotesize{$p=30$}}}&\includegraphics[width=.31\textwidth]
  {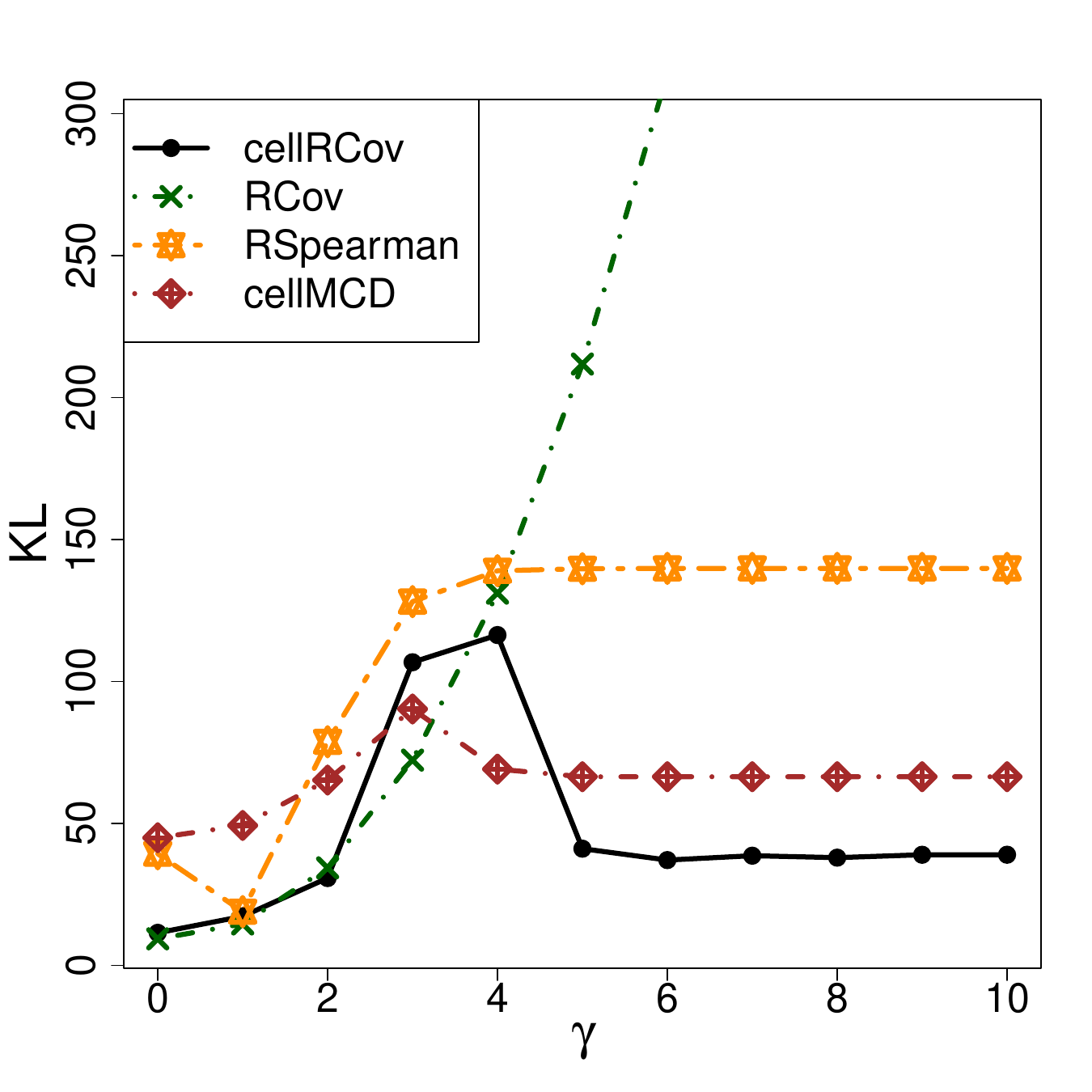}&\includegraphics[width=.31\textwidth]
  {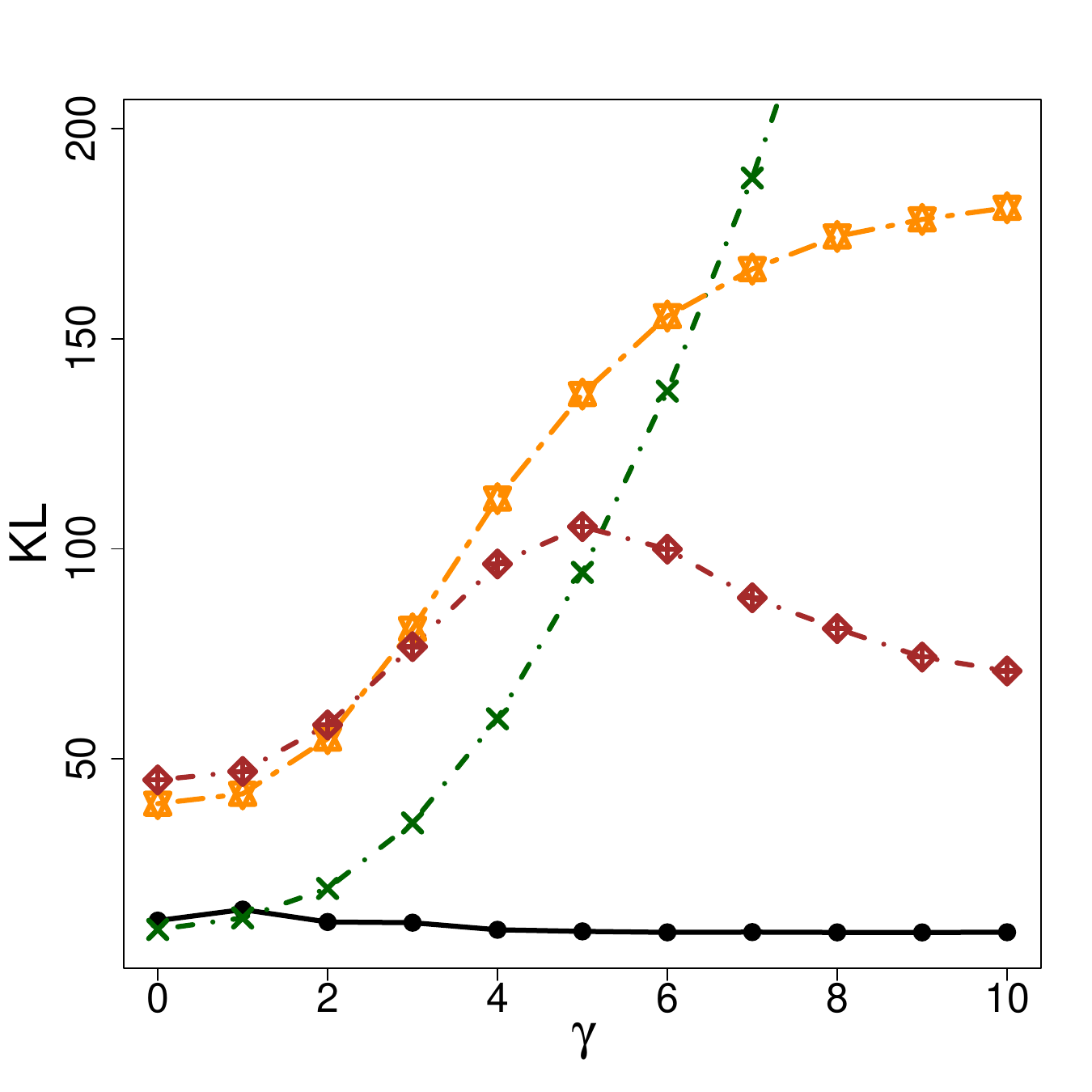}&\includegraphics[width=.31\textwidth]
  {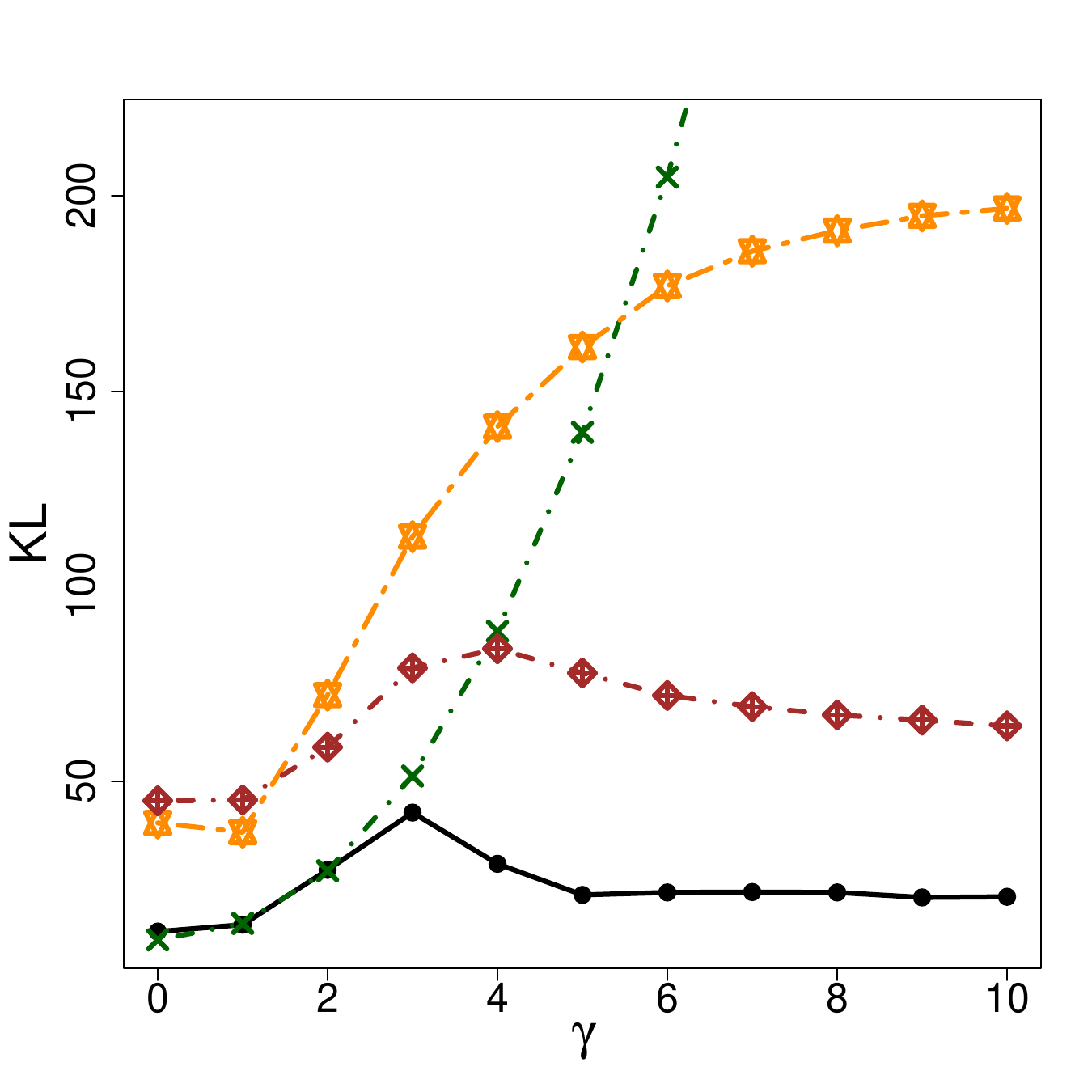}   \\ [-4mm]  \rotatebox{90}{\textbf{\footnotesize{$p=60$}}}&\includegraphics[width=.31\textwidth]
  {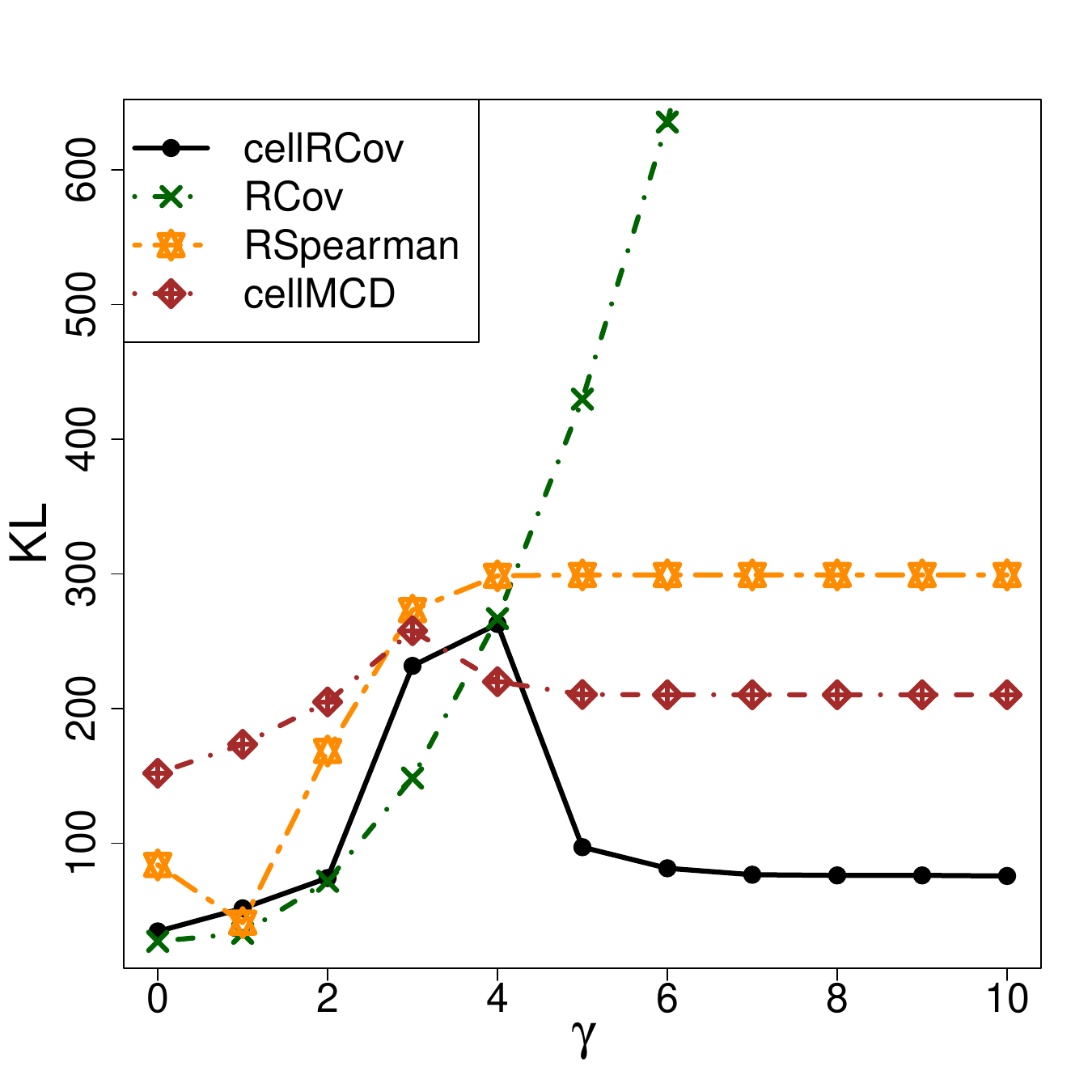}&\includegraphics[width=.31\textwidth]
  {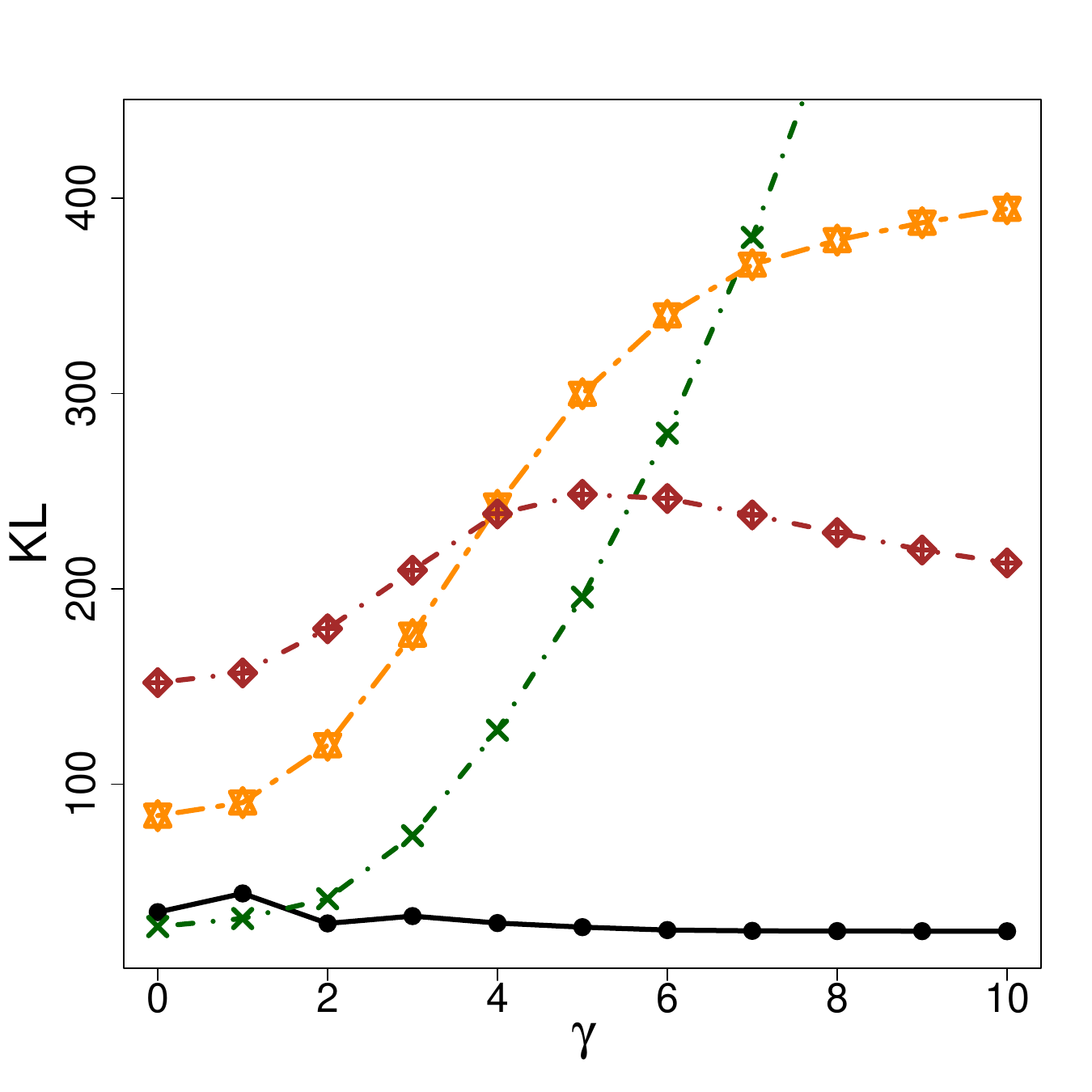}&\includegraphics[width=.31\textwidth]
  {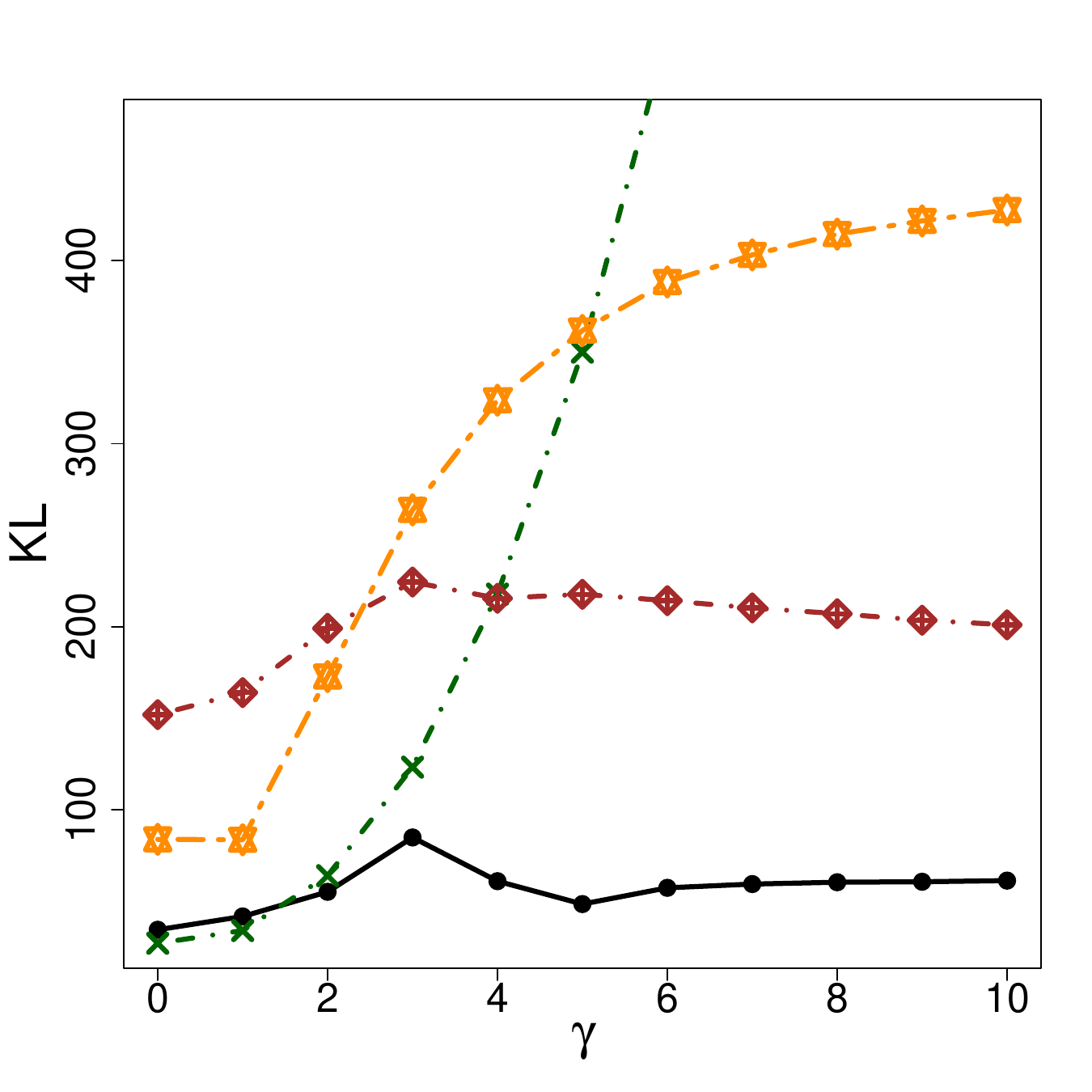}   \\ [-4mm]  \rotatebox{90}{\textbf{\footnotesize{$p=120$}}}&\includegraphics[width=.31\textwidth]
  {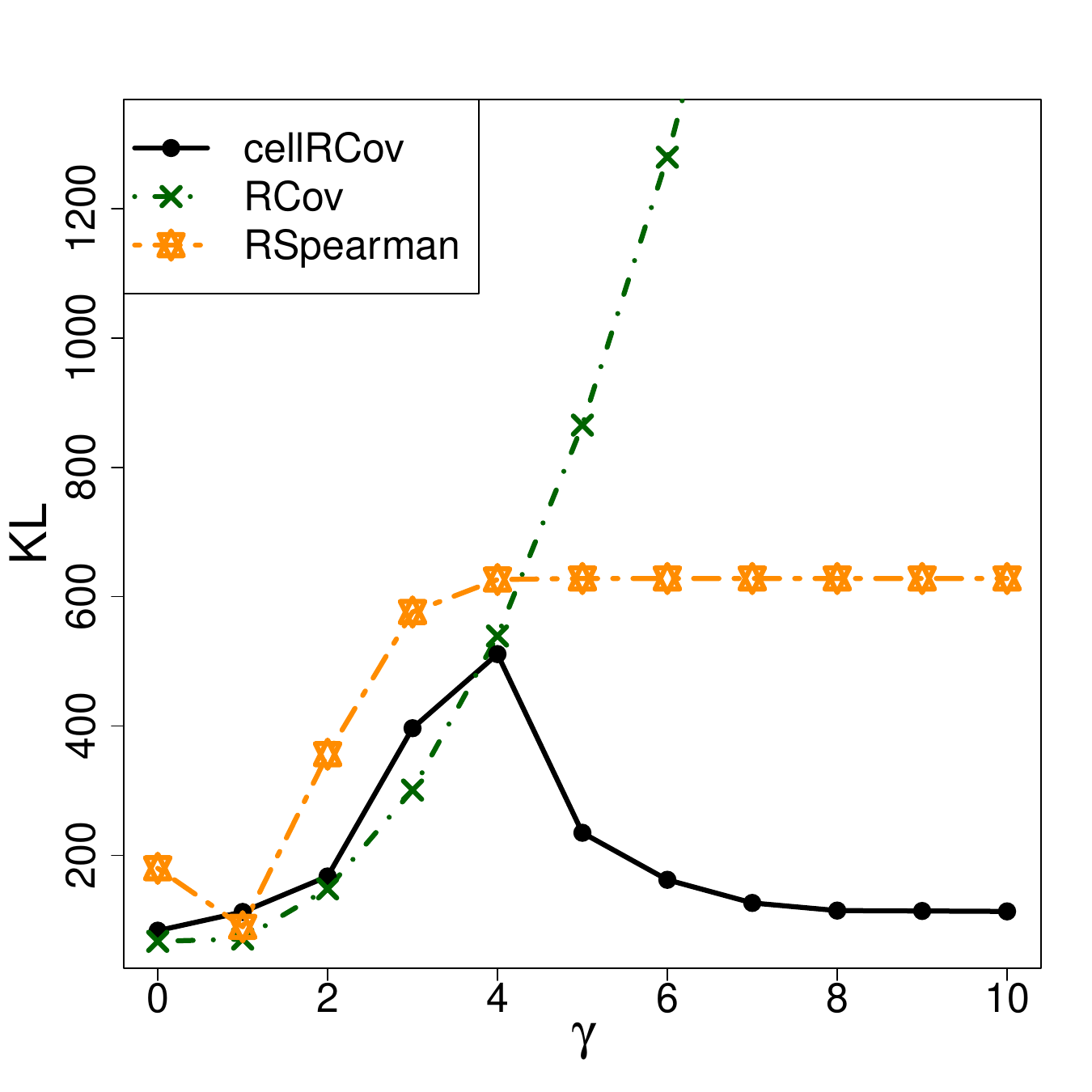}&\includegraphics[width=.31\textwidth]
  {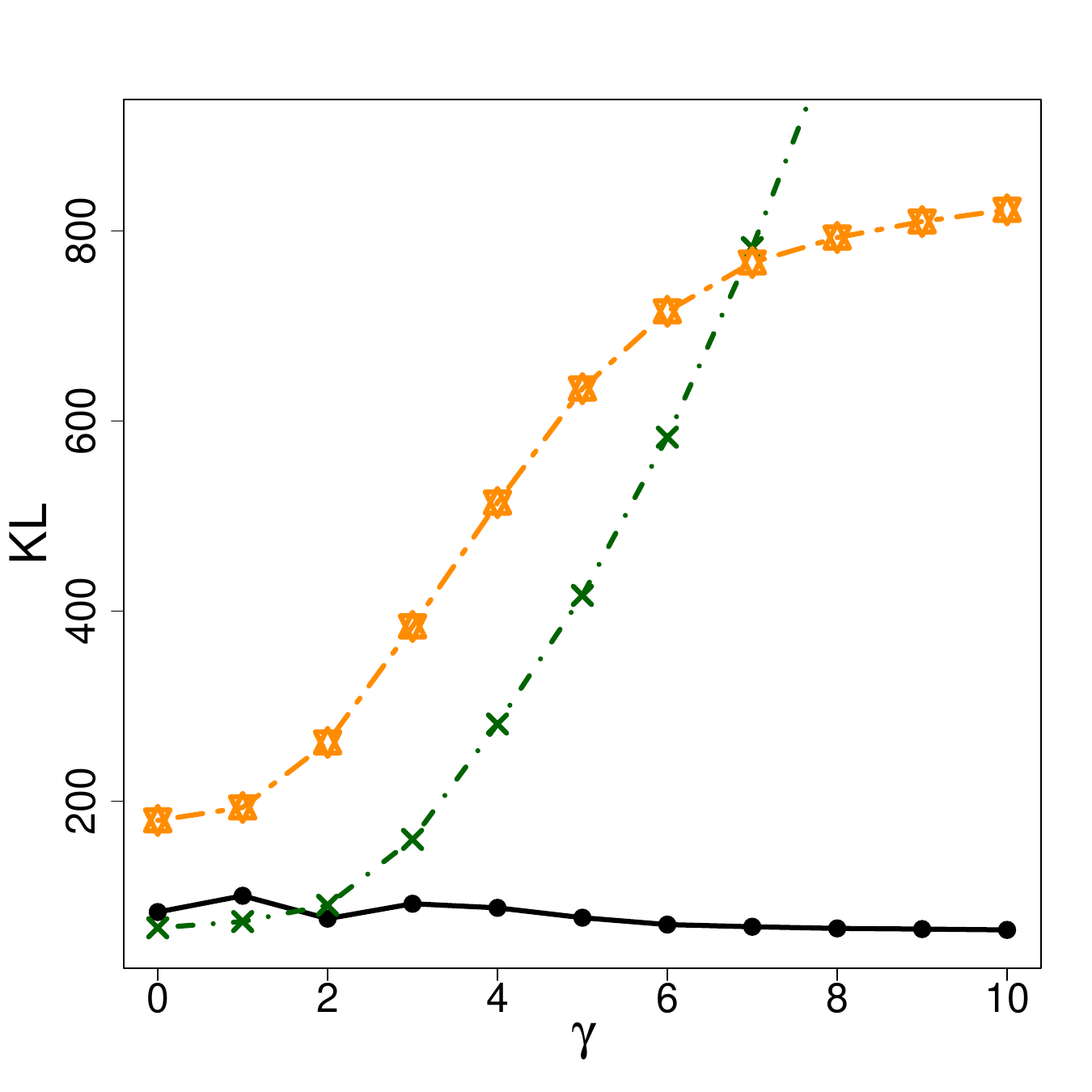}&\includegraphics[width=.31\textwidth]
  {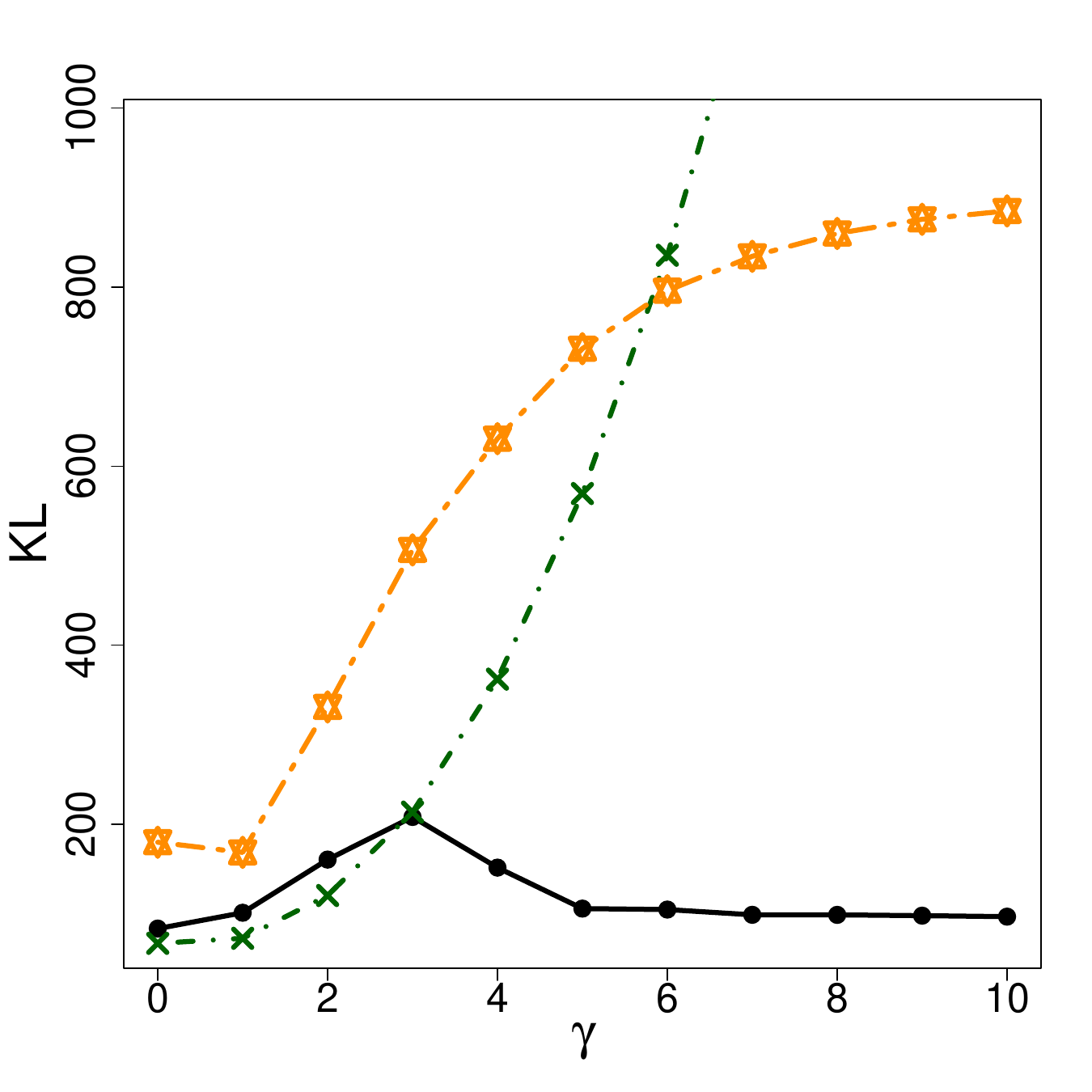}   
\end{tabular}
\caption{Average KL attained by cellRCov, 
RCov, RSpearman, and cellMCD in the presence of either 
cellwise outliers, casewise outliers, or both, 
for the A06 covariance model in dimensions 
$p$ in $\lbrace30,60,120\rbrace$, with 20\p of 
missing cells.}
\label{fig:results_A06_NA}
\end{figure}

\begin{figure}[!ht]
\centering
 
\begin{tabular}{M{0.0005\textwidth}M{0.29\textwidth}M{0.29\textwidth}M{0.32\textwidth}}
   &\large \textbf{Cellwise}  & \large \textbf{Casewise} &\large{\textbf{Casewise \& Cellwise}} \\
   [-4mm]
   \rotatebox{90}{\textbf{\footnotesize{$p=30$}}}&\includegraphics[width=.31\textwidth]
  {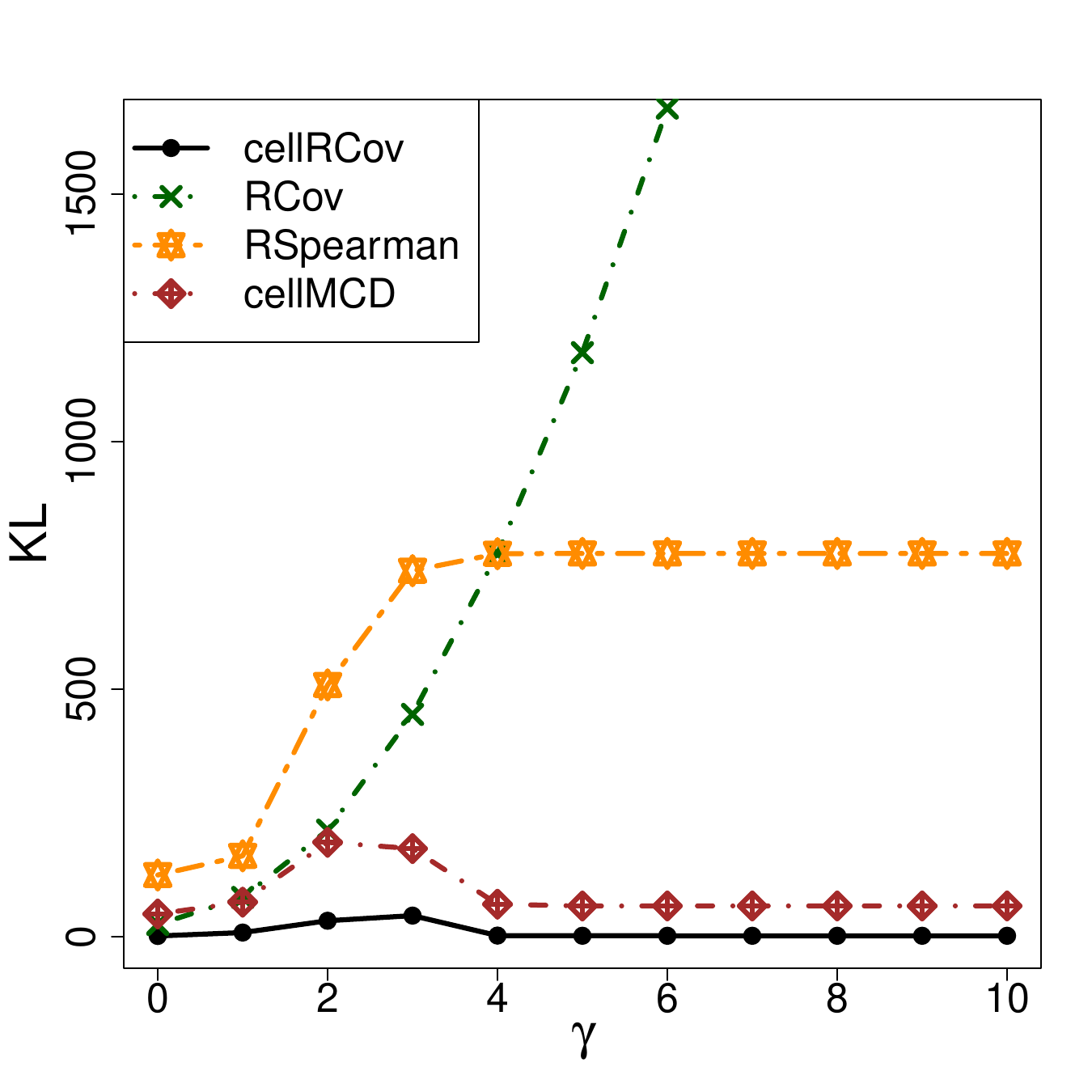}&\includegraphics[width=.31\textwidth]
  {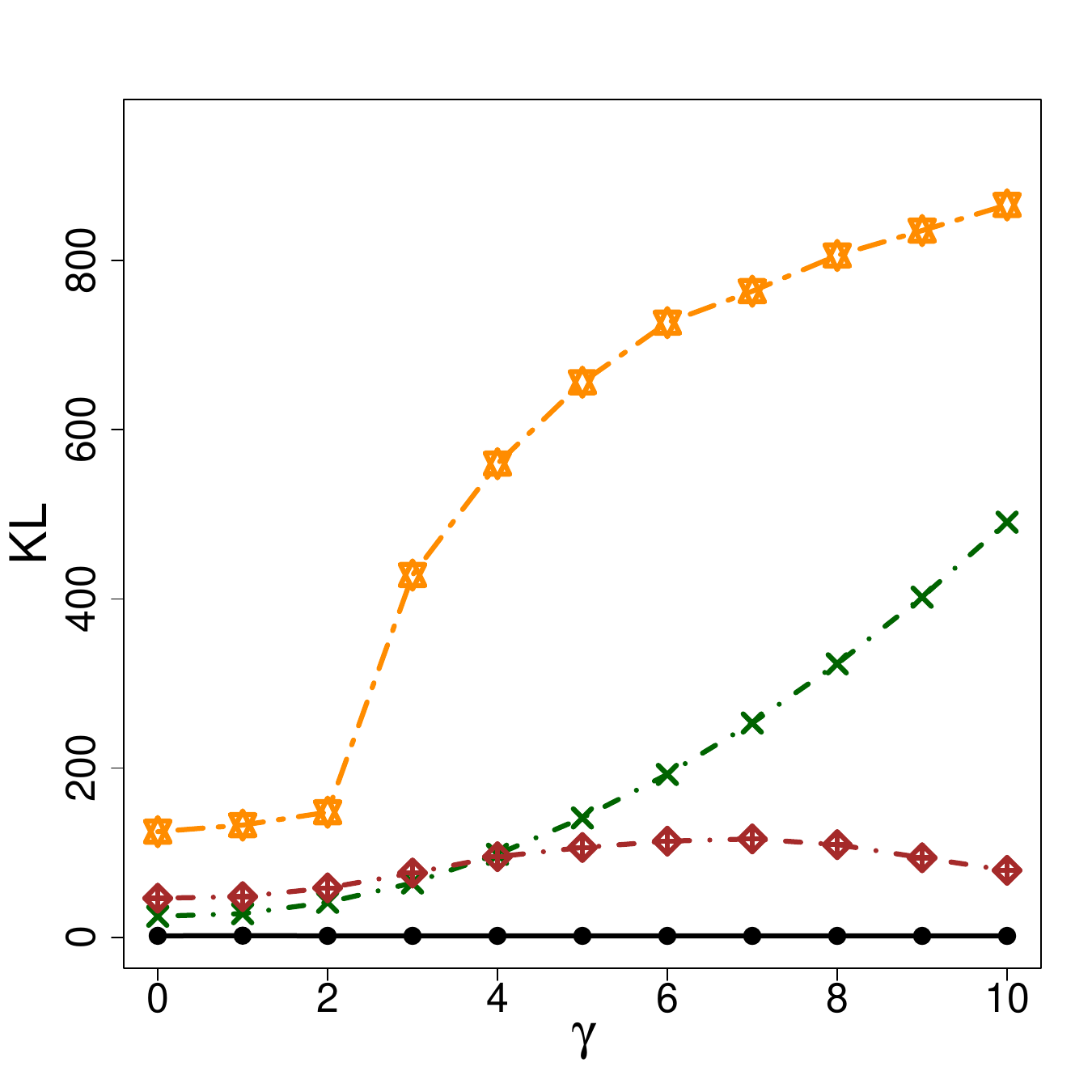}&\includegraphics[width=.31\textwidth]
  {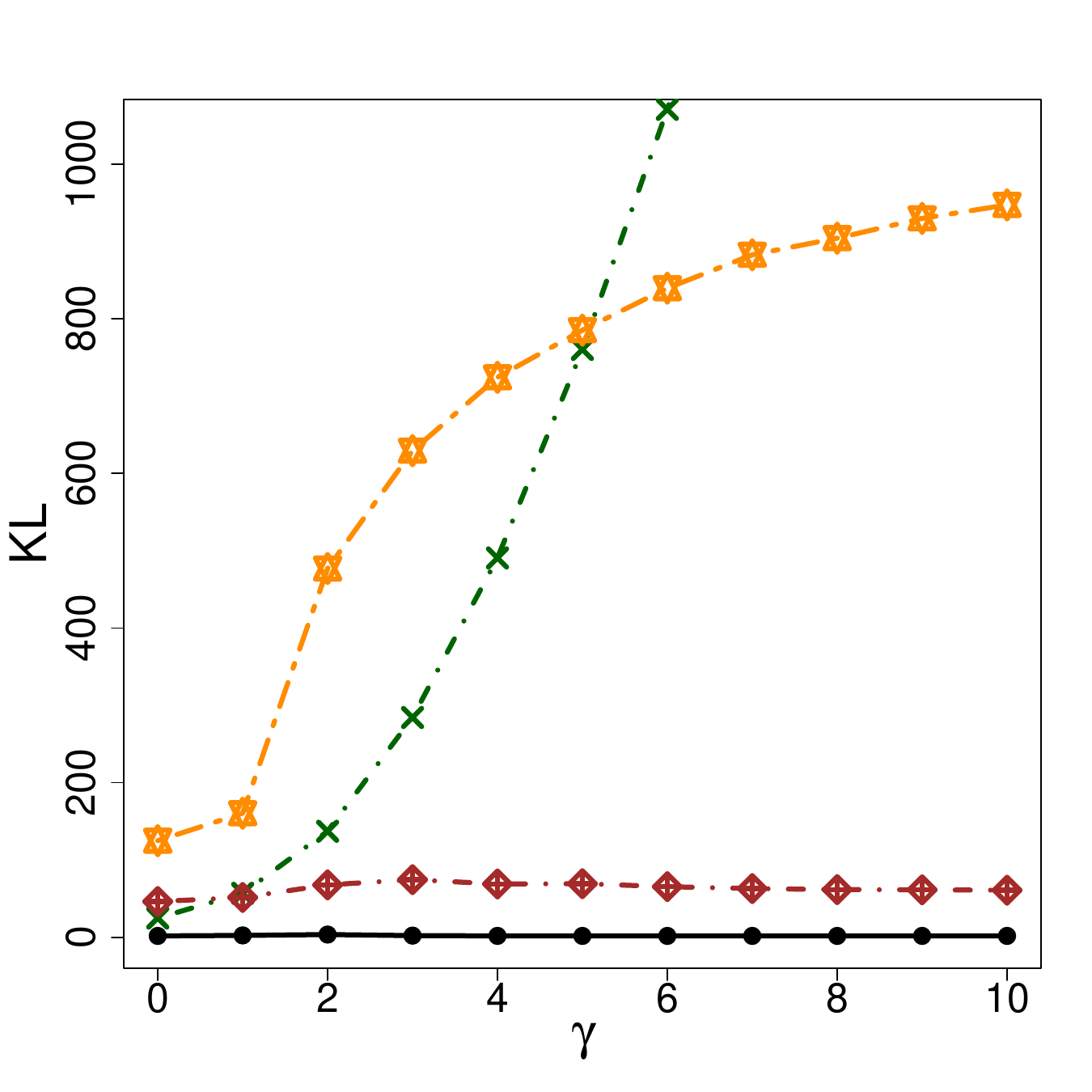}   \\ [-4mm]  \rotatebox{90}{\textbf{\footnotesize{$p=60$}}}&\includegraphics[width=.31\textwidth]
  {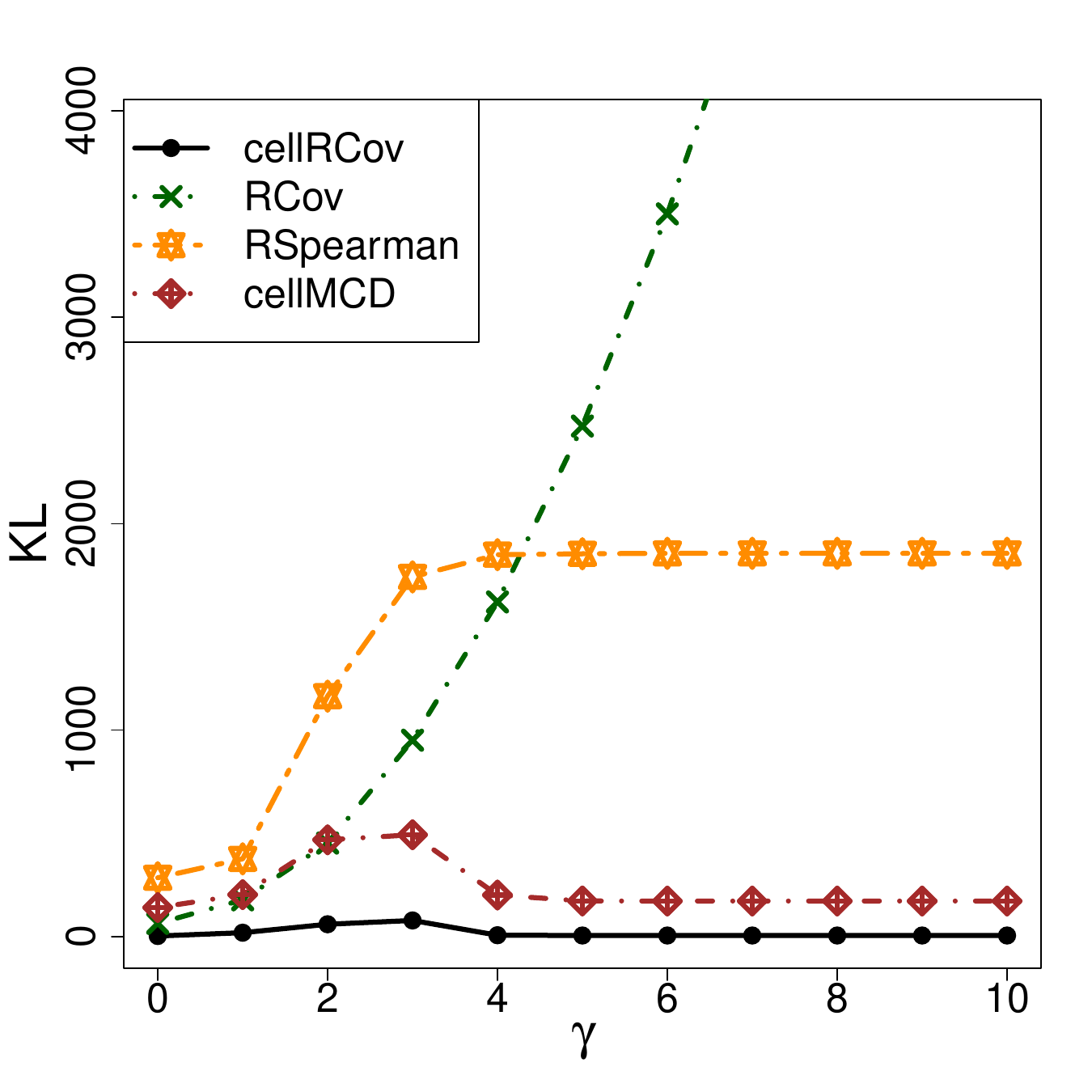}&\includegraphics[width=.31\textwidth]
  {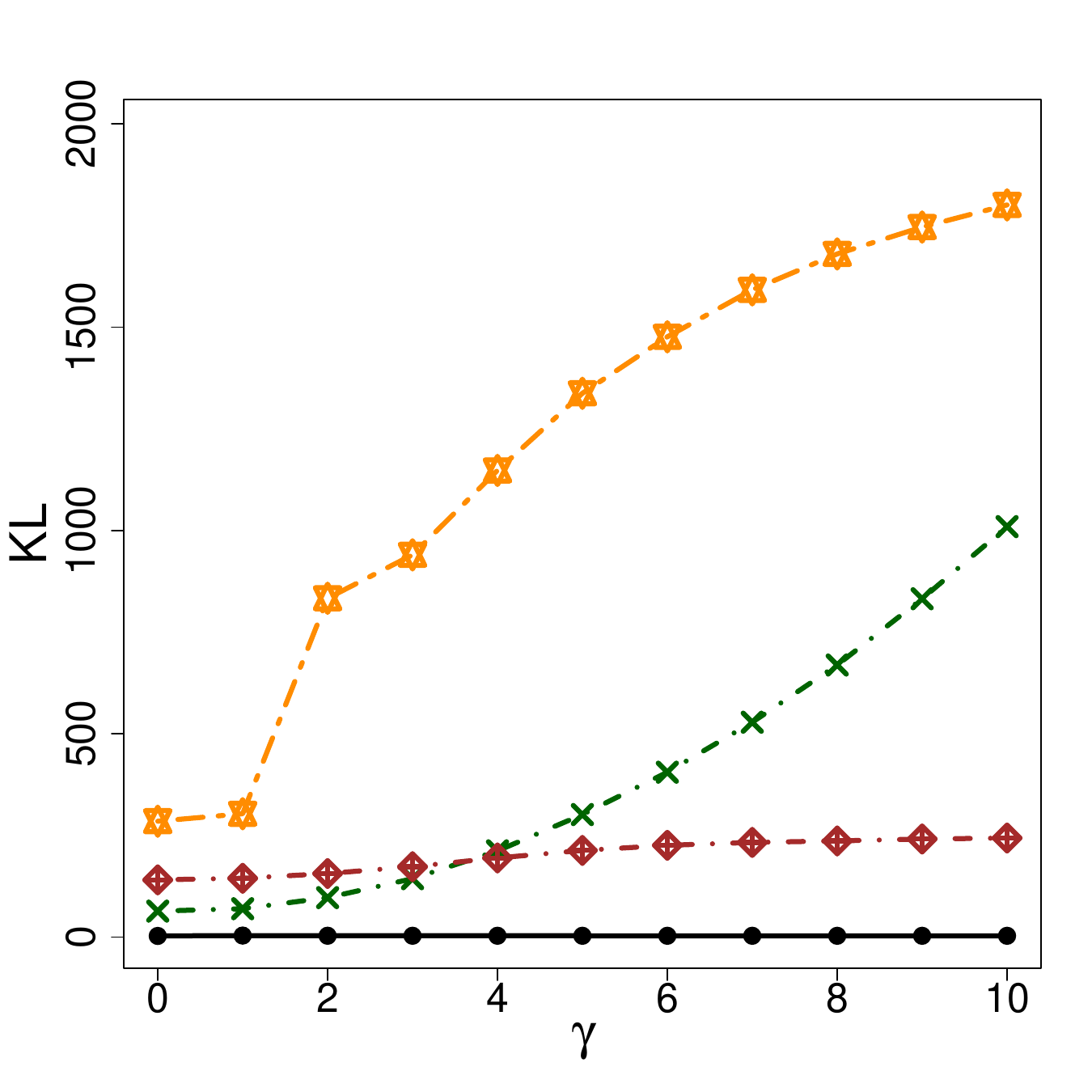}&\includegraphics[width=.31\textwidth]
  {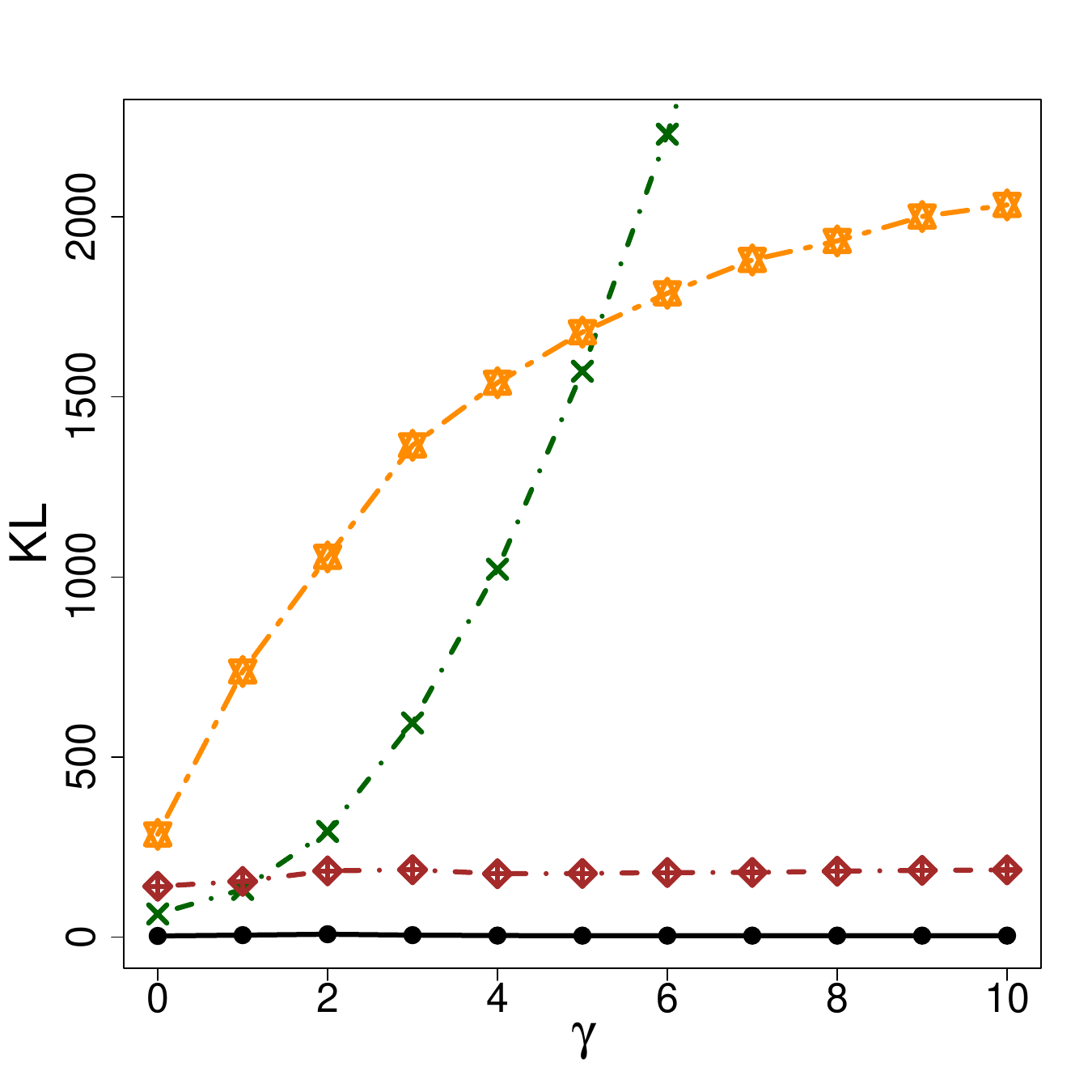}   \\ [-4mm]  \rotatebox{90}{\textbf{\footnotesize{$p=120$}}}&\includegraphics[width=.31\textwidth]
  {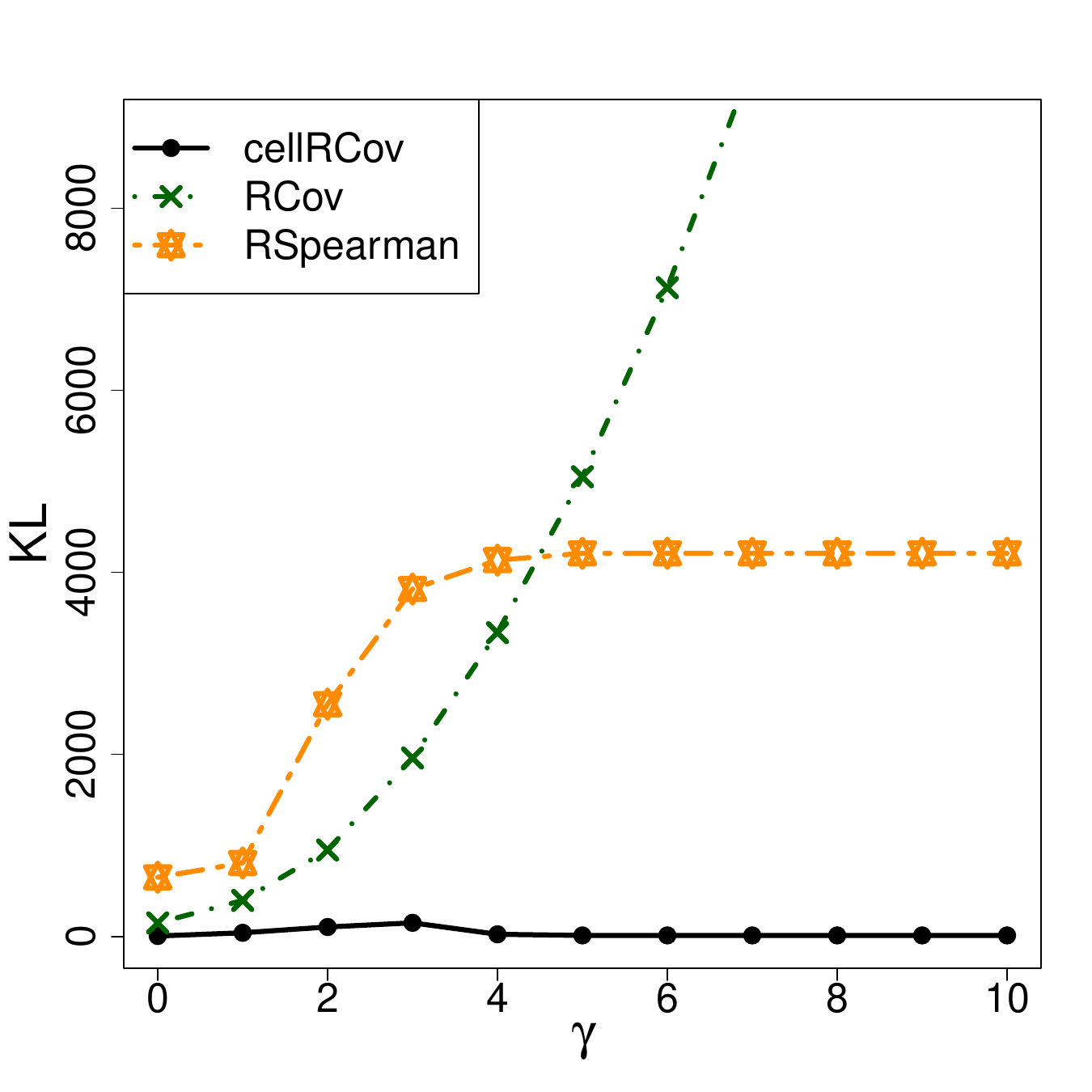}&\includegraphics[width=.31\textwidth]
  {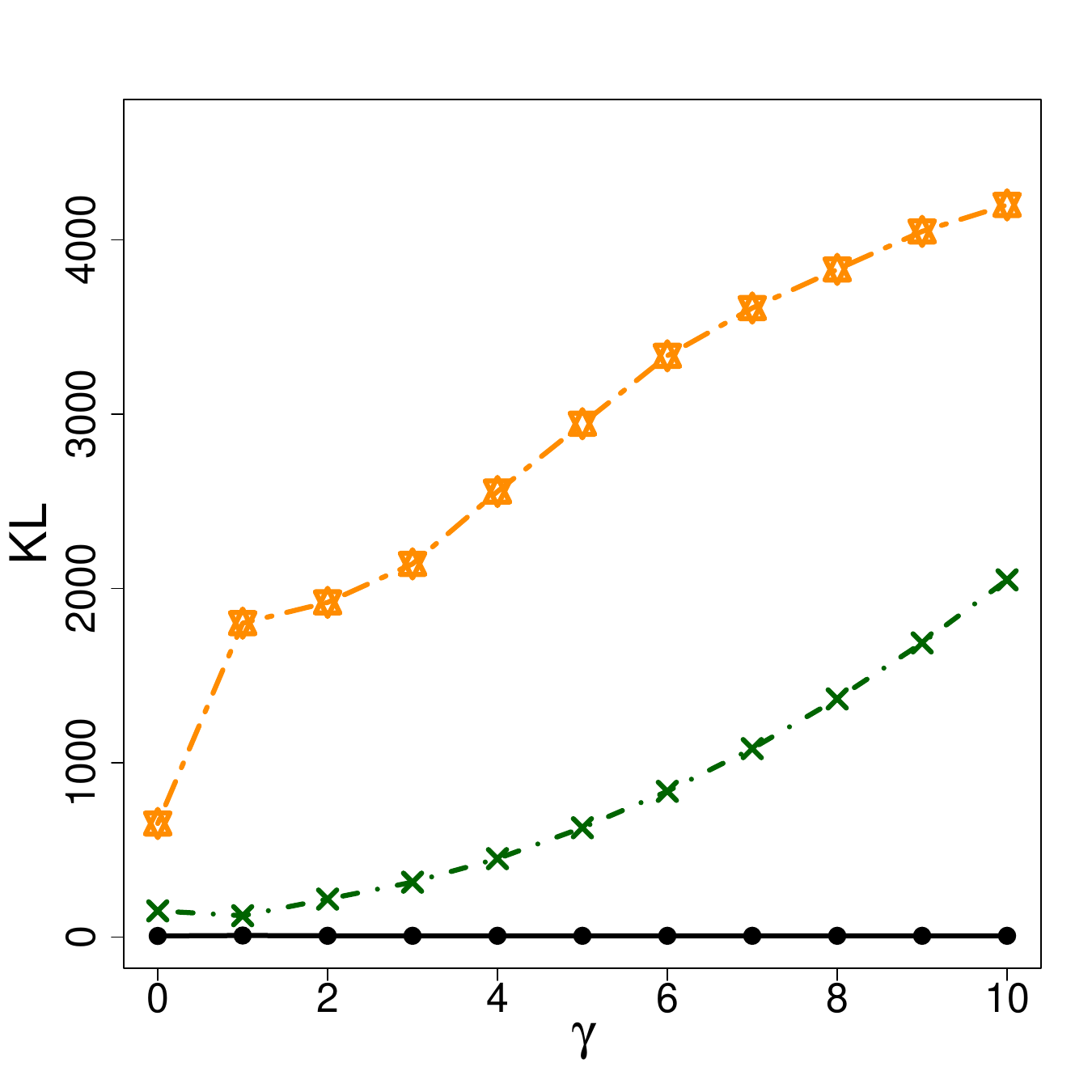}&\includegraphics[width=.31\textwidth]
  {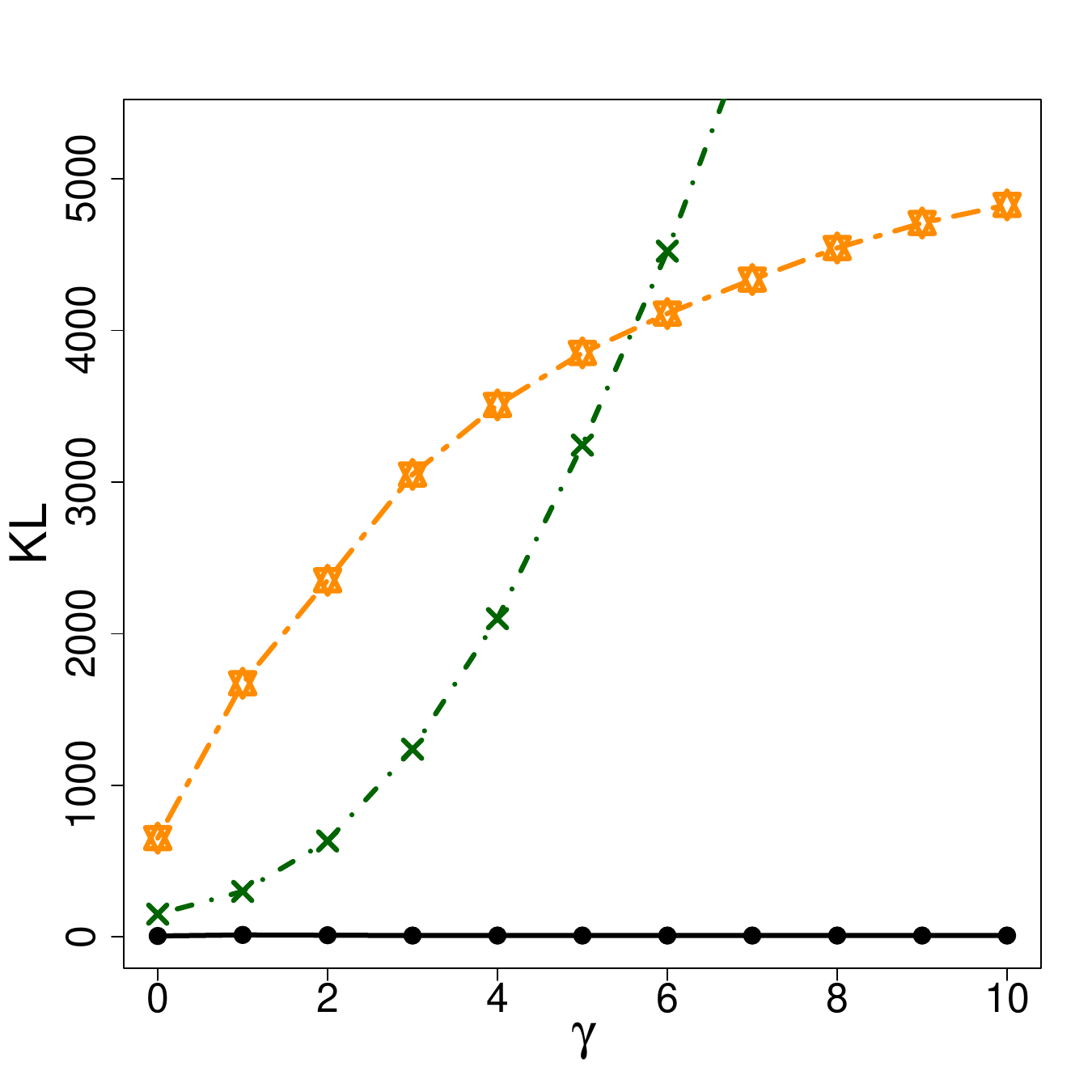}   
\end{tabular}
\caption{Average KL attained by cellRCov, 
RCov, RSpearman, and cellMCD in the presence of either 
cellwise outliers, casewise outliers, or both, 
for the planar covariance model in dimensions  
$p$ in $\lbrace30,60,120\rbrace$, with 20\p 
of missing cells.}
\label{fig:results_planar_NA}
\end{figure}

\begin{figure}[!ht]
\centering
 
\begin{tabular}{M{0.0005\textwidth}M{0.29\textwidth}M{0.29\textwidth}M{0.32\textwidth}}
   &\large \textbf{Cellwise}  & \large \textbf{Casewise} &\large{\textbf{Casewise \& Cellwise}} \\
   \rotatebox{90}{\textbf{\footnotesize{$p=30$}}}&\includegraphics[width=.31\textwidth]
  {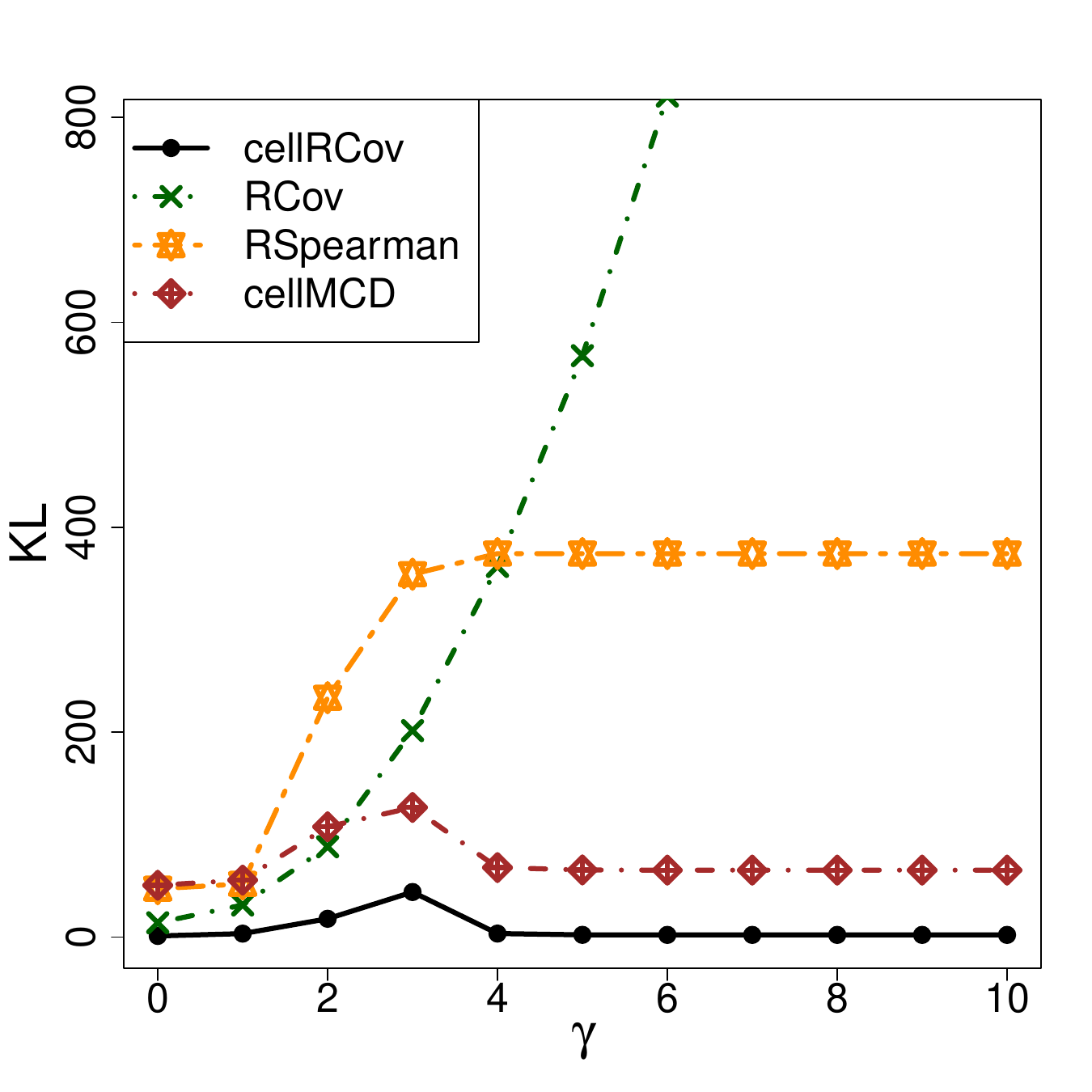}&\includegraphics[width=.31\textwidth]
  {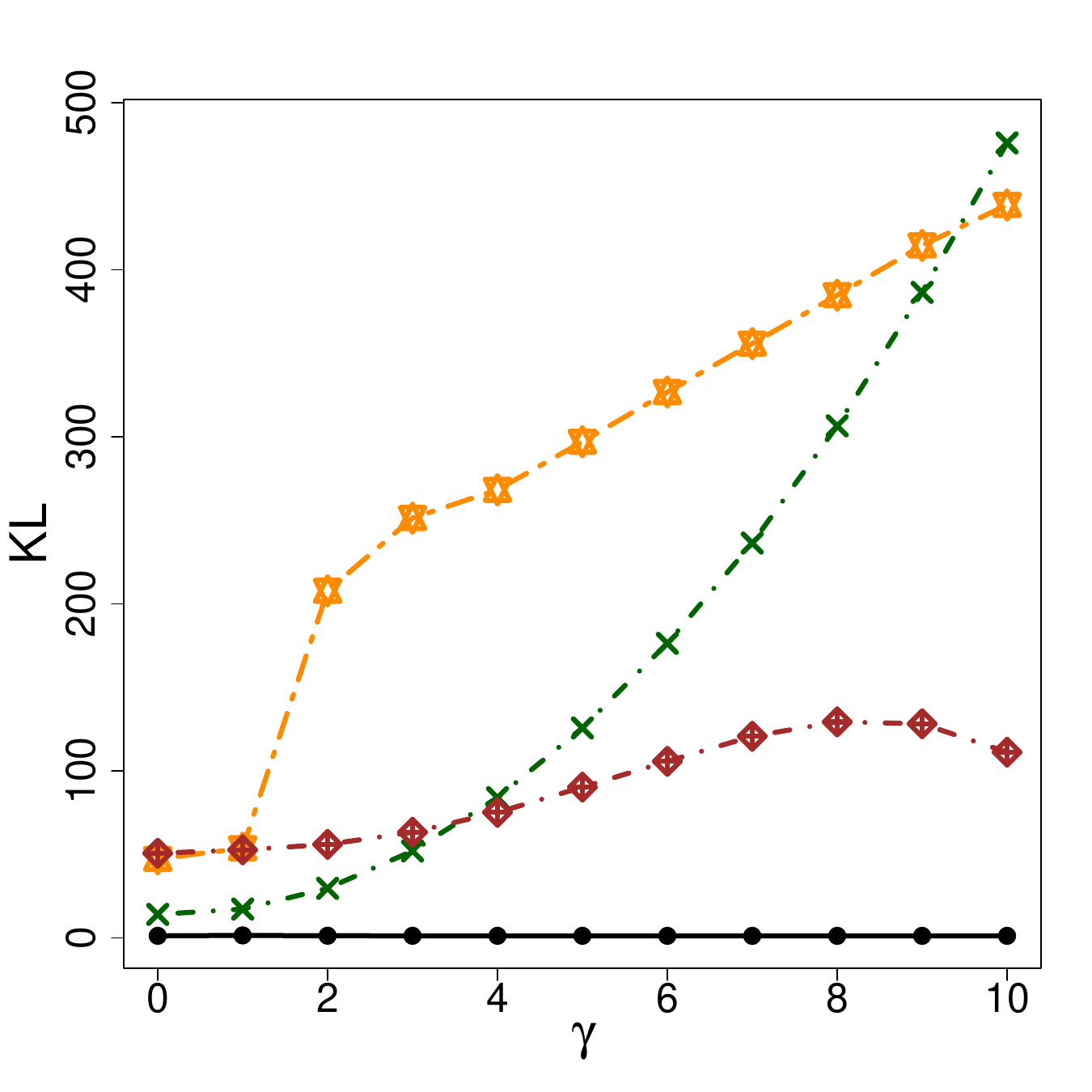}&\includegraphics[width=.31\textwidth]
  {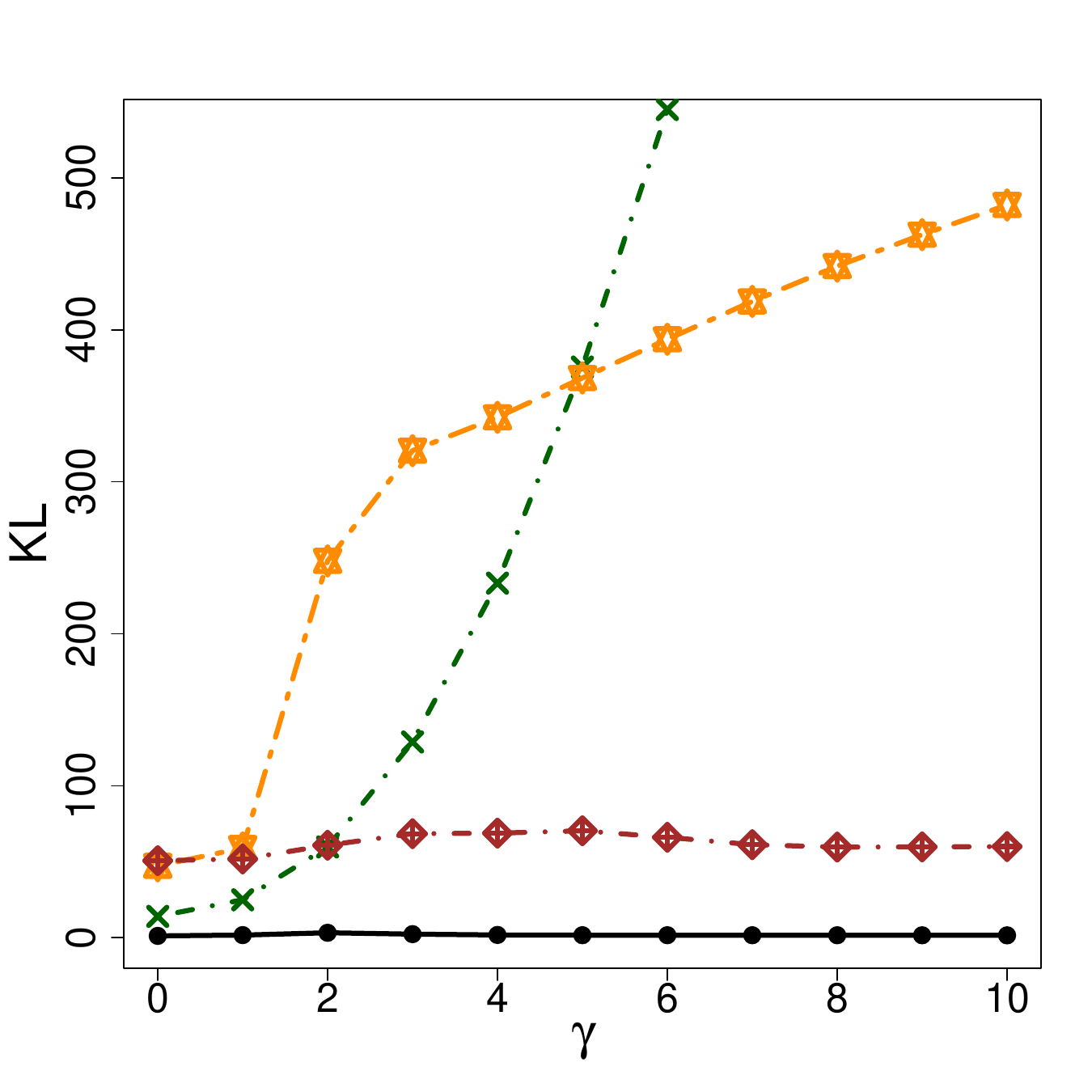}   \\ [-4mm]  \rotatebox{90}{\textbf{\footnotesize{$p=60$}}}&\includegraphics[width=.31\textwidth]
  {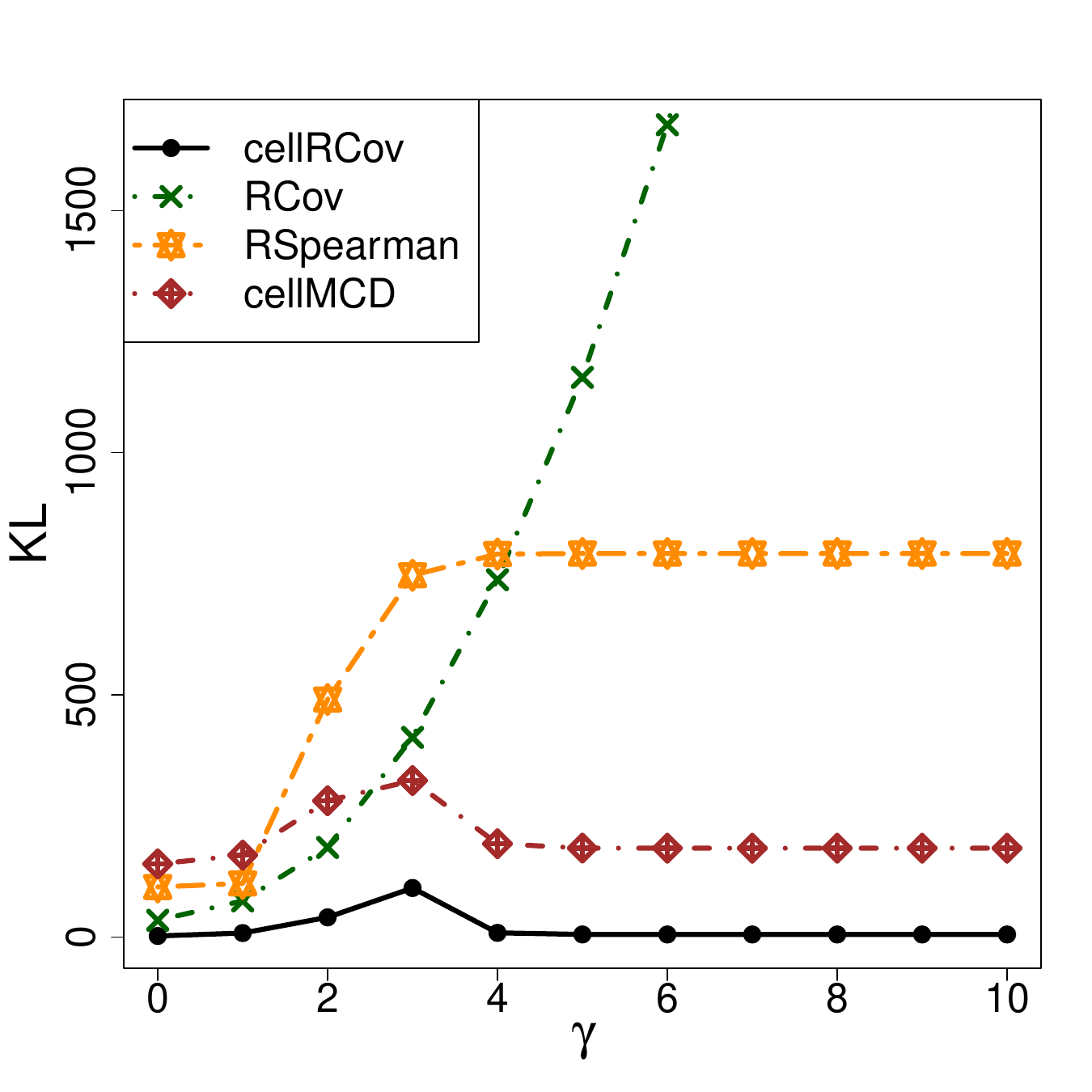}&\includegraphics[width=.31\textwidth]
  {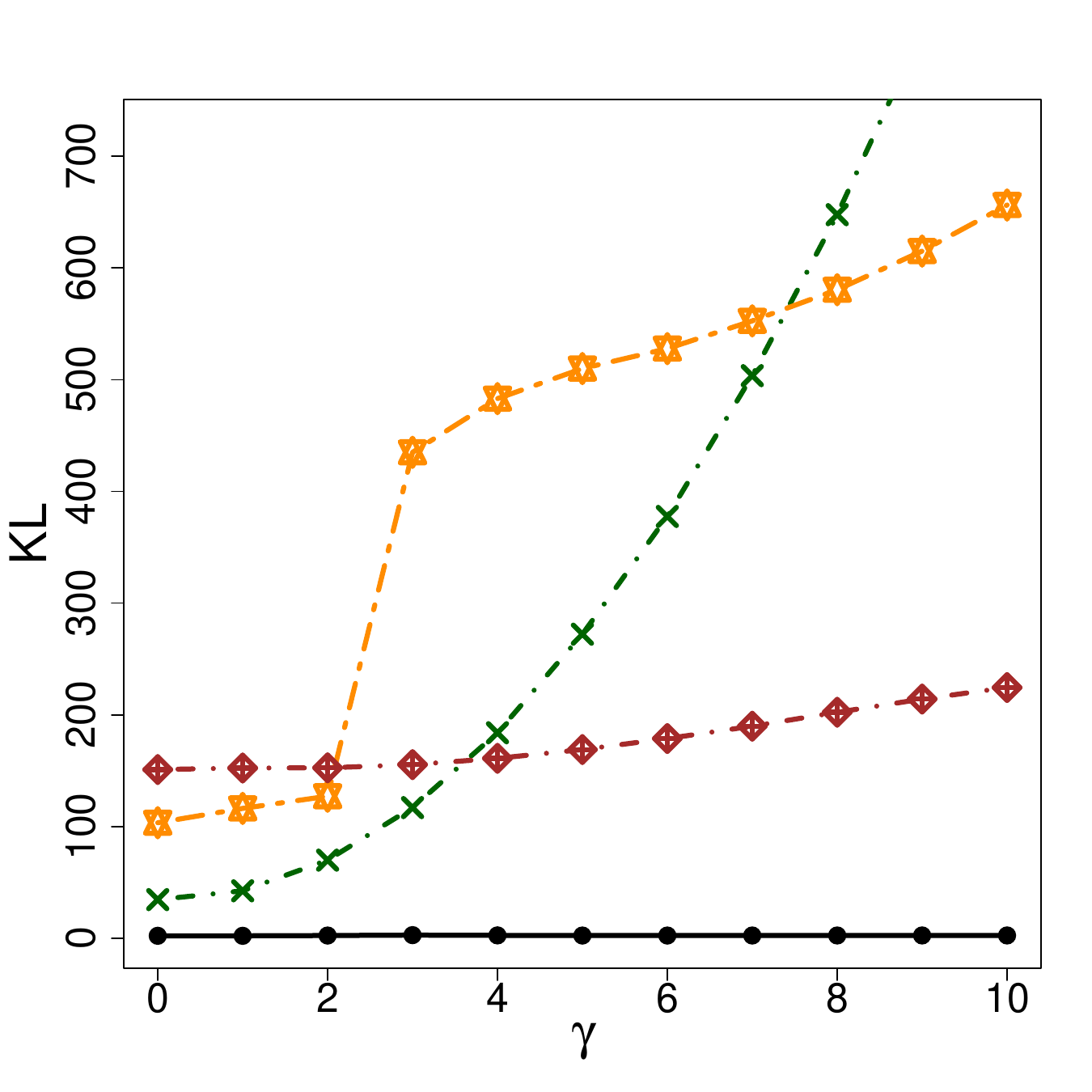}&\includegraphics[width=.31\textwidth]
  {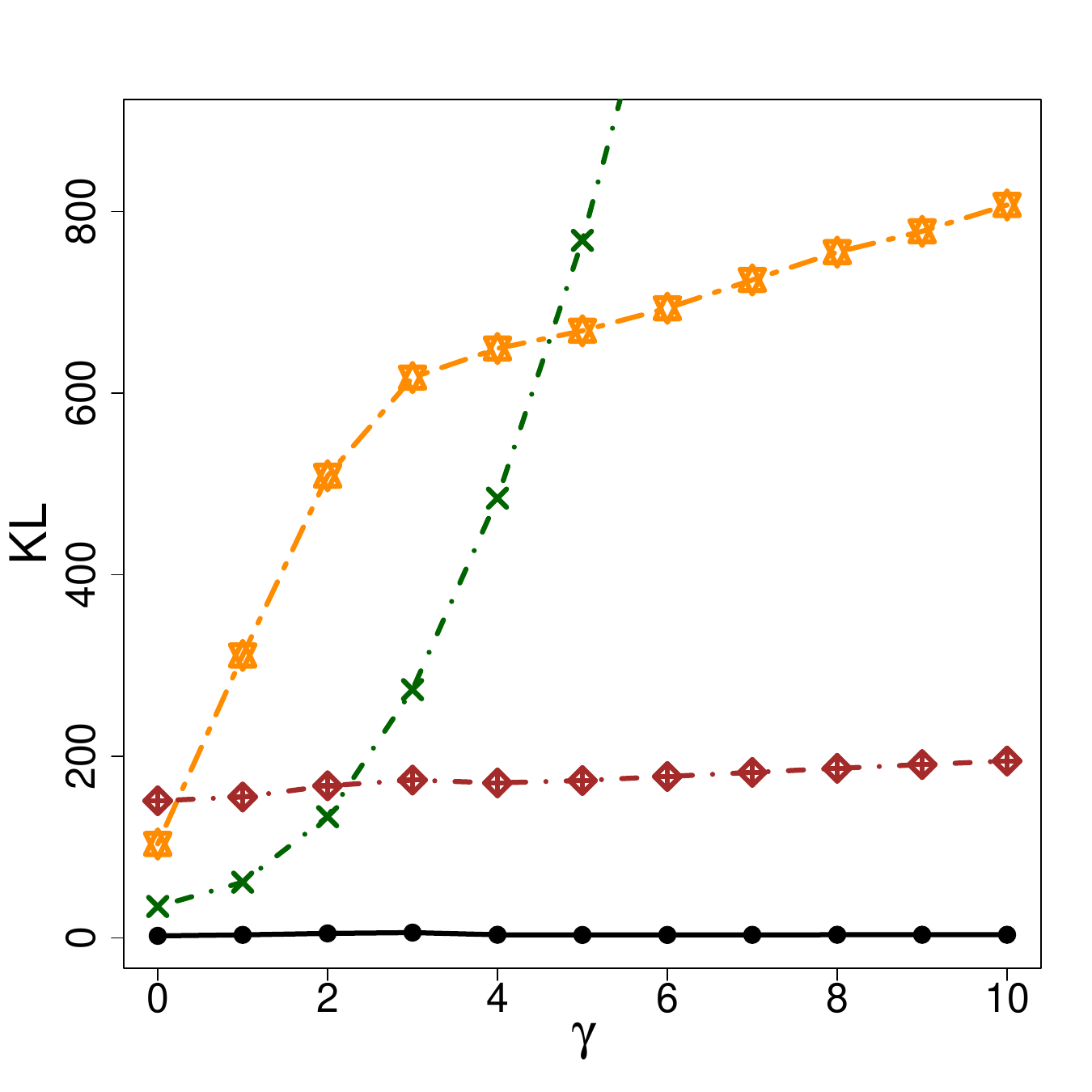}   \\ [-4mm]  \rotatebox{90}{\textbf{\footnotesize{$p=120$}}}&\includegraphics[width=.31\textwidth]
  {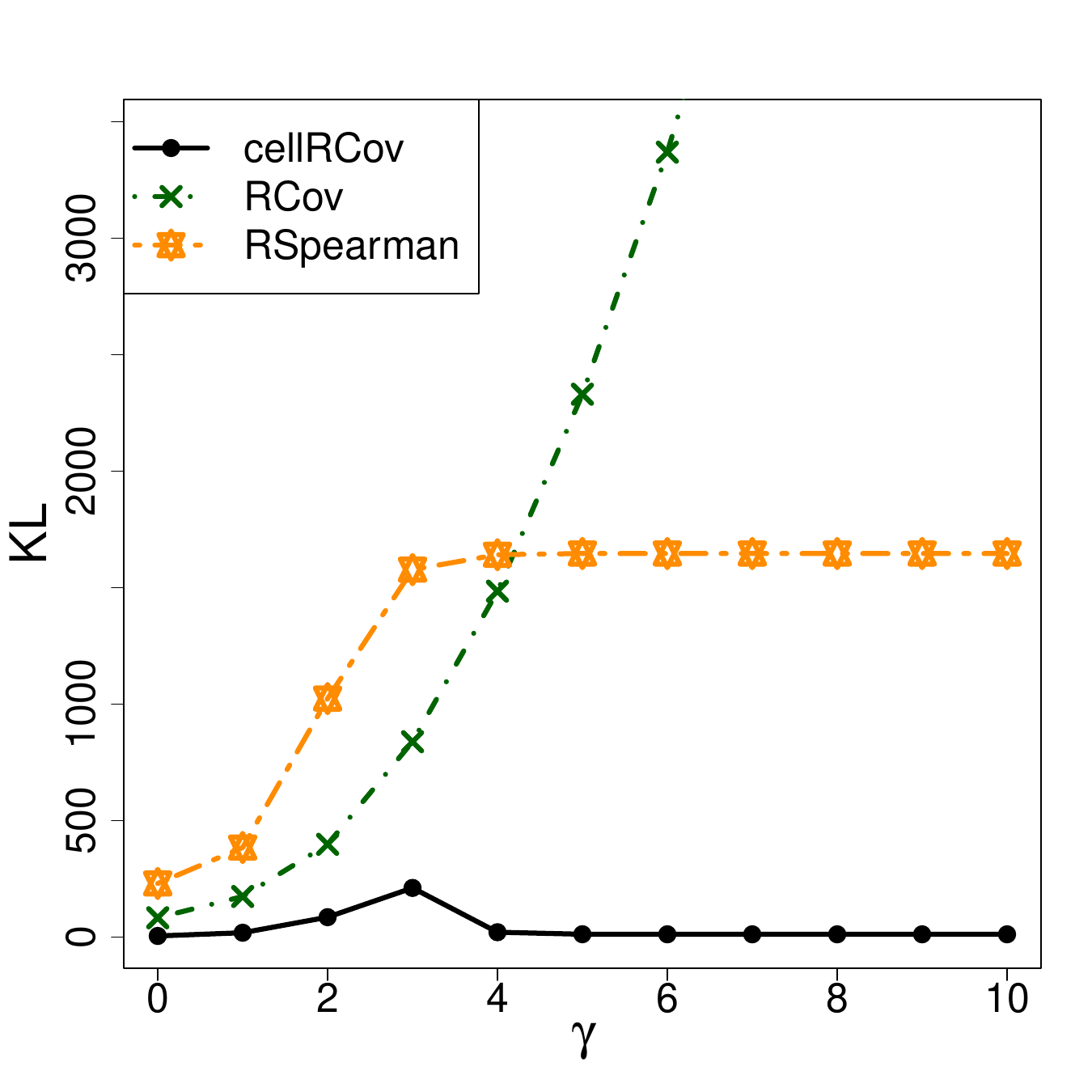}&\includegraphics[width=.31\textwidth]
  {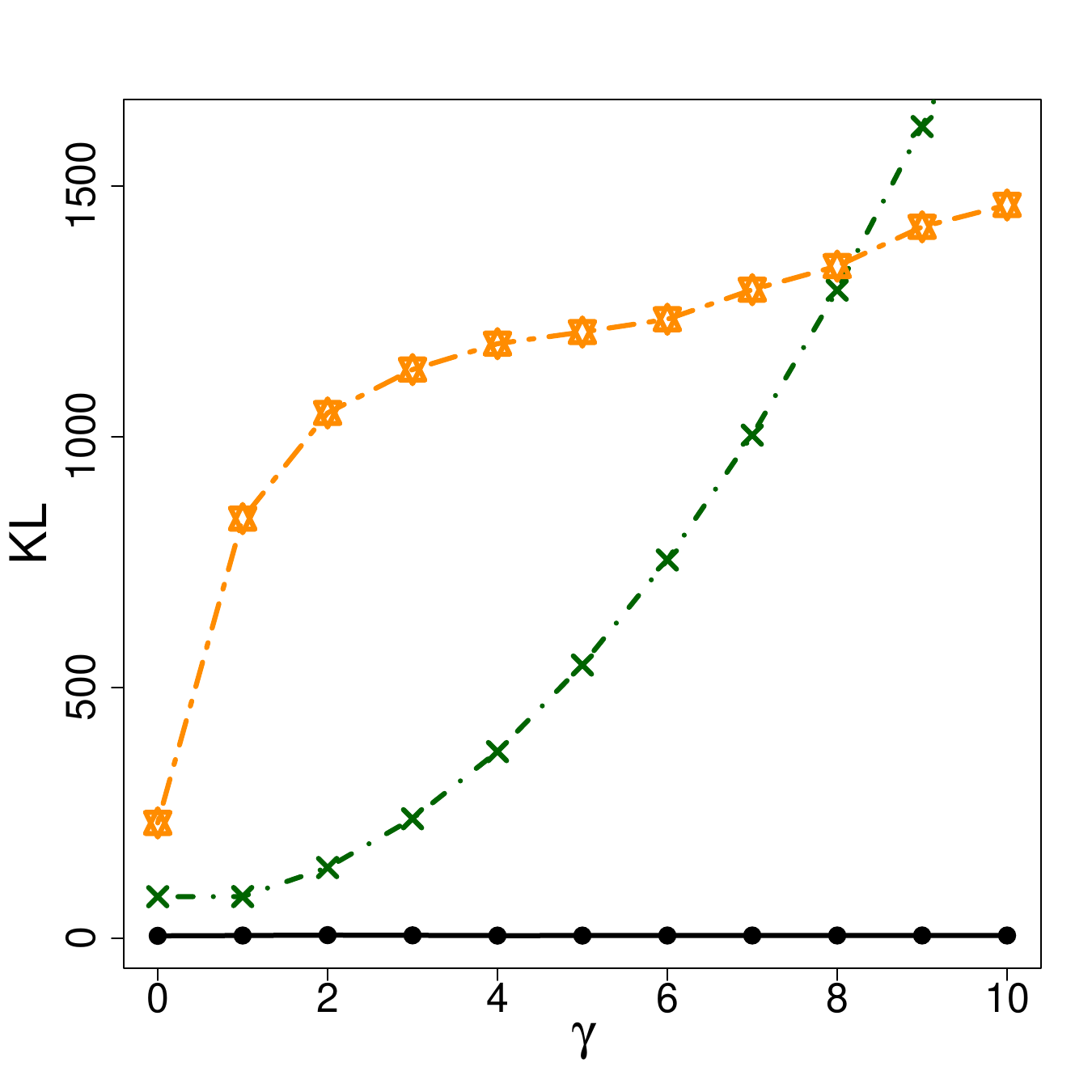}&\includegraphics[width=.31\textwidth]
  {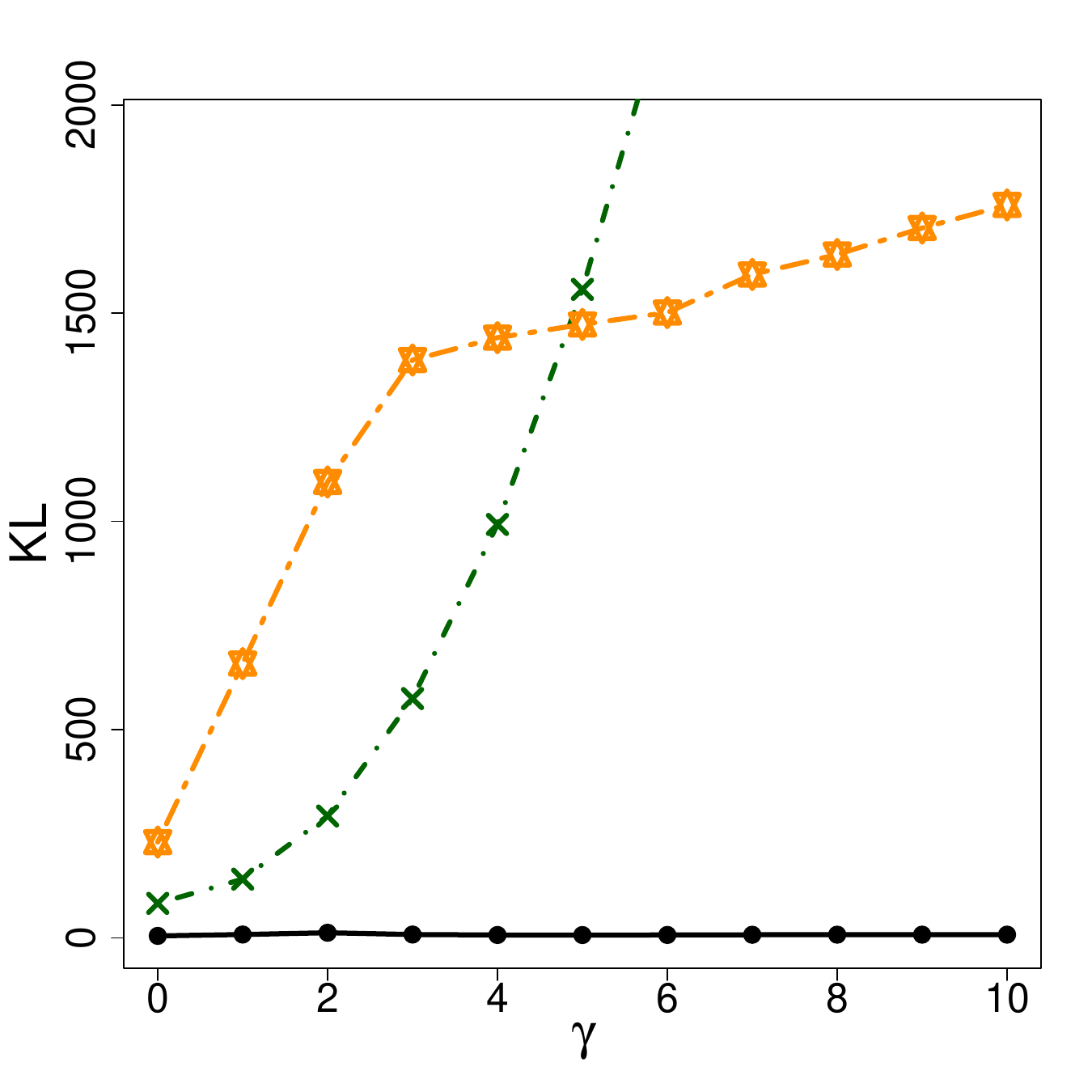}   
\end{tabular}
\caption{Average KL attained by cellRCov, 
RCov, RSpearman, and cellMCD in the presence of either 
cellwise outliers, casewise outliers, or both,
for the dense covariance model in dimensions  
$p$ in $\lbrace30,60,120\rbrace$, with 20\p 
of missing cells.}
\label{fig:results_dense_NA}
\end{figure}

\clearpage
\section{\large Sensitivity to tuning parameters}
\label{app:sensitivity}

To assess the sensitivity of cellRCov to the 
tuning parameters we perform an additional 
simulation study, where we investigate the 
cellRCov performance for different values 
of $\rk$ and $\delta$. 
Figures~\ref{fig:sens_A09}--\ref{fig:sens_dense} 
report the sensitivity of cellRCov to the 
tuning parameters $\rk$ and $\delta$ for the 
A09, A06, planar, and dense covariance models 
for $p=30,60,120$. In each figure, the rows
correspond to uncontaminated data ($\gamma=0$) 
and contaminated data with both casewise and 
cellwise outliers with $\gamma=6$. Each panel
reports the average $\log(\mathrm{KL})$ as a 
function of $\rk$, with different curves 
corresponding to different values of $\delta$.

The figures show that the optimal choice of 
$\rk$ depends on the underlying covariance 
structure. For the A09 and A06 covariance 
models, the curves are typically U-shaped: 
very small values of $\rk$ do not capture 
enough of the dominant low-dimensional 
structure, whereas overly large values of 
$\rk$ start to include unnecessary variation 
in the subspace component, which deteriorates 
the covariance estimate. 

For the planar covariance model, the best 
performance is obtained for very small values 
of $\rk$, which is consistent with the fact 
that this model has a strong low-dimensional 
structure. For the dense covariance model, 
increasing $\rk$ generally does not improve 
the performance and often leads to larger KL 
values, suggesting that the dependence is 
better handled by the residual covariance 
component than by forcing a larger 
low-rank fit.

The sensitivity to $\delta$ depends on the 
covariance model and on the dimension. Small 
values of $\delta$ leave the residual 
covariance close to the unregularized 
weighted estimate, whereas larger values 
shrink it more strongly toward the diagonal 
target. The plots show that the optimal 
amount of shrinkage varies. In 
several settings, especially for larger $p$, 
some regularization improves stability, 
but excessive shrinkage can increase the 
KL loss.

When the data are contaminated, the best 
performance tends to be achieved for slightly 
larger values of $\rk$, suggesting that a 
somewhat richer low-rank fit can be useful 
for computing robust fitted points and 
residuals. In contrast, the optimal value 
of $\delta$ appears broadly similar between 
the contaminated and uncontaminated settings.

\begin{figure}[H]
\centering
\begin{tabular}{M{0.0005\textwidth}M{0.31\textwidth}M{0.31\textwidth}M{0.31\textwidth}}
   & \large \textbf{$p=30$}  & \large \textbf{$p=60$} & \large{\textbf{$p=120$}} \\ [-4mm]
   \rotatebox{90}{\textbf{\footnotesize{$\gamma=0$}}}
   &\includegraphics[width=.31\textwidth]{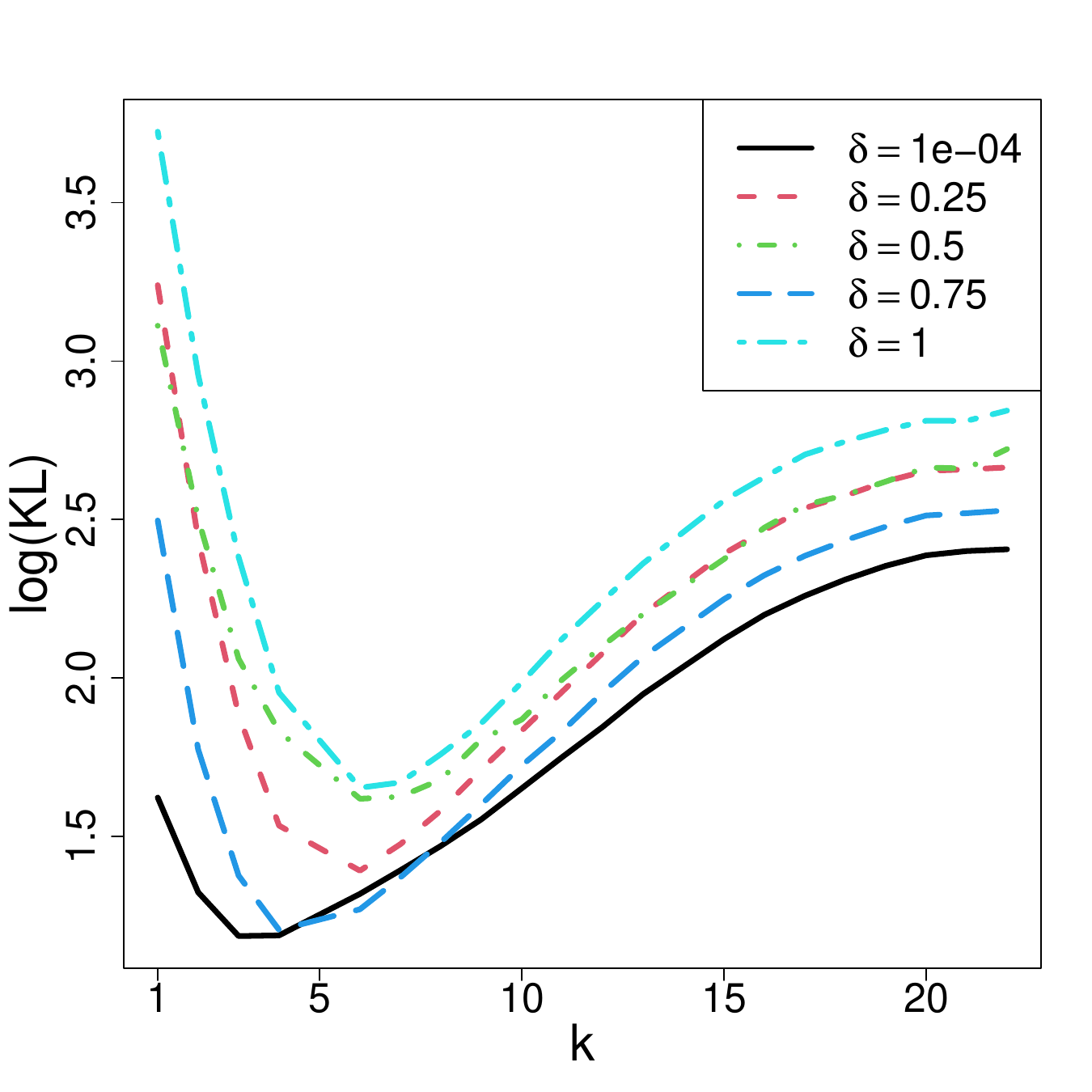}
   &\includegraphics[width=.31\textwidth]{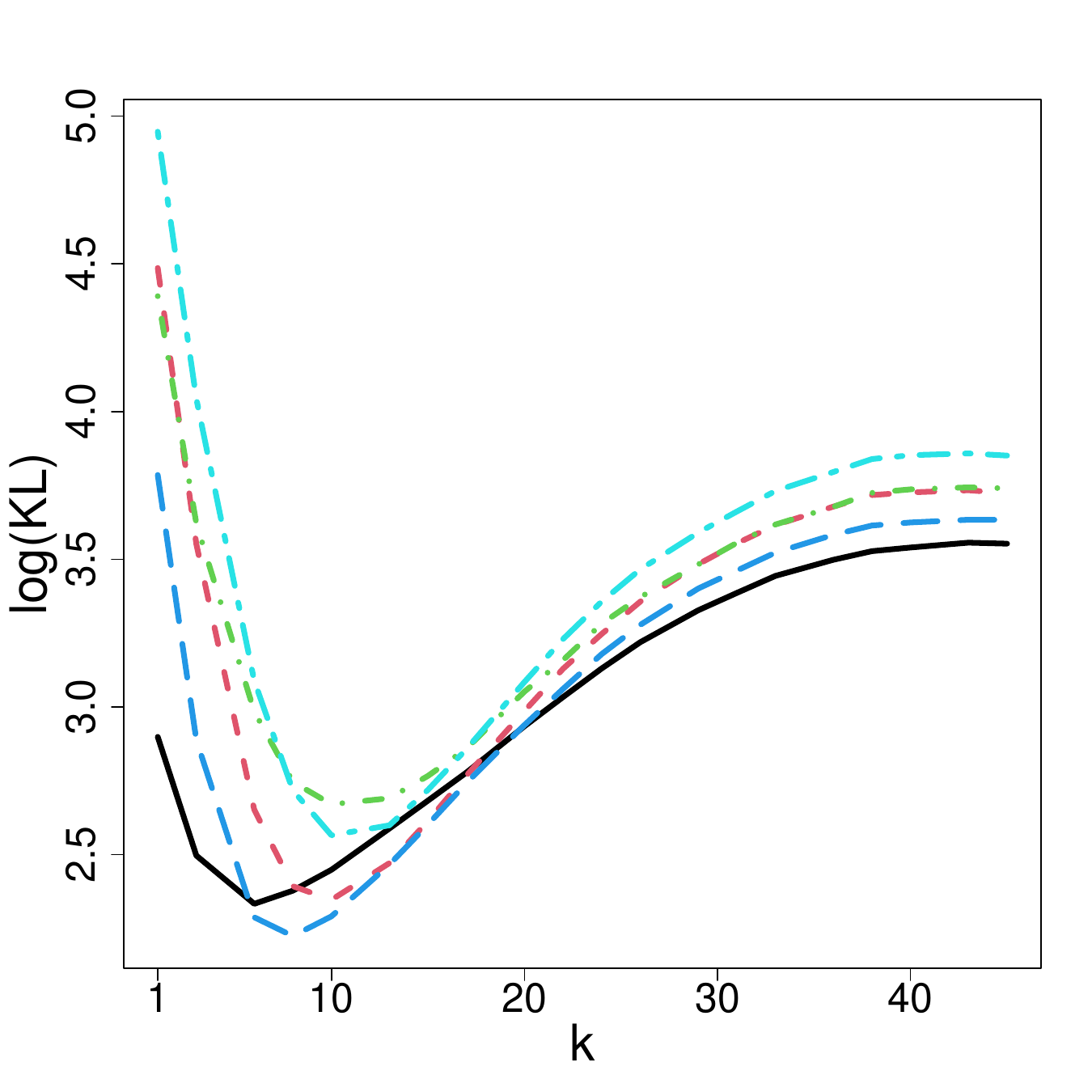}
   &\includegraphics[width=.31\textwidth]{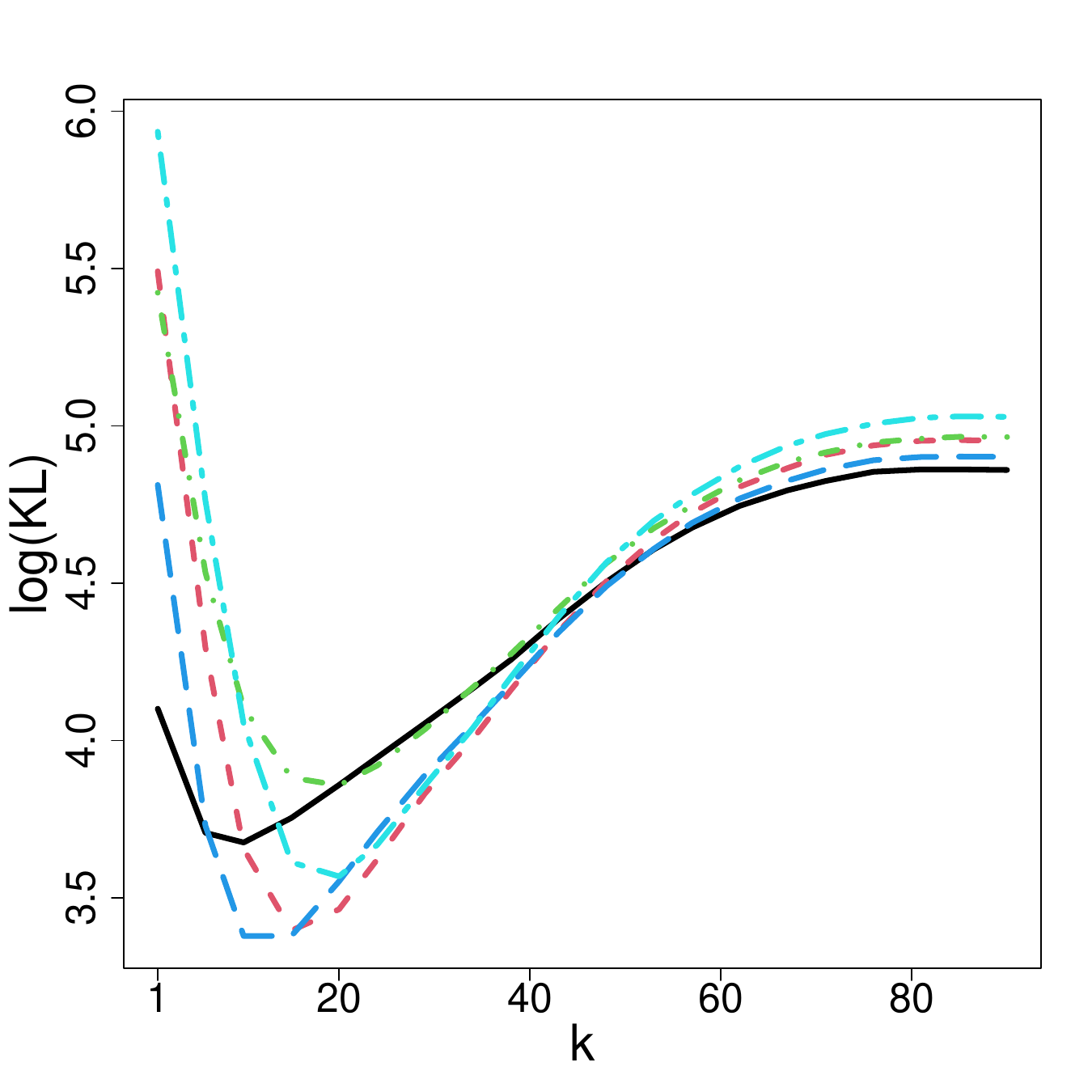} \\ [-4mm]
   \rotatebox{90}{\textbf{\footnotesize{$\gamma=6$}}}
   &\includegraphics[width=.31\textwidth]{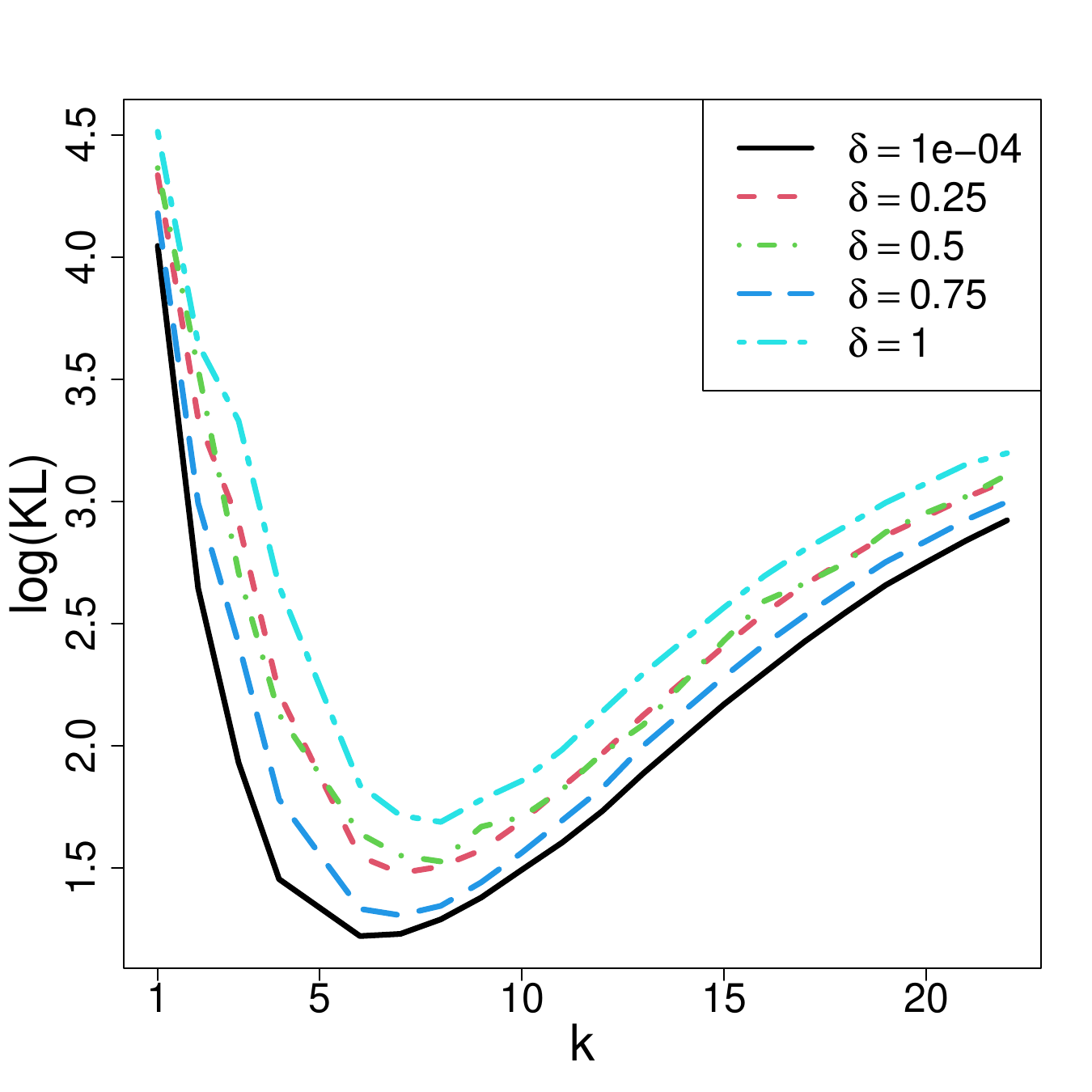}
   &\includegraphics[width=.31\textwidth]{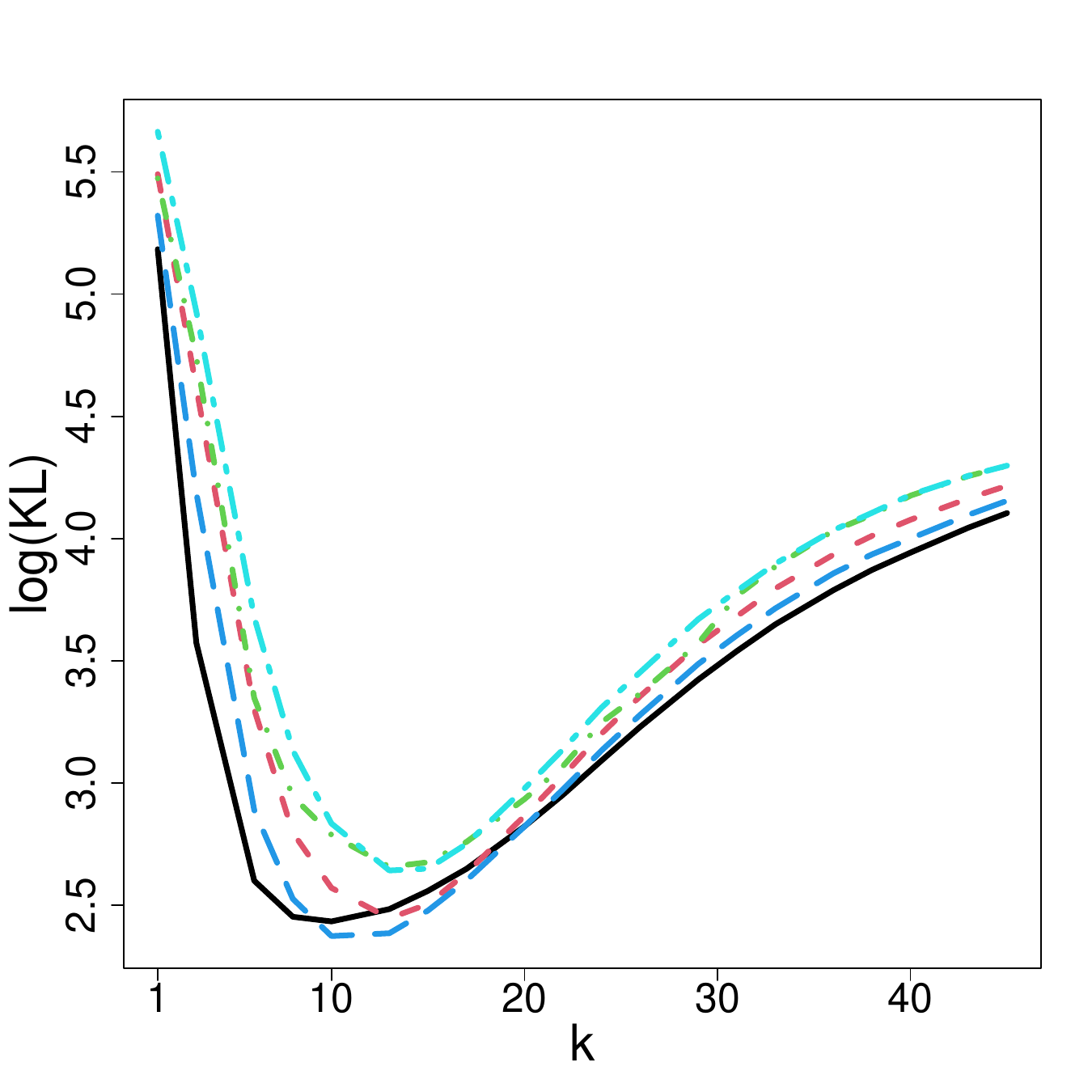}
   &\includegraphics[width=.31\textwidth]{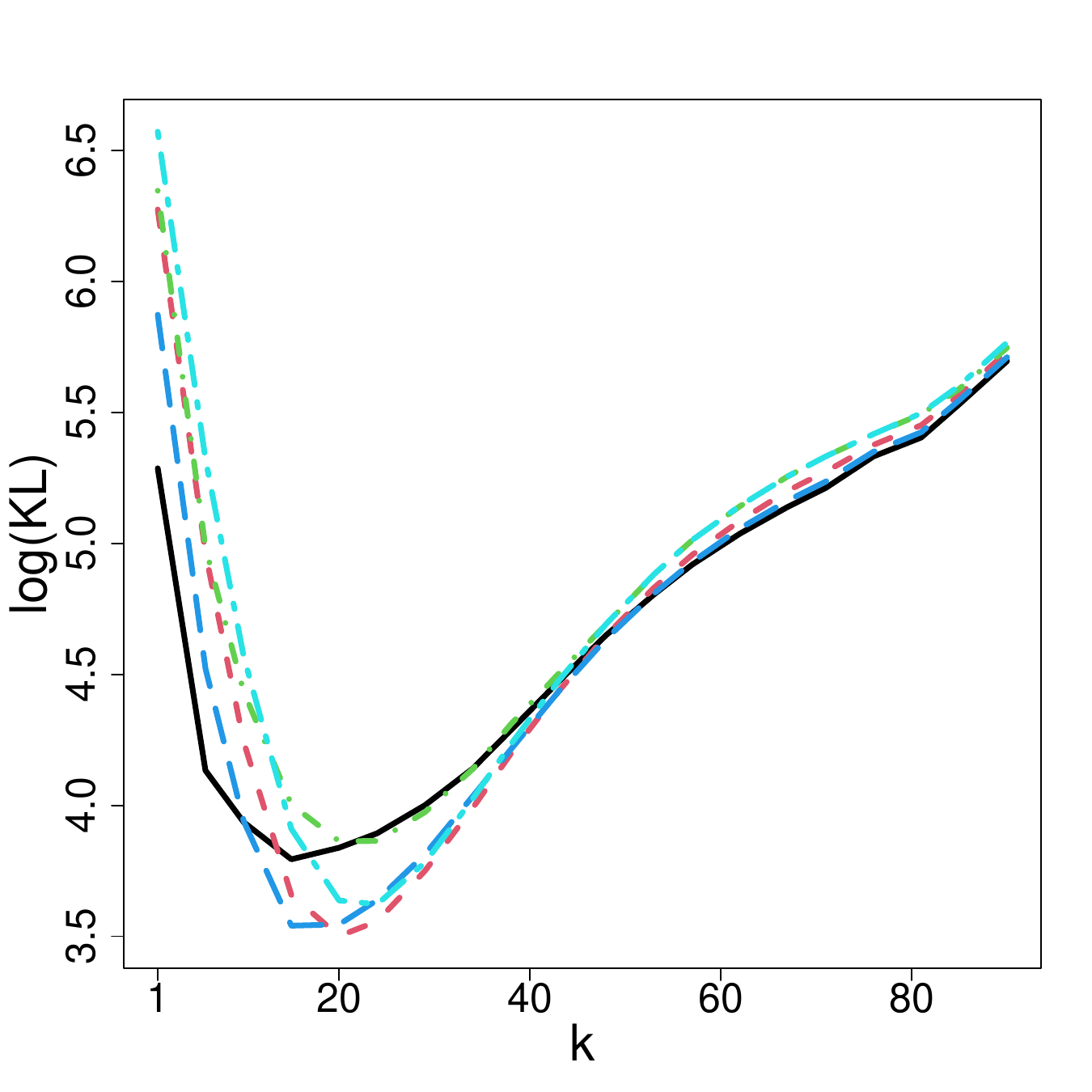}
\end{tabular}
\caption{Average log(KL) attained by cellRCov as a function of $\rk$ and $\delta$ for the A09 covariance
model for to uncontaminated ($\gamma=0$) and contaminated data
($\gamma=6$) in dimensions 
$p$ in $\lbrace30,60,120\rbrace$.}
\label{fig:sens_A09}
\end{figure}

\begin{figure}[H]
\centering
\begin{tabular}{M{0.0005\textwidth}M{0.31\textwidth}M{0.31\textwidth}M{0.31\textwidth}}
   & \large \textbf{$p=30$}  & \large \textbf{$p=60$} & \large{\textbf{$p=120$}} \\ [-4mm]
   \rotatebox{90}{\textbf{\footnotesize{$\gamma=0$}}}
   &\includegraphics[width=.31\textwidth]{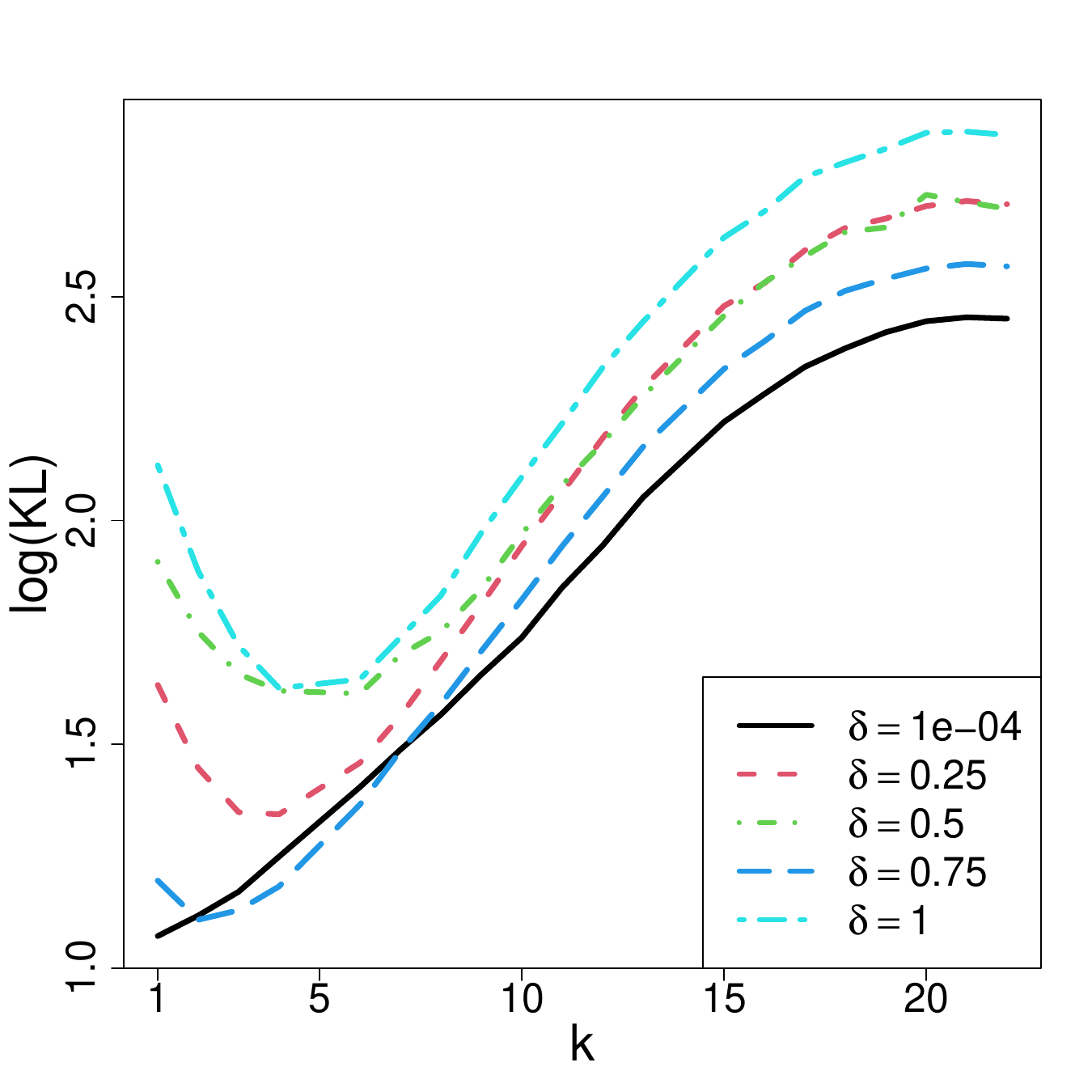}
   &\includegraphics[width=.31\textwidth]{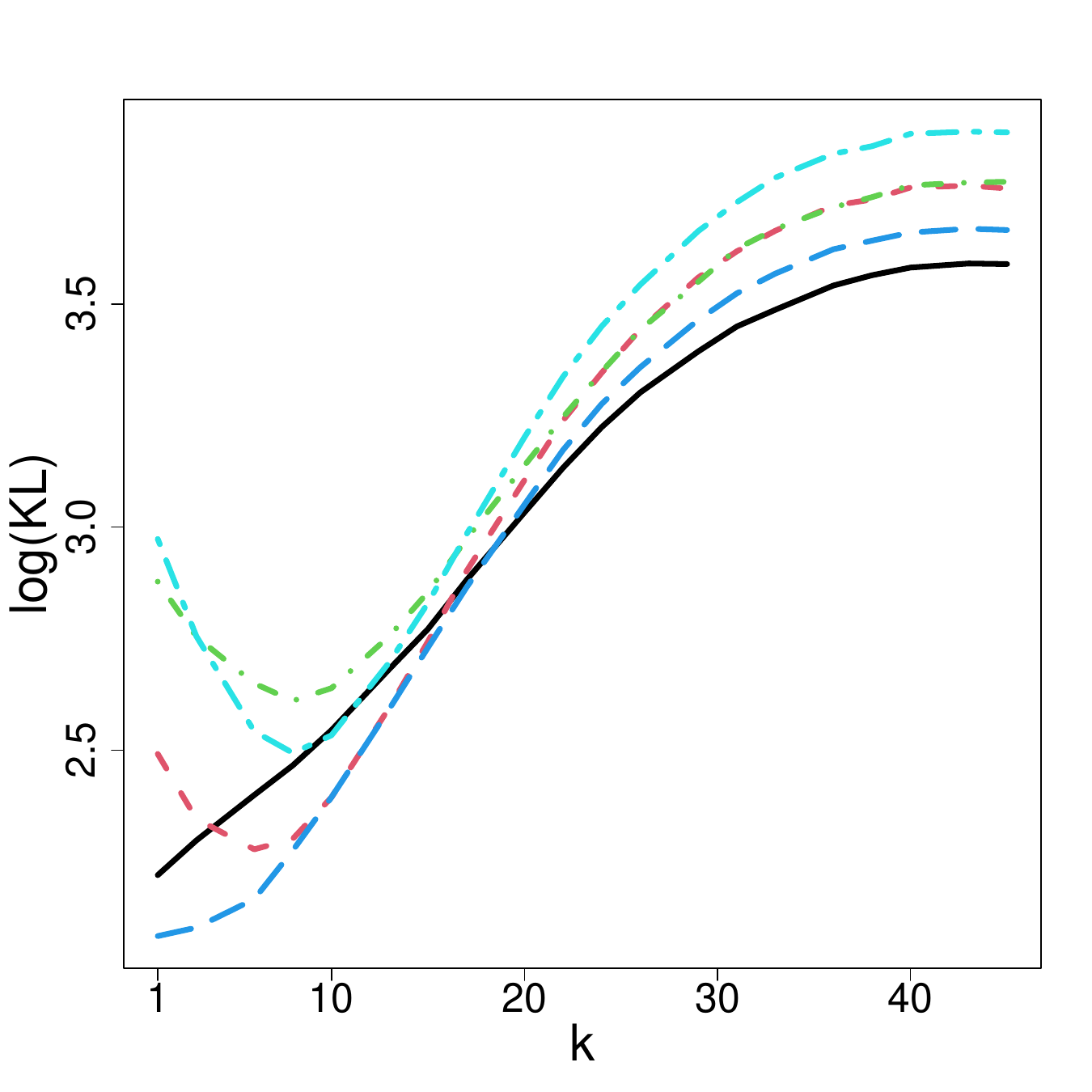}
   &\includegraphics[width=.31\textwidth]{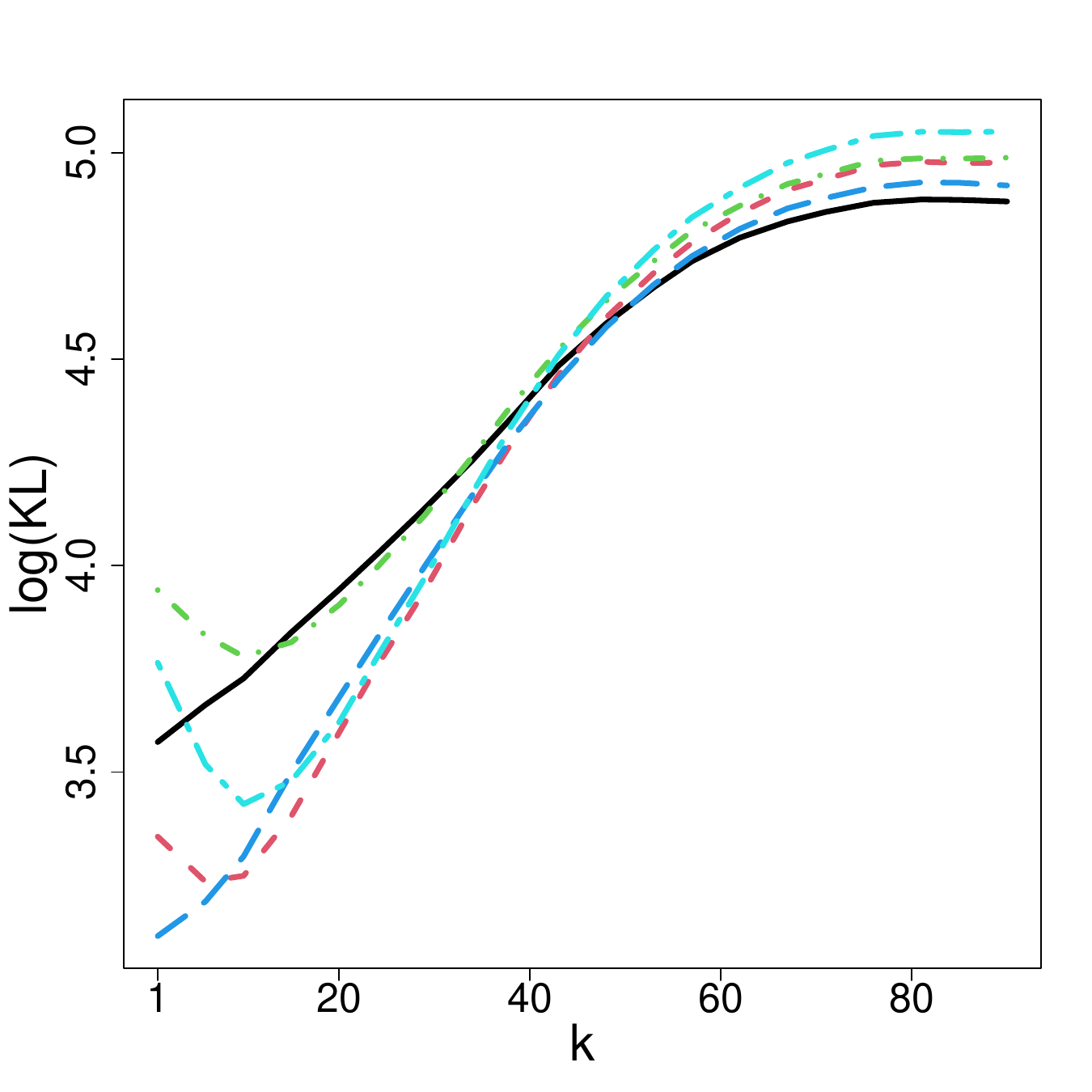} \\ [-4mm]
   \rotatebox{90}{\textbf{\footnotesize{$\gamma=6$}}}
   &\includegraphics[width=.31\textwidth]{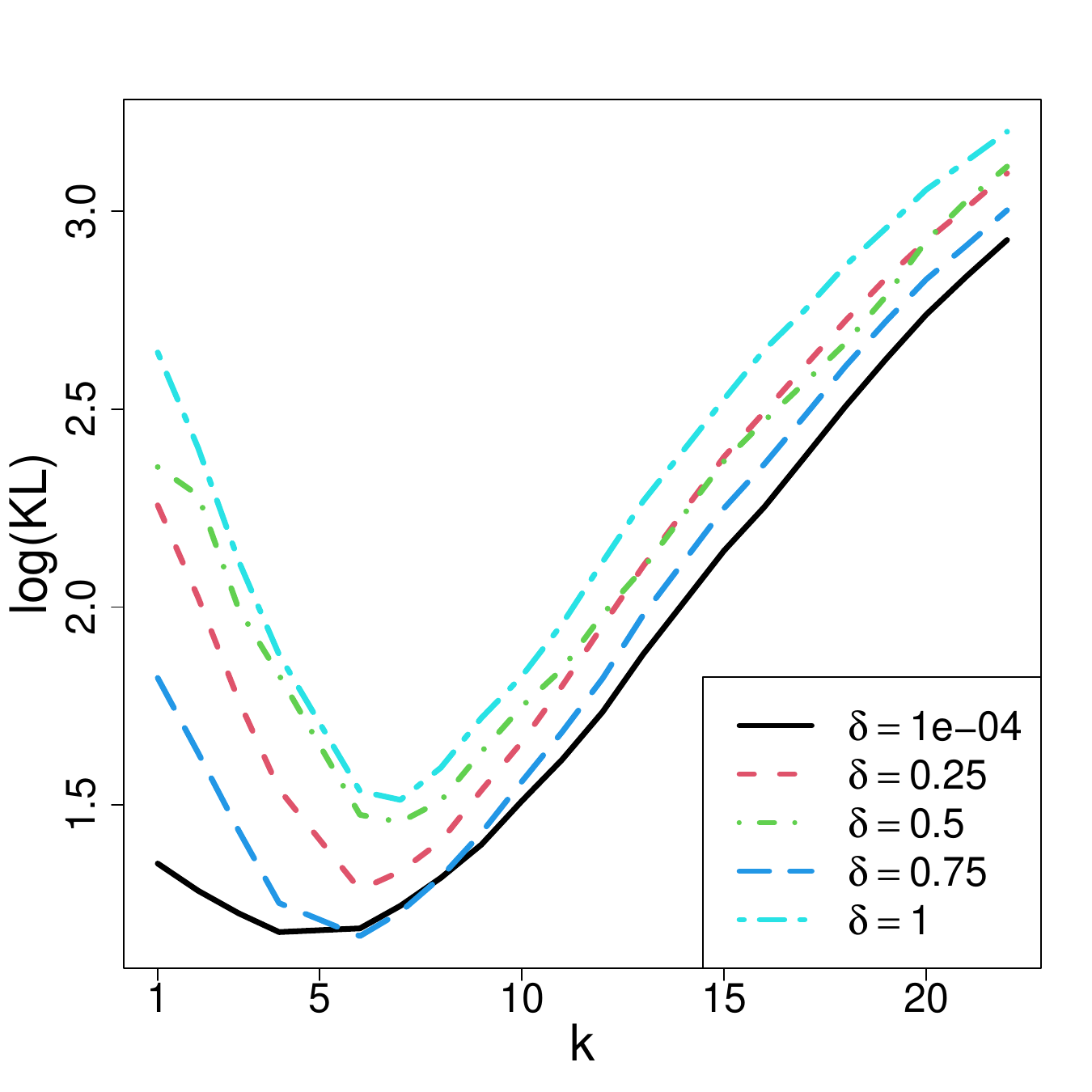}
   &\includegraphics[width=.31\textwidth]{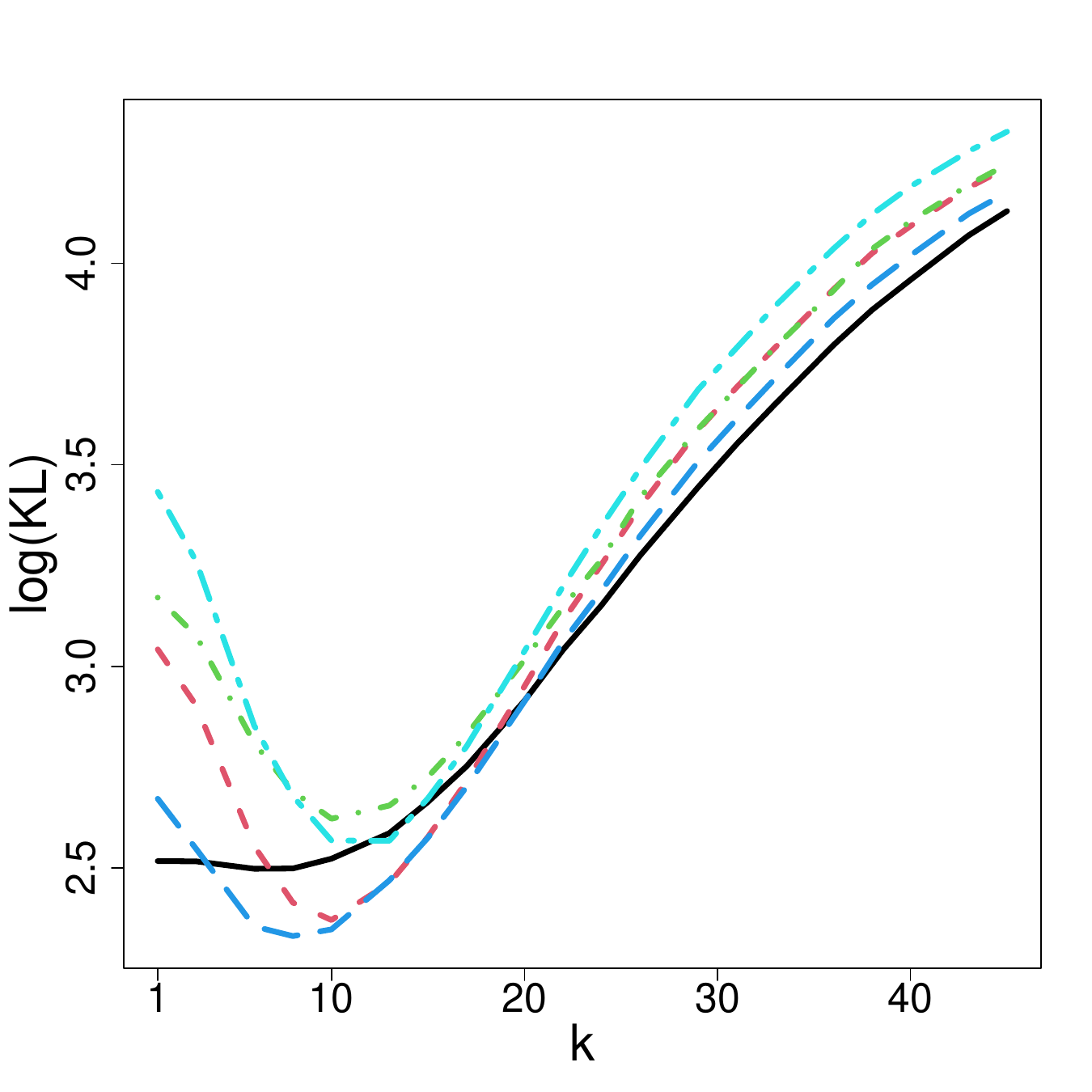}
   &\includegraphics[width=.31\textwidth]{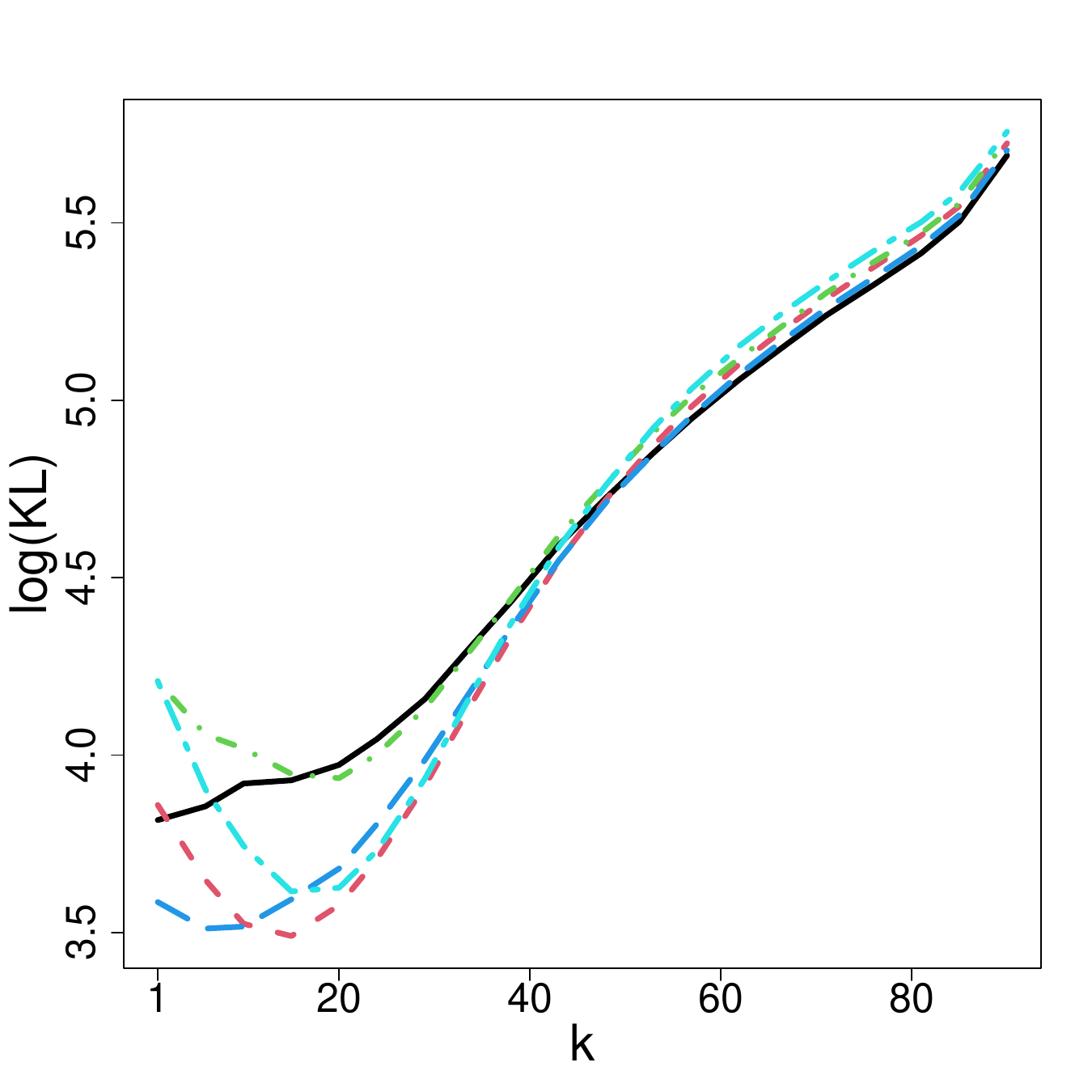}
\end{tabular}
\caption{Average log(KL) attained by cellRCov as a function of $\rk$ and $\delta$ for the A06 covariance
model for to uncontaminated ($\gamma=0$) and contaminated data
($\gamma=6$) in dimensions 
$p$ in $\lbrace30,60,120\rbrace$.}
\label{fig:sens_A06}
\end{figure}

\begin{figure}[H]
\centering
\begin{tabular}{M{0.0005\textwidth}M{0.31\textwidth}M{0.31\textwidth}M{0.31\textwidth}}
   & \large \textbf{$p=30$}  & \large \textbf{$p=60$} & \large{\textbf{$p=120$}} \\ [-4mm]
   \rotatebox{90}{\textbf{\footnotesize{$\gamma=0$}}}
   &\includegraphics[width=.31\textwidth]{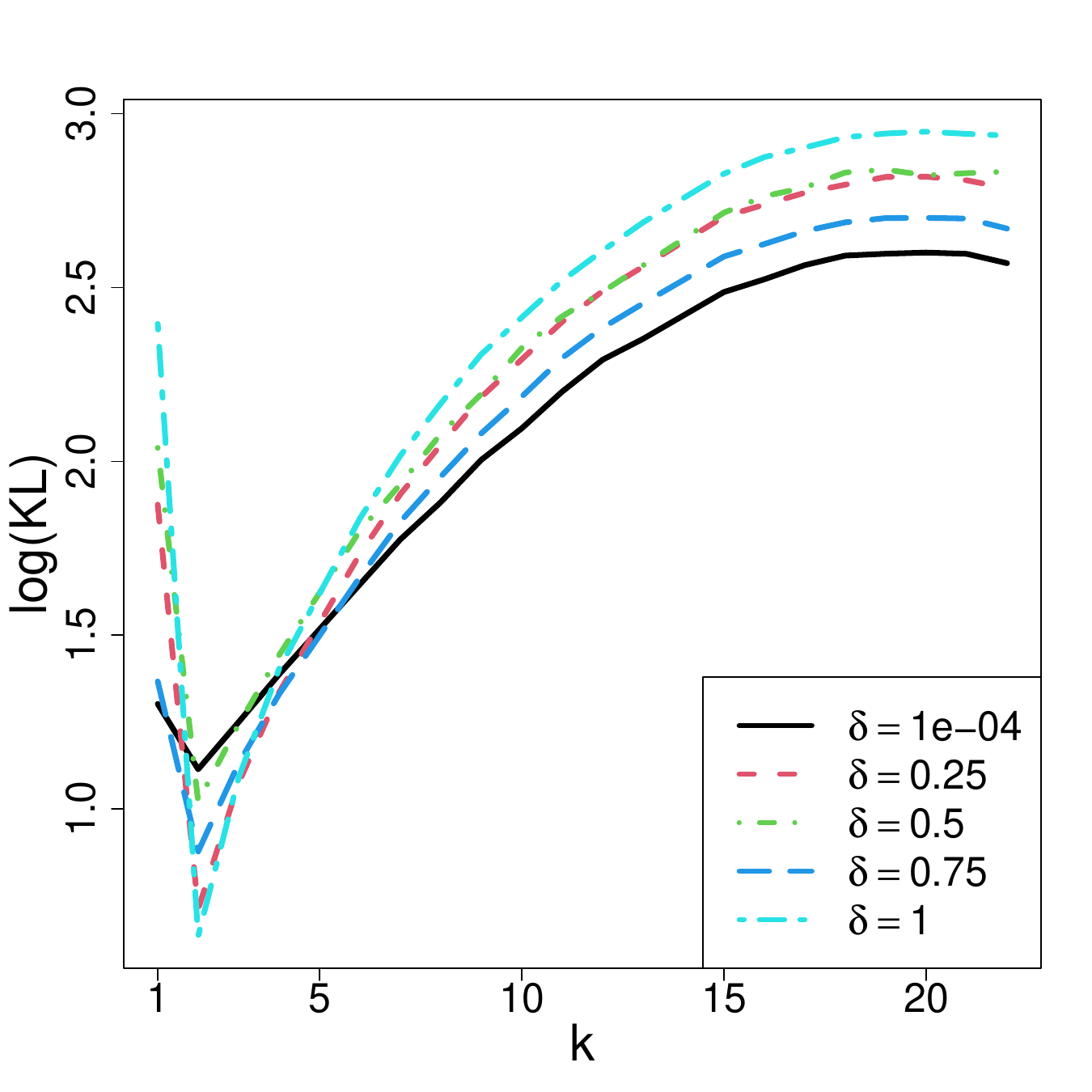}
   &\includegraphics[width=.31\textwidth]{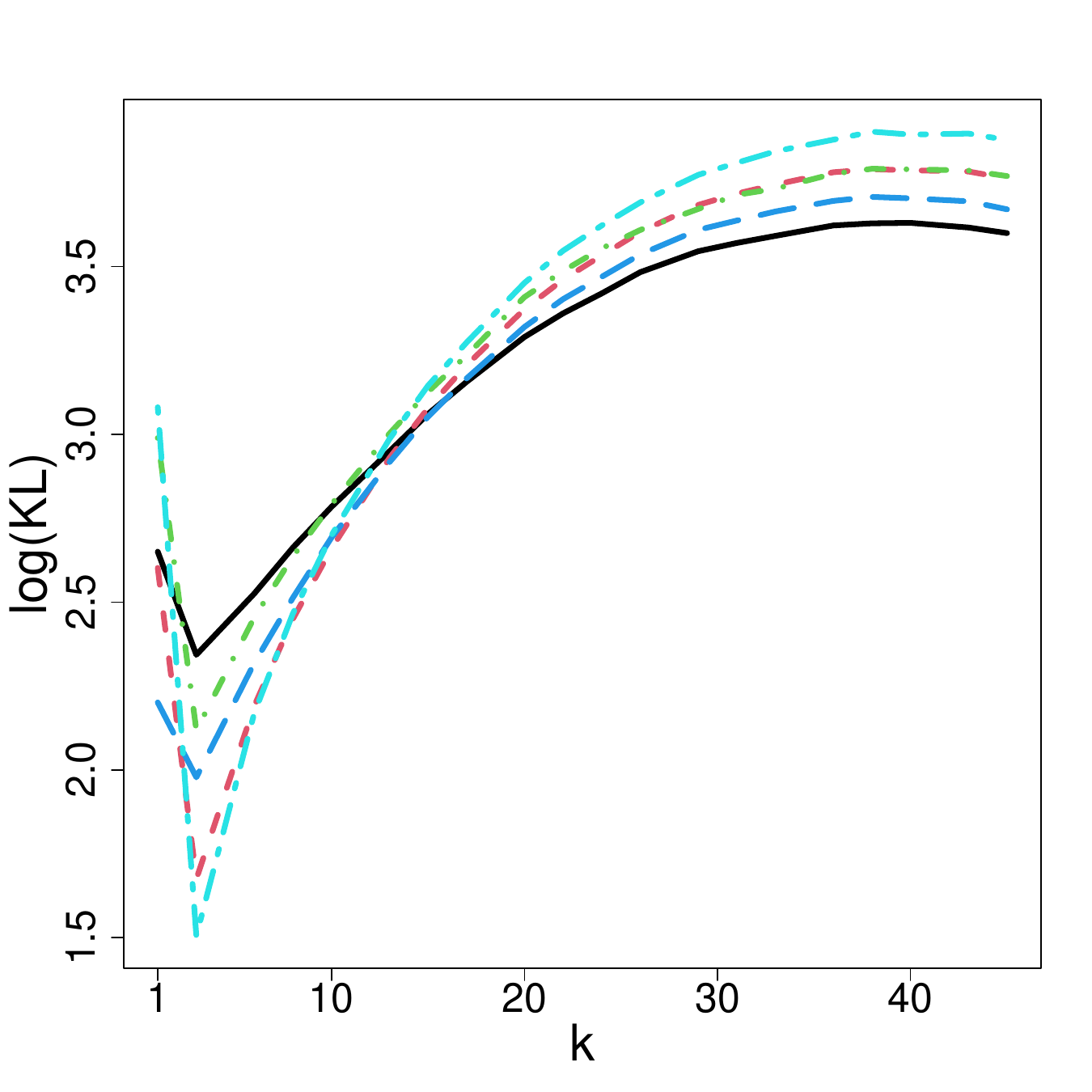}
   &\includegraphics[width=.31\textwidth]{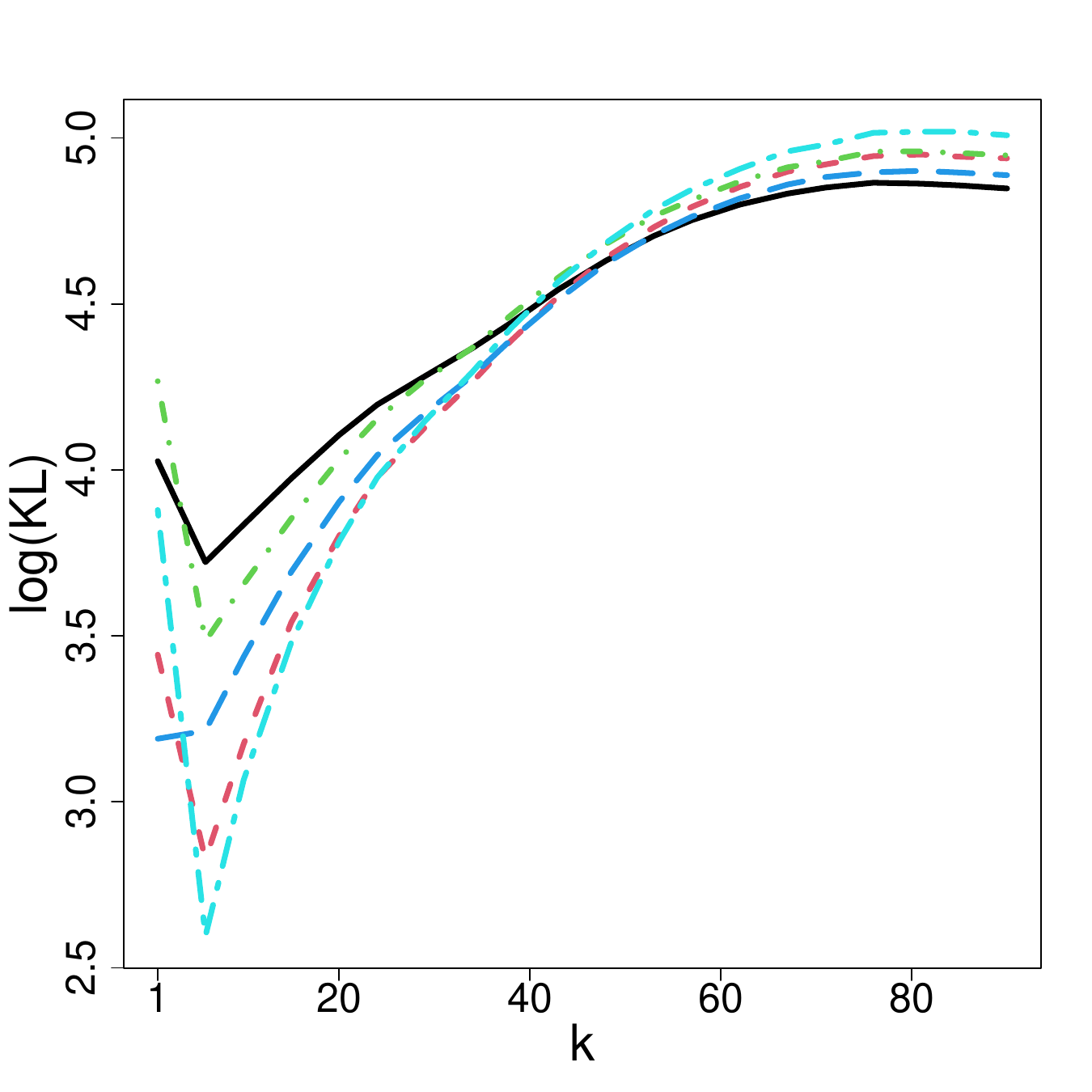} \\ [-4mm]
   \rotatebox{90}{\textbf{\footnotesize{$\gamma=6$}}}
   &\includegraphics[width=.31\textwidth]{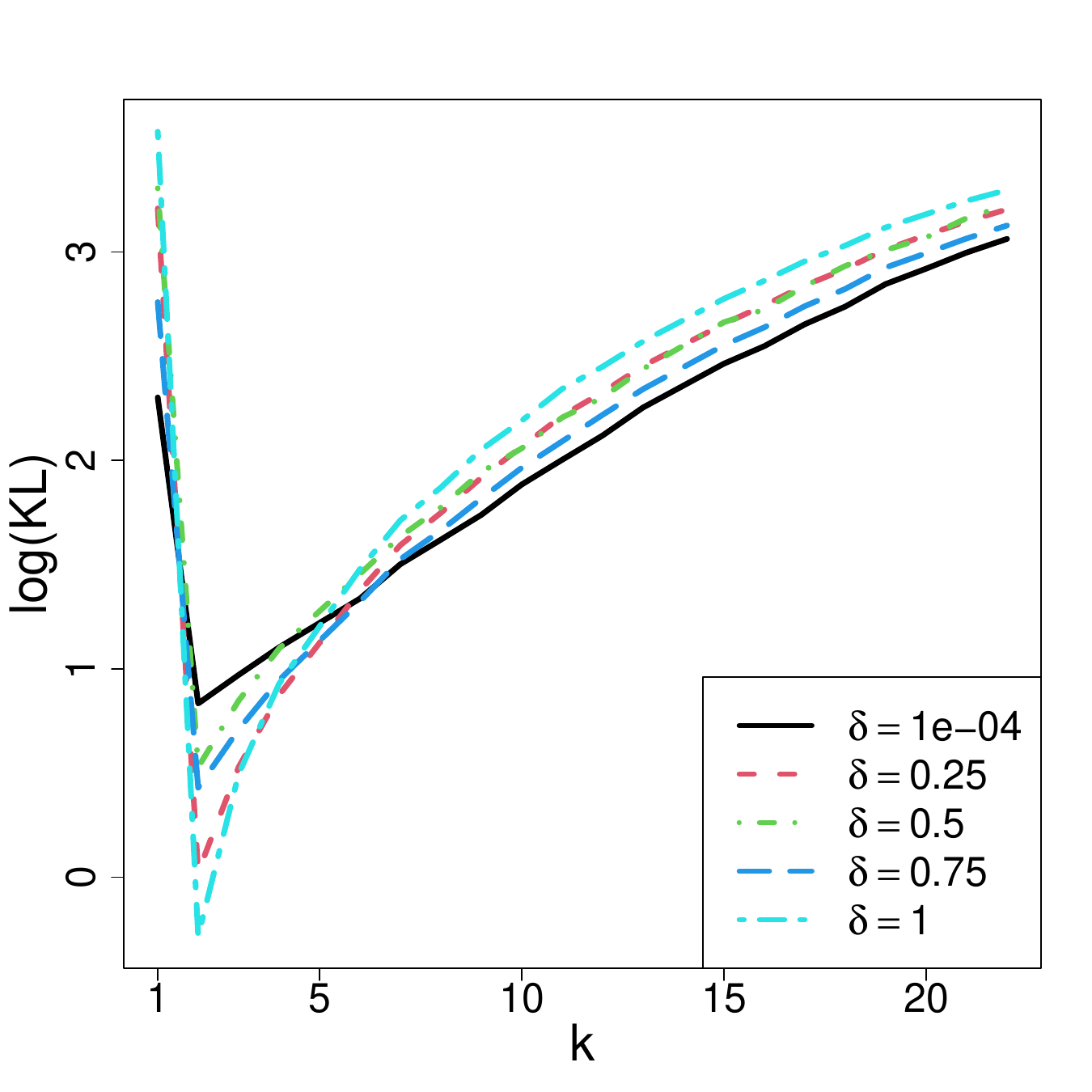}
   &\includegraphics[width=.31\textwidth]{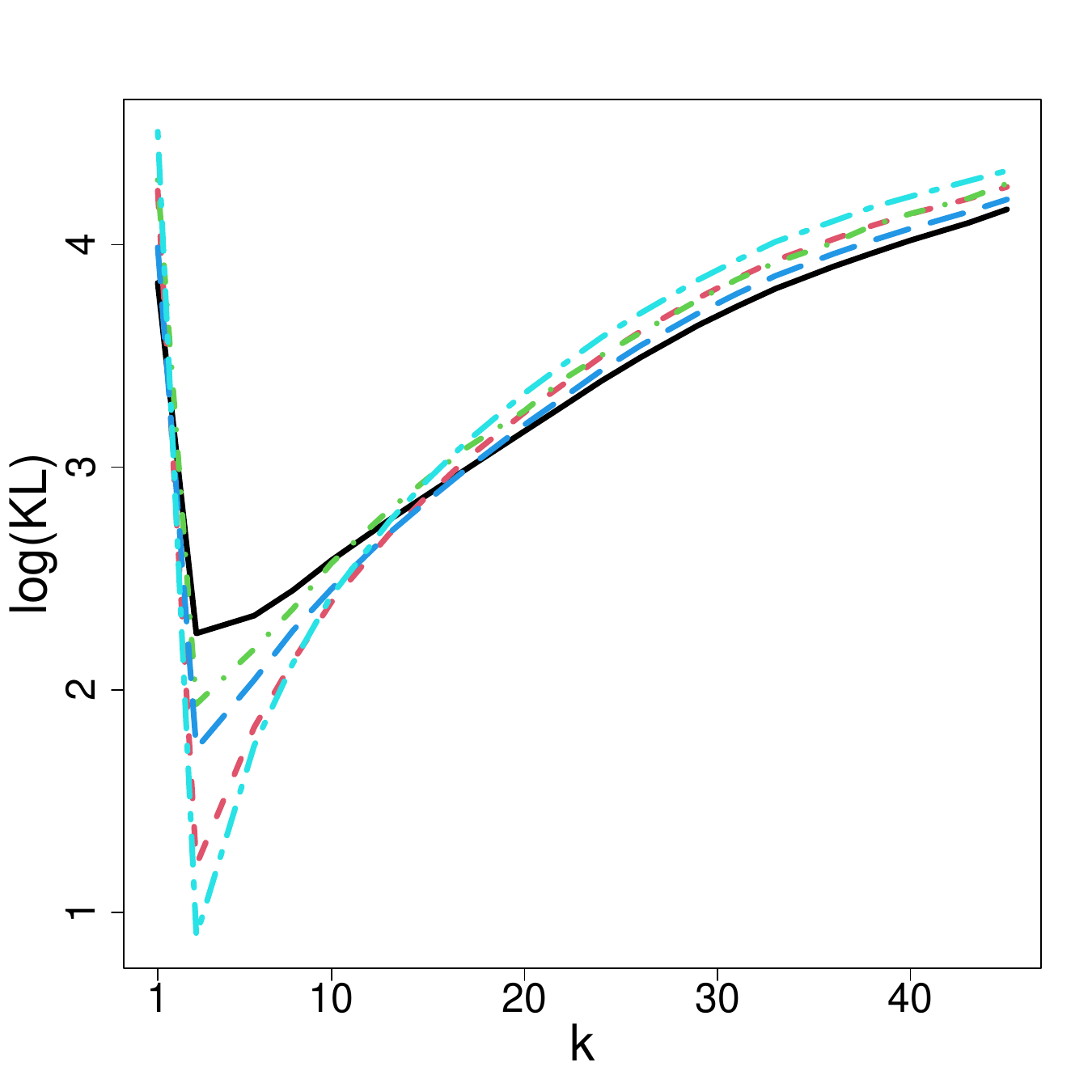}
   &\includegraphics[width=.31\textwidth]{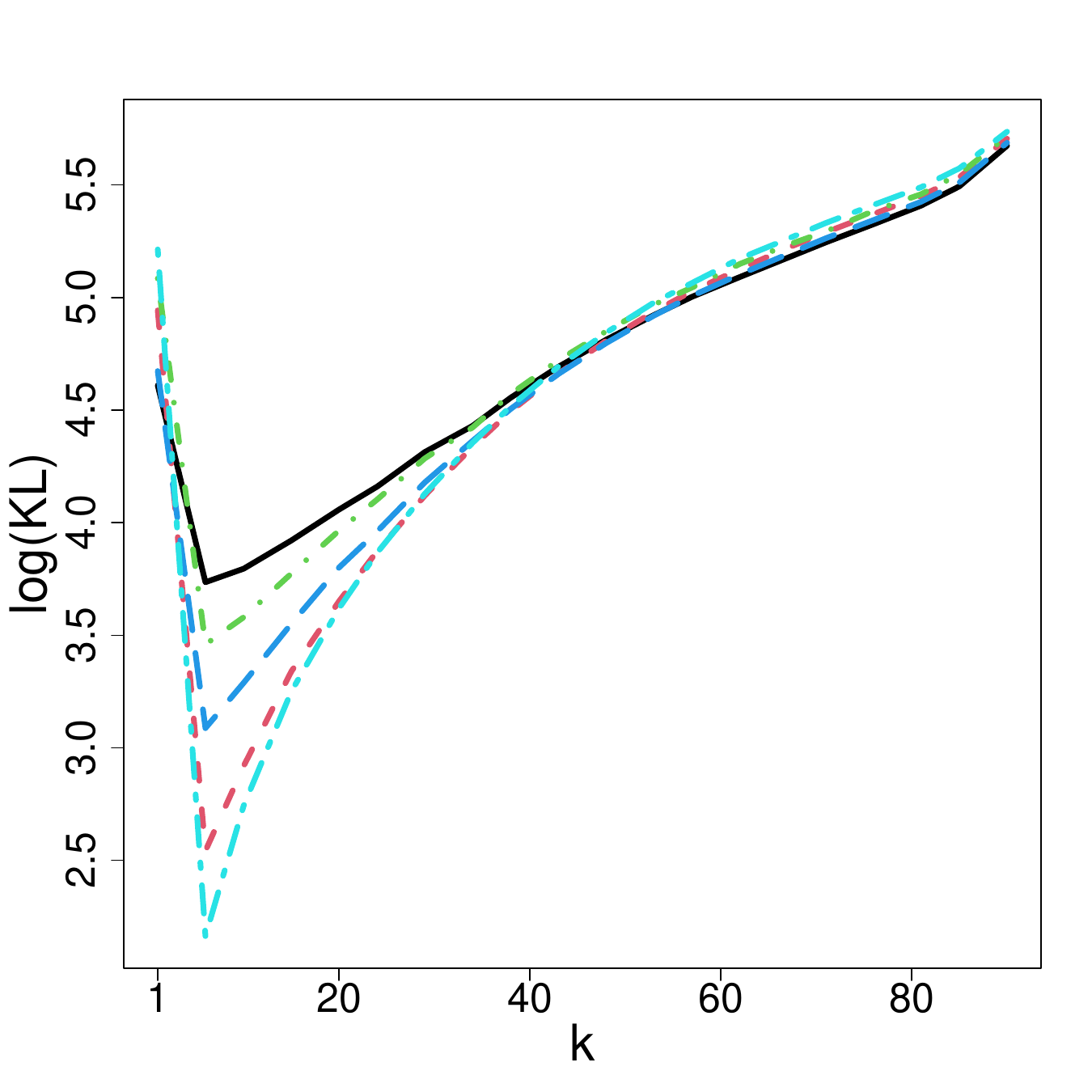}
\end{tabular}
\caption{Average log(KL) attained by cellRCov as a function of $\rk$ and $\delta$ for the planar covariance
model for to uncontaminated ($\gamma=0$) and contaminated data
($\gamma=6$) in dimensions 
$p$ in $\lbrace30,60,120\rbrace$.}
\label{fig:sens_planar}
\end{figure}

\begin{figure}[H]
\centering
\begin{tabular}{M{0.0005\textwidth}M{0.31\textwidth}M{0.31\textwidth}M{0.31\textwidth}}
   & \large \textbf{$p=30$}  & \large \textbf{$p=60$} & \large{\textbf{$p=120$}} \\ [-4mm]
   \rotatebox{90}{\textbf{\footnotesize{$\gamma=0$}}}
   &\includegraphics[width=.31\textwidth]{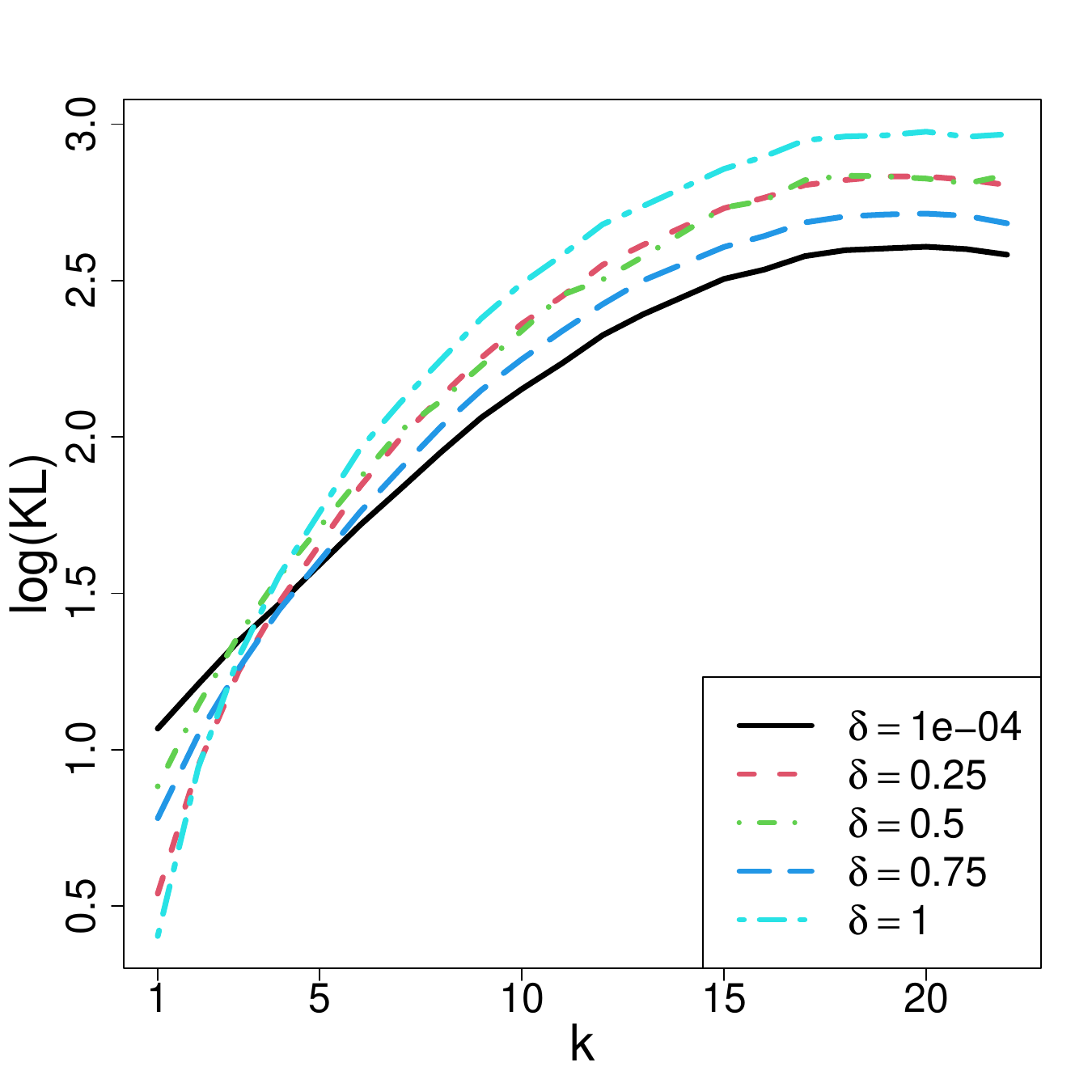}
   &\includegraphics[width=.31\textwidth]{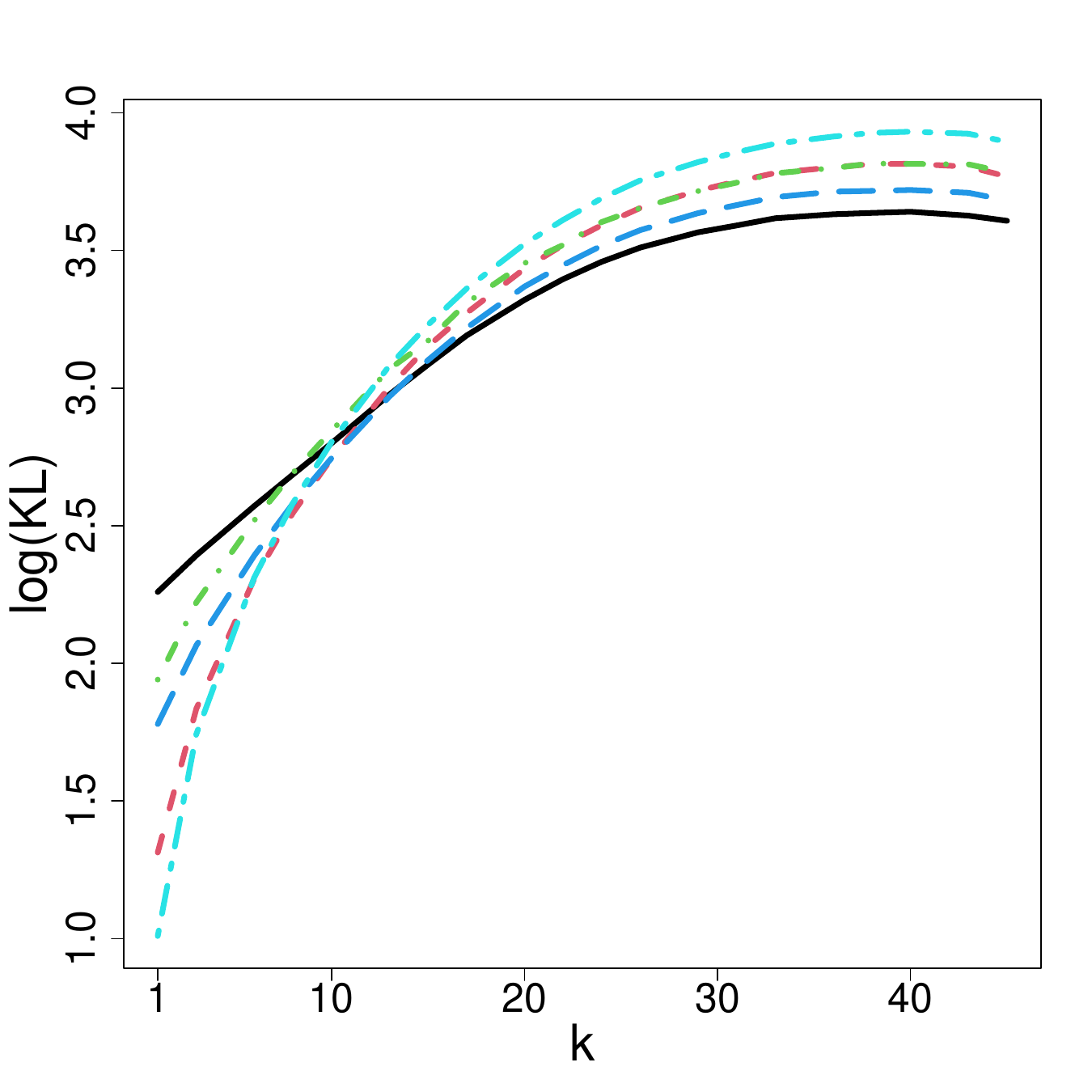}
   &\includegraphics[width=.31\textwidth]{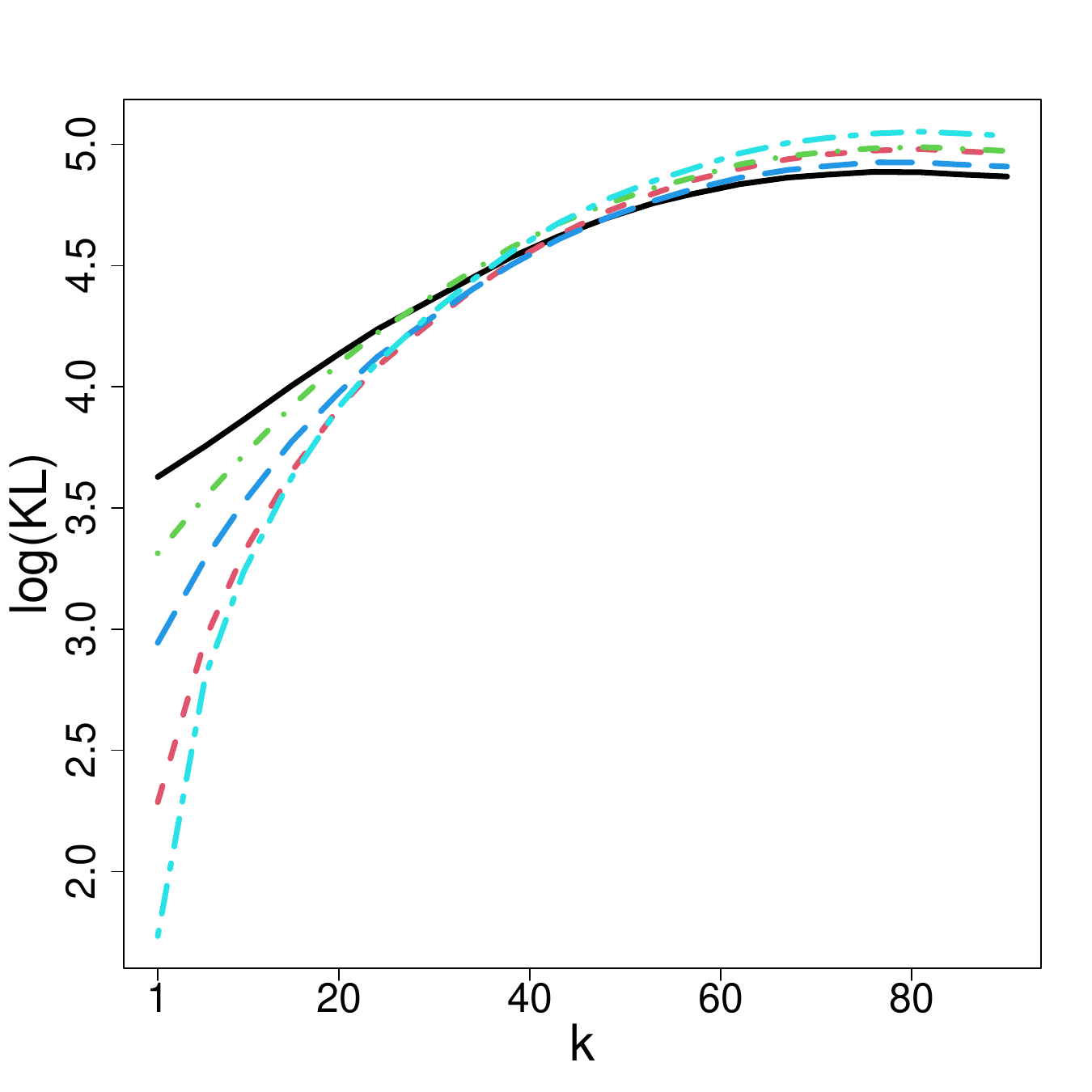} \\ [-4mm]
   \rotatebox{90}{\textbf{\footnotesize{$\gamma=6$}}}
   &\includegraphics[width=.31\textwidth]{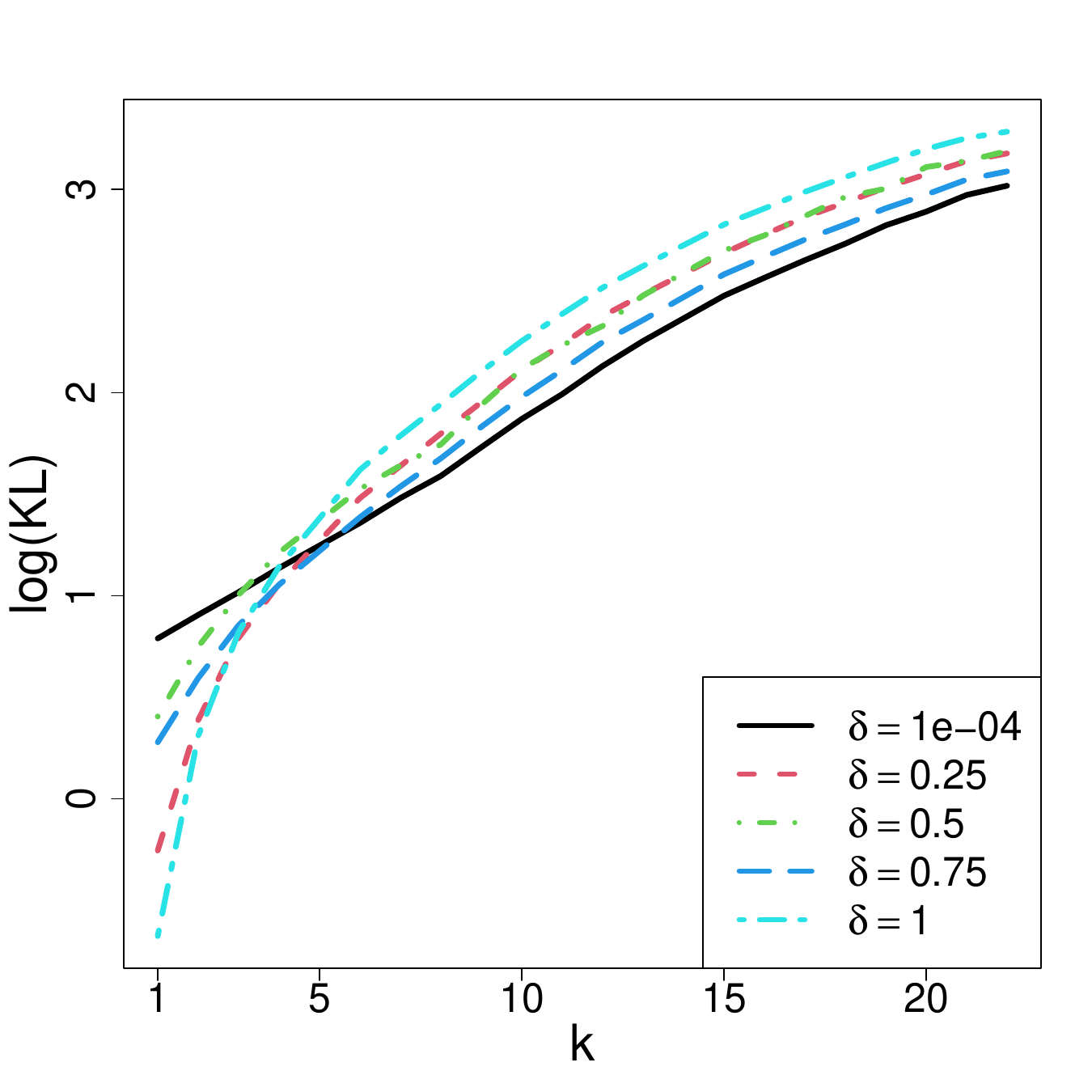}
   &\includegraphics[width=.31\textwidth]{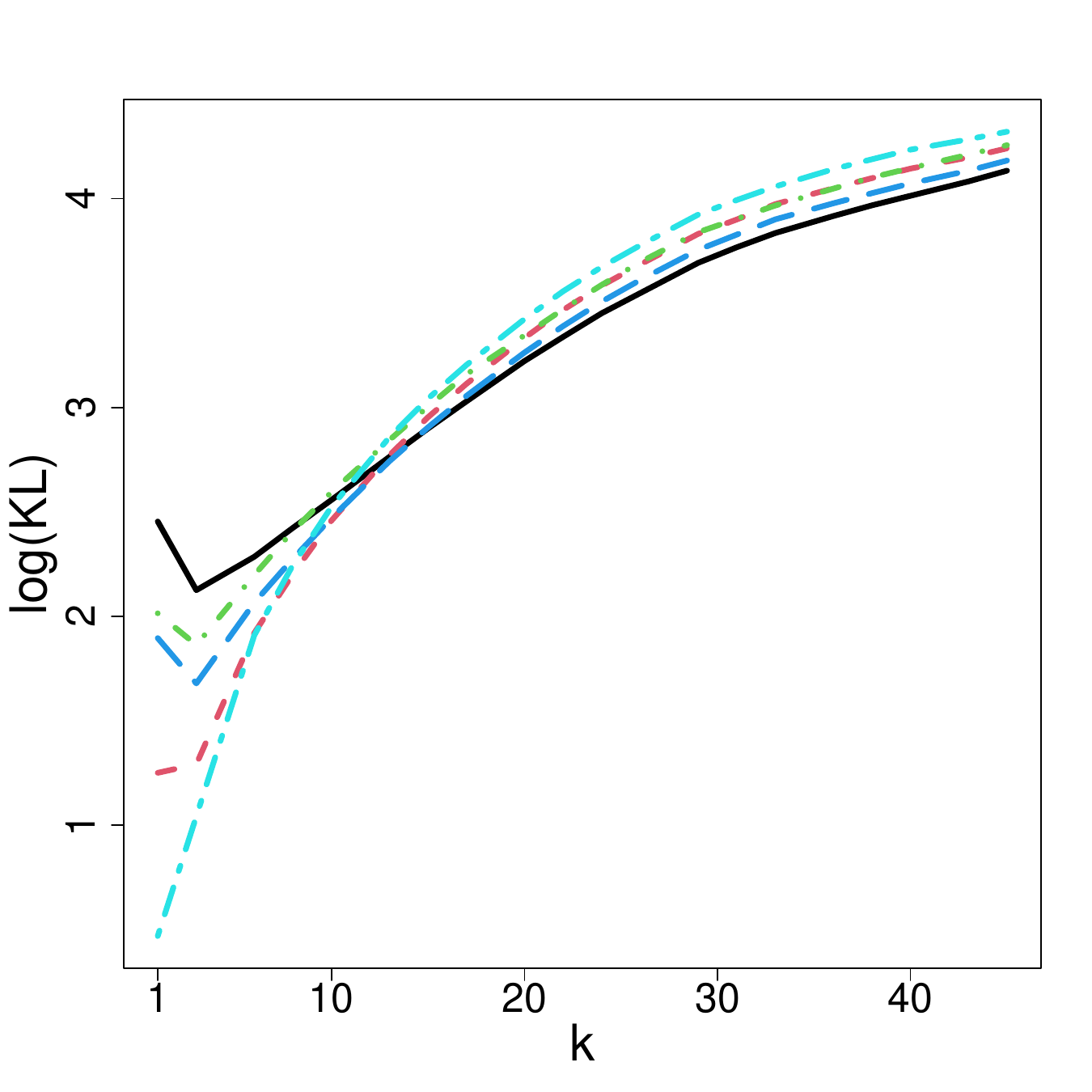}
   &\includegraphics[width=.31\textwidth]{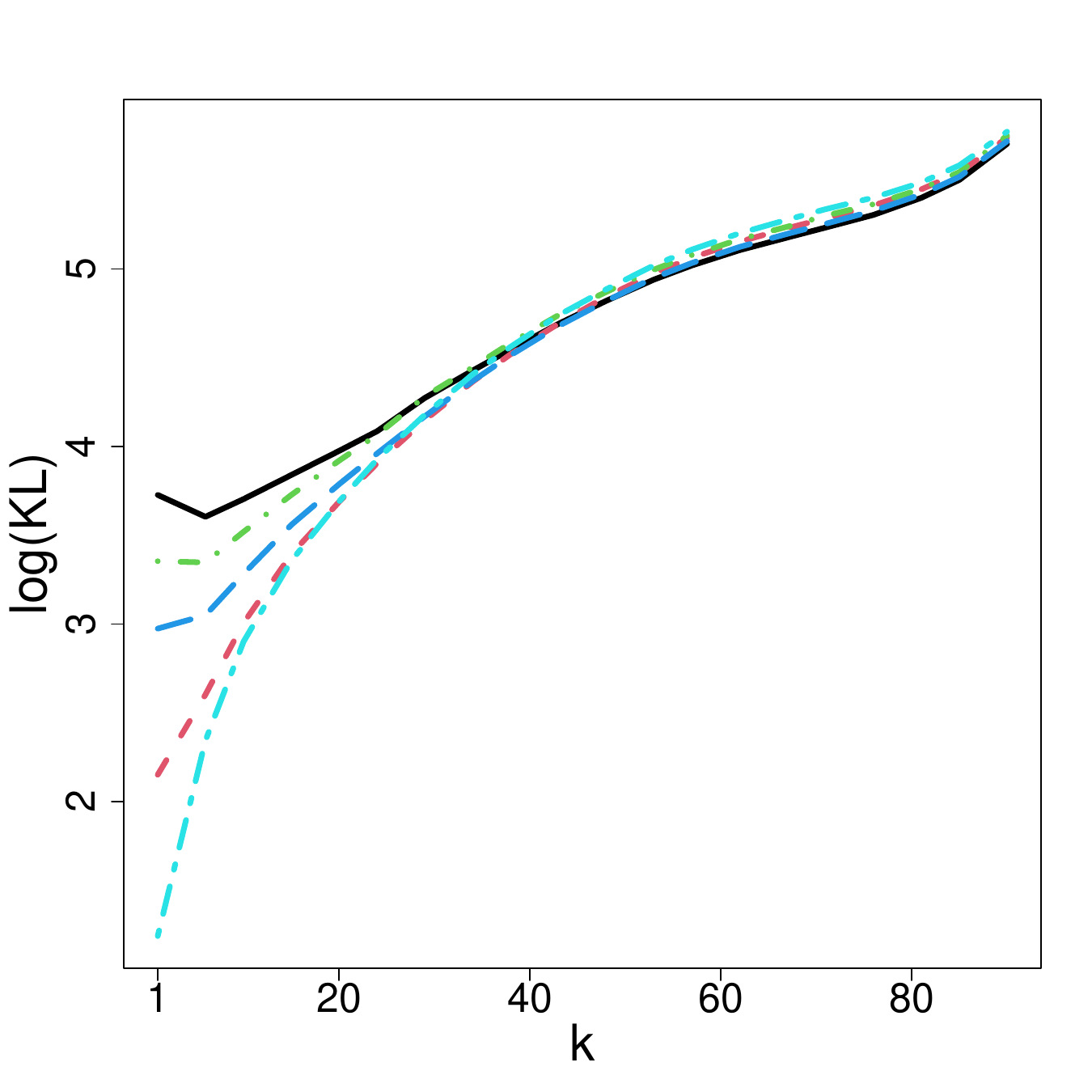}
\end{tabular}
\caption{Average log(KL) attained by cellRCov as a function of $\rk$ and $\delta$ for the dense covariance
model for to uncontaminated ($\gamma=0$) and contaminated data
($\gamma=6$) in dimensions 
$p$ in $\lbrace30,60,120\rbrace$.}
\label{fig:sens_dense}
\end{figure}

\section{\large Computation times}
\label{app:times}

We also assessed the computational feasibility 
of cellRCov by recording the running times in 
the simulation study. The timings were 
performed for the A09 covariance model without 
contamination and without missing values, for
dimensions $p=30,60,90$. The reported times 
include the full cellRCov procedure, with the 
automatic selection of the rank $\rk$ and  
the regularization parameter $\delta$. 
The computations were run on the KU 
Leuven/UHasselt Tier-2 wICE  cluster, with
IceLake thin nodes. Each node contains two 
Intel Xeon Platinum 8360Y CPUs at 2.4 GHz, 
with 36 cores per CPU and 256 GB RAM.
Figure~\ref{fig:runtime_cellRCov} shows the 
boxplots of the running times in seconds, 
averaged over the Monte Carlo replications. 

As expected, the computation time increases 
with the dimension $p$. This is consistent 
with the complexity analysis in 
Section~\ref{sec:supp-complexity-cellRCov},
which shows that the cost of cellRCov
contains the term $O(np^2)$ due to the 
computation of the full residual covariance 
matrix. Nevertheless, the procedure remains
computationally feasible for the dimensions
considered, even when the tuning parameters 
are selected automatically.

Finally, we note that the computation speed 
of the procedure can be substantially improved 
in several ways. The selection of $\rk$ can 
be parallelized, since the cellPCA fits for 
different candidate $\rk$ do not depend on
each other. In addition, the current 
implementation of cellPCA has not yet been 
extensively optimized, and implementing the 
most computationally intensive steps in C++
is expected to reduce the computation time 
considerably.

\begin{figure}[H]
\centering
\includegraphics[width=.47\textwidth]
  {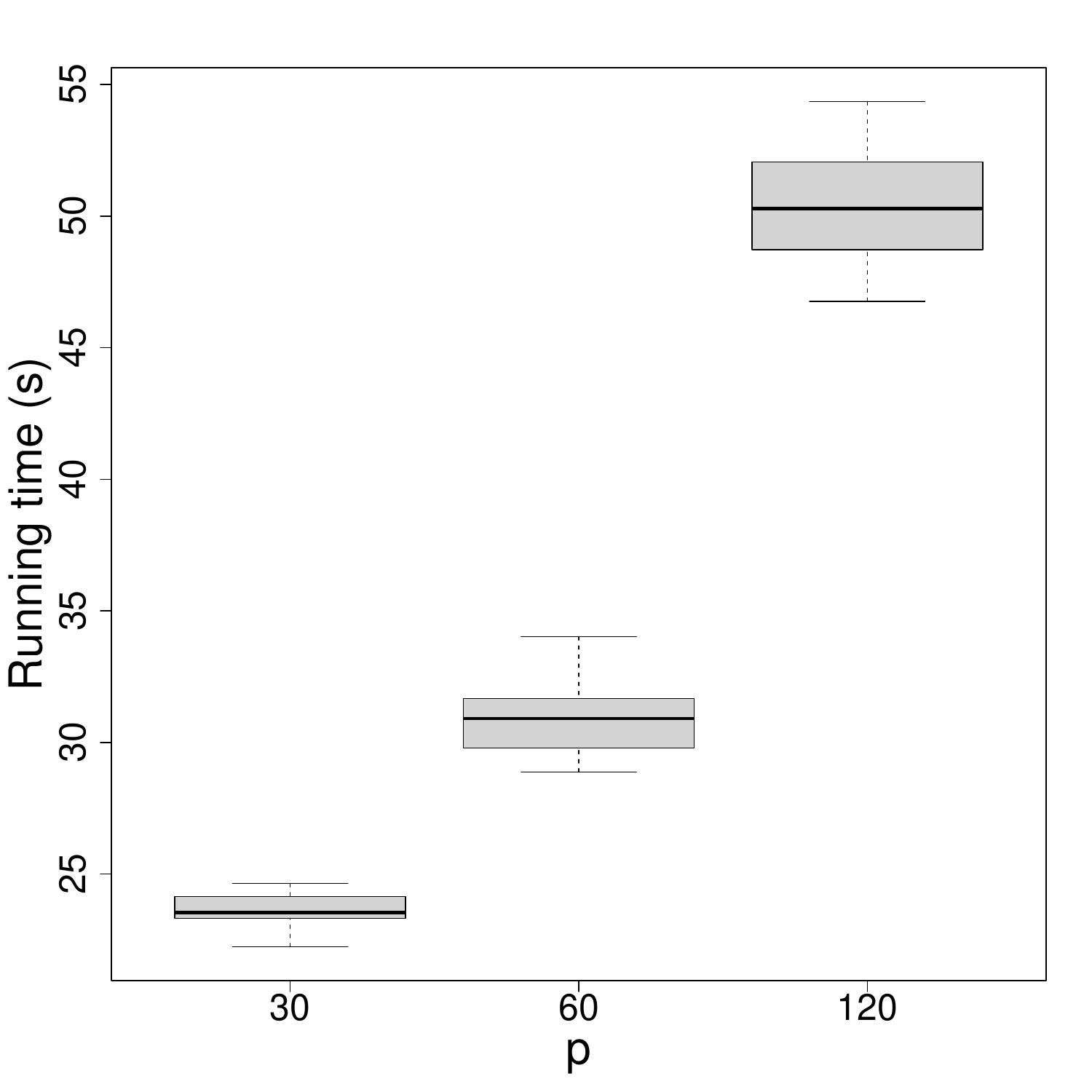}
\caption{Running time of cellRCov in 
  seconds, for $n=100$ and dimensions 
  $p = 30,60,120$\,.}
\label{fig:runtime_cellRCov}
\end{figure}

\renewcommand{\refname}{Additional Reference}


\begin{thebibliography}{}

\bibitem[\protect\citeauthoryear{Aerts and Wilms}{Aerts and Wilms}{2017}]{aerts2017cellwise}
Aerts, S. and I.~Wilms (2017).
\newblock Cellwise robust regularized discriminant analysis.
\newblock {\em Statistical Analysis and Data Mining\/}~{\em 10\/}(6), 436--447.

\bibitem[\protect\citeauthoryear{Agostinelli, Leung, Yohai, and Zamar}{Agostinelli et~al.}{2015}]{agostinelli2015robust}
Agostinelli, C., A.~Leung, V.~J. Yohai, and R.~H. Zamar (2015).
\newblock Robust estimation of multivariate location and scatter in the presence of cellwise and casewise contamination.
\newblock {\em Test\/}~{\em 24}, 441--461.

\bibitem[\protect\citeauthoryear{Alfons, Croux, and Filzmoser}{Alfons et~al.}{2017}]{Alfons_MaxAssoc}
Alfons, A., C.~Croux, and P.~Filzmoser (2017).
\newblock Robust maximum association estimators.
\newblock {\em Journal of the American Statistical Association\/}~{\em 112}, 436--445.

\bibitem[\protect\citeauthoryear{Alqallaf, Van~Aelst, Yohai, and Zamar}{Alqallaf et~al.}{2009}]{alqallaf2009}
Alqallaf, F., S.~Van~Aelst, V.~J. Yohai, and R.~H. Zamar (2009).
\newblock Propagation of outliers in multivariate data.
\newblock {\em The Annals of Statistics\/}~{\em 37}, 311--331.

\bibitem[\protect\citeauthoryear{Bickel and Levina}{Bickel and Levina}{2008}]{bickel2008regularized}
Bickel, P.~J. and E.~Levina (2008).
\newblock {Regularized estimation of large covariance matrices}.
\newblock {\em The Annals of Statistics\/}~{\em 36\/}(1), 199 -- 227.

\bibitem[\protect\citeauthoryear{Boudt, Rousseeuw, Vanduffel, and Verdonck}{Boudt et~al.}{2020}]{boudt2020minimum}
Boudt, K., P.~J. Rousseeuw, S.~Vanduffel, and T.~Verdonck (2020).
\newblock The minimum regularized covariance determinant estimator.
\newblock {\em Statistics and Computing\/}~{\em 30\/}(1), 113--128.

\bibitem[\protect\citeauthoryear{Branco, Croux, Filzmoser, and Oliveira}{Branco et~al.}{2005}]{branco2005robust}
Branco, J.~A., C.~Croux, P.~Filzmoser, and M.~R. Oliveira (2005).
\newblock Robust canonical correlations: A comparative study.
\newblock {\em Computational Statistics\/}~{\em 20}, 203--229.

\bibitem[\protect\citeauthoryear{Buja and Eyuboglu}{Buja and Eyuboglu}{1992}]{buja1992remarks}
Buja, A. and N.~Eyuboglu (1992).
\newblock Remarks on parallel analysis.
\newblock {\em Multivariate Behavioral Research\/}~{\em 27\/}(4), 509--540.

\bibitem[\protect\citeauthoryear{Butler, Davies, and Jhun}{Butler et~al.}{1993}]{butler1993asymptotics}
Butler, R., P.~Davies, and M.~Jhun (1993).
\newblock Asymptotics for the minimum covariance determinant estimator.
\newblock {\em The Annals of Statistics\/}~{\em 21\/}(3), 1385--1400.

\bibitem[\protect\citeauthoryear{Cator and Lopuha{\"a}}{Cator and Lopuha{\"a}}{2012}]{cator2012central}
Cator, E.~A. and H.~P. Lopuha{\"a} (2012).
\newblock Central limit theorem and influence function for the {MCD} estimators at general multivariate distributions.
\newblock {\em Bernoulli\/}~{\em 18}, 520--551.

\bibitem[\protect\citeauthoryear{Centofanti, Hubert, Palumbo, and Rousseeuw}{Centofanti et~al.}{2025}]{centofanti2025multivariate}
Centofanti, F., M.~Hubert, B.~Palumbo, and P.~J. Rousseeuw (2025).
\newblock Multivariate singular spectrum analysis by robust diagonalwise low-rank approximation.
\newblock {\em Journal of Computational and Graphical Statistics\/}~{\em 34\/}(1), 360--373.

\bibitem[\protect\citeauthoryear{Centofanti, Hubert, and Rousseeuw}{Centofanti et~al.}{2026}]{centofanti2026robust}
Centofanti, F., M.~Hubert, and P.~J. Rousseeuw (2026).
\newblock Robust principal components by casewise and cellwise weighting.
\newblock {\em Technometrics\/}~{\em (in press)}, 1--25, \url{https://doi.org/10.1080/00401706.2026.2643216}.

\bibitem[\protect\citeauthoryear{Clarke}{Clarke}{2018}]{clarke2018robustness}
Clarke, B.~R. (2018).
\newblock {\em Robustness Theory and Application}.
\newblock John Wiley \& Sons.

\bibitem[\protect\citeauthoryear{Croux and Öllerer}{Croux and Öllerer}{2016}]{croux2016robust}
Croux, C. and V.~Öllerer (2016).
\newblock Robust and sparse estimation of the inverse covariance matrix using rank correlation measures.
\newblock In C.~Agostinelli, A.~Basu, P.~Filzmoser, and D.~Mukherjee (Eds.), {\em Recent Advances in Robust Statistics: Theory and Applications}, pp.\  35--55. New Delhi: Springer India.

\bibitem[\protect\citeauthoryear{Engel, Buydens, and Blanchet}{Engel et~al.}{2017}]{engel2017overview}
Engel, J., L.~Buydens, and L.~Blanchet (2017).
\newblock An overview of large-dimensional covariance and precision matrix estimators with applications in chemometrics.
\newblock {\em Journal of Chemometrics\/}~{\em 31\/}(4), e2880.

\bibitem[\protect\citeauthoryear{Fernholz}{Fernholz}{2012}]{fernholz2012mises}
Fernholz, L.~T. (2012).
\newblock {\em Von Mises Calculus for Statistical Functionals}.
\newblock Springer.

\bibitem[\protect\citeauthoryear{Hampel, Ronchetti, and Rousseeuw}{Hampel et~al.}{1981}]{tanh1981}
Hampel, F.~R., E.~M. Ronchetti, and P.~J. Rousseeuw (1981).
\newblock {The Change-of-Variance Curve and Optimal Redescending M-Estimators}.
\newblock {\em Journal of the American Statistical Association\/}~{\em 76}, 643--648.

\bibitem[\protect\citeauthoryear{Hampel, Ronchetti, Rousseeuw, and Stahel}{Hampel et~al.}{1986}]{hampel1986}
Hampel, F.~R., E.~M. Ronchetti, P.~J. Rousseeuw, and W.~A. Stahel (1986).
\newblock {\em {Robust Statistics: the Approach based on Influence Functions}}.
\newblock Wiley.

\bibitem[\protect\citeauthoryear{Hirari, Centofanti, Hubert, and Van~Aelst}{Hirari et~al.}{2025}]{hirari2025robust}
Hirari, M., F.~Centofanti, M.~Hubert, and S.~Van~Aelst (2025).
\newblock Robust multilinear principal component analysis, {arXiv} preprint arxiv:2503.07327.

\bibitem[\protect\citeauthoryear{Horn}{Horn}{1965}]{horn1965rationale}
Horn, J.~L. (1965).
\newblock A rationale and test for the number of factors in factor analysis.
\newblock {\em Psychometrika\/}~{\em 30}, 179--185.

\bibitem[\protect\citeauthoryear{Hubert, Rousseeuw, and {Van den Bossche}}{Hubert et~al.}{2019}]{hubert2019macropca}
Hubert, M., P.~J. Rousseeuw, and W.~{Van den Bossche} (2019).
\newblock Macro{PCA}: an all-in-one {PCA} method allowing for missing values as well as cellwise and rowwise outliers.
\newblock {\em Technometrics\/}~{\em 61\/}(4), 459--473.

\bibitem[\protect\citeauthoryear{Hubert, Rousseeuw, and Verdonck}{Hubert et~al.}{2012}]{hubert2012deterministic}
Hubert, M., P.~J. Rousseeuw, and T.~Verdonck (2012).
\newblock A deterministic algorithm for robust location and scatter.
\newblock {\em Journal of Computational and Graphical Statistics\/}~{\em 21\/}(3), 618--637.

\bibitem[\protect\citeauthoryear{Johnson and Wichern}{Johnson and Wichern}{2002}]{johnson2002applied}
Johnson, R.~A. and D.~W. Wichern (2002).
\newblock {\em {Applied Multivariate Statistical Analysis}}.
\newblock Upper Saddle River, NJ: Prentice-Hall.

\bibitem[\protect\citeauthoryear{Kalivas}{Kalivas}{1997}]{kalivas1997two}
Kalivas, J.~H. (1997).
\newblock Two data sets of near infrared spectra.
\newblock {\em Chemometrics and Intelligent Laboratory Systems\/}~{\em 37\/}(2), 255--259.

\bibitem[\protect\citeauthoryear{Lam}{Lam}{2020}]{lam2020high}
Lam, C. (2020).
\newblock High-dimensional covariance matrix estimation.
\newblock {\em Wiley Interdisciplinary Reviews: Computational Statistics\/}~{\em 12\/}(2), e1485.

\bibitem[\protect\citeauthoryear{Ledoit and Wolf}{Ledoit and Wolf}{2004}]{ledoit2004well}
Ledoit, O. and M.~Wolf (2004).
\newblock A well-conditioned estimator for large-dimensional covariance matrices.
\newblock {\em Journal of Multivariate Analysis\/}~{\em 88\/}(2), 365--411.

\bibitem[\protect\citeauthoryear{Ledoit and Wolf}{Ledoit and Wolf}{2012}]{ledoit2012nonlinear}
Ledoit, O. and M.~Wolf (2012).
\newblock {Nonlinear shrinkage estimation of large-dimensional covariance matrices}.
\newblock {\em The Annals of Statistics\/}~{\em 40\/}(2), 1024 -- 1060.

\bibitem[\protect\citeauthoryear{Leurgans, Moyeed, and Silverman}{Leurgans et~al.}{1993}]{leurgans1993canonical}
Leurgans, S.~E., R.~A. Moyeed, and B.~W. Silverman (1993).
\newblock Canonical correlation analysis when the data are curves.
\newblock {\em Journal of the Royal Statistical Society Series B: Statistical Methodology\/}~{\em 55\/}(3), 725--740.

\bibitem[\protect\citeauthoryear{Maronna, Martin, Yohai, and Salibi{\'a}n-Barrera}{Maronna et~al.}{2019}]{maronna2019robust}
Maronna, R.~A., R.~D. Martin, V.~J. Yohai, and M.~Salibi{\'a}n-Barrera (2019).
\newblock {\em Robust Statistics: {T}heory and Methods (with R)}.
\newblock John Wiley \& Sons.

\bibitem[\protect\citeauthoryear{Mart{\'\i}n, Pereda, Santos, and Gal{\'a}n}{Mart{\'\i}n et~al.}{2014}]{martin2014assessment}
Mart{\'\i}n, {\'O}., M.~Pereda, J.~I. Santos, and J.~M. Gal{\'a}n (2014).
\newblock Assessment of resistance spot welding quality based on ultrasonic testing and tree-based techniques.
\newblock {\em Journal of Materials Processing Technology\/}~{\em 214\/}(11), 2478--2487.

\bibitem[\protect\citeauthoryear{Mikno, Pilarczyk, Korzeniowski, Kustro{\'n}, and Ambroziak}{Mikno et~al.}{2018}]{mikno2018analysis}
Mikno, Z., A.~Pilarczyk, M.~Korzeniowski, P.~Kustro{\'n}, and A.~Ambroziak (2018).
\newblock Analysis of resistance welding processes and expulsion of liquid metal from the weld nugget.
\newblock {\em Archives of Civil and Mechanical Engineering\/}~{\em 18}, 522--531.

\bibitem[\protect\citeauthoryear{Portnoy and He}{Portnoy and He}{2000}]{portnoyhe2000}
Portnoy, S. and X.~He (2000).
\newblock {A Robust Journey in the New Millennium}.
\newblock {\em Journal of the American Statistical Association\/}~{\em 95\/}(452), 1331--1335.

\bibitem[\protect\citeauthoryear{Pourahmadi}{Pourahmadi}{2013}]{pourahmadi2013high}
Pourahmadi, M. (2013).
\newblock {\em High-dimensional Covariance Estimation}.
\newblock John Wiley \& Sons.

\bibitem[\protect\citeauthoryear{Raymaekers and Rousseeuw}{Raymaekers and Rousseeuw}{2024}]{raymaekers2024cellMCD}
Raymaekers, J. and P.~J. Rousseeuw (2024).
\newblock The cellwise minimum covariance determinant estimator.
\newblock {\em Journal of the American Statistical Association\/}~{\em 119\/}(548), 2610--2621, \url{https://doi.org/10.1080/01621459.2023.2267777}.

\bibitem[\protect\citeauthoryear{Raymaekers and Rousseeuw}{Raymaekers and Rousseeuw}{2026}]{raymaekersChallenges}
Raymaekers, J. and P.~J. Rousseeuw (2026).
\newblock Challenges of cellwise outliers.
\newblock {\em Econometrics and Statistics\/}~{\em 38}, 6--25, \url{https://doi.org/10.1016/j.ecosta.2024.02.002}.

\bibitem[\protect\citeauthoryear{Rothman, Levina, and Zhu}{Rothman et~al.}{2009}]{rothman2009generalized}
Rothman, A.~J., E.~Levina, and J.~Zhu (2009).
\newblock Generalized thresholding of large covariance matrices.
\newblock {\em Journal of the American Statistical Association\/}~{\em 104\/}(485), 177--186.

\bibitem[\protect\citeauthoryear{Rousseeuw}{Rousseeuw}{1984}]{rousseeuw1984least}
Rousseeuw, P.~J. (1984).
\newblock Least median of squares regression.
\newblock {\em Journal of the American Statistical Association\/}~{\em 79\/}(388), 871--880.

\bibitem[\protect\citeauthoryear{Rousseeuw and {Van den Bossche}}{Rousseeuw and {Van den Bossche}}{2018}]{DDC2018}
Rousseeuw, P.~J. and W.~{Van den Bossche} (2018).
\newblock Detecting deviating data cells.
\newblock {\em Technometrics\/}~{\em 60\/}(2), 135--145, \url{https://doi.org/10.1080/00401706.2017.1340909}.

\bibitem[\protect\citeauthoryear{Sch{\"a}fer and Strimmer}{Sch{\"a}fer and Strimmer}{2005}]{schafer2005shrinkage}
Sch{\"a}fer, J. and K.~Strimmer (2005).
\newblock A shrinkage approach to large-scale covariance matrix estimation and implications for functional genomics.
\newblock {\em Statistical Applications in Genetics and Molecular Biology\/}~{\em 4\/}(Article 32), 1--30.

\bibitem[\protect\citeauthoryear{Sundaram}{Sundaram}{1996}]{sundaram1996first}
Sundaram, R.~K. (1996).
\newblock {\em A First Course in Optimization Theory}.
\newblock Cambridge University Press.

\bibitem[\protect\citeauthoryear{Tarr, M{\"u}ller, and Weber}{Tarr et~al.}{2016}]{tarr2016robust}
Tarr, G., S.~M{\"u}ller, and N.~C. Weber (2016).
\newblock Robust estimation of precision matrices under cellwise contamination.
\newblock {\em Computational Statistics \& Data Analysis\/}~{\em 93}, 404--420.

\bibitem[\protect\citeauthoryear{Van~Aelst, Vandervieren, and Willems}{Van~Aelst et~al.}{2011}]{van2011stahel}
Van~Aelst, S., E.~Vandervieren, and G.~Willems (2011).
\newblock {Stahel-Donoho estimators with cellwise weights}.
\newblock {\em Journal of Statistical Computation and Simulation\/}~{\em 81\/}(1), 1--27.

\bibitem[\protect\citeauthoryear{{von Mises}}{{von Mises}}{1947}]{mises1947asymptotic}
{von Mises}, R. (1947).
\newblock On the asymptotic distribution of differentiable statistical functions.
\newblock {\em The Annals of Mathematical Statistics\/}~{\em 18\/}(3), 309--348.

\bibitem[\protect\citeauthoryear{Warton}{Warton}{2008}]{warton2008penalized}
Warton, D.~I. (2008).
\newblock Penalized normal likelihood and ridge regularization of correlation and covariance matrices.
\newblock {\em Journal of the American Statistical Association\/}~{\em 103\/}(481), 340--349.

\bibitem[\protect\citeauthoryear{Wilms and Croux}{Wilms and Croux}{2015}]{WilmsCroux2015}
Wilms, I. and C.~Croux (2015).
\newblock Sparse canonical correlation analysis from a predictive point of view.
\newblock {\em SSRN Electronic Journal\/}~{\em 57}, 1--24.

\bibitem[\protect\citeauthoryear{Zaccaria, Garc{\'\i}a-Escudero, Greselin, and Mayo-{\'I}scar}{Zaccaria et~al.}{2025}]{zaccaria2024cellwise}
Zaccaria, G., L.~A. Garc{\'\i}a-Escudero, F.~Greselin, and A.~Mayo-{\'I}scar (2025).
\newblock Cellwise outlier detection in heterogeneous populations.
\newblock {\em Technometrics\/}~{\em 67}, 643--654.

\end{thebibliography}

\begin{thebibliography}{}
\providecommand{\natexlab}[1]{#1}
\providecommand{\url}[1]{\texttt{#1}}
\expandafter\ifx\csname urlstyle\endcsname\relax
  \providecommand{\doi}[1]{doi: #1}\else
  \providecommand{\doi}{doi: \begingroup \urlstyle{rm}\Url}\fi

\bibitem[\protect\citeauthoryear{Rio Branco~de Oliveira}{Rio Branco~de Oliveira}{2012}]{Riobranco2012}
Rio Branco~de Oliveira, O. (2012).
\newblock The implicit and the inverse function theorems: easy proofs, arXiv preprint 1212.2066.
\end{thebibliography}
\end{document}